\documentclass[11pt]{article}
\usepackage{fullpage}
\usepackage{amsmath}
\usepackage{amsthm}  %for \align
\usepackage{amssymb, algorithm, algpseudocode,eqparbox,thmtools,thm-restate,enumerate,verbatim,bm,array,tabu,yfonts,graphicx,subcaption,kbordermatrix}
\usepackage[usenames,dvipsnames,svgnames,table]{xcolor}
\definecolor{darkgreen}{rgb}{0.0,0,0.9}

\usepackage[colorlinks=true,pdfpagemode=UseNone,citecolor=OliveGreen,linkcolor=BrickRed,urlcolor=BrickRed,
pagebackref]{hyperref} %,pagebackref=true]{hyperref}

\usepackage{tikz}
\usetikzlibrary{arrows,matrix,calc}

\newcommand{\Input}{\item[{\bf Input:}]}
\newcommand{\Output}{\item[{\bf Output:}]}
\renewcommand{\Return}{\item[{\bf return}]}
\def\mE{\mathbb{E}}

\renewcommand{\P}[1]{{\mathbb{P}}\left[#1\right]}

\newcommand{\I}[1]{{\mathbb{I}}\left[#1\right]}
\newcommand{\E}{{\mathbb{E}}}

\newcommand{\iEE}[2]{{\mathbb E}_{#1}[#2]}
\newcommand{\uE}[1]{\underset{#1}{\E}}
\newcommand{\EE}[2]{{\mathbb{E}}_{#1}\left[#2\right]}

\newcommand{\st}{\mbox{\rm s.t. }}
\providecommand{\norm}[1]{\left\lVert#1\right\rVert}
\newcommand{\tab}[4]{\Big[\begin{tabular}{cc} $#1$ & $#2$ \\ $#3$ & $#4$ \end{tabular}\Big]}

\def\equationautorefname~#1\null{%
  equation~(#1)\null
}

\declaretheorem[numberwithin=section]{theorem}
\declaretheorem[sibling=theorem]{lemma}
\declaretheorem[sibling=theorem]{proposition}
\declaretheorem[sibling=theorem]{claim}
\declaretheorem[sibling=theorem]{conjecture}
\declaretheorem[sibling=theorem]{corollary}

\declaretheorem[sibling=theorem]{fact}
\declaretheorem[sibling=theorem]{definition}
\declaretheorem[sibling=theorem]{example}

\def\setminus{-}
\def\br{\delta}

\newenvironment{proofof}[1]{{\medbreak\noindent \em Proof of #1.  }}{\hfill\qed\medbreak}

\def\bone{{\bf 1}}
\def\bzero{{\bf 0}}
\def\eps{{\epsilon}}
\def\expandertree{locally connected hierarchy}
\def\expandertrees{locally connected hierarchies}

\def\Expandertrees{Locally Connected Hierarchies}
\def\avoiding{avoiding}
\def\Avoiding{Avoiding}
\def\R{\mathbb{R}}
\def\Z{\mathbb{Z}}

\def\fbag{{\textup{FBag}}}
\def\bx{{\bf x}}
\def\cE{{\cal E}}
\def\cB{{\cal B}}
\def\expansion{\alpha}

\def\cC{{\cal C}}
\def\X{{\mathcal X}}
\def\cT{{\cal T}}
\def\maincp{\textup{Tree-CP}}
\def\averagecp{\textup{Average-CP}}
\def\maxcp{\textup{Max-CP}}

\def\cZ{{\cal Z}}
\def\cA{{\cal A}}
\def\cB{{\cal B}}
\def\setdeg{{\cal P}}
\def\talpha{\tilde{\alpha}}
\def\tlambda{\tilde{\lambda}}
\def\l{\ell}
\def\Pr{t_P}
\def\Ue{U^{e}}
\def\Uf{U^{f}}

\def\reff{{\cal R}\textup{eff}}

\def\cut{\mathcal{O}}
\def\deg{\partial}
\def\Xb{\mathbf{X}}
\def\Yb{\mathbf{Y}}

\def\vol{d}
\def\bad{\textup{bad}}
\def\conflictset{\textup{conflict set}}

\DeclareMathOperator{\diam}{diam}

\DeclareMathOperator{\rank}{rank}

\DeclareMathOperator{\poly}{poly}
\DeclareMathOperator{\loglog}{loglog}
\DeclareMathOperator{\polylog}{polylog}
\DeclareMathOperator{\polyloglog}{polyloglog}
\DeclareMathOperator{\bag}{Bag}
\DeclareMathOperator{\token}{token}

\DeclareMathOperator{\argmax}{argmax}
\DeclareMathOperator{\argmin}{argmin}
\DeclareMathOperator{\trace}{Tr}

\DeclareMathOperator{\border}{Bor}
\DeclareMathOperator{\interior}{Int}
\DeclareMathOperator{\shrink}{shrink}
\DeclareMathOperator{\excess}{excess}

\title{Effective-Resistance-Reducing Flows, Spectrally Thin Trees,\\ and Asymmetric TSP}
\author{Nima Anari
\thanks{Computer Science Division, UC Berkeley.
Email: \protect\url{anari@berkeley.edu}}
\and
Shayan Oveis Gharan
\thanks{Department of Computer Science and Engineering, University of Washington. This work was partly done while the author was a postdoctoral Miller fellow at UC Berkeley. Email: \protect\url{shayan@cs.washington.edu}}
}
\begin{document}

\maketitle
\begin{abstract}
We show that the integrality gap of the natural LP relaxation of the Asymmetric Traveling Salesman Problem is $\polyloglog(n)$. In other words, there is a polynomial time algorithm that approximates the {\em value} of the optimum tour within a factor of $\polyloglog(n)$, where $\polyloglog(n)$ is a bounded degree polynomial of $\loglog(n)$.
We prove this by showing that any $k$-edge-connected unweighted graph has a $\polyloglog(n)/k$-thin spanning tree.

Our main new ingredient is a procedure, albeit an exponentially sized convex program, that ``transforms'' graphs that do not admit any {\em spectrally} thin trees into those that provably have spectrally thin trees. More precisely, given a $k$-edge-connected graph $G=(V,E)$ where $k\geq 7\log(n)$,
we show that there is  a matrix $D$ that ``preserves'' the structure of all cuts of $G$ such that for a set $F\subseteq E$ that induces an $\Omega(k)$-edge-connected graph, the effective resistance of every edge in $F$ w.r.t. $D$ is at most $\polylog(k)/k$. Then, we use a recent extension of the seminal work of Marcus, Spielman, and Srivastava \cite{MSS13} by the authors \cite{AO14b} to prove the existence of a   $\polylog(k)/k$-spectrally thin tree with respect to $D$. Such a tree
is $\polylog(k)/k$-combinatorially thin with respect to $G$ as $D$ preserves the structure of cuts of~$G$.
\end{abstract}

\thispagestyle{empty}
\newpage 
\tableofcontents
\thispagestyle{empty}
\newpage
\setcounter{page}{1}

\section{Introduction}
\label{sec:introduction}
In the Asymmetric Traveling Salesman Problem (ATSP)
we are given a set $V$ of $n:=|V|$ vertices and a nonnegative cost
function $c:V\times V\to\R_+$. The goal is to find the shortest tour that visits every vertex {\em at least} once.

If the cost function is symmetric, i.e., $c(u,v) = c(v,u)$ for all $u,v\in V$, then the problem is known as the Symmetric Traveling Salesman Problem (STSP).
There is a $3/2$ approximation algorithm by Christofides \cite{Chr76} for STSP. 

There is a natural LP relaxation for ATSP proposed by Held and Karp \cite{HK70},
\begin{equation}
\begin{aligned}
\min & \sum_{u,v\in V} c(u,v) x_{u,v} &\\
\st  & \sum_{u\in S, v\notin S} x_{u,v} \geq 1 & \forall S\subseteq V,\\
& \sum_{v\in V} x_{u,v} =\sum_{v\in V} x_{v,u} = 1 & \forall u\in V,\\
& x_{u,v} \geq 0 & \forall u,v\in V.
\end{aligned}
\label{lp:tsp}
\end{equation}
It is conjectured that the integrality gap of the above LP relaxation is a constant, i.e., the optimum value of the above LP relaxation is within a constant factor of the length of the optimum ATSP tour. 
Until very recently, we had a very limited understanding of the solutions of the above LP relaxation.
To this date, the best known lower bound on the integrality gap of the above LP is 2 \cite{CGK06}.

Despite many efforts, there is no known constant factor approximation algorithm for ATSP. Recently, Asadpour, Goemans, Madry, the second author and Saberi \cite{AGMOS10}
designed an $O(\log n/\log\log n)$ approximation algorithm
for ATSP that broke the $O(\log n)$ barrier from Frieze, Galbiati, and Maffioli \cite{FGM82} and subsequent improvements \cite{Bla02, KLSS05, FS07}.
The result of \cite{AGMOS10} also upper-bounds the integrality gap of the Held-Karp LP relaxation by $O(\log n/\log\log n)$. 
Later, the second author with Saberi~\cite{OS11} and subsequently Erickson and Sidiropoulos~\cite{ES14}
designed  constant factor approximation algorithms for ATSP on planar and bounded genus graphs. 
 
\paragraph{Thin Trees.} The main ingredient of all of the above
recent developments is the construction of a ``thin'' tree. 
Let $G=(V,E)$ be an unweighted undirected $k$-edge-connected graph with $n$ vertices. Recall that $G$ is $k$-edge-connected if there are at least $k$ edges in every cut of $G$, see \autoref{subsec:graphtheory} for properties of $k$-edge-connected graphs.
We allow $G$ to have an arbitrary number of parallel edges, so we think of $E$ as a multiset of edges.
Roughly speaking, a spanning tree $T\subseteq E$ is $\alpha$-thin with respect to $G$ if it does not contain
more than $\alpha$-fraction of the edges of any cut in $G$.
\begin{definition}
A spanning tree $T\subseteq E$ is $\alpha$-thin with respect to a (unweighted) graph $G=(V,E)$, if for each set $S\subseteq V$,
$$ |T(S,\overline{S})| \leq \alpha \cdot |E(S,\overline{S})|,$$
where $T(S,\overline{S})$ and $E(S,\overline{S})$ are the set of edges of $T$ and $G$ in the cut $(S,\overline{S})$ respectively.
\end{definition}
\noindent One can analogously define $\alpha$-thin edge covers, $\alpha$-thin paths, etc. Note that thinness is a downward closed property, that is any subgraph of an $\alpha$-thin subgraph of $G$ is also $\alpha$-thin. In particular, any spanning tree of an $\alpha$-thin connected subgraph of $G$ is an $\alpha$-thin spanning tree of $G$. See \autoref{fig:thintreeexample} for two examples of thin trees.
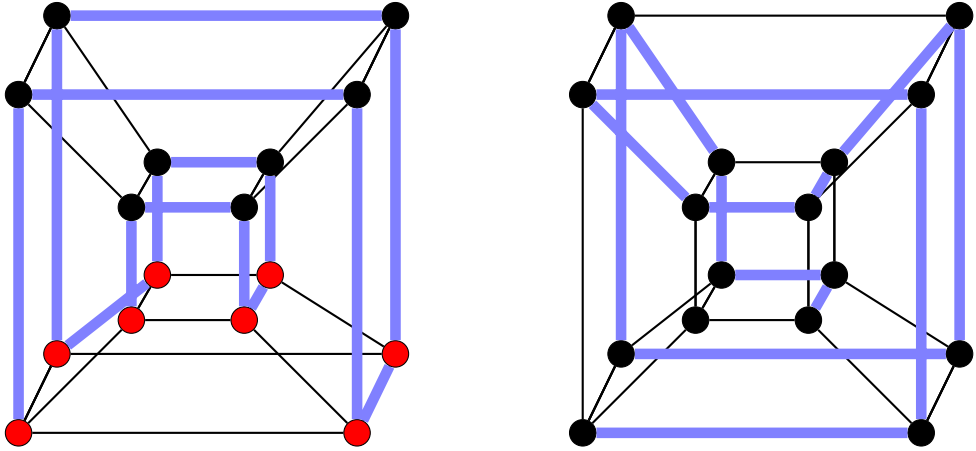
\begin{figure}\centering
\begin{tikzpicture}[scale=1.5]
	 \tikzstyle{vertex}=[draw,fill,circle,minimum size=10pt,inner sep=0pt]
	 \tikzstyle{cutv} = [vertex, fill=red]
	 \tikzstyle{ncutv} = [vertex, fill=black]
	 \tikzstyle{selected edge} = [draw,line width=4pt,-,blue!50]
	 \tikzstyle{edge} = [draw,thick,-,black]
\foreach \i/\c in {0/cutv,1/ncutv}{	 
	 \begin{scope}[shift={(\i*5,0)}]
	 \node[\c] (v0) at (0,0) {}; %{$0000$};
	 \node[vertex] (v1) at (0,1) {}; %{$0001$};
	 \node[\c] (v2) at (1,0) {}; %{$0010$};
	 \node[vertex] (v3) at (1,1) {}; %{$0011$};
	 \node[\c] (v4) at (0.23, 0.4) {}; %{$0100$};
	 \node[vertex] (v5) at (0.23,1.4) {}; %{$0101$};
	 \node[\c] (v6) at (1.23,0.4) {};% {$0110$};
	 \node[vertex] (v7) at (1.23,1.4) {}; %{$0111$};
	 \node[\c] (v8) at (-1,-1) {}; %{$1000$};
	 \node[vertex] (v9) at (-1,2) {}; %{$1001$};
	 \node[vertex] (v13) at (-0.66,2.7) {}; %{$1101$};
	 \node[\c] (v12) at (-0.66,-0.3) {}; %{$1100$};
	 \node[\c] (v10) at (2,-1) {}; %{$1010$};
	 \node[\c] (v14) at (2.34,-0.3) {}; %{$1110$};
	 \node[vertex] (v11) at (2,2) {}; %{$1011$};
	 \node[vertex] (v15) at (2.34,2.7) {}; %{$1111$};
	 \draw[edge] (v0) -- (v1) -- (v3) -- (v2) -- (v0);
	 \draw[edge] (v0) -- (v4) -- (v5) -- (v1) -- (v0);
	 \draw[edge] (v2) -- (v6) -- (v7) -- (v3) -- (v2);
	 \draw[edge] (v4) -- (v6) -- (v7) -- (v5) -- (v4);
	 \draw[edge] (v8) -- (v9) -- (v13) -- (v12) -- (v8);
	 \draw[edge] (v0) -- (v4) -- (v12) -- (v8) -- (v0);
	 \draw[edge] (v1) -- (v9) -- (v13) -- (v5) -- (v1);
	 \draw[edge] (v2) -- (v10) -- (v14) -- (v6) -- (v2);
	 \draw[edge] (v8) -- (v10) -- (v14) -- (v12) -- (v8);
	 \draw[edge] (v3) -- (v11) -- (v15) -- (v7) -- (v3);
	 \draw[edge] (v10) -- (v11) -- (v15) -- (v14) -- (v10);
	 \draw[edge] (v9) -- (v11) -- (v15) -- (v13) -- (v9);
	 \ifthenelse {\equal{\i}{0}} 
	 	{\draw[selected edge] (v0)--(v1)
	  (v1)--(v3) (v3)--(v2) (v2)--(v6) (v6)--(v7) (v7)--(v5)
	  (v5)--(v4) (v4)--(v12) (v12) --(v13) (v13)--(v15)
	  (v15)--(v14) (v14)--(v10) (v10)--(v11) (v11)--(v9) (v9)--(v8);}
	  {\draw[selected edge] (v2) -- (v6) (v6) -- (v4) 
	  	(v4) -- (v5) (v5) -- (v13) (v13) -- (v12)
	  	 (v12) -- (v14) (v14) -- (v15) (v15) -- (v7)
	  	  (v7) -- (v3) (v3) -- (v1) (v1) -- (v9)
	  	   (v9) -- (v11) (v11) -- (v10) (v10) -- (v8);}
	  \end{scope}
	}
\end{tikzpicture}

%\subfloat[]{
%\centering
%\begin{subfigure}[b]{0.3\textwidth}
%	\includegraphics[width=5cm]{cube4-bad}
%\end{subfigure}
%\begin{subfigure}[b]{0.3\textwidth}
%%\subfloat[]{
%\includegraphics[width=5cm]{cube4-good}
%%}
%\end{subfigure}
\caption{Two spanning trees of  4-dimensional hypercube that is $4$-edge-connected. Although both of the trees are Hamiltonian paths, the left spanning tree is $1$-thin because all of the edges of the cut separating red vertices from the black ones are in the tree while the right spanning tree is $0.667$-thin.}	
\label{fig:thintreeexample}
\end{figure}

A key lemma in \cite{AGMOS10} shows that one can obtain an approximation algorithm
for ATSP by finding a thin tree of small cost with respect to the graph defined by the fractional solution
of the LP relaxation. In addition, proving the existence of a thin tree provides a bound
on the integrality gap of the Held-Karp LP relaxation for ATSP.

Later, in \cite{OS11}  this connection is made more concrete. Namely, to break the $\Theta(\frac{\log(n)}{\log\log(n)})$ barrier, it suffices to ignore the costs of the edges and construct a thin tree in every $k$-edge-connected graph for $k=\Theta(\log(n))$.
\begin{restatable}{theorem}{atspthinnessthm}
\label{thm:agmos}
For any $\alpha>0$ (which can be a function of $n$), and $k\geq \log n$, a polynomial-time construction of  an $\alpha/k$-thin tree 
in any $k$-edge-connected graph gives an $O(\alpha)$-approximation algorithm for ATSP. In addition, even an existential proof gives an $O(\alpha)$ upper bound on the integrality gap of the LP relaxation. 
\end{restatable}
\noindent See \autoref{app:missingintroduction} for the proof of the above theorem.
%a construction of an $O(f(n)/k)$ thin tree in any $k$-edge connected graph gives an $O(f(n))$ approximation algorithm for ATSP.  In addition, if $O(f(n)/k)$ thin trees exists, then the optimum solution of LP \eqref{lp:tsp} estimates the value of the optimum of ATSP up to an $O(f(n))$ multiplicative factor.
The above theorem shows that to understand the solutions of LP \eqref{lp:tsp} it is enough to understand the thin tree problem in graphs with low connectivity. 

It is easy to show that any $k$-edge-connected graph has an $O(\log(n)/k)$-thin tree  \cite{GHJS09} using the independent randomized rounding method of Raghavan and Thompson~\cite{RT87}. It is enough to sample each edge of $G$ independently with probability $\Theta(\log(n)/k)$ and then choose an arbitrary spanning tree of the sampled graph. 

Asadpour et al.~\cite{AGMOS10} employ a more sophisticated randomized rounding algorithm and show that any $k$-edge-connected graph has a $\frac{\log(n)}{k\cdot \log\log(n)}$-thin tree. The basic idea of their algorithm is to use a correlated distribution, that is to sample edges almost independently while preserving the connectivity of the sampled set. More precisely, they sample a random spanning tree from a distribution where the edges are negatively correlated, so they get connectivity for free, and they only use the upper tail of the Chernoff types of bounds. The $1/\loglog(n)$ gain comes from the fact that the upper tail of the Chernoff bound is slightly stronger than the lower tail, 

\medskip 
Independently of the above applications of thin trees, Goddyn formulated the thin tree conjecture because of the close connections to several long-standing open problems regarding nowhere-zero flows.
\begin{conjecture}[Goddyn~\cite{God04}]
\label{conj:thintree}
There exists a function $f(\alpha)$ such that, for any $0<\alpha<1$, every $f(\alpha)$-edge-connected graph (of arbitrary size) has an $\alpha$-thin spanning tree.
\end{conjecture}
Goddyn's conjecture in the strongest form postulates that
for a sufficiently large $k$ that is independent of the size of $G$, every $k$-edge-connected graph has an $O(1/k)$-thin tree.
Goddyn proved that if the above conjecture holds for an arbitrary function $f(.)$, it implies a weaker version of Jaeger's conjecture on the existence of circular nowhere-zero flows \cite{jaeger:flowconj}. Very recently, Thomassen proved a weaker version of Jaeger's conjecture \cite{Tho12,LTWZ13}, but his proof has not yet shed any light on the resolution of the thin tree conjecture.

 To this date, \autoref{conj:thintree} is only proved for planar and bounded genus graphs  \cite{OS11,ES14} and edge-transitive graphs\footnote{A graph $G=(V,E)$ is edge-transitive, if for any pair of edges $e,f\in E$ there is an automorphism of $G$ that maps $e$ to $f$.} \cite{MSS13,HO14}
for $f(\alpha)=O(1/\alpha)$. 
%It is worth noting that we are not aware of any proof of even the following weaker version of the above thin tree conjecture. 
%\begin{conjecture}
%\label{conj:weakthintree}
%There exists  numbers $k_0>1$  and $0< \alpha <1$ such that any $k_0$ edge connected graph contains an $\alpha$-thin spanning tree.
%\end{conjecture}
We remark that if Goddyn's thin tree conjecture holds for an arbitrary function $f(.)$, we get an  upper bound of $O(\log^{1-\Omega(1)}(n))$ on the integrality gap  of the LP relaxation of ATSP. 

\medskip
\paragraph{Summary of our Contribution.} In this paper, we show that any $k$-edge-connected graph has a $\polyloglog(n)/k$-thin tree.
Using \autoref{thm:agmos} for $\alpha=\polyloglog(n)$ and $k=\log(n)$ this implies that the integrality gap of the LP relaxation is $\polyloglog(n)$.
Note that this does not resolve Goddyn's conjecture. 
%but it implies that the integrality gap of the LP \eqref{lp:tsp} is $\polyloglog(n)$. 
Perhaps, one of the main consequences of our work is that we can  round (not necessarily in polynomial time) the solutions of the LP relaxation exponentially better than the  randomized rounding in the worst case.

The key to our proof is to rigorously relate  the thin tree problem to a seemingly related spectral question that is known as the Kadison-Singer problem in operator theory \cite{Wea04} and then to use tools in spectral (graph) theory to solve the new problem. Until very recently, the best solution to the Kadison-Singer problem and the Weaver conjecture was based on the randomized rounding technique and matrix Chernoff bounds and incurred a loss of $\log(n)$ \cite{Rud99,AW02}. Marcus, Spielman, and Srivastava \cite{MSS13} in a breakthrough managed to resolve the conjecture using spectral techniques with no cost that is dependent on $n$. As we will elaborate in the next section, the Kadison-Singer problem can be seen as an ``$L_2$'' version of the thin tree question, or thin tree question can be seen as an $L_1$ version of the Kadison-Singer problem. So, we can summarize our contribution as an $L_1$ to $L_2$ reduction. 

We construct this $L_1$ to $L_2$ reduction using a convex program that symmetrizes the $L_2$ structure of a given graph while preserving its $L_1$ structure. More precisely, a convex program that equalizes the {\em effective resistance} of the edges while preserving the cut structure of $G$. We expect to see several other applications of this convex program in combinatorial optimization and approximation algorithms. In addition to that, we extend the result of Marcus, Spielman, and Srivastava to a larger family of distributions known as {\em strongly Rayleigh} distributions \cite{AO14b}. Strongly Rayleigh distributions are a family of probability distributions with the strongest forms of negative dependence properties \cite{BBL09}. They have been used also in a recent work of the second author, Saberi, and Singh \cite{OSS11} to improve the Christofides approximation algorithm for STSP on graph metrics. We refer the interested readers to \cite{AO14b} for more information.

\paragraph{Subsequent Work.} Subsequent to our work, Svensson \cite{Sve15} employed a sophisticated cycle cover idea and designed a constant factor approximation algorithm for ATSP when $c(.,.)$ is the shortest path metric of an unweighted graph. 
It is unclear if a combination of the ideas in this work and \cite{Sve15} can lead to constant factor approximation algorithms for general ATSP.

The rest of this section is organized as follows: In \autoref{subsec:spectralthintree} we overview the connections of the thin tree problem and graph sparsifiers and in particular the Kadison-Singer problem. Then, in \autoref{subsec:ourcontribution} we present our main theorems. Finally, in \autoref{subsec:proofoverview} we highlight the main ideas of the proof. 

\subsection{Spectrally Thin Trees}
\label{subsec:spectralthintree}
As mentioned before, thin trees are the basis for the best-known approximation algorithms
for ATSP on planar, bounded genus, or general graphs. This follows from their intuitive definition
and the fact that they eliminate the difficulty arising from the underlying asymmetry and the cost function. On the other hand, the major challenge in constructing thin trees or proving their existence
is that we are not aware of any efficient algorithm for measuring or certifying the thinness of a given tree exactly. In order to verify the thinness of a given tree, it seems that one has to look at exponentially many cuts. 

One possible way to avoid this difficulty is to study a stronger definition of thinness, namely the {\em spectral} thinness. 
First, we define some notation. For a set $S\subseteq V$ we use $\bone_S\in \R^V$ to denote the indicator (column) vector of the set $S$. For a vertex $v\in V$, we abuse notation and write $\bone_v$ instead of $\bone_{\{v\}}$. %the vector that is $1$ in the $v$ coordinate and zero otherwise. 
For any edge $e=\{u,v\}\in E$ we fix an arbitrary orientation, say $u\to v$, and we define $\X_e := \bone_u-\bone_v$. 
The Laplacian of $G$, $L_G$, is defined as follows:
$$ L_G:=\sum_{e\in E} \X_e\X_e^\intercal.$$ 
If $G$ is weighted, then we scale up each term $\X_e\X_e^\intercal$ according to the weight of the edge $e$.
Also, for a set $T\subseteq E$ of edges, we write
$$L_T:=\sum_{e\in T} \X_e\X_e^\intercal.$$ 
%to denote the Laplacian matrix of a set $T\subseteq E$ of edges. 

We say a spanning tree, $T$, is $\alpha$-spectrally thin with respect to $G$ if %for any vector $x\in \R^n$,
\begin{equation}
\label{eq:spectralthinness}
 L_T \preceq \alpha\cdot L_G, \text{ i.e., for all $x\in \R^n,$ }x^\intercal L_T x \leq \alpha\cdot x^\intercal L_G x. 
 \end{equation}
 We also say $G$ has a spectrally thin tree if it has an $\alpha$-spectrally thin tree for some $\alpha<1/2$.
Observe that  if $T$ is $\alpha$-spectrally thin, then it is also $\alpha$-(combinatorially) thin. To see that,
note that for any set $S\subseteq V$, $\bone_S^\intercal L_T \bone_S=|T(S,\overline{S})|$ and
$\bone_S^\intercal L_G \bone_S=|E(S,\overline{S})|$. 

One can verify spectral thinness of $T$ (in polynomial time) by finding the smallest $\alpha\in \R$ such that
$$ L_G^{\dagger/2} L_T L_G^{\dagger/2} \preceq \alpha\cdot I,$$
i.e., by computing the largest eigenvalue of $L_G^{\dagger/2} L_T L_G^{\dagger/2}$.
Recall that $L_G^\dagger$ is the pseudoinverse of $L_G$, and $L_G^{\dagger/2}$ is the square root of the pseudoinverse of $L_G$; $L_G^{\dagger/2}$ is well-defined because $L_G^\dagger \succeq 0$. So, unlike the combinatorial thinness, spectral thinness can be computed {\em exactly} in polynomial time. 

The notion of spectral thinness is closely related to spectral sparsifiers of graphs, which have been studied extensively in the past few years \cite{ST04,SS11,BSS14,FHHP11}. Roughly speaking, a spectrally thin tree is a one-sided spectral sparsifier. A spectrally thin tree $T$ would be a true spectral sparsifier if in addition to \eqref{eq:spectralthinness},
it satisfies $\alpha\cdot (1-\eps) x^\intercal L_G x \preceq L_T$ for some constant $\eps$. Until the recent breakthrough of Batson, Spielman, and Srivastava, all constructions of spectral sparsifiers used at least $\Omega(n\log(n))$ edges of the graph \cite{ST04,SS11,FHHP11}. Because of this they are of no use for the particular application of ATSP. Batson, Spielman, and Srivastava \cite{BSS14} managed to construct a spectral sparsifier that uses only $O(n)$ edges of $G$. But in their construction, they assign different weights to the  edges of the sparsifier which again makes their contribution not helpful for ATSP.

Indeed, it was observed by  several people that there is an underlying barrier for the construction of spectrally thin trees and {\em unweighted} spectral sparsifiers. Many families of $k$-edge-connected graphs do not admit spectrally thin trees (see \cite[Thm 4.9]{HO14}).  
Let us elaborate on this observation.
The {\em effective resistance} of an edge $e=\{u,v\}$ in $G$, $\reff_{L_G}(e)$,
is the \emph{energy} of the electrical flow that sends 1 unit of current from $u$ to $v$
%the electrical resistance measured across the endpoints of $e$ 
when the network represents an electrical circuit with each edge being a resistor of resistance 1 
(and if $G$ is weighted, the resistance is the inverse of the weight of $e$).  
See \cite[Ch. 2]{LP13} for background on electrical flows and effective resistance.
Mathematically, the effective resistance can be computed using $L_G^{\dagger}$,
$$\reff_{L_G}(e):=\X_e^\intercal L_G^\dagger \X_e.$$
It is not hard to see that the spectral thinness of any spanning tree $T$ of $G$ is at least the maximum effective resistance of the edges of $T$ in $G$. 
\begin{lemma}
\label{lem:necessaryspectralthinness}
For any graph $G=(V,E)$, the spectral thinness of any spanning tree $T\subseteq E$ is at least
$\max_{e\in T} \reff_{L_G}(e).$
\end{lemma}
\begin{proof}
Say the spectral thinness of $T$ is $\alpha$. Obviously, by the downward closedness of spectral thinness, the spectral thinness of any subset of edges of $T$ is at most $\alpha$, i.e., for any edge $e\in T$,
%By definition, for any edge $e\in T$,
$$ L_{\{e\}} \preceq L_T \preceq \alpha \cdot L_G.$$
%$$  L_G^{\dagger/2} \X_e \X_e^\intercal L_G^{\dagger/2}  = L_G^{\dagger/2}L_{\{e\}} L_G^{\dagger/2} \preceq L_G^{\dagger/2} L_T L_G^{\dagger/2} \preceq \alpha\cdot I. $$
But, the spectral thinness of an edge is indeed its effective resistance. More precisely, multiplying $L_G^{\dagger/2}$ on both sides of the above inequality we have
$$ L_G^{\dagger/2} \X_e \X_e^\intercal L_G^{\dagger/2} = L_G^{\dagger/2} L_{\{e\}} L_G^{\dagger/2} \preceq \alpha\cdot L_G^{\dagger/2} L_G L_G^{\dagger/2} \preceq \alpha\cdot I.$$
Since the matrix on the LHS has rank one, its only eigenvalue is equal to its trace; therefore,
$$ \trace(\X_e^\intercal L_G^{\dagger} \X_e) = \trace(L_G^{\dagger/2} \X_e\X_e^\intercal L_G^{\dagger/2}) \leq \alpha. $$
The lemma follows by the fact that $\reff_{L_G}(e) = \trace(\X_e^\intercal L_G^{\dagger} \X_e)$.
%for a proof of the first matrix inequality. So, for any vector $x$,
%$$ x^\intercal L_G^{\dagger/2} \X_e \X_e^\intercal L_G^{\dagger/2}x \leq \alpha\cdot x^\intercal x.$$
%Letting $x=L_G^{\dagger/2} \X_e$, we get 
%$$ \big(\X_e^\intercal L_G^{\dagger} \X_e\big)^2 \leq \alpha\cdot \X_e^\intercal L_G^{\dagger} \X_e.$$
%Therefore, $\reff_G(e)\leq \alpha$. 
\end{proof}
In light of the above lemma, a necessary condition for $G$ to have a spanning tree with spectral thinness bounded away from $1$ is that every cut of  $G$ must have at least one edge with effective resistance bounded away from 1. In other words, any graph $G$ with at least one cut where the effective resistance of every edge is very close to 1 has no spectrally thin tree (see \autoref{fig:edgecover} for an example of a graph  where the effective resistance of every edge in a cut is very close to 1). 

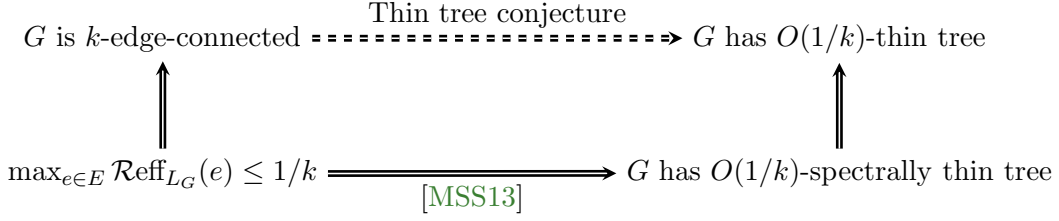
\begin{figure}
\centering
\begin{tikzpicture}
  \matrix (m) [matrix of nodes,
  row sep=3em,column sep=10em,minimum width=2em]
  {
      $G$   is  $k$-edge-connected & $G$ has  $O(1/k)$-thin tree \\
     $\max_{e\in E} \reff_{L_G}(e) \leq 1/k$  & $G$ has  $O(1/k)$-spectrally thin tree \\};
  \path[-stealth]
    (m-1-1) edge [double, dashed,line width=1] node [above] {Thin tree conjecture} (m-1-2)
    (m-2-1) edge [double, line width=1] (m-1-1)
    (m-2-1) edge [double, line width=1] node [below] {\cite{MSS13}} (m-2-2)
    (m-2-2) edge [double, line width=1] (m-1-2);
    %(m-1-1) edge [bend right=30, color=red, double, line width=1] node [left] {?} (m-2-1);
\end{tikzpicture}
\caption{A summary of the relationship between spectrally thin trees and combinatorially thin trees before our paper.} %In this paper we study if the red row.}
\label{fig:overviewspectrallythintree}
\end{figure}

In a very recent breakthrough, Marcus, Spielman, and Srivastava \cite{MSS13} proved the Kadison-Singer conjecture. As a byproduct of their result, it was shown in \cite{HO14} that a stronger version of the above condition is sufficient  for the existence of spectrally thin trees. 
\begin{theorem}[\cite{MSS13}]\label{thm:mss}
Any connected   graph $G=(V,E)$ has a spanning tree with spectral thinness
$O(\max_{e\in E} \reff_{L_G}(e))$. 
\end{theorem}
See \cite[Appendix E]{HO14} for a detailed proof of the above theorem.
It follows from the above theorem that every $k$-edge-connected edge-transitive graph has an $O(1/k)$-spectrally thin tree. This is because in any edge-transitive graph, by symmetry, the effective resistances of all edges are equal.

%We  conclude this part by
Let us summarize the relationship between spectrally thin trees and combinatorially thin trees that has been in the literature before our work. Goddyn conjectured that every $k$-edge-connected graph has an $O(1/k)$-thin tree. The result of \cite{MSS13} shows that a stronger assumption implies an stronger conclusion, i.e., if the maximum effective resistance of edges of $G$ is at most $1/k$, then $G$ has an $O(1/k)$-spectrally thin tree (see \autoref{fig:overviewspectrallythintree}).

We emphasize that $\max_{e\in E} \reff_{L_G}(e)\leq 1/k$ is a stronger assumption than $k$-edge-connectivity. %To see that first we recall an equivalent definition of effective resistance. Say we construct a resistor network by replacing each edge of $G$ with a resistor of resistance 1. The effective resistance of an edge $e=\{u,v\}$ is the {\em energy} of the electrical flow that sends 1 unit from $u$ to $v$ (see \cite[Chapter 2]{LP13} for background on electrical flows). 
If $\reff_{L_G}(u,v)\leq 1/k$,  it means that when we send one unit of flow from $u$ to $v$, the electric current divides and goes through at least $k$ parallel paths connecting $u$ to $v$, so, there are $k$ edge-disjoint paths between $u,v$. But the converse of this does not necessarily hold. If there are $k$ edge-disjoint paths from $u$ to $v$, the electric current may just use one of these paths if the rest are very long, so the effective resistance can be very close to 1. 
Therefore, if $\max_{e\in E}\reff_{L_G}(e)\leq 1/k$, there are $k$ edge-disjoint paths between each pair of vertices of $G$, and $G$ is $k$-edge-connected, but the converse does not necessarily hold.
For example in the graph in the top of \autoref{fig:edgecover}, even though there are $k$ edge-disjoint paths from $u_1$ to $v_1$, a unit electrical flow from $u_1$ to $v_1$ almost entirely goes through the edge $\{u_1,v_1\}$, so $\reff(u_1,v_1)\approx 1$. 

%In particular, for an edge $e=\{u,v\}$, $r\reff_G(e)\leq 1/k$ only if there are at least $k$-edge disjoint paths from $u$ to $v$ in $G$. Recall that $k$-edge connectivity is equivalent to having $k$-edge disjoint paths between each pair of the vertices, so $\max_{e\in E} \reff_G(e) \leq 1/k$ implies $k$-edge connectivity. The converse of this does not necessarily hold, because 

\iffalse
\medskip
We conclude this part by stating a simple extension of \cite{MSS13}. It follows from the simple extension of \cite{MSS13}, by Akemann and Weaver \cite{AW13}, that even if a constant fraction of edges in every cut of $G$ have small effective resistance, then $G$ has a spectrally thin tree.
\begin{theorem}
\label{thm:avgeffthintree}
Given a connected graph $G=(V,E)$, if
%suppose that there is a set $E'\subseteq E$ 
for any set $S\subseteq V$,
$$ \EE{e\sim E(S,\overline{S})}{\reff_{L_G}(e)} \leq \alpha,$$
%$|E'(S,\overline{S})| \geq \alpha\cdot |E(S,\overline{S})|$. %If $\max_{e\in E'} \reff_G(e)\ll \eps^4$, 
then $G$ has an $O( \sqrt[4]{\alpha})$-spectrally thin spanning tree. In the above, the expectation is over uniform distribution among all of the edges in the cut $(S,\overline{S})$.
\end{theorem}
%\noindent We prove the above theorem in \autoref{app:missingintroduction}.
We omit the proof of the above theorem as we have shown a stronger variant of it in our recent extension of \cite{MSS13} (see \autoref{thm:effrestothintree} below).
\fi

As a side remark, note that
 the sum of effective resistances of all edges of any connected graph $G$ is $n-1$, 
 $$ \sum_{e\in E} \X_e^\intercal L_G^{\dagger} \X_e = \sum_{e\in E} \trace(L_G^{\dagger/2} \X_e\X_e^\intercal L_G^{\dagger/2}) = \trace\Big(\sum_{e\in E}L_G^{\dagger/2} \X_e\X_e^\intercal L_G^{\dagger/2}\Big) =\trace(L_G^{\dagger/2} L_G L_G^{\dagger/2}) = n-1.  $$
In the last identity we use that $L_G^{\dagger/2} L_G L_G^{\dagger/2}$ is an identity matrix on the space of vectors that are orthogonal to the all-1s vector.

 If $G$ is $k$-edge-connected, by Markov's inequality, at most a quarter of the edges have effective resistance more than $8/k$. Therefore, by an application of \cite{MSS13}, any $k$-edge-connected graph $G$ has an $O(1/k)$-spectrally thin set of edges, $F\subset E$ where $|F|\geq \Omega(n)$ \cite{HO14}. Unfortunately, the corresponding subgraph $(V,F)$ may have $\Omega(n/k)$ connected components. So, this does not give any improved bounds on the approximability of ATSP.

%\medskip

%If $G$ has edges with effective resistance very close to 1
%in which the effective resistance of all edges is at most has a $1/2$-spectral thin tree. 
%For $G$ to have  a spectral thin tree, it must be that every cut has
% at least one edge with effective resistance bounded away from 1. 
%Conversely, if the effective resistance of all edges of $G$ are bounded away from 1 then it has a $1/2$-spectral thin tree. 

\subsection{Our Contribution}
\label{subsec:ourcontribution}
%Our goal in this paper is to use the recent developments on the existence of spectral thin sets of edges to find (combinatorially) thin 
In this paper we introduce a procedure to ``transform'' graphs that do not admit spectrally thin trees into those that {\em provably} have these trees. Then, we use our recent extension of \cite{MSS13} to {\em strongly Rayleigh distributions} \cite{AO14b} to find
spectrally thin trees in the transformed ``graph''. Finally, we show that any spectrally thin tree of the transformed ``graph'' is a (combinatorially) thin tree in the original graph.
%In this paper we use the seminal result of \cite{MSS13} to find (combinatorial) thin set of edges of a given graph by finding a spectral thin set of edges in a different graph. To do this, we design an algorithm to change an input graph into a new graph that admits spectral thin  trees. 
From a high level perspective, our transformation massages the graph to equalize the effective resistance of the  edges, while keeping the cut structure of the graph intact.

For two matrices $A,B\in\R^{n\times n}$, we write $A \preceq_{\square} B$,
if for any set $\emptyset\subset S\subsetneq V$,
$$ \bone_S^\intercal A\bone_S \leq \bone_S^\intercal B\bone_S.$$
Note that $A\preceq B$ implies $A\preceq_{\square} B$, but the converse is not necessarily true. 
We say a graph $D$ is a \emph{shortcut} graph with respect to $G$ if $L_D \preceq_\square L_G$. 
We say a positive definite (PD) matrix $D$ is a shortcut matrix with respect to $G$ if $D\preceq_{\square} L_G$.

Our ideal plan is as follows: 
Show that there is a (weighted) shortcut graph $D$  such that
for any edge $e\in E$, $\reff_{L_D}(e) \leq \tilde{O}(1/k)$.
%Since every edge of $G$ has effective resistance $\tilde{O}(1/k)$ and $L_D\preceq_{\square} L_G$, the average effective resistance across any cut is $\tilde{O}(1/k)$. %for any set $S\subseteq V$,
%$$ $$
Then, use  a simple extension of \autoref{thm:mss} such as \cite{AW13} to show that there is a spanning tree $T\subseteq E$ such that %has a spectrally thin tree supported on the edges of $E$. 
$$ L_T \preceq_\square \alpha \cdot (L_G+L_D),$$
for $\alpha=O(\max_{e\in E} \reff_{L_G+L_D}(e))=\tilde{O}(1/k)$.
But, since $L_D\preceq_{\square} L_G$, any $\alpha$-spectrally thin tree of $D+G$ is a $2\alpha$-combinatorially thin tree of $G$. 
In summary, the graph $D$ allows us to bypass  
the spectral thinness barrier that we described in \autoref{lem:necessaryspectralthinness}.

Let us give a clarifying example. Consider the $k$-edge-connected planar graph $G$ illustrated at the top of \autoref{fig:edgecover}. In this graph, all edges in the cut $(\{v_1,\dots,v_n\},\{u_1,\dots,u_n\})$ have effective resistance very close to 1. Now, let $D$ consist of the red edges  shown at the bottom. Observe that $L_D\preceq_{\square} L_G$. The effective resistance of every {\em black} edge in $G+D$ is  $O(1/\sqrt{k})$. Roughly speaking, this is because the red edges {\em shortcut} the long paths between the endpoints of vertical edges. This reduces the energy of the corresponding electrical flows. 
So, $G+D$ has a spectrally thin tree $T\subseteq E$. Such a tree is combinatorially thin with respect to $G$. 
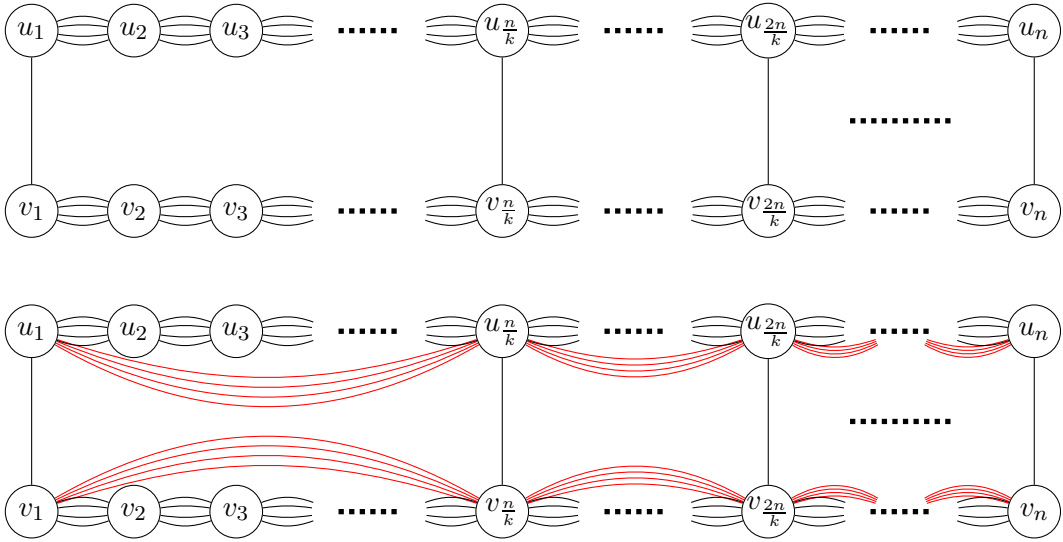
\begin{figure}\centering
\begin{tikzpicture}[scale=0.8]
\def\elen{1.5}
\tikzstyle{every node} = [draw,circle,minimum size=5.5mm,inner sep=0];
\foreach \y in {0,-5}{
\begin{scope}[shift={(0,\y)}]
\def\elen{1.7}
\tikzstyle{every node} = [draw,circle,minimum size=7mm,inner sep=0];
%\foreach \y in {0,-4}{
%\begin{scope}[shift={(0,\y)}]
\foreach \i/\j in {2/1,3/2,4/3}{
	\node at (\elen*\i,3) (u_\i)  {$u_\j$};
	\node at (\elen*\i,0) (v_\i)  {$v_\j$};
}
\foreach \i/\x in {5/5,6/5.6,8/7.6,9/8.2,11/10.2,12/10.8,14/10.5,15/10.5}{
	\node at (\elen*\x,3) [draw=none] (u_\i) {};
	\node at (\elen*\x,0) [draw=none] (v_\i) {};
}
\node at (\elen*6.6,3) (u_7) {$u_{\frac{n}{k}}$} ;
\node at (\elen*9.2,3) (u_10) {$u_{\frac{2n}{k}}$};
\node at (\elen*11.8,3) (u_13) {$u_n$};
\node at (\elen*6.6,0) (v_7) {$v_{\frac{n}{k}}$};
\node at (\elen*9.2,0) (v_10) {$v_{\frac{2n}{k}}$};
\node at (\elen*11.8,0) (v_13) {$v_n$};
\ifthenelse{\equal{\y}{0}}{}{
\foreach \i/\j in {7/2, 10/7, 14/10, 13/14}{
	\foreach \k in {18,22,26,30}{
		\path (u_\i) edge [bend left=\k,color=red] (u_\j)
			(v_\i) edge [bend right=\k,color=red] (v_\j);
	}
}}
\foreach \i/\j in {2/3, 3/4, 4/5, 6/7,7/8, 9/10, 10/11,12/13}{
	\foreach \a in {20, 8, -8, -20}{
		\path (u_\i) edge [bend left=\a] (u_\j) (v_\i) edge [bend left=\a] (v_\j);
	}
}
\path (v_2) edge (u_2) (v_7) edge (u_7) (v_10) edge (u_10) (v_13) edge (u_13);
\draw [dotted, line width=2] 
(\elen*5,0) -- (\elen*5.6,0) (\elen*5,3) -- (\elen*5.6,3) 
(\elen*7.6,0) -- (\elen*8.2,0) (\elen*7.6,3) -- (\elen*8.2,3)
(\elen*10.2,0) -- (\elen*10.8,0) (\elen*10.2,3) -- (\elen*10.8,3)
 (\elen*10,1.5) -- (\elen*11,1.5);
\end{scope}
}
\end{tikzpicture}
\caption{The top shows a $k$-edge-connected planar graph that has no spectrally thin tree. There are $k+1$ vertical edges, $(u_1,v_1), (u_{n/k},v_{n/k}), \dots, (u_n,v_n)$. For each $1\leq i\leq n-1$ there are $k$ parallel edges between  $u_i,u_{i+1}$ and $v_i,v_{i+1}$. The effective resistances of all vertical edges are $1-O(k^2/n)$.
The bottom shows a graph $G+D$ where the effective resistance of every black edge is $O(1/\sqrt{k})$. The red edges are edges in $D$ and there are $k$ parallel edges between the endpoints of consecutive vertical edges.
Note that $L_D\preceq_{\square} L_G$ by construction.}
%Therefore, by \autoref{lem:necessaryspectralthinness}
%any subgraph with at least one edge incident to vertex 1 has spectral thinness at least $1-O(k/n)$.}
\label{fig:edgecover}
\end{figure}

%Therefore, we take a detour. 
%Although we can not  Nonetheless, 
%Our main theorem proves (a stronger version of) the conclusion of the above lemma {\em unconditionally}.
It turns out that there are $k$-edge-connected graphs where it is impossible to reduce the effective resistance of all edges by a shortcut graph $D$ (see \autoref{sec:dual} for details).
So, in our main theorem, we prove a  weaker version of the above ideal plan. Firstly, instead of finding a shortcut graph $D$, we find a  PD shortcut matrix $D$.
The matrix $D$  does not necessarily represent the Laplacian matrix of a graph as it may have positive off-diagonal entries. 
Secondly, the shortcut matrix  reduces the effective resistance of only a set $F\subseteq E$ of edges, that we call \emph{good} edges, where $(V,F)$ is $\Omega(k)$-edge-connected.
% We show that there is  a {\em matrix} $D$ that reduces the effective resistance of only a set $F\subseteq E$ of \emph{good} edges where $(V,F)$ is $\Omega(k)$-edge-connected.

\begin{theorem}[Main]
\label{thm:hierarchydecomgeneral}
For any $k$-edge-connected graph $G=(V,E)$ where $k\geq 7\log(n)$, there is a shortcut matrix $0\prec D\preceq_{\square} L_G$ and a set of good edges $F\subseteq E$ such that the graph $(V,F)$ is $\Omega(k)$-edge-connected 
and that for any edge $e\in F$, 
$$\reff_D(e) \leq \tilde{O}(1/k),\footnote{For  functions $f(.),g(.)$ we write  $g=\tilde{O}(f)$ if $g(n) \leq \polylog(f(n))\cdot f(n)$ for all sufficiently large $n$.}$$
where $\reff_D(e)=\X_e^\intercal D^{-1}\X_e$.
\end{theorem}
Note that in the above we upper bound the effective resistance of good edges with respect to $D$ as opposed to $D+L_G$; this is sufficient because $\reff_{L_G+D}(e) \leq \reff_D(e)$. 
We remark that the dependency on $\log(n)$ in the statement of the theorem is because of a limitation of  our current proof techniques. We expect that a corresponding statement  without any dependency on $n$ holds for any $k$-edge-connected graph $G$. Such  a statement would resolve Goddyn's thin tree conjecture \ref{conj:thintree} and may lead to improved bounds on the integrality gap of LP \eqref{lp:tsp}. Finally, the logarithmic dependency on $k$ in the upper bound on the effective resistance of the edges of $F$ is necessary.

Unfortunately, the good edges in the above theorem may be very sparse with respect to $G$, i.e., $G$ may have cuts $(S,\overline{S})$ such that 
$$|F(S,\overline{S})| \ll |E(S,\overline{S})|.$$ 
%In other words, although $G$ is $k$-edge-connected and $F$ is $\Omega(k)$ connected, $G$ may have cuts $(S,\overline{S})$ 
%where $|E(S,\overline{S})| \geq n$ while $|F(S,\overline{S})| =\Theta(k)$.
So, if we  use \autoref{thm:mss} or its simple extensions as in \cite{AW13}, we get a thin set of edges $T\subseteq E$ that may have $\Omega_k(n)$ many connected components. 
Instead, we use a theorem, that we proved in our recent extension of \cite{MSS13}, that shows that as long as $F$ is $\Omega(k)$-edge-connected, $G$ has a spanning tree $T$ that is $\tilde{O}(1/k)$-spectrally thin with respect to $D+L_G$.
\begin{theorem}[\cite{AO14b}]
\label{thm:effrestothintree}
Given a graph $G=(V,E)$, a PD matrix $D$ and $F\subseteq E$ such that $(V,F)$ is $k$-edge-connected, if for $\eps>0$, 
$$ \max_{e\in F} \reff_D(e) \leq \eps,$$
then $G$ has a spanning tree $T\subseteq F$ s.t., %and for all $e\in F$, $\norm{x_e}^2\leq \eps$, then $G$ has a spanning tree $T$ such that
$$ L_T \preceq O(\eps+1/k)(D+L_G).$$
\end{theorem}
Putting \autoref{thm:hierarchydecomgeneral} and \autoref{thm:effrestothintree} together implies that any $k$-edge-connected graph has a $\polyloglog(n)/k$-thin tree.

\begin{corollary}
Any $k$-edge-connected graph $G=(V,E)$, has a $\polyloglog(n)/k$-thin tree.	
\end{corollary}
\begin{proof}
First, observe that by theorems \ref{thm:hierarchydecomgeneral}, \ref{thm:effrestothintree}
any $7\log(n)$ connected graph has a $\polyloglog(n)/\log(n)$-thin tree. 

Now, if $G$ is $k$-edge-connected and $k \gg \log(n)$, then we simply construct a $7\log(n)$ connected subgraph of $G$ that is $7\log(n)/k$ thin by sampling each edge independently with probability $\Theta(\log n/k)$ (see the proof of \autoref{thm:agmos} for the details of the analysis). Then, we use the aforementioned statement  to prove the existence of a thin tree in the sampled graph.
	
	Otherwise, if $k \ll \log(n)$, then we add $7\log(n)/k$ copies of each edge of $G$ and make a new graph $H$ that is $7\log(n)$ connected, then we use the previous corollary to find a $\polyloglog(n)/\log(n)$-thin tree of $H$. Such a tree is $\polyloglog(n)/k$-thin with respect to $G$.
\end{proof}

% nonetheless it is 
%T can not simply use \autoref{thm:avgeffthintree} to find a spectrally thin tree with respect to the ``graph''
%The above theorem can be seen as a corollary of the thin tree conjecture.
%We conjecture that any graph $G=(V,E)$ with a set $F\subseteq E$ such that $(V,F)$ is $k$-edge-connected and $\max_{e\in F}\reff_G(e) \ll 1$ has a $1/2$-spectrally thin tree. %i.e., that the above theorem implies the thin tree conjecture. 
%In fact, a stronger variant of the main theorem of \cite{MSS13} together with our main theorem implies that the integrality of LP \eqref{lp:tsp} is $\tilde{O}(\sqrt{\log(n)})$.

%We emphasize that all of the above results are existential.

We remark that, the above theorems do not resolve Goddyn's thin tree conjecture because of the dependency on $n$.

At first inspection, it would seem that there are two nonalgorithmic ingredients in our proof. The first one is the exponential-sized convex program that we will use to find the shortcut matrix $D$; this is because verifying $D\preceq_\square L_G$ is equivalent to $2^n$ many linear constraints. 
Secondly, we need to have  a constructive (in polynomial time)  proof of \autoref{thm:effrestothintree}. 
The following theorem shows we can get around the first barrier. 
\begin{theorem}\label{thm:algorithmicatsp}
	Assume that there is an oracle that  takes an input graph $G=(V,E)$, PD matrix $D$, and a $k$-edge-connected  $F\subseteq E$, such that $\max_{e\in F} \reff_D(e)\leq \eps$, and returns the spanning tree $T$ promised by \autoref{thm:effrestothintree}, i.e., $L_T\preceq O(\eps+1/k)(D+L_G)$.
	For any $\l\leq \log\log n$,
 there is a $\polyloglog(n)\cdot \log(n)^{1/\l}$-approximation algorithm for ATSP that runs in time $n^{O(\l)}$ (and  makes at most $n^{O(\l)}$ oracle calls).
\end{theorem}
We will prove this theorem in \autoref{subsec:algorithmicatsp}

%\newpage

\subsection{Main Components of the Proof}
\label{subsec:proofoverview}
Our proof has three main components, namely the thin basis problem, the effective resistance reducing convex programs, and the locally connected hierarchies. 
In this section we summarize the high-level interaction of these three components.

\paragraph{The Thin Basis Problem.}
Let us start by an overview of the proof of
\autoref{thm:effrestothintree} which has appeared  in a companion paper \cite{AO14b}. 
The thin basis problem is defined as follows: Given a set of vectors $\{x_e\}_{e\in E}\in \R^d$, 
%in the sub-isotropic position, 
%$$ \sum_{e\in E} x_ex_e^\intercal = I,$$
what is a sufficient condition for the existence of an $\alpha$-thin basis, namely, $d$ linearly independent set of vectors $T\subseteq E$ such that 
$$ \norm{\sum_{e\in T} x_ex_e^\intercal} \leq \alpha?$$
It follows from the work of Marcus, Spielman, and Srivastava \cite{MSS13} that a sufficient condition for the existence of an $\alpha$-thin basis is that  the vectors are in isotropic position, 
$$\sum_{e\in E} x_ex_e^\intercal = I,$$
and for all $e\in E$, $\norm{x_e}^2 \leq c\cdot \alpha$ for some universal constant $c<1$. 

The thin basis problem is closely related to the existential problem of spectrally thin trees. Say we want to see if a given graph $G=(V,E)$ has a spectrally thin tree. We can define a vector $y_e=L_G^{\dagger/2} \X_e$ for each edge $e\in E$. 
It turns out that these vectors are in isotropic position; in addition, if all edges of $G$ have  effective resistance at most $\eps$, then $\norm{y_e}^2=\X_e^\intercal L_G^\dagger \X_e\leq \eps$. So, these vectors contain an $O(\eps)$-thin basis. It is easy to see that such a basis corresponds to an $O(\eps)$-spectrally thin tree of $G$ (see \cite{AO14b} for details).

As alluded to in the introduction, if $G$ is a  $k$-edge-connected graph, it may have many edges of large effective resistance, so $\norm{y_e}^2$ in the above argument may be very close to 1. We use the shortcut matrix $D$ that is promised in \autoref{thm:hierarchydecomgeneral} to reduce the squared norm of the vectors. We assign a vector $y_e=(L_G+D)^{-1/2}\X_e$ to any good edge $e\in F$. It follows that
$$ \norm{y_e}^2 \leq \X_e^\intercal D^{-1} \X_e \leq \tilde{O}(1/k).$$
But, since the good edges are only a subset of the edges of $G$, the set of vectors $\{y_e\}_{e\in F}$ are not necessarily in an isotropic position; they are rather in a sub-isotropic position,
$$ \sum_{e\in F} y_ey_e^\intercal \preceq I.$$

In \cite{AO14b} we prove a weaker sufficient condition for the existence of a thin basis. If the  vectors $\{x_e\}_{e\in E}$ are in a sub-isotropic position, each of them has a squared norm at most $\eps$, and they contain $k$ disjoint bases, then there exists an $O(\eps+1/k)$-thin basis $T\subset E$
$$ \norm{\sum_{e\in E} x_ex_e^\intercal} \leq O(\eps+1/k).$$
Since, the set $F$ of good edges promised in \autoref{thm:hierarchydecomgeneral} is $\Omega(k)$-edge-connected, it contains $\Omega(k)$ edge-disjoint spanning trees, so the set of vectors $\{y_e\}_{e\in F}$ defined above contains $\Omega(k)$ disjoint bases. So, $\{y_e\}_{e\in F}$ contains a $\tilde{O}(1/k)$-thin basis $T$; this  corresponds to a $\tilde{O}(1/k)$-spectrally thin tree of $L_G+D$ and a $\tilde{O}(1/k)$-thin tree of $G$.
\paragraph{Effective Resistance Reducing Convex Programs.}
%As mentioned in the previous section, 
%our main idea is to ``symmetrize'' the $L_2$ structure of the graph while preserving its $L_1$ structure. 
%the thin tree conjecture can be seen as a special case of the $L_1$ version 
%of the Weaver conjecture. 
As illustrated in the previous section, at the heart of our proof we find a PD shortcut matrix $D$ to reduce the effective resistance of a subset of edges of $G$.

It turns out that the problem of finding the best shortcut matrix $D$ that reduces the maximum effective resistance of the edges of $G$ is  convex. 
%by writing a convex program to
%minimize the effective resistance of the edges of the graph, while  preserving the value of every cut up to a constant factor.
This is because %to a convex optimization problem,  because 
for any fixed vector $x$ and $D\succ 0$,
$x^\intercal D^{-1} x$ is a convex function of $D$.
See \autoref{lem:convexityeffres} for the proof.
The problem of minimizing the sum of effective resistances of all pairs of vertices in a given graph was previously studied in \cite{GBS08}.

The following (exponentially sized) convex program finds the best shortcut matrix $D$ that minimizes the maximum effective resistance of the edges of $G$ while preserving the cut structure of~$G$.
%\begin{equation}
%\label{cp:maxeffres}
\begin{table}[htb]\centering\begin{tabular}{l}
{\bf \maxcp:}\\
$\begin{aligned}
\min ~~~& ~~{\cal E},\\
\st ~~~&   
%\X_e^\intercal D^{-1} \X_e \leq {\cal E} 
\reff_{D}(e)\leq \cE
& \forall e\in E, \\
&  D \preceq_{\square} L_G, &\\
%\frac12 \bone_S^\intercal L_G \bone_S \leq  \bone_S^\intercal D \bone_S \leq \bone_S^\intercal L_G \bone_S & \forall \emptyset \subset S\subset V\\
& D \succ 0. & 
\end{aligned}$
%\end{equation} 
\end{tabular}\end{table}

%Note that the optimum matrix $D$ is not necessarily the Laplacian of a graph and indeed for our application of finding spectrally thin trees, any positive definite matrix $D$ that approximates the cut structure of $G$ is enough. 
%We can simplify the  program measuring the effective resistance with respect to $D$ as opposed to $D
%Note that $\X_e^\intercal D^\dagger \X_e$ is a convex function of $D$, so the above program is indeed a convex program. 

%+L_G$. %by dropping the constraint $\frac12  L_G \preceq_{\square} D$, and instead averaging out the optimum $D$ with the matrix $L_G$. 
Note that if we replace the constraint $D\preceq_{\square} L_G$ with $D\preceq L_G$, i.e., if we require $D$ to be upper-bounded by $L_G$ in the PSD sense, then the optimum $D$ for any graph $G$ is exactly $L_G$ and the optimum value is the maximum effective resistance of the edges of~$G$.  

Unfortunately, the optimum of the above program can be very close to 1 even if the input graph $G$ is $\log(n)$-edge-connected. A bad graph is shown in \autoref{fig:maxeffresbadexample}. In \autoref{thm:avgeffres} we show that the optimum of the above convex program for the family of graphs in \autoref{fig:maxeffresbadexample} is close to 1   by constructing a feasible solution of the dual. 
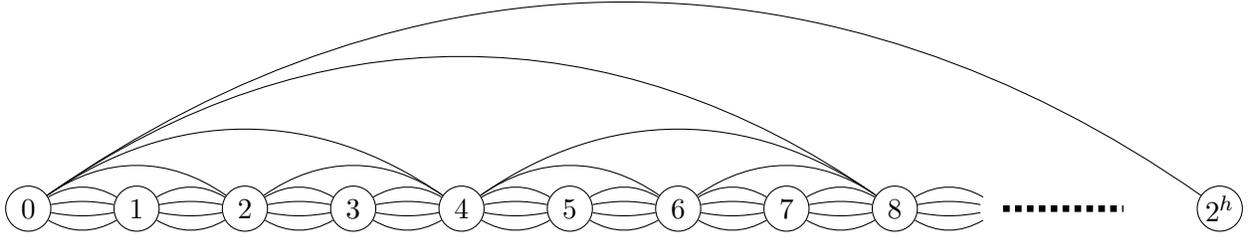
\begin{figure}
\centering
\begin{tikzpicture}[scale=0.8]
\tikzstyle{every node} = [draw,circle,minimum size=6mm,inner sep=0];
\foreach \i/\j in {1/0,2/1,3/2,4/3,5/4,6/5,7/6,8/7,9/8}{
	\node at (1.8*\i,0) (v_\i) {\j};
}
\node at (1.8*10,0) [draw=none] (v_10) {};
%\node at (1.8*9,0) (v_n1) {n-1};
\node at (1.8*12,0) (v_n) {$2^h$};
\foreach \i/\j in {1/2, 2/3, 3/4, 4/5, 5/6, 6/7, 7/8, 8/9, 9/10}{
	\foreach \a in {30, 10, -10, -30}{
		\path (v_\i) edge [bend left=\a] (v_\j);
	}
}
\path (v_1) edge [bend left=35] (v_3)  (v_3) edge [bend left=35] (v_5) (v_1) edge [bend left=35] (v_5) (v_5) edge [bend left=35] (v_7) (v_1) edge [bend left=35] (v_n)
(v_7) edge [bend left=35] (v_9) (v_5) edge [bend left=35] (v_9) (v_1) edge [bend left=35] (v_9);
\draw [dotted, line width=2.5] (18,0) -- (20,0);
\end{tikzpicture}
\caption{A tight example for \autoref{thm:avgeffres}. The graph has $2^h+1$ vertices labeled with $\{0,1,\ldots,2^{h}\}$. There are $k$ parallel edges
connecting each pair of consecutive vertices. In addition, for any $1\leq i\leq h$ and any $0\leq j< 2^{h-i}$ there is an edge $\{j\cdot 2^i,(j+1)\cdot 2^i\}$.}
\label{fig:maxeffresbadexample}
\end{figure}

To prove our main theorem, we study a variant of the above convex program that reduces the effective resistance of only a subset of edges of $G$ to $\tilde{O}(1/k)$. We will use combinatorial objects called locally connected hierarchies as discussed in the next paragraph to feed a carefully chosen set of edges into the convex program.
To show that the optimum value of the program is $\tilde{O}(1/k)$, we analyze its dual. The dual problem corresponds  to proving an upper bound on the ratio involving distances of pairs of vertices of $G$ with respect to an $L_1$ embedding of the vertices in a high-dimensional space.  We refrain from going into the details at this point. We will provide a more detailed overview in \autoref{sec:dual}.

\paragraph{Locally Connected Hierarchies.}
The main difficulty in proving \autoref{thm:hierarchydecomgeneral} is that the good edges, $F$, are unknown a priori. 
If we knew $F$ then we could use Max-CP to minimize the maximum effective resistance of edges of $F$ as opposed to $E$.
In addition,
the   $k$-th smallest  effective resistance of the edges of a cut of $G$ is not a convex function of $D$.
So, we cannot write a single program that gives us the best matrix $D$ for which there are at least $\Omega(k)$ edges of small effective resistance in every cut of $G$.

So, we take a detour. 
We use combinatorial structures that we call locally connected hierarchies that allow us to find an $\Omega(k)$-edge-connected set of good edges that may be very sparse with respect to $G$ in some of the cuts.
 Let us give an informal definition of locally connected hierarchies. 
Consider a {\em laminar} structure on the vertices of $G$, say $S_1,S_2, \dots \subseteq V$, where by a laminar structure we mean that there is no $i\neq j$ such that $S_i\cap S_j, S_i\setminus S_j, S_j\setminus S_i\neq \emptyset$. 
Modulo some technical conditions, if for all $i$, the induced subgraph on $S_i$, $G[S_i]$, is $k$-edge-connected, then we call $S_1,S_2,\dots$ a locally connected hierarchy. 

Let $S_{i^*}$ be the smallest set that is a superset of $S_i$ in the family, and let $\cut(S_i)=E(S_i, S_{i^*}\setminus S_i)$ be the set of edges leaving $S_i$ in the induced graph $G[S_{i^*}]$. 
In our main technical theorem we show that for any locally connected hierarchy we can find a shortcut matrix $D$ that reduces the maximum of the average effective resistance of all $\cut(S_i)$'s. In other words, the shortcut matrix $D$ reduces the effective resistance of at least half of the edges of each $\cut(S_i)$. %Let $F_1$ be . 
Unfortunately, these small effective resistance edges may have $\Omega(n)$ connected components. 

To prove \autoref{thm:hierarchydecomgeneral} we  choose $\polyloglog(n)$ many locally connected hierarchies adaptively, such that the following holds: Let the laminar family $S^j_1,S^j_2,\dots$ be the $j$-th locally connected hierarchy, and $D_j$ be a shortcut matrix that reduces the maximum average effective resistance  of $\cut(S^j_i)$'s. We let $F_j$ be the set of small effective resistance edges in $\cup_i \cut(S^j_i)$. 
We choose our  locally connected hierarchies such that $F=\cup_j F_j$ is $\Omega(k)$-edge-connected in $G$. To ensure this we  use several tools in graph partitioning.

\subsection{Organization}
The rest of the paper is organized as follows: We start with an overview of linear algebraic tools and graph theoretic tools that we use in the paper.
%In \autoref{sec:KS} we prove Theorems \ref{thm:KSapplication}, \ref{thm:edgecover}
%and \ref{thm:minorfree}.
In \autoref{sec:LCH} we given a high-level overview of our approach; we formally define \expandertrees~and we describe the main technical theorem \ref{thm:main}.
Then in \autoref{sec:partitioning} we prove
the main theorem \ref{thm:hierarchydecomgeneral} assuming the main technical theorem \ref{thm:main}.
The rest of the paper is dedicated to the proof of \autoref{thm:main}.
In \autoref{sec:dual} we characterize the dual of \maincp~and we prove \autoref{thm:avgeffres}, then in  the last two sections we upper-bound the value of the dual.
%\autoref{sec:balls} we upper-bound the numerator of the dual of \maincp~and in \autoref{sec:charging} we lower-bound the denominator of the dual of \maincp.

%such that the union of good edges that we obtain induces an $\Omega(k)$-edge-connected subgraph of $G$.

% say $F_1$, $F_2$, $\dots$, and we find a matrix $D$ that minimizes the maximum of the average effective resistance of edges across all $F_i$'s
%we use a combinatorial structure that we call expander trees 
% and the set $F$if the 

% !TEX root = main.tex

\section{Preliminaries}
%\subsection{Notation}
For an integer $k\geq 1$, we use $[k]$ to denote the set $\{1,\ldots,k\}$.
Unless otherwise specified, we assume that $G=(V,E)$ is an {\em unweighted} $k$-edge-connected graph with $n$ vertices. 
%Note that conventionally the term ``$k$-connectivity'' refers to $k$-vertex-connectivity, here we abuse the terminology and we write $k$-connectivity to refer to $k$-edge-connectivity. %for the sake brevity.
For a set $S\subseteq V$,
we use $G[S]$ to denote the induced subgraph of $G$ on $S$.
All graphs that we work with are unweighted with no loops but they may have an arbitrary number of parallel edges between every pair of vertices.

For a matrix $A\in\R^{m\times n}$ we write $A_i$ to denote the $i$-th column of $A$, $A^i$ to denote the $i$-th row of $A$ and $A_{i,j}$ to denote the $i,j$-th entry of $A$.

Throughout the paper we assume that there is a fixed ordering on the edges of $G$.
For an edge $e=\{u,v\}$ we use $\X_e=\bone_u - \bone_v$. We also write,
$$ L_{u,v} =\X_{u,v}\X_{u,v}^\intercal.$$
We use $\X\in\R^{V\times E}$ to denote the matrix where the $e$-th column is $\X_e$.
%Similarly, for a matrix $U$, and an edge $e$ of $G$, we write $U^e$ to denote the $e$-th row of $U$.

For disjoint sets $S,T\subseteq V$ we write
$$ E(S,T):=\{\{u,v\}: u\in S, v\in T\}.$$
We say two sets $S,T\subseteq V$ cross if $S\cap T, S\setminus T, T\setminus \neq \emptyset$.
For a set $S$ of elements we write
$\EE{e\sim S}{.}$ to denote the expectation under the uniform distribution over the elements of $S$.
We think of a permutation of a set $S$ as a bijection mapping the elements
of $S$ to $1,2,\ldots,|S|$. 
 For a vector $x\in\R^d$, we write
\begin{eqnarray*} 
\norm{x} &=& \sqrt{\sum_{i=1}^d x_i^2},\\
\norm{x}_1 &=& \sum_{i=1}^d |x_i|.
\end{eqnarray*}
We will use the following inequality in many places.
For any sequence of nonnegative numbers $a_1,\dots,a_m$ and $b_1,\dots,b_m$
\begin{equation}\label{eq:averagefracs} \min_{1\leq i\leq m}\frac{a_i}{b_i}\leq \frac{a_1+a_2 + \dots+a_m}{b_1+b_2+\dots+b_m} \leq \max_{1\leq i\leq m} \frac{a_i}{b_i}.
\end{equation}

\subsection{Balls and High-Dimensional Geometry}
\label{sec:prelim:balls}
For $x\in \R^d$ and $r\in \R$, an $L_1$ ball is the set of points at $L_1$ distance less than $r$ of $x$,
$$B(x,r):=\{ y\in \R^d: 0< \norm{x-y}_1 < r\}.$$
Unless otherwise specified, any ball that we consider in this paper is an $L_1$ ball.
We may also work with $L_2$ or $L_2^2$ balls and by that we are referring to  a set of points whose $L_2$ or $L_2^2$ distance from a center is bounded by $r$. 

An $L_1$ {\em hollowed} ball is a ball with part of it  removed; for $0\leq r_1<r_2$, we  define the hollowed ball $B(x,r_1\|r_2)$ as follows:
$$ B(x,r_1\|r_2) :=\{y\in \R^d: r_1 < \norm{x-y}_1 < r_2\}.$$
Observe that $B(x,r)=B_1(x,0\|r)$.
The {\em width} of $B(x,r_1\|r_2)$ is $r_2-r_1$.

We say a point $y\in \R^d$ is inside a hollowed ball $B=B(x,r_1\|r_2)$
if $$ r_1 < \norm{x-y}_1 < r_2,$$
and we say it is outside of $B$ otherwise. 
We also say a (hollowed) ball $B_1$ is inside a (hollowed) ball $B_2$
if every point $x\in B_1$ is also in $B_2$.

For a (finite) set of points $S\subseteq \R^d$, the $L_1$ diameters of $S$, $\diam(S)$ is defined as the maximum $L_1$ distance between points in $S$,
$$ \diam(S) = \max_{x,y\in S} \norm{x-y}_1.$$

For a set $S$ of elements we say $X:S\to \R^h$ is an $L^2_2$ metric if for any three elements $u,v,w\in S$,
$$ \norm{X_u-X_w}^2 \leq \norm{X_u-X_v}^2 + \norm{X_v-X_w}^2.$$
A cut metric of $S$ is a mapping $X:S\to\{0,1\}^h$ equipped with the $L_1$ metric. Note that any cut metric of $S$ is also a $L_2^2$ metric because for any two elements $u,v\in S$,
$$ \norm{X_u-X_v}_1 = \norm{X_u-X_v}^2.$$
Similarly, we define a weighted cut metric, $X:S\to\{0,1\}^h$ together with nonnegative weights $w_1,\dots,w_h$, to be
the be the points $\{X_v\}_{v\in S}$ where equipped with the weighted $L_1$ norm:
$$ \norm{x}_1 = \sum_{i=1}^h w_i\cdot |x_i|, \text{ for all } x\in\R^h.$$
If all the weights are $1$ we simply get an (unweighted) cut metric. It is easy to see that any weighted cut metric can be embedded, with arbitrarily small loss, (up to scaling) in an unweighted cut metric of a (possibly) higher dimension.

We can  look at an embedding $X$ as a matrix where there is a column $X_u$ for any vertex $u$.
We also write
$$ \Xb = X \X.$$
Therefore, for any edge $e=\{u,v\} \in E$ (oriented from $u$ to $v$),
$$ \Xb_e =  X\X_e = X_u-X_v. $$
%For a graph $G=(V,E)$, a cut metric $X:V\to\{0,1\}^h$,
%and an edge $e=\{u,v\}\in E$ (oriented from $u$ to $v$), we write
%$$ X_e = X\X_e = X_u-X_v.$$

\subsection{Facts from Linear Algebra}
%We use $J:=\frac1n \bone \bone^\intercal$. For a symmetric PSD matrix $A$ we use 
%$$ A^{\dagger/2} := (A^\dagger)^{1/2}.$$
%For a matrix $A$, let $\ker(A):=\{x\in \R^n: Ax=\bzero\}$.
%Also, $\image(A)$ be the set of all linear combinations of columns of $A$.
We use $I$ to denote the identity matrix
and $J$ to denote the all $1$'s matrix. 
%For a matrix $A\in\R^{m\times n}$ we write 
%$$ \trace(A):=\sum_{i=1}^{\min\{m,n\}} A_{i,i}.$$
A matrix $U\in\R^{n\times n}$ is called orthogonal/unitary if $UU^\intercal=U^\intercal U=I$.
An orthogonal matrix is a nonsingular square matrix whose singular values are all 1.
It follows by definition that orthogonal operators preserve $L_2$ norms of vectors, i.e., for any vector $x\in\R^n$,
$$ \norm{Ux} = \sqrt{(Ux)^\intercal Ux}=\sqrt{x^\intercal U^\intercal U x} = \sqrt{x^\intercal x}=\norm{x}. $$

A (not necessarily square) matrix $U$ is called semiorthogonal if $UU^{\intercal}=I$, i.e. the rows are orthonormal, and the number of rows is less than the number of columns. For any semiorthogonal $U\in \R^{m\times n}$, we can extend $U$ to an actual orthogonal matrix by adding $n-m$ rows.

For two matrices $A,B$ of the same dimension we define the matrix inner product $A\bullet B:=\trace(AB^\intercal)$.

For any matrix $A\in\R^{m\times n}$ and $B\in\R^{n\times m}$,
$$ \trace(AB)=\trace(BA). $$
For any two matrices $A\in\R^{m\times n},B\in\R^{n\times m}$, then the nonzero eigenvalues of $AB$ and $BA$ are the same with the same multiplicities. 

\begin{lemma}
\label{lem:kernelpsd}
If $A$, $B$ are positive semidefinite matrices of the same dimension,
then
$$ \trace(AB)\geq 0.$$
\end{lemma}
\begin{proof}
$$\trace(AB) = \trace(A B^{1/2} B^{1/2}) = \trace(B^{1/2} A B^{1/2}) \geq 0.$$
\end{proof}

\begin{fact}[{Schur's Complement \cite[Section A.5]{BV06}}]
For any symmetric positive-definite matrix $A \in\R^{n\times n}$ a (column) vector $x\in \R^n$   and $c\geq 0$,
we have $x^\intercal A^{-1} x \leq c$ if and only if
$$  \tab{c}{x^\intercal}{x}{A}  \succeq 0.$$
\end{fact}

The following  lemma proving the operator-convexity of the inverse of PD matrices is well-known. 
\begin{lemma}
\label{lem:convexityeffres}
For any two  symmetric $n\times n$ matrices  $A, B \succ 0$, %and $v\in\R^n$, $v^\intercal A^{-1} v$ is a convex function of $A$, i.e., if $A=\lambda B+(1-\lambda) C$ for PD matrices $B,C\in\R^{n\times n}$, then,
%$$ v^\intercal A^{-1} v \leq \lambda\cdot  v^\intercal B^{-1} v + (1-\lambda)\cdot v^\intercal C^{-1} v.$$
$$ \Big(\frac12 A + \frac12 B\Big)^{-1} \preceq \frac12 A^{-1} + \frac12 B^{-1}.$$
\end{lemma}
\begin{proof}
%If $v\in \ker(B)$ or $v\in \ker(C)$ then the lemma holds trivially. Otherwise,
For any vector $x\in\R^n$,
$$ 
\frac12  \tab{x^\intercal A^{-1} x}{x^\intercal}{x}{A} + \frac12 \tab{x^\intercal B^{-1} x}{x^\intercal}{x}{B} = 
\tab{\frac12 x^\intercal A^{-1} x+\frac12 x^\intercal B^{-1} x}{x^\intercal}{x}{ \frac12 A+ \frac12 B}.
$$
By Schur complement both of the matrices on the LHS of above equality are PSD.
Therefore, by convexity of PSD matrices, the  matrix in RHS is also PSD.
By another application of Schur complement to the matrix in RHS we obtain the lemma.
\end{proof}

\begin{definition}[Matrix Norms]
\label{def:matrixnorms}
The trace norm (or nuclear norm) of  a matrix $A\in\R^{m\times n}$ is defined as follows:
$$ \norm{A}_* := \trace((A^\intercal A)^{1/2}) = \sum_{i=1}^{\min\{m,n\}} \sigma_i,$$
where $\sigma_i$'s are the singular values of $A$.
The Frobenius norm of $A$ is defined as follows:
$$ \norm{A}_F :=\sqrt{\sum_{1\leq i\leq m, 1\leq j\leq n} A_{i,j}^2} = \sqrt{\sum_{i=1}^{\min\{m,n\}} \sigma_i^2}.$$
\end{definition}
The following lemma is a well-known fact about the trace norm. 
\begin{lemma}
\label{lem:matrixtrace}
For any matrix $A\in\R^{n\times m}$ such that $n\geq m$,
$$ \norm{A}_* = \max_{\text{Semiorthogonal } U} \trace(UA),$$
where the maximum is over all semiorthogonal matrices $U\in\R^{m\times n}$.
In particular, $\trace(A) \leq \norm{A}_*$.
\end{lemma}
\begin{proof}
Let the singular value decomposition of $A$ be the following
$$A=\sum_{i=1}^m \sigma_i u_i v_i^{\intercal},$$
where $s_1, \dots, s_m$ are the singular values and $u_1,\dots, u_m\in \R^n$ are the left singular vectors and $v_1,\dots ,v_m\in \R^m$ are the right singular vectors. Now let 
$$U=\sum_{i=1}^m v_i u_i^{\intercal}.$$
It is easy to observe that $U\in \R^{m\times n}$ is semiorthogonal, i.e. $UU^{\intercal}=I$. Now observe that
$$UA=\sum_{i=1}^{m} \sigma_i v_i \langle u_i, u_i\rangle v_i^\intercal = \sum_{i=1}^m \sigma_i v_i v_i^\intercal.$$
It is easy to see that $\trace(UA)=\sum_{i=1}^m \sigma_i=\norm{A}_*$.

It remains to prove the other side of the equation. By von Neumann's trace inequality \cite{vonneumann}, for any semiorthogonal matrix $U\in\R^{m\times n}$ we can write
$$ \trace(UA) \leq \sum_{i} 1\cdot \sigma_i = \norm{A}_*,$$
where $\sigma_1,\ldots,\sigma_m$ are the singular values of $A$.
\end{proof}

\begin{theorem}[Hoffman-Wielandt Inequality]
\label{thm:hoffmanineq}
Let $A,B\in\R^{n\times n}$ have singular values $\sigma_1\leq\sigma_2\leq\ldots\sigma_n$
and $\sigma'_1\leq\sigma'_2\leq\ldots\leq\sigma'_n$. Then,
$$ \sum_{i=1}^n (\sigma_i-\sigma'_i)^2 \leq \norm{A-B}_F^2.$$
\end{theorem}

\subsection{Background in Graph Theory}
\label{subsec:graphtheory}
For a graph $G=(V,E)$, 
and a set $S\subseteq V$, we define 
$$ \phi_G(S):=\frac{\deg_G(S)}{\vol_G(S)}$$
where 
$\deg_G(S):=|E(S,V\setminus S)|$
is the number of edges that leave $S$, and 
$ \vol_G(S)$ is the sum of the degrees (in $G$) of vertices
of $S$.
Note that, by definition, $\vol_G(v)=\deg_G(\{v\})$ for any vertex.
 If the graph is clear in the context we drop the subscript $G$.
The expansion of $G$ is defined as follows:
$$ \phi(G) := \min_{S\subset V} \frac{\deg_G(S)}{\min\{\vol_G(S),\vol_G(V\setminus S)\}} = \min_{S\subset V} \max\{\phi_G(S),\phi_G(V\setminus S)\},$$
We say a graph $G$ is an $\eps$-expander, if $\phi(G)\geq \eps$. Recall that in an expander graph, $\phi(G)= \Omega(1)$.

\medskip
An (unweighted) graph $G=(V,E)$ is $k$-edge-connected
if and only if for any pair of vertices $u,v\in V$,
there are at least $k$ edge-disjoint paths between $u,v$ in $G$. Equivalently, $G$ is $k$-edge-connected if for any set $\emptyset \subsetneq S\subsetneq V$, $\deg(S)\geq k$. 

There is a well-known theorem by Nash-Williams that gives an almost (up to a factor of 2) necessary and sufficient condition for $k$-connectivity.
\begin{theorem}[\cite{NW61}]
\label{thm:Nash-Williams}
For any $k$-edge-connected graph, $G=(V,E)$, there are at least $k/2$ disjoint spanning trees in $G$. 
\end{theorem}
Note that any union of $k/2$ edge-disjoint spanning trees is a $k/2$-edge-connected graph. So, the above theorem does not give a necessary and sufficient condition for $k$-connectivity. A cycle gives a tight example for the loss
of $2$ in the above theorem.

Given a graph $G=(V,E)$, and a set $S\subseteq V$,
we write $G/S$ to denote the graph where the set $S$ is {\em contracted}, i.e., we remove all vertices $v\in S$
and add a new vertex $u$ instead, and for any vertex $w\notin S$, we let $|E(S,\{w\})|$ be the number of  (parallel) edges between $u$ and $w$. We also remove any self-loops that result from this operation.
The following fact will be used throughout the paper.
\begin{fact}
\label{fact:kconinvariant}
For any $k$-edge-connected graph $G=(V,E)$ and any set $S\subseteq V$, $G/S$ is $k$-edge-connected.
\end{fact}

\medskip

Throughout the paper we may use a natural decomposition of a graph $G$ (that is not necessarily $k$-edge-connected) into $k$-edge-connected subgraphs as defined below. 
\begin{definition}
\label{def:decompkconnected}
For a graph $G=(V,E)$ a natural decomposition into $k$-edge-connected subgraphs is defined as follows: Start with a partition $S_1=V$.
While there is a nonempty set $S_i$ in the partition such that $G[S_i]$ is not $k$-edge-connected, find an induced cut $(S_{i,1},S_{i,2})$ in $G[S_i]$  of size less than $k$,
remove $S_i$ and add $S_{i,1}, S_{i,2}$ as  new sets in the partition.
\end{definition}
The following fact follows directly from the above definition.
\begin{lemma}
\label{lem:decompkconnected}
For any natural decomposition of a graph $G=(V,E)$ into $k$-edge-connected subgraphs $S_1,\ldots,S_\ell$ and any $I\subseteq [\ell]$,
$$ \sum_{i_1, i_2\in I: i_1<i_2} |E(S_{i_1}, S_{i_2})| \leq (k-1)(|I|-1).$$
Consequently, 
$$ \sum_{i=1}^\ell \deg(S_i)= 2\sum_{i_1,i_2\in[\ell]: i_1<i_2} |E(S_{i_1},S_{i_2})|  \leq 2(k-1)(\ell-1).$$
\end{lemma}
\begin{proof}
Let $S=\cup_{i\in I} S_i$. 
%The lemma simply follows from the fact that the number of edges between the sets whose index is in $I$ is at most $(k-1)\cdot (|I|-1)$.
%$S_1,\ldots,S_j$ is at most $(k-1)(j-1)$; 
A natural decomposition of the induced subgraph, $G[S]$ into $k$-edge-connected subgraphs gives exactly all set $S_i$ where $i\in I$. This decomposition partitions  $G[S]$ exactly $|I|-1$ times and each time adds at most $k-1$ new edges  between the sets in the partition. 
\end{proof}

\section{Overview of Our Approach}
\label{sec:LCH}
In this section we give a high-level overview of our approach. We will motivate and formally define locally connected hierarchies and we describe our main technical theorem. 
In this section we will not overview the proof of the main technical theorem \ref{thm:main}, see \autoref{sec:dual} for the explanation.

As alluded to in the introduction, in \autoref{thm:avgeffres} we will show that it is not possible to reduce the maximum effective resistance of the edges of every $k$-edge-connected graph using a shortcut matrix. 

The first idea that comes to mind is to reduce the maximum  average effective resistance amongst all cuts of $G$. 
We can use the following convex program to find the best such shortcut matrix.
%The following convex program finds the best shortcut matrix $D$ that minimizes the maximum of average effective resistance of edges across all cuts while preserving the cut structure of $G$.
%\begin{equation}
%\label{cp:avgeffres}
\begin{table}[htb]\centering\begin{tabular}{l}
{\bf \averagecp}:\\
$\begin{aligned}
\min ~~~& ~~{\cal E}\\
\st ~~~&   
%\sum_{e\in E(S,\overline{S})} \frac{1}{|E(S,\overline{S})|} \cdot 
\uE{e\sim E(S,\overline{S})} 
%\X_e^\intercal D^{-1} \X_e \leq {\cal E} 
\reff_D(e)\leq \cE
&~~~~ \forall \emptyset\subsetneq S\subsetneq V, \\
& D\preceq_{\square} L_G, &\\ 
% \bone_S^\intercal D \bone_S \leq \bone_S^\intercal L_G \bone_S & \forall \emptyset \subset S\subset V\\
& D \succ 0. & 
\end{aligned}$
%\end{equation} 
\end{tabular}\end{table}

Note that if the optimum is small, it means that there are at least $k/2$ good edges in every cut of $G$, so the set $F$ of good edges is $\Omega(k)$-edge-connected and we are done.
Unfortunately, as we will show in \autoref{thm:avgeffres} the same example  shows that the optimum of the above convex program 
is very close to 1 for an $\Omega(\log(n))$-edge-connected graph. 
In fact, in the proof of \autoref{thm:avgeffres},
we lower-bound the optimum of \averagecp. 

The above impossibility result shows that it is not possible to reduce  the average effective resistance of all cuts of $G$. 
Our approach is to recognize families of subsets of edges for which it is possible to reduce the maximum average effective resistance.
%we reduce average effective resistance of carefully chosen sets of edges of $G$.

In the first step, we observe that for any partitioning of the vertices of a $k$-edge-connected graph $G$ into $S_1,S_2,\dots$ we can use a variant of the above convex program to reduce the maximum average effective resistance of the sets
$$ E(S_1,\overline{S_1}), E(S_2,\overline{S_2}), \text{ and so on} $$
to $\tilde{O}(1/k)$.
Next, we illustrate why this is useful using an example. Later, we will see that our main technical theorem implies a stronger version of this statement.

\begin{example}\label{ex:cycleexpander} Assume that $G$ is defined as follows: Start with a $k$-regular $\eps$-expander on $\sqrt{n}$ vertices and replace each vertex with a cycle of length $\sqrt{n}$ repeated $k$ times where the endpoints of the expander edges incident to each cycle are equidistantly distributed. This graph is $k$-edge-connected by definition and all expander edges have effective resistance close to 1. 

If we use the $\sqrt{n}$ cycles as our partition, by the above observation, we can reduce the average effective resistance of edges coming out of each cycle to some $\alpha =\tilde{O}(1/k)$. Let $F$ be the union of all of the cycle edges and the  expander edges of effective resistance at most $2\alpha/\eps$. Now, we show that $F$ is $\Omega(k)$-edge-connected. For any cut that cuts at least one of the cycles, obviously there are at least $k$ cycle edges in $F$. For the rest of the cuts, at least $\eps$-fraction of the expander edges incident to the cycles on the small side of the cut cross the cut; among these edges at least half of them are in $F$, so $F$ has at least $\Omega(k)$ edges in the cut.
\end{example}
We can use the above observation in any $k$-edge-connected graph repeatedly to gradually make $F$ $\Omega(k)$-edge-connected as follows: Start with partitioning into singletons; let $D_1$ be a shortcut matrix that reduces the average effective resistance of degree cuts to $\alpha=\tilde{O}(1/k)$, and let $F_1$ be the edges of effective resistance at most $2\alpha$. In the next step, let the partitioning $S_1,S_2,\dots$ be a natural decomposition of $(V,F_1)$ into $k/2$-edge-connected components. Similarly, define $D_2$ and let $F_2$ be the edges connecting $S_1,S_2,\dots$ of effective resistance at most $2\alpha$. This procedure ends in $\l=O(\log n)$ iterations. It follows that $\cup_{i=1}^\l F_i$ is $\Omega(k)$-edge-connected and the average of shortcut matrices, $\mE_{i} D_i$, is a shortcut matrix that reduces the effective resistance of all edges of $F$ to $O(\l\cdot \alpha)$. Therefore, if $\l =\polyloglog(n)$ we are done. 

Unfortunately there are $k$-edge-connected graphs where the above procedure  ends in $\Theta(\log n)$ steps because each time the size of the partition may reduce only by a factor of 2. Note that this procedure defines a laminar family over the vertices. Let $S_1,S_2,\dots$ be all of the sets in all partitions; observe that they form a laminar family; let $S_{i^*}$ be the smallest set that is a superset of $S_i$. Also, let $\cut(S_i)=E(S_i,S_{i^*}\setminus S_i).$ 

Suppose  we write a convex program to {\em simultaneously} reduce the maximum average effective resistance of all $\cut(S_i)$'s; then we may obtain a $k$-edge-connected set $F$ of good edges in a single shot. As we will see next, modulo some technical conditions, this is what we prove in our main technical theorem. Such a statement is not enough to get a $k$-edge-connected set of good edges, but it is enough to get $F$ in $\polyloglog(n)$ steps.

\begin{figure}
\centering
\begin{tikzpicture}[scale=0.8]
\tikzstyle{every node} = [draw,circle,minimum size=6mm,inner sep=0];
\def\anglet{0.6}
\node at (0, 0) (v_0) {$0$};
\foreach \i/\l in {1/1,2/2,3/3,5/2^h}{
	\node at (\i+1, \anglet*\i-\anglet) (v_\i) {$\l$};
	\node at (\i, \anglet*\i) (t_\i) {$t_{\l}$} edge (v_\i);
%	\node at (\i+1, \i-1) (v_\i) {$\i$};
}
\node at (4, \anglet*4) [draw=none] (t_4) {};

\path (t_1) edge (v_0) (t_2) edge (t_1) (t_3) edge (t_2) (t_5) edge (t_4);
%\node at (1,1) (t_0) {$t_0$} edge (v_0) edge (v_1);
%\node at (5,1) (t_1) {$t_1$} edge (v_2) edge (v_3);
%\node at (3,2) (t_2) {$t_2$} edge (t_0) edge (t_1) edge (5,3);
%\node at (10, 0) (v_4) {$2^h$} edge (9,1);
\draw [dotted, line width=1.5] (3.5,\anglet*3.5) -- (4,\anglet*4);
\end{tikzpicture}
\caption{A $\cT(k,1/2,\{1,2,\ldots,2^h\})$-\expandertree~of 
the graph of \autoref{fig:maxeffresbadexample}.}
\label{fig:hierarchicaltree}
\end{figure}
\subsection{Locally Connected Hierarchies}
\label{subsec:expandertree}
For a graph $G=(V,E)$,  
a hierarchy, $\cT$, is a tree where every non-leaf node has at least two children and  each leaf corresponds to a unique vertex of $G$. We use the terminology {\em node} to refer to vertices of $\cT$. For each node $t\in \cT$ let $
V(t)\subseteq V$ be the set of vertices of $G$ that are mapped to the leaves of the subtree of $t$, $E(t)$ be the set of edges between the vertices of  $V(t)$, and
$$ G(t)=G[V(t),E(t)],$$
be the induced subgraph of $G$ on $V(t)$.
Let $\setdeg(t):=E(V(t),  \overline{V(t)})$ be the set of edges that leave $V(t)$ in $G$.
Throughout the paper we use $t^*$ to denote the parent of a node $t$. We define  $\cut(t) := E(V(t), V(t^*)\setminus V(t))$ as the set of edges that leave $V(t)$ in $G(t^*)$.
We abuse notation and use $\cT$ to also denote the set of nodes of $\cT$.

Let us give a clarifying example. Say $G$ is the ``bad'' graph of \autoref{fig:maxeffresbadexample}.
In  \autoref{fig:hierarchicaltree} we give a \expandertree~of $G$.
For each node $t_i$, $V(t_i)=\{0,1,\ldots,i\}$.
For each $1\leq i\leq 2^h$, the set $\cut(i)$ is the set of edges from vertex $i$ to all vertices $j$ with $j < i$. In addition, since $t_i$ has exactly two children,  $\cut(i)=\cut(t_{i-1})$. 
Finally, $\setdeg(i)$ is all edges incident to vertex $i$ and $\setdeg(t_i)$ is the set of edges $E(\{0,1,\ldots,i\},\{i+1,\ldots,2^h\})$.

For an integer $k>1$, $0<\lambda<1$, and $T\subseteq \cT$, we say $\cT$ is a $(k,\lambda, T)$-\expandertree~of $G$, or $(k,\lambda,T)$-LCH if
\begin{enumerate}
\item For each node $t\in \cT$, the induced graph $G(t)$ is $k$-edge-connected.
\item For any node $t\in \cT$ that is not the root, $|\cut(t)| \geq k$.  This property follows from 1 because $\cut(t)=E(V(t), V(t^*)\setminus V(t))$ is a cut of $G(t^*)$.
\item For any node $t\in T$, $|\cut(t)| \geq \lambda\cdot |\setdeg(t)|$. Note that unlike the other two properties, this one only holds for a subset $T$ of the nodes of $\cT$.
\end{enumerate}
 We say $\cT$ is a $(k,\lambda,\cT)$-LCH if $T$ is the set of all nodes of $\cT$. 
For example, the hierarchy of \autoref{fig:hierarchicaltree} is a $(k,1/2,\{1,2,\dots,2^h\})$-LCH of the graph illustrated in \autoref{fig:maxeffresbadexample}.
Condition 1 holds because there are $k$ parallel edges between any pair of vertices $i-1,i$, so
$G(V(t_i))$ is $k$-edge-connected. Condition 2 holds because,  
$$ |\cut(i)| = |\cut(t_{i-1})| =|E(\{0,\ldots,i-1\},\{i\})| \geq k. $$
Lastly, it is easy to see that condition 3 holds for any leaf node $i\in T$,
$|\cut(i)|\geq d(i)/2=|\setdeg(i)|/2$.
%See \autoref{sec:partitioning} for the construction of \expandertree s.

We will use the following terminology mostly in \autoref{sec:charging}. For two nodes $t,t'$ of an \expandertree, $\cT$, 
we say $t$ is an {\em ancestor} of $t'$, if $t\neq t'$ and $t'$ is a node of a subtree of $t$.
We say $t$ is a {\em weak ancestor} of $t'$ if either $t=t'$ or $t$ is an ancestor of $t$.
We say $t$ is a {\em descendant} of $t'$ if $t'$ is an ancestor of $t$.
We say $t,t'\in\cT$ are {\em ancestor-descendant} if either $t$ is a weak ancestor of $t'$ or $t'$ is a weak ancestor of $t$.

\paragraph{\Expandertrees~and Good Edges.}
Let $\cT$ be a hierarchy of $G$. Let $t\in\cT$ have children $t_1,\dots,t_j$. Define $$G\{t\}:=G(t)/V(t_1)/V(t_2)/\dots/V(t_j)$$ to be the graph obtained from $G(t)$ by contracting each $V(t_i)$ into a single vertex. 
We may call  $G\{t\}$ an internal subgraph of $G$.
Let $V\{t\}$ be the vertex set of $G\{t\}$;
we can also identify this set with the children of $t$ in $\cT$. 
Also, let $E\{t\}$ be the edge set of $V\{t\}$. 

The following property  of \expandertrees~is crucial in our proof. Roughly speaking, if a subset $F$ of edges of $G$ is  $k$-edge-connected in each internal subgraph, then it is globally $k$-edge-connected.
\begin{lemma}
\label{lem:LCHinternalconnectivity}
Let $\cT$ be a hierarchy of a graph $G=(V,E)$ and $F\subseteq E$. If for any internal node $t$, the subgraph $(V\{t\}, F\cap E\{t\})$ is $k$-edge-connected, then $(V,F)$ is $k$-edge-connected.
\end{lemma}
\begin{proof}
	Consider any cut $(S,\overline{S})$ of $G$. Observe that there exists an internal node $t\in \cT$ such that $S$ crosses $V(t)$. Let $t_0$ be the deepest such node in $\cT$ (root has depth $0$).	But then,
	$$ F(S,\overline{S}) \supseteq F(S\cap V(t_0), \overline{S}\cap V(t_0)),$$
	and the size of the set on the RHS is at least $k$ by the assumption of the lemma.
\end{proof}
To prove \autoref{thm:hierarchydecomgeneral} we will find a good set of edges which satisfy the assumption of the above lemma. 
Note that the assumption of the above lemma does not imply that $F$ is dense in $G$. This is crucial because \autoref{thm:avgeffres} shows that there is  no shortcut matrix $D$ which has a dense set of good edges. 
%In particular, even if $G$ is only $2k$-edge-connected, we may have cuts where $G$ has significantly more edges than $F$.

%Recall that in \autoref{thm:avgeffres} we show that
%there is not shortcut matrix $D$ for which at least half of the edges of every cut of $G$ is good. The crucial property of the \expandertrees is that we may be able to find a set $F$ of good edges where $F$ is $\Omega(k)$-edge-connected in each internal subgraphs while 

\paragraph{Construction of an LCH for Planar Graphs.}
%Let us construct \expandertrees~ for planar graphs. 
In this section we give a universal construction of \expandertrees~for $k$-edge-connected planar graphs. 
\begin{lemma}\label{lem:planarLCH}
Any $k$-edge-connected planar graph $G=(V,E)$ has a
$(k/5,1/5,T)$-LCH $\cT$ where $\cT$ is a binary tree, and $T$ contains at least one child of each nonleaf node of $\cT$.
\end{lemma}

We will use the following fact about planar graphs, whose proof easily follows from the fact that {\em simple} planar graphs have at least one vertex with degree at most $5$.

\begin{fact}
\label{fact:planardegree}
In any $k$-edge-connected planar graph $G=(V,E)$, there is a pair of vertices $u,v\in V$ with at least $k/5$ parallel edges between them.
\end{fact}

%We construct a $(1, 5)$-expanding 
%We construct a
%$\cT=\cT(k/5,.,.)$ \expandertree~of $G$. 
The details of the construction are given in \autoref{alg:minorfree}.
\begin{algorithm}
\begin{algorithmic}[1]
\Input A $k$-edge-connected planar graph $G$.
\Output A  $(k/5,.,.)$-LCH  of $G$.
\State For each vertex $v\in V$, add a unique leaf node to $\cT$ and map $v$ to it. Let $W$ be the set of these leaf nodes.
\Comment{We keep the invariant that $W$ is the nodes of $\cT$ that do not have a parent yet, but their subtree is fixed, i.e., $V(t)$ is well-defined for any $t\in W$.}
\While {$|W| > 1$}
\State Add a new node $t^*$ to $W$. 
\State \parbox[t]{\dimexpr\linewidth-\algorithmicindent}{
Let $G_{t^*}$ be the graph where for each node $t\in W$, $V(t)$ is contracted to a single vertex; identify each $t\in W$ with the corresponding contracted vertex.
\Comment{Note that $G_{t^*}$ is also a planar graph, because for any $t\in W$, the induced graph $G[V(t)]$ is connected.}
\strut}
\State Let $t_1$ be a vertex with at most $5$ neighbors in $G_{t^*}$.  
\Comment{$t_1$ exists by \autoref{fact:planardegree}.}
\State \parbox[t]{\dimexpr\linewidth-\algorithmicindent}{
Let $t_2$ be a neighbor of $t_1$ such that $\{t_1,t_2\}$ has the largest number of  parallel edges among all neighbors of $t_1$. 
\Comment{Note that $t_1,t_2$ are not necessarily vertices of $G$, so parallel edges between them do not correspond to parallel edges of $G$.}
\strut}
\State \parbox[t]{\dimexpr\linewidth-\algorithmicindent}{
 Make $t^*$ the parent of $t_1,t_2$;  remove $t_1,t_2$ from $W$, and add $t_1$ to $T$. 
\Comment{So, $V(t^*)=V(t_1)\cup V(t_2)$.} 
\strut}
\EndWhile
\Return $\cT$.
\end{algorithmic}
\caption{Construction of a \expandertree~for planar graphs.}
\label{alg:minorfree}
\end{algorithm}
Observe that the algorithm terminates after exactly $n-1$ iterations of the loop, because 
 any non-leaf node of $\cT$ has exactly two children,
 so $|W|$ decreases by 1 in each iteration.
 We show that $\cT$ is $\cT(k/5,1/5,T)$-LCH.
First of all, for any non-leaf node $t$ of $\cT$, $G(t)$ is $k/5$-edge-connected. We prove this by induction. Say, $t_1,t_2$ are the two children of $t^*$, and by induction, $G(t_1)$ and $G(t_2)$ are $k/5$-edge-connected. By the selection of $t_2$, there are at least $k/5$ parallel edges between $t_1,t_2$, so $G(t^*)$ is $k/5$-edge-connected. 
Secondly, we need to show that $\cut(t_1) \geq \setdeg(t_1)/5$. This is because by the selection of $t_2$, $1/5$ of the edges incident to $t_1$ in $G_{t^*}$ are $\{t_1,t_2\}$.
This completes the proof of \autoref{lem:planarLCH}.

\subsection{Main Technical Theorem}
Given a $(k,\lambda,T)$-LCH $\cT$ of $G$, in our main technical theorem we minimize the maximum average effective resistance of $\cut(t)$'s among all nodes $t\in T$.
%In our main technical theorem, we show that there is a matrix $D$ that reduces the average effective resistance of edges in all sets $\cut(t)$ to $\tilde{O}(1/k)$. 

The following convex program finds a shortcut matrix $0\prec D\preceq L_G$ that minimizes the maximum of the average effective resistance of edges in $\cut(t)$ for all $t\in T$.
\begin{table}[htb]
\centering
\begin{tabular}{l}
%\begin{equation}
{\bf \maincp}($\cT\in (k,\lambda,T)$-LCH): \\
$\begin{aligned}
\min ~~~& ~~{\cal E}\\
\st ~~~&  
%\sum_{e\in \cut(t)} \frac{1}{|\cut(t)|} \cdot 
\uE{e\sim \cut(t)}
%\X_e^\intercal D^{-1} \X_e 
\reff_D(e)
\leq \cE & \forall t\in T, \\
& D\preceq_{\square} L_G, &\\ 
%& \bone_S^\intercal D \bone_S \leq \bone_S^\intercal L_G \bone_S & \forall \emptyset \subset S\subset V\\
& D \succ 0. & 
\end{aligned}$
%\end{equation}
\end{tabular}
\end{table}
%We emphasize that, unlike classical approximation algorithms, the above program is not a {\em relaxation} of the thin tree problem. 
%Note that each edge $e\in E$ is contained in at most two sets $\cut(t_1),\cut(t_2)$. So, 

%The following is our main technical theorem.
\begin{theorem}[Main Technical]
\label{thm:main}
For any $k$-edge-connected graph $G$, and any $\cT(k,\lambda,T)$-LCH, $\cT$, of $G$,  there is a PD shortcut matrix $D$
such that  for any $t\in T$,
$$ \uE{e\sim \cut(t)}
%\X_e^\intercal D^{-1} \X_e 
\reff_D(e)
\leq \frac{f_1(k,\lambda)}{k},$$
where $f_1(k,\lambda)$ is a poly-logarithmic function of $k,1/\lambda$. 
\end{theorem}
Note that the statement of the above theorem does not have any dependency on the size of $G$. 

If we apply the above theorem to the $(k/5,1/5,T)$-LCH $\cT$ of a $k$-edge-connected planar graph as constructed in \autoref{alg:minorfree}, we obtain a shortcut matrix $D$ for which the small effective resistance edges are $\Omega(k)$-edge-connected.
Let us elaborate on this. 
Let $F=\{e: \reff_D(e) \leq \frac{2f_1(k/5,1/5)}{k/5}\}$.
First, note that by \autoref{lem:planarLCH}, $\cT$ is a binary tree and at least one child of each internal node of $\cT$ is in $T$. Say $t$ is an internal node with children $t_1,t_2$ and $t_1\in T$. Then, by Markov's inequality 
$$ |F\cap \cut(t_1)| \geq |\cut(t_1)|/2 \geq \frac{k/5}{2}.$$
Since $t$ has only two children, this implies $G(V\{t\}, F\cap E\{t\})$ is $k/10$-edge-connected. Now, by \autoref{lem:LCHinternalconnectivity},  $(V,F)$ is $k/10$-edge-connected. 
%Assuming $G(t_1),G(t_2)$ are $k/10$-edge-connected, the above inequality implies $G(t)$ is also $k/10$-edge-connected.

It is natural to expect that for every $k$-edge-connected graph $G$, one can find a \expandertree~$\cT$ such that one application of the above theorem produces a set $F$ of good edges such that for any $t\in \cT$, $G(V\{t\}, F\cap E\{t\})$ is $\Omega(k)$-edge-connected. 
By \autoref{lem:LCHinternalconnectivity} this would imply $(V,F)$ is $\Omega(k)$-edge-connected.
However, the following example shows that this may not be the case.

\begin{example} \label{ex:hypercube} Let $G=(V,E)$ be the $k$-dimensional hypercube ($n=2^k$). Note that $G$ is $k$-edge-connected. Let $\cT$ be a $(\Omega(k), ., .)$-LCH for $G$. Consider an internal node $t_0\in\cT$, all of whose children are leaves. By definition $G(t_0)$ is $\Omega(k)$-edge-connected. Consider a dimension cut of the hypercube that cuts $G(t_0)$ into $(S, V(t_0)\setminus S)$. Imagine a solution $D$ of $\maincp(\cT)$ which reduces the effective resistance of all edges except those in the cut $(S, V(t_0)\setminus S)$. In such a solution, $\mE_{e\sim \cut(t)}\reff_D(e)$ is small for all $t$. This is because each vertex $v\in G(t)$ has at most one of its $\Omega(k)$ neighboring edges in the cut $(S, V(t_0)\setminus S)$. But note that the small effective resistance edges are disconnected in $G\{t_0\}=G(t_0)$.
\end{example}

Consider a $(\Omega(k),.,.)$-LCH $\cT$ of $G$ and let $t$ be an internal node. \autoref{thm:main} promises that the average effective resistance of all degree cuts of the internal graph $G\{t\}$ are small. If $G\{t\}$ is an \emph{expander} this implies that the good edges are $\Omega(k)$-edge-connected in $G\{t\}$. Therefore, if we can find a \expandertree~whose internal subgraphs are expanding we can find an $\Omega(k)$-edge-connected set of good edges by a single application of \autoref{thm:main}. This is exactly what we proved in the case of planar graphs. The above hypercube example shows that such a \expandertree~does not necessarily exist in all $k$-edge-connected graphs.

\subsection{Expanding \Expandertrees}
In this section we define expanding \expandertrees~and we describe our plan to prove \autoref{thm:hierarchydecomgeneral} using the main technical theorem. 
\begin{definition}[Expanding \Expandertrees]\label{def:expandinghierarchy}
For a node $t$ with children $t_1,\dots,t_j$ in a \expandertree~$\cT$ of a graph $G=(V,E)$,   
an internal node $t$ (or the internal subgraph $G\{t\}$) is called $(\expansion, \beta)$-expanding, if $G\{t\}$ is an $\expansion$-expander and is $\beta$-edge-connected. A subset of the nodes $T$ is called $(\expansion,\beta)$-expanding iff each one of them is $(\expansion, \beta)$-expanding and similarly the \expandertree, $\cT$, is $(\expansion,\beta)$-expanding iff all of its nodes are $(\expansion,\beta)$-expanding.
\end{definition}
Recall that \expandertrees~already guarantee $k$-edge-connectivity of the internal subgraphs for some $k$. So, we always have $\beta \geq k$. If $\beta=k$, we omit it from the notation and write $(\expansion, .)$-expanding; otherwise, the $(\expansion, \beta)$-expanding property guarantees slightly stronger connectivity for a \emph{subset} of the internal subgraphs.

For  example, observe that the \expandertrees~that we constructed in \autoref{alg:minorfree} for $k$-edge-connected planar graphs are $(1,k/5)$-expanding.
In \autoref{thm:constructexpandertree} we construct an $(\Omega(1/k),\Omega(k))$-expanding $(\Omega(k),\Omega(1),\cT)$-LCH for any $k$-edge-connected graph where $k\geq 7\log n$.
But \autoref{ex:hypercube} shows that this is essentially the best possible, as the $k$-dimensional hypercube does not have any $(\omega(1/k),\Omega(k))$-expanding \expandertree. 

It follows that if $G$ has an $(\alpha,\Omega(k))$-expanding \expandertree~then there is a shortcut matrix $D$ and an $\Omega(k)$-edge-connected set $F$ of edges such that
$$ \max_{e\in F}\reff_D(e) \leq  O(\maincp(\cT)/\alpha).$$
Recall the argument in \autoref{ex:cycleexpander} for details.
Since the best $\alpha$ we can hope for is $O(1/\log n)$ this argument by itself does not work.

Our approach is to apply \autoref{thm:main} to an adaptively chosen sequence of \expandertrees. Each time we recognize the internal subgraphs of the  \expandertree~in which the set of good edges found so far are not $\Omega(k)$-edge-connected. Then, we apply \autoref{thm:main} to the nodes in these internal subgraphs. We ``refine'' these internal subgraphs by a natural decomposition of the newly found good edges to get the next \expandertree.
At the heart of the argument we show that this refinement procedure improves the expansion of the aforementioned internal subgraphs by a constant factor. Therefore, this procedure stops after $O(\log(1/\alpha))=\polyloglog(n)$ steps in the worst case.

We conclude this section by describing an instantiation of the above procedure in the special case of a $k$-dimensional  hypercube for demonstration purposes. Let $G$ be a $k$-dimensional hypercube. We let $\cT_1$ be a star, i.e., it has only one internal node and the vertices of $G$ are the leaves. This means that in $\maincp(\cT_1)$ we minimize the maximum average effective resistance of degree cuts of $G$. Let $F_1$ be the edges of effective resistance at most twice the optimum of $\maincp(\cT_1)$. It follows that half the edges incident to each vertex are in $F_1$. Now, we find a natural decomposition of the good edges $F_1$. In the ``worst case'', edges of $F_1$ form $k/2$ dimensional sub-hypercubes and all edges connecting these sub-hypercubes are not in $F_1$. Note that if we contract these sub-hypercubes, we get a $k/2$-dimensional hypercube which is a $2/k$-expander, twice more expanding than $G$. Of course, we cannot contract, because we need good edges having small effective resistance with respect to the original vertex set, but the expansion is our measure of progress. 

In the next iteration, we construct a $(.,.,T_2)$-LCH $\cT_2$ where the vertices of each $k/2$-dimensional sub-hypercube are connected to a unique internal node, and the root is connecting  all internal nodes, i.e., $\cT_2$ has height 2. We let $T_2$ be the set of all internal nodes (except the root). Note that if we delete the leaves, then $\cT_2$ would be the same as $\cT_1$ for a $k/2$-dimensional sub-hypercube. Similarly, we solve $\maincp(\cT_2)$, and in the worst case the new good edges form $k/4$ dimensional sub-hypercubes. Continuing this procedure after $\log(k)=\log\log n$ iterations  the good edges span an $\Omega(k)$-edge-connected subset of $G$.

In the next section, we will use expanding \expandertrees~to prove the main theorem \ref{thm:hierarchydecomgeneral} using the main technical theorem \ref{thm:main}. In the remaining sections we will prove the main technical theorem \ref{thm:main}.

%\input{spectral-to-combinatorial}
% !TEX root = main.tex

\section{Proof of the Main Theorem}
\label{sec:partitioning}
In this section we prove our main theorem \ref{thm:hierarchydecomgeneral} assuming the main technical theorem \ref{thm:main}.
Lastly, we will prove the algorithmic theorem \ref{thm:algorithmicatsp}.
First, in \autoref{subsec:constructtree} we show that
for $k\geq 7\log n$, any $k$-edge-connected graph has a 
$(1/k,.)$-expanding $(k/20,1/4,\cT)$-LCH.
Then, in \autoref{subsec:extractsubgraph}, we show that if a given graph $G=(V,E)$ has an $(\expansion, .)$-expanding $(k, ., .)$-LCH, then there exists a PD shortcut matrix $D$, and an $\Omega(k)$-edge-connected subset $F$ of good edges, such that for any $e\in F$, 
$$\reff_D(e)\leq \frac{\polylog(k,1/\expansion)}{k}.$$
\subsection{Construction of \Expandertrees}
\label{subsec:constructtree}
In this section, we prove the following theorem. We remark that this is the only place in the entire paper where we depend on $k$ being $\Omega(\log(n))$.

\begin{theorem}
\label{thm:constructexpandertree}
Given a $k$-edge-connected graph $G$, with $k\geq 7\log(n)$, one can construct a $(\frac{1}{k}, .)$-expanding $(\frac{k}{20}, \frac{1}{4}, \cT)$-LCH $\cT$.
\end{theorem}

The proof of the theorem will be an adaptation of the proof for the special case of $k$-edge-connected planar graphs that we saw in \autoref{lem:planarLCH}. Given a graph $G$, we iteratively find $\Omega(k)$-edge-connected $\Omega(1/k)$ induced expanders, i.e., a set $S\subseteq V$ where  $G[S]$  is $\Omega(k)$-edge-connected and $\phi(G[S])\geq \Omega(1/k)$. We also need to make sure that $G[S]$ satisfies the following definition to ensure that we get a $(.,\lambda,.)$-LCH.
\begin{definition}
An induced subgraph $H$ of an unweighted graph $G=(V,E)$ is $\lambda$-dense if for any $v\in V(H)$,
$$ d_H(v) \geq \lambda\cdot d_G(v),$$
where we use $V(H)$ to denote the vertex set of $H$.
\end{definition}
The following proposition is the main technical statement that we need for the proof.
\begin{proposition}\label{prop:expsizedenseexpander}
Any $k\geq 7\log n$-edge-connected graph $G=(V,E)$ (with $n$ vertices) has an induced $k/20$-edge-connected, $1/4$-dense subgraph $G[S]$ that is an $1/k$-expander. 
\end{proposition}
Note that for every edge $\{u,v\}\in E$, the induced graph $G[\{u,v\}]$ is a $1$-expander. But, if there is only one edge between $u,v$ in $G$, then this induced graph is only $1$-edge-connected and $O(1/k)$-dense. 
It is instructive to compare the statement of the above proposition to the planar case. Recall that \autoref{fact:planardegree} asserts that in any $k$-edge-connected planar graph there is a pair of vertices with $k/5$ parallel edges. Such an induced graph is a $k/5$-edge-connected $1$-expander.
Of course, this fact does not necessarily hold for a general $k$-edge-connected graph as $G$ may not have any parallel edges at all. 

Note that, in the above proposition, the condition $k\geq 7\log n$ is necessary up to a  constant; a tight example  is the $\log n$-dimensional hypercube, which  is a $k$-edge-connected for any $k\leq \log n$, but every $\Omega(1)$-dense induced subgraph is no better than $O(1/\log n)$-expanding.

%We cannot  directly apply the same construction to general $k$-edge-connected graphs, because we cannot use \autoref{fact:planardegree}. 

%Instead of finding $\Omega(k)$-connected $1$-expanders   (i.e. parallel edges), which do not necessarily exist, we find  $\Omega(k)$-connected $\Omega(1/k)$-expanders, assuming that $k\geq 7\log(n)$. In the following preposition, we show that these subgraphs always exist. Subsequently, we use it to prove \autoref{thm:constructexpandertree}.
We use proof by contradiction. Suppose $G$ does not have any induced subgraph satisfying  the statement of the proposition. Then, invoking the following lemma with $H=G$ and $\phi^*=1/k$, we obtain that $G$ must have more than $2^{3k/20}$ vertices. But this contradicts the fact that $k\geq 7\log n$.
%We use the following lemma to prove the above proposition.
\begin{lemma}
\label{lem:expsizedenseexpander}
Given a $k$-edge-connected graph $G$,
if every $k/20$-edge-connected $1/4$-dense subgraph $G[S]$ of $G$ satisfies $\phi(G[S]) < \phi^*$, then
for any induced subgraph $H$ of $G$,
$$ \log_2(|V(H)|) \geq \frac{3/10 - \phi_G(V(H))}{2\phi^*}.$$
\end{lemma}

\begin{proof}
%Let $n:=|V(G)|$.
We prove the lemma by induction on the number of vertices of $H$. 
%Suppose the claim holds for any graph $G$ with less than $n$ vertices. 
Fix an induced subgraph $H=G[U]$. Without loss of generality, assume that $\phi_G(U) < 3/10$. We consider two cases, and in the end we show that one of them always happens.

\paragraph{Case 1: There is a vertex $v\in U$ such that $d_H(v) \leq 7d_G(v)/20$.}
%This means that $H$ has a vertex $v$ such that $\partial_H(v) < \partial_G(v)/3$. 
We show that $\phi_G(U)$ decreases when we remove $v$ from $U$.
%By the assumption, $\partial_H(v) \leq 
$$ \phi_G(U) = \frac{\deg_G(U\setminus \{v\}) + d_G(v) - 2d_H(v)}{\vol_G(U\setminus \{v\}) + d_G(v)} \geq \frac{\deg_G(U\setminus \{v\}) + 6d_G(v)/20}{\vol_G(U\setminus \{v\}) + d_G(v)} \geq \phi_G(U\setminus \{v\})$$
The last inequality uses that $\phi_G(U)<3/10$.
By induction, 
$$ \log_2(|U|) \geq \log_2(|U-\{v\}|) \geq \frac{3/10-\phi_G(U-\{v\})}{2\phi^*}\geq \frac{3/10-\phi_G(U)}{2\phi^*},$$
and we are done.
Note that if this case does not happen, then $H$ is $\frac{7}{20}$-dense in $G$. 

\paragraph{Case 2: For some $S\subset U$,  $\max\{\phi_H(S), \phi_H(U\setminus S)\} < \phi^*$.}
%Let $x:=d(S)/|S|$, $y:=d(\overline{S})/|\overline{S}|$ and $\alpha:=|S|/n.$
%Then,
%$$D = \frac{d(S) + d(\overline{S})}{n} = \alpha x + (1-\alpha) y.$$
Let $T:=U\setminus S$. 
Observe that if $\phi_G(S)\leq \phi_G(U)$ or $\phi_G(T)\leq \phi_G(U)$, then
%that if the average degree of the induced subgraph of $G[S]$ or $G[\overline{S}]$
%is at least $D$ then 
we are done by induction. So assume that none of the two conditions hold. We show
that $\phi_G(S), \phi_G(T) \leq \phi_G(U)+ 2\phi^*$.

First, it follows from
$$ \phi_G(U) = \frac{\deg_G(S) + \deg_G(T) - 2\deg_H(T)}{\vol_G(S) + \vol_G(T)}$$
and $\frac{\deg_G(S)}{\vol_G(S)} = \phi_G(S) > \phi_G(U)$ that 
\begin{equation}
\label{eq:phiTupper1}
 \phi_G(U) > \frac{\deg_G(T) - 2\deg_H(T)}{\vol_G(T)} = \phi_G(T) - 2\frac{\deg_H(T)}{\vol_G(T)} \geq \phi_G(T) - 2\phi_H(T) .
 \end{equation}
 Therefore, $\phi_G(T) \leq \phi_G(U) + 2\phi^*$.
Similarly, we can show $\phi_G(S) \leq \phi_G(U)+2\phi^*$.
So, by induction,
$$ \log_2(|U|) =  \log_2(|S| + |T|) \geq 1+\log_2(\min\{|S|,|T|\}) \geq 1+\frac{3/10-\phi_G(U) - 2\phi^*}{2\phi^*} = \frac{3/10-\phi_G(U)}{2\phi^*}.
$$
%2^{\frac{\ln(D(1-2\phi^*)/2k)}{4\phi^*}}$$
%So, we need to show that
%$$ 1 + \frac{\ln(D(1-2\phi^*)/2k}{4\phi^*} \geq \frac{\ln(D/2k)}{4\phi^*}.$$
%Or, equivalently, $1 \geq -\ln(1-2\phi^*)/4\phi^*$.
%The latter holds simply because $\phi^*<0.3$.

\medskip

We  now show that one of the above cases (Case 1 and Case 2) need to happen. Suppose towards contradiction that none of the above cases happens.  Then $H$ is $7/20$-dense and for all $S\subset U: \max \{\phi_H(S), \phi_H(U\setminus S)\} \geq \phi^*$. In other words, $\phi(H) \geq \phi^*$. Therefore, by the assumption of the lemma, there must be a set $S \subset U$ such that $\partial_H(S) < k/20$ (we can also assume that $\phi_H(S) \geq \phi_H(U\setminus S)$, otherwise just take the other side). We now  show that this cannot happen.

%\paragraph{Case 3: For some $S\subset U$, $\deg_H(S) < k/20$ but $\phi_H(S) > \phi^*$.}
%Without loss of generality we assume $D>2k$ (otherwise there is nothing to prove). 
%We show that this case does not happen.
Note that $H$ is $7/20$-dense in $G$, so for each $v\in U$,
\begin{equation}\label{eq:Hdegklowerbound} 
	d_H(v) \geq 7d_G(v)/20 \geq 7k/20,
\end{equation}
where we used the $k$-edge-connectivity of $G$. 
%First, assume that $G$ has a vertex $v$ of degree $\vol(v) < k$. Since $D>2k$ if we delete $v$ from $G$, the remaining graph has a larger average degree and we are done by induction.

%Otherwise let $\vol(S) \leq \vol(\overline{S})$. 
%We reach a contradiction showing that $G$ must have a vertex of degree less than $k/4$.
We start with a natural decomposition of the induced graph $G[S]$ into $k/20$-edge-connected subgraphs, $S_1,\ldots,S_\ell$, as defined in \autoref{def:decompkconnected}.
%Observe that this means that we have partitioned $G[S]$ exactly $\l-1$ times. 
We show that for each $i$, $\deg_H(S_i) \geq k/10$. This already gives a contradiction, because
by \autoref{lem:decompkconnected}
\begin{eqnarray}
  \frac{k}{20} + 2(\ell-1) \frac{k}{20} & > & \deg_H(S) + \sum_{i=1}^\ell \deg_{G[S]}(S_i) \nonumber \\
  & = & \sum_{i=1}^\ell \deg_H(S_i) \geq  \ell \cdot \frac{k}{10}.
  \label{eq:expandercase3}
\end{eqnarray}

It remains to show that $\deg_H(S_i) \geq k/10$. For the sake of the contradiction, suppose that
$\deg_H(S_i) < k/10$ for some $i$. First, observe that $S_i$ cannot be a singleton, because the induced degree of each vertex of $H$ is at least $7k/20 > k/10$. We reach a contradiction by showing that
$G[S_i]$ is a $1/4$-dense, $k/20$-edge-connected induced subgraph of $G$ with expansion $\phi(G[S_i]) \geq \phi^*$.
By definition, $G[S_i]$ is $k/20$-edge-connected.
Next, we show $G[S_i]$ is dense. For every vertex $v\in S_i$,
$$ d_{G[S_i]}(v) \geq d_H(v)  - \deg_H(S_i) \geq d_H(v) - k/10 \geq \frac{7d_G(v)}{20} - \frac{d_G(v)}{10} \geq d_G(v)/4,$$
where the third inequality uses \eqref{eq:Hdegklowerbound}. Therefore $G[S_i]$ is $1/4$-dense.
%such that $v$ sends at least $3 \vol(v)/4$ edges to out of $S_i$. But since $\deg_H(S_i) < k/10$, $v$ sends at least $(3/4-1/10)\vol(v)$ edges to outside of $H$. But this is a contradiction with the assumption that every vertex of $H$ sends less than $7\vol(v)/10$ fraction of its edges to outside of $H$.

Finally, we show that $G[S_i]$ is a $\phi^*$-expander. This is because for any set $T\subseteq S_i$,
$$ \phi_{G[S_i]}(T) \geq \frac{\deg_{G[S_i]}(T)}{d_H(T)} \geq \frac{k/20}{d_H(T)} \geq \frac{\deg_H(S)}{d_H(S)} = \phi_H(S) \geq \phi^*.$$
Therefore, $G[S_i]$ is a $k/20$-edge-connected, $1/4$-dense and $\phi^*$-expander, which is a contradiction. So, $\deg_H(S_i)\geq k/10$, which gives a contradiction by \eqref{eq:expandercase3}.
\end{proof}
This completes the proof of \autoref{prop:expsizedenseexpander}.
We are now ready to prove \autoref{thm:constructexpandertree}.
The details of our construction are given in \autoref{alg:consthierarchy}.

%\begin{proof}
\begin{algorithm}[htb]
\begin{algorithmic}[1]
\Input A $k$-edge-connected graph $G=(V,E)$ where $k\geq 7\log(n)$.
\Output A $(1/k,.)$-expanding $(k/20,1/4,\cT)$-LCH $\cT$ of  $G$.
\State For each vertex $v\in V$, add a unique (leaf) node to $\cT$ and map $v$ to it. Let $W$ be the set of these leaf nodes. \Comment{Throughout the algorithm, we keep the invariant that $W$ consists of the nodes of $\cT$ that do not have a parent yet, but their corresponding subtree is fixed, i.e., $V(t)$ is well-defined for any $t\in W$.}
\While {$|W| >1$}
\State Add a new node $t^*$ to $W$. \label{step:constexpandernewtstar}
\State \parbox[t]{\dimexpr\linewidth-\algorithmicindent}{
Let $G_{t^*}$ be the graph where for each node $t\in W$, $V(t)$ is contracted to a single vertex, and identify $t$ with the corresponding contracted vertex.
\Comment{$G_{t^*}$ is $k$-edge-connected by \autoref{fact:kconinvariant}.} \strut}
\State \parbox[t]{\dimexpr\linewidth-\algorithmicindent}{
Let $H_{t^*}=G_{t^*}[U_{t^*}]$ be the $k/20$-edge-connected, $1/4$-dense $1/k$-expanding induced subgraph of $G_{t^*}$ promised by \autoref{prop:expsizedenseexpander}. 
\label{step:useofdensesubgraph}
%\Comment{By \autoref{prop:expsizedenseexpander}, $H_{t^*}$ always exists because $k> 20\log(n)/3$.} 
\strut}
%Add a new node $t^*$ to $T$ that corresponds to $H_i$.
\State \parbox[t]{\dimexpr\linewidth-\algorithmicindent}{
Let $W=W\setminus U_{t^*}$,  and make $t^*$ the parent of all nodes in $U_{t^*}$. \Comment{So, $V(t^*)=\cup_{t\in U_{t^*}} V(t)$ and $G\{t^*\}=H_{t^*}$.}
\strut}
%\State If $|T|=1$, {\bf return} $\cT$.
\EndWhile
\Return $\cT$.
\end{algorithmic}
\caption{Construction of an \expandertree~for a $7\log(n)$-edge-connected graph.}
\label{alg:consthierarchy}
\end{algorithm}

%Note that in step \ref{step:useofdensesubgraph}
%we crucially use the fact that $k\geq 7\log(n)$. 
First of all, observe that the algorithm always terminates in at most $n-1$ iterations of the loop, because in each iteration $|W|$ decreases by at least 1.
The properties of $H_{t^*}$ in step \ref{step:useofdensesubgraph} translate to the properties of $\cT$ as follows: 
\begin{itemize}
\item $1/k$-expansion of $H_{t^*}$ guarantees that $\cT$ is $(1/k, .)$-expanding. 
\item The $k/20$-edge-connectivity of $H_{t^*}$ implies that $\cT$ is $(k/20,.,.)$-LCH.
\item Finally, the fact that $H_{t^*}$ is $1/4$-dense with respect to $G_{t^*}$ implies that $\cT$ is $(.,1/4,\cT)$-LCH.
\end{itemize}
%Therefore, it suffices to see that the output of the algorithm, $\cT$, is a $(k/20,1/4,\cT)$-LCH of $G$.
%Fix a node $t^*$ constructed in step \ref{step:constexpandernewtstar}. We show that  $G(t^*)$ is $k/20$-edge-connected.
%and that $|\cut(t^*)| \geq \frac14\cdot|\setdeg(t^*)|$. 
%Suppose, by induction, that for any node $t'$ in $G\{t\}$, $G(t')$ is $k/20$-connected. Then, since $G\{t\}$ is $k/20$-connected, $G[t]$ is $k/20$-connected. Now, suppose $t^*$ is the parent of $t$. Then, $\cut(t)$ is simply the edges adjacent to $t$ in $G\{t^*\}$. Since $G\{t^*\}$ is $1/4$-dense with respect to $G_{t^*}$, $|\cut(t)|\geq \frac14\cdot|\setdeg(t)|$. This completes the proof of \autoref{thm:constructexpandertree}.

\subsection{Extraction of an $\Omega(k)$-Edge-Connected Set of Good Edges}
\label{subsec:extractsubgraph}
In this part we prove the following theorem.
\begin{theorem}
\label{thm:expandertreeboosting}
If $G=(V,E)$ has an $(\expansion,.)$-expanding  $(k,\lambda,\cT)$-LCH, then there exists a PD shortcut matrix $D$, and a $k/4$-edge-connected set $F$ of good edges such that
$$ \max_{e\in F}\reff_D(e) \leq \frac{f_2(k,\lambda,\expansion)}{k}, $$
where $f_2(k,\lambda,\expansion)=f_1(k,\lambda\alpha)\cdot O(\log(1/\expansion))$. 
\end{theorem}
\noindent The main theorem of the paper , \autoref{thm:hierarchydecomgeneral}, follows from the above theorem together with \autoref{thm:constructexpandertree}.

Let $\cT$ be the $(\expansion,.)$-expanding $(k,\lambda,\cT)$-LCH given to us. First, observe that it is very easy to prove a weaker version of the above theorem where 
$$\reff_D(e)\leq \frac{2 f_1(k,\lambda)}{k\cdot \expansion}$$ 
for edges of $F$ by a single application of \autoref{thm:main}. Let $D$ be the optimum of $\maincp(\cT)$; we let $F\subseteq E$ be the edges where $\reff_D(e) \leq \frac{2 f_1(k,\lambda)}{k\cdot \expansion}$. Let $G'=(V,F)$. It follows that for any node $t$ of $\cT$, $G'\{t\}$ is $k/2$-edge-connected, so by \autoref{lem:LCHinternalconnectivity} $G'$ is $k/2$-edge-connected and we are done. 
%Note that, if $G$ were a planar graph, the above theorem could be proved by a single application of the main technical theorem. 

The main difficulty in proving the above theorem is to reduce the inverse polynomial dependency on $\expansion$ in the above argument to a polylogarithmic function of $\expansion$.
To achieve that, we apply \autoref{thm:main} to $\log(1/\expansion)$  
 \expandertrees, $\cT_0,\dots,\cT_{\log(1/\alpha)}$, of our graph. For each $\cT_i$, $W_i$ is the set of bad internal nodes of $W_{i-1}$, i.e., those where their internal subgraph is not yet $\Omega(k)$-edge-connected with respect to the good edges found so far. Originally, $W_0$ contains all internal nodes of $\cT_0$ and it is a $(1/k,.)$-expanding set. For each $i$, we will make sure that $\cT_i$ is $(k/4,\lambda\alpha^i, T_i)$-LCH and $W_i$  is $(2^i\expansion,k)$-expanding.
 In other words,
   each $\cT_i$ is a ``refinement'' of $\cT_{i-1}$ whose $W_i$ nodes are twice more expanding.

\begin{algorithm}[tb]
\begin{algorithmic}[1]
\Input A graph $G=(V,E)$ and a $(\expansion,.)$-expanding $(k,\lambda,\cT)$-LCH $\cT$.
\Output A PD shortcut matrix $D$ and a $k/4$-edge-connected set $F$ of good edges. 
%and $\reff_D(e)\leq \beta f_2(k,\lambda,\expansion)/k$ for all $e\in F$.
\State Let $W_0$ be all internal nodes of $\cT$, $W_i=\emptyset$ for $i>0$, and $\cT_0=\cT$, and $G'=(V,\emptyset)$. 
\For {$i=0\to \log(1/\expansion)$}
\State Let $D_i$ be the optimum of $\maincp(\cT_{i})$.
\State \parbox[t]{\dimexpr\linewidth-\algorithmicindent}{
Say $\cT_i$ is a $(k',\lambda',T_i)$-LCH of $G$; let 
\begin{equation}
\label{def:FireffDi}
F_i:=\left\{e\in E: \reff_{D_i}(e) \leq \frac{16 f_1(k',\lambda')}{k'}\right\}, 
\end{equation}
add all edges of $F_i$ to $G'$.
\strut}
\State \parbox[t]{\dimexpr\linewidth-\algorithmicindent}{
For any node $t\in W_i$, let $S_{t,1},\ldots,S_{t,\ell(t)}$ be
a natural decomposition of $G'_{\cT_i}\{t\}$ into $k/4$-edge-connected components as defined in \autoref{def:decompkconnected}.
% and assume that  $S_{t,\l(t)}$ has the largest volume in $G_{\cT_i}\{t\}$ among all sets in the partition.
If $\ell(t) > 1$, then we add $t$ to $W_{i+1}$. 
\Comment{Note that if $\l(t)=1$ it means that $G'\{t\}$ is $k/4$-edge-connected.}
\strut}
\State \parbox[t]{\dimexpr\linewidth-\algorithmicindent}{
We construct a $(.,.,T_{i+1})$-LCH of $G$, called $\cT_{i+1}$,  by modifying   $\cT_{i}$. For any node $t\in W_{i+1}$  we add $\ell(t)$ new nodes $s_{t,1},\dots,s_{t,\ell(t)}$ to
$\cT_{i+1}$ and we make all nodes of $S_{t,j}$ children of $s_{t,j}$ and we make $t$ the parent of $s_{t,j}$. Therefore, $t$ has exactly $\ell(t)$ children in $\cT_{i+1}$. See \autoref{fig:cT_2} for an example. The set $T_{i+1}$ is the union of all nondominating nodes children of all nodes of $W_i$. %For any node $t\in W_{i+1}$, we  add $s_{t,1},\dots,s_{t,\ell(t)-1}$ to $T_{i+1}$. If $d_{G_{\cT_{i}}\{t\}}(S_{t,\ell(t)}) \leq \frac12 d_{G_{\cT_{i}}\{t\}}(V_{\cT_{i}}\{t\})$ then we also add $s_{t,\ell(t)}$ to~$T_{i+1}$. 
\strut}
\EndFor
\Return the PD shortcut matrix $\mE_i D_i$ and the good edges $\cup_i F_i$.
\end{algorithmic}
\caption{Extracting Small Effective Resistance Edges}
\label{alg:expandertreeboosting}
\end{algorithm}
Throughout the algorithm we also make sure that all (except possibly one) children of each node in $W_i$ are in $T_i$.
Let us elaborate on this statement. Let $t\in W_i$ and let $t_0,t_1,\dots$ be the children of $t$. Since $W_i\subseteq W_0$, $t\in W_0$. Consider the graph $G_{\cT_0}\{t\}$; by the theorem's assumptions $G_{\cT_0}\{t\}$ is a $k$-edge-connected $\expansion$-expander. 
It follows that $G_{\cT_i}\{t\}$ can be obtained from $G_{\cT_{0}}\{t\}$ by contracting a set $U_{t_j}\subset V_{\cT_0}\{t\}$ corresponding to each children $t_j$ of $t$. 
We use the notation 
$$\vol_0(t_j)=\sum_{t'\in U_{t_j}} d_{G_{\cT_0}\{t\}}(t')$$ 
to denote the sum of the degrees of nodes in $S_{t_j}$ in the noncontracted graph $G_{\cT_0}\{t\}$.
We say a child $t_\l$ of $t$ is {\em dominating} if 
$$\vol_0(t_\l) > \frac12 \sum_{j} \vol_0(t_j).$$
It follows that each node $t\in W_i$ can have at most one dominating child. 
In addition, if $t_\l$ is a dominating child, it may not satisfy $\cut(t_\l) \gtrsim \setdeg(t_\l)$, so we may not add $t_\l$ to $T_i$. Because of this we need to treat the dominating children (of nodes of $W_i)$ differently throughout the algorithm and the proof. 
In our construction $T_i$ consists of all nondominating children of all nodes of $W_i$. %For any node $t\in W_i$ we make sure that all nondominating children of $t$ are in $T_i$. 
It is easy to see that for any nondominating child $t_\l$ of $t\in W_i$,
$$ \cut_{\cT_i}(t_\l) = \deg_{G_{\cT_0}\{t\}}(U_{t_\l}) \geq \expansion  \cdot \vol_{G_{\cT_0}\{t\}}(U_{t_\l}) = \expansion\cdot \sum_{t'\in U_{t_\l}} \cut_{\cT_0}(t') \geq \expansion\cdot\lambda\cdot \sum_{t'\in U_{t_\l}} \setdeg_{\cT_0}(t') \geq \expansion\lambda \setdeg_{\cT_i}(t_\l), $$
where the first inequality uses the fact that $G_{\cT_0}\{t\}$ is an $\expansion$-expander and the second inequality uses the fact that $\cT_0$ is
$(.,\lambda,\cT_0)$-LCH.
The following claim is immediate
\begin{claim}
If $\cT_0$ is an $(\expansion,.)$-expanding $(.,\lambda,\cT_0)$-LCH, then for  any $i\geq 1$, $\cT_i$ is a $(.,\lambda\expansion,T_i)$-LCH, where $T_i$ consists of all nondominating children of the nodes of $W_i$, .
\end{claim}

%By the assumption of the theorem $G_{\cT_0}\{t\}$ is $k$-edge-connected and $\alpha$-expander. 
% In particular, for a node $t\in W_i$ with children $\{t_0,t_1,\dots,t_j\}$ we say $t_0$ is a dominating child if
 %$\cut(t_0)>\frac12 \sum_i \cut(t_i)$, i.e., the number of edges
% In particular, if a node $t\in W_i$ has a child $t'$ such that $\vol_{G\{t\}}(

At the end of the algorithm, we obtain PD shortcut matrices
$D_0,\dots,D_{\log(1/\expansion)}$ and sets $F_0,\dots,F_{\log(1/\expansion)}$ such that the edges of each $F_i$
 have small effective resistance with respect to $D_i$, and $\cup_{i=0}^{\log(1/\expansion)} F_i$ is $\Omega(k)$-edge-connected.
Then, we let $D$ be the average of $D_0,\dots,D_{\log(1/\expansion)}$ and $F$ be the union of $F_0,\dots,F_{\log(1/\expansion)}$.
The details of the construction of these matrices and sets are given in \autoref{alg:expandertreeboosting}.

%Let $W_i$ be the set $W$ after iteration $i$ of the main loop,
%and $W_0$ be the non-leaf nodes of $\cT_0$. 
We prove the claim by induction on $i$.
In the first step we show $\cT_{i+1}$ is a $(k/4,.,.)$-LCH.
Then, we show that $W_{i+1}$ is $(2^{i+1}\alpha,k)$-expanding.
Then, we show that $W_{\log(1/\alpha)}$ is empty and we conclude by showing that $G'=(V,\cup_i F_i)$ is $\Omega(k)$-edge-connected.
% that for any $i\geq 1$, $W_i$ is a $(2^i\expansion,\beta)$-expanding set of nodes of $\cT_i$ and
%$\cT_i$ is a $\cT(k/4\beta,\lambda\expansion^i,T_i)$ \expandertree~of $G$.
%Subsequently, we show that as soon as $W_i$ becomes empty, $G'$ is $k/4\beta$-connected and that $W_i$ always becomes empty in at most $1/\log(1/\alpha)$ iterations. 

\begin{claim}
If  $\cT_i$ is a $(k/4,.,.)$-LCH of $G$, then
  $\cT_{i+1}$ is a $(k/4,.,.)$-LCH of $G$.
 In addition, if $W_i$ is $(.,k)$-expanding, then $W_{i+1}$ is $(.,k)$-expanding.
\end{claim}
\begin{proof}
%We prove this by induction. Suppose $\cT_i$ is a $\cT(k/4\beta,\lambda\expansion^i,T_i)$ \expandertree~of $G$.  
First, for any node $t\in \cT_{i+1}$ that is also in $\cT_i$, $G_{\cT_{i+1}}(t) = G_{\cT_i}(t)$; so, $G_{\cT_{i+1}}(t)$ is $k/4$-edge-connected by induction. So, $G_{\cT_{i+1}}\{t\}$ is also $k/4$-edge-connected.
For any new node $s_{t,j}\in\cT_{i+1}$, since $S_{t,j}$ is a $k/4$-edge-connected subgraph of $G'_{\cT_i}\{t\}$, $G_{\cT_{i+1}}(s_{t,j})$ is $k/4$-edge-connected. 
Therefore, $\cT_{i+1}$ is a $(k/4,.,.)$-LCH of $G$.

Similarly, observe that $W_{i+1}$ is $(.,k)$-expanding, because $W_{i+1}\subseteq W_i$ and for any node $t\in \cT_i$, $G_{\cT_{i+1}}(t) = G_{\cT_i}(t)$.
\end{proof}

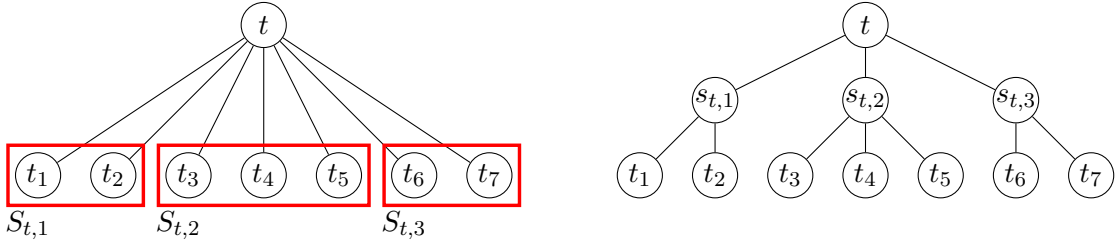
\begin{figure}[tb]
\centering
\begin{tikzpicture}
\tikzstyle{every node} = [draw,circle,minimum size=6mm,inner sep=0];
\node at (4,2) (t) {$t$};
\foreach \i in {1,...,7}{
\node at (\i,0) (t_\i) {$t_\i$} edge (t);
}
\draw [color=red,line width=1.3pt] (0.6,-0.4) node[below right,draw=none,color=black] {$S_{t,1}$} rectangle (2.4,0.4)
(2.6,-0.4) node[below right,draw=none,color=black] {$S_{t,2}$} rectangle (5.4,0.4)
(5.6,-0.4) node[below right,draw=none,color=black] {$S_{t,3}$} rectangle (7.4,0.4);
\begin{scope}[shift={(8,0)}]
\node at (4,2) (t) {$t$};
\foreach \i/\x in {1/2, 2/4, 3/6}{
	\node at (\x,1) (s_\i) {$s_{t,\i}$} edge (t);
}
\foreach \i/\f in {1/1,2/1,3/2,4/2,5/2,6/3,7/3}{
	\node at (\i,0) (t_\i) {$t_\i$} edge (s_\f);
}
\end{scope}
\end{tikzpicture}
\caption{A node $t$ and its children, $t_1,t_2,\dots,$ in $\cT_{i-1}$ are illustrated in left. The right diagram shows the tree $\cT_i$ when the new nodes $s_{t,1},s_{t,2},s_{t,3}$ corresponding to the sets $S_{t,1},S_{t,2},S_{t,3}$ are added.}
\label{fig:cT_2}
\end{figure}

We slightly strengthen our induction; instead of showing that $G_{\cT_i}\{t\}$ is $(2^i\alpha,.)$-expanding for all  $t\in W_i$, we show that for any $t\in W_i$ and any $S\subseteq V_{\cT_i}\{t\}$ where
$ d_0(S) \leq \frac12 d_0(V_{\cT_i}\{t\}),$
$$\phi_{G_{\cT_i}\{t\}}(S) \geq 2^i\expansion.$$

% By this change definition we make sure to avoid including the possible dominating child of $t$ in $S$. Note that the above definition is identical to \autoref{def:expandinghierarchy} for $W_0$ but it is slightly different for $i>0$.
For a set of indices $I\subseteq [\l]$ we use $ S_I=\cup_{i\in I} S_i.$
The following is the key lemma of the proof of this section.

%For a node $t\in \cT_i$ and a set of indices $J\subseteq [\ell(t)]$ let $V_{\cT_i}\{t\}_{|J}=\cup_{j\in J} S_{t,j}$.
\begin{claim}
\label{cl:HSexpander}
For any $i\geq 0$, $t\in W_i$, and any $S\subseteq V_{\cT_i}\{t\}$ where $\vol_0(S) \leq \frac12 \vol_0(V_{\cT_i}\{t\})$,
$$ \phi_{G_{\cT_i}\{t\}}(S) \geq \min\{2^i\alpha,1/8\}.$$
Therefore, for any $i\geq 1$, $W_i$ is $(2^i\alpha,.)$-expanding.
%If $W_i$ is a $(\expansion,k)$-expanding set of nodes of $\cT_i$, $\expansion<1/8$, and $\cT_i$ is a $(k/4,\lambda,T_i)$-LCH of $G$, 
%then $W_{i+1}$ is $(2\expansion,k)$-expanding. 
\end{claim}
\begin{proof}
We prove this by induction. Note that the statement obviously holds for $i=0$ because $G_{\cT_0}\{t\}$ is an $\expansion$-expander for all $t\in\cT_0$.
Suppose the statement holds for $i$.
Fix a node $t\in W_{i+1}$ and let $S_{t,1},\dots,S_{t,\l(t)}$ be the natural decomposition of $G_{\cT_i}\{t\}$ into $k/4$-edge-connected components. We abuse notation and drop the subscript $t$ and name these sets $S_1,\dots,S_{\l(t)}$. Choose $I\subset [\l(t)]$ such that $\vol_0(S_I) \leq \frac12 \vol_0(V_{\cT_i}\{t\})$. If $\phi_{G_{\cT_i}\{t\}}(S_I)\geq 1/8$ there is nothing to prove. Otherwise, we invoke \autoref{lem:expansionboosting} for the $k$-edge-connected graph $G=G_{\cT_i}\{t\}$, $F=\cup_{j=1}^i F_j$ and the natural decomposition $S_{1},\dots,S_{\l(t)}$ of $(V_{\cT_i}\{t\},F)$ into $k/4$-edge-connected components. 
The lemma shows that $\phi_{G_{\cT_{i+1}}\{t\}}(S_I) \geq 2^{i+1}\alpha$.

We just need to verify the assumptions of the lemma. 
By the induction hypothesis $\phi_{G_{\cT_i}\{t\}}(S_I) \geq 2^i\expansion$.
%Since $W_i$ is $(\expansion,.)$-expanding, $\phi_{G\{t\}}(S_I)\geq \expansion$.
In addition, $S_I$ only contains nondominating nodes of $t$, i.e., $S_I \subset T_i$.
Therefore, by the main technical theorem \ref{thm:main}, 
%we have
%$$ \mE_{e\sim \cut(t')} \reff_{D_i}(e) \leq \frac{f_1(k/4,\lambda\expansion^i)}{k/4}. $$
%Therefore, by 
equation \eqref{def:FireffDi},  and the Markov inequality,  at least $15/16$ fraction of the edges incident to each $t'\in T_i$ are  in $F_i$. So, $\deg_{F_i}(S_I) \geq \frac{15}{16}\vol(S_I) \geq \frac78 \vol(S_I)$.
\end{proof}

\begin{lemma}[Expansion Boosting Lemma]\label{lem:expansionboosting}
Given a $k$-edge-connected graph $G=(V,E)$, a set $F\subseteq E$ and a  natural-decomposition of $(V,F)$ into $k/4$-edge-connected components $S_1,\dots,S_\l$.
For any $I\subseteq [\l]$ if $\vol_F(S_I) \geq 7\vol(S_I)/8$, and  $\phi(S_I) <1/8$, then
	$$ \frac{\deg(S_I)}{\sum_{i\in I} \deg(S_i)} \geq 2\phi(S_I) .$$
%that is an $\alpha$-expander for $\alpha < 1/8$. 
% then
%$H=G(V/ S_1 / \dots / S_\l)$ is a $2\alpha$-expander.
\end{lemma}
\begin{proof}
%	For a set of indices $I\subseteq [\l]$ let
%	We show that if $\vol(S_I) \leq \vol(V)/2$, then
%	Note that the denominator of the above ratio is the number of edges adjacent to the vertices corresponding to sets $\{S_i\}_{i\in I}$ in the contracted graph $H$.
	Think of the edges in $F$ as good edges and the edges not in $F$, $E\setminus F$ as the bad edges. 
	%The assumption of the lemma, $|F\cap d(v)| \geq 7d(v)/8$ implies that most of the edges incident to every vertex are good.
	We can write the denominator of the above as follows:
	\begin{equation}\label{eq:goodedgehalfcont} \sum_{i\in I} \deg(S_i) = \deg_{F}(S_I) + 2\sum_{i,j\in I, i<j} |F(S_i,S_j)| + \sum_{i\in I} \deg_{E\setminus F}(S_i)
	\end{equation}
	where we used $\deg_F(S)$ to denote the edges of $F$ leaving a set $S$.
	
	First, we observe that by the natural decomposition lemma \ref{lem:decompkconnected}, the middle term on the RHS, i.e., the number of good edges between $\{S_i\}_{i\in I}$ is small,
	\begin{eqnarray*}
		\sum_{i,j\in I, i<j} |F(S_i,S_j)| \leq (|I|-1)(k/4) \leq \frac14 \sum_{i\in I} \deg(S_i),
	\end{eqnarray*}
	where the second inequality follows by $k$-edge-connectivity of $G$.
	Subtracting twice the above inequality from  \eqref{eq:goodedgehalfcont}  we get
	\begin{eqnarray} \label{eq:removegoodmiddleedges}
		\deg_{F}(S_I) + \sum_{i\in I} \deg_{E\setminus F}(S_i) \geq \frac12 \sum_{i\in I} \deg(S_i).
	\end{eqnarray}
	
	Secondly, by the lemma's assumption, %since $\vol(S_I) \leq \vol(V)/2$, for any $v\in S_I$,  $|F\cap d(v)| \leq 7d(v)/8$. Therefore, the number of bad edges incident to  $S_I$ is at most
	\begin{equation}
		\sum_{i\in I} \partial_{E\setminus F}(S_i) \leq \sum_{i\in I} \vol_{E\setminus F}(S_i) = \vol_{E\setminus F}(S_I)=  \vol(S_I) - \vol_F(S_I)%\sum_{v\in S_I} d(v)/8 
		\leq \frac18 \vol(S_I).
	\end{equation}
	Putting the above two inequalities together we get,
	\begin{eqnarray*}
		\frac12 \sum_{i\in I} \deg(S_i) \leq \deg_F(S_I) + \frac18 \vol(S_I)
	\end{eqnarray*}
	Dividing both sides of the above inequality by $\deg(S_I)$ we get
	\begin{eqnarray*} \frac{1}{2\deg(S_I)}{\sum_{i\in I} \deg(S_i)} &\leq& \frac{\deg_F(S_I)}{\deg(S_I)} + \frac{\vol(S_I)}{8\deg(S_I)}\\
	& \leq & 1 + \frac{1}{8\phi(S_I)} \leq \frac{1}{4\phi(S_I)},
	\end{eqnarray*}
	where  the last inequality uses that $\alpha \leq 1/8$.
\end{proof}

%Now, we are ready to show that $W_{\log(1/\expansion)}$ is empty.
\begin{claim}
\label{cl:Wlogempty}
$W_{\log(1/\expansion)}$ is empty.
\end{claim}
\begin{proof}
Let $i$ be the smallest integer such that $2^i\alpha \geq 1/8$.
Note that $i < \log(1/\expansion)$. By \autoref{cl:HSexpander},
for any $t\in W_i$, 
\begin{equation}
\label{eq:GTiphi18}
\phi(G_{\cT_{i}}\{t\}) \geq 1/8.
\end{equation}
We show that $W_{i+1}$ is empty.
Fix a node $t\in W_i$. 
Similar to the previous claim, at least $15/16$ fraction of the edges adjacent to any nondominating  child of $t$ are in $F_i$.
For a set $I\subset [\l(t)]$ such that $d_0(S_I) \leq \frac12 d_0(V_{\cT_0}\{t\})$, we have $\phi_{G_{\cT_i}\{t\}}(S_I) \geq 1/8$; therefore at least half of the edges in the cut $(S_I, V_{\cT_i}\{t\}\setminus S_I)$ are in $F_{i}$. By $k$-edge-connectivity of $G_{\cT_i}\{t\}$, 
$F_{i}$ has at least $k/2$ edges in this cut. So, $(V_{\cT_{i}}\{t\},F_{i})$ is $k/2$-edge-connected.
\end{proof}

\begin{claim}
At the end of the algorithm $G'$ is $k/4$-edge-connected.
\end{claim}
\begin{proof}
We show that for any $i$ and any node $t\notin W_i$, $G'_{\cT_i}\{t\}$ is 
$k/4$-edge-connected. 
Then, the claim follows by \autoref{cl:Wlogempty}.

At any iteration $i$, for any new node $s_{t,j}$, $G'_{\cT_i}\{s_{t,j}\}$ is $k/4$-edge-connected because $S_{t,j}$ is a $k/4$-edge-connected component of $G_{\cT_i}\{t\}$; this subgraph remains $k/4$-edge-connected in the rest of the algorithm because we never delete edges from $G'$. On the other hand, when we remove a node $t$ from $W_i$, we are guaranteed that $G_{\cT_i}\{t\}$ is $k/4$-edge-connected.
\end{proof}

Now, \autoref{thm:expandertreeboosting}
follows from the above claim and that for any $e\in \cup_i F_i$,
$$ \reff_{\mE_i D_i}(e) \leq \log(1/\expansion)\cdot \min_i \reff_{D_i}(e) \leq \frac{16 f_1(k/4,\lambda\cdot \expansion)\log(1/\expansion)}{k/4}.$$

\subsection{Algorithmic Aspects}
\label{subsec:algorithmicatsp}
In this part we prove \autoref{thm:algorithmicatsp}.
We emphasize that our algorithm does not necessarily find a thin tree. As alluded to in the introduction, the main barrier is that verifying the thinness is a variant of the sparsest cut problem for which the best known algorithm only gives an $O(\sqrt{\log n})$-approximation factor. Instead, we use the fact that ``directed thinness'', as defined in \eqref{eq:directedthinness} of \autoref{thm:agmos_version}, is polynomially testable and it is enough to solve ATSP. We refrain  from giving the details and we refer interested readers to \cite{AGMOS10}.
Our rough idea is as follows: We run the ellipsoid algorithm on the convex program $\maincp$ by first  discarding the $2^n$ constraints $\bone_S^\intercal D\bone_S \leq \bone_S^\intercal L_G \bone_S$ that verify $D$ is a shortcut matrix. If the directed thinness of the output tree fails, the undirected thinness fails as well, so we get a set $S$ for which $\bone_S^\intercal D\bone_S > \bone_S^\intercal L_G \bone_S$. That corresponds to a violating constraint of the convex program which the ellipsoid algorithm can use in the same way that it uses separation oracles. Repeating this procedure, either the ellipsoid algorithm converges, i.e., we find an actual undirected thin tree, or we find an ATSP tour along the way.

\begin{algorithm}
\begin{algorithmic}[1]
	\Input A $k\geq 7\log n$-edge-connected graph $G=(V,E)$.
	\Output A $k/20$-edge-connected, $1/4$-dense induced subgraph that is an $\Omega(1/k^2)$-expander.
	\State Let $U\leftarrow V$. We always let $H$ be the induced subgraph on $U$.
	\Loop \label{step:exptreeloop}
		\If {there is a vertex $v\in U$ such that $d_H(v) \leq 7d_G(v)/20$}
			\State Let $U\leftarrow U\setminus \{v\}$ and {\bf goto} \ref{step:exptreeloop}. 
		\EndIf \Comment{If this case does not happen, $H$ is $7/20$-dense.}
		\State Let $S$ be the output of the spectral partitioning algorithm on $H$, and let $T=U\setminus S$.
		\If {$\phi_G(S)\leq \phi_G(U)$ or $\phi_G(T)\leq \phi_G(U)$}
		\State Let $U=S$ or $U=T$ whichever has the smallest $\phi_G(.)$, and {\bf goto} \ref{step:exptreeloop}.
		\EndIf %\Comment{If this case does not happen, then $\max\{\phi_G(S),\phi_G(T)
		\If {$\max\{\phi_H(S),\phi_H(T)\}<1/k$}
			\State Let $U=S$ or $U=T$ whichever has fewer vertices, and {\bf goto} \ref{step:exptreeloop}.
		\EndIf \Comment{If this case does not happen, by Cheeger's inequality, $H$ is an $\Omega(1/k^2)$-expander. }
		\State \parbox[t]{\dimexpr\linewidth-\algorithmicindent}{If $H$ is $k/20$-edge-connected, {\bf return} $H$. Otherwise, let $S\subseteq U$ be such that $\partial_H(S) <k/20$ and $\phi_H(S) \geq \phi_H(U\setminus S)$. \Comment{So, $\phi_H(S) \geq \Omega(1/k^2)$.}\strut}
		\State \parbox[t]{\dimexpr\linewidth-\algorithmicindent}{Let $S_1,S_2,\dots$ be a natural decomposition of $G[S]$ into $k/20$-edge-connected components. By \eqref{eq:expandercase3} there is $S_i$ such that $\partial_H(S_i) < k/10$. {\bf Return} $G[S_i]$.\strut}
	\EndLoop
\end{algorithmic}
	\caption{Expander Extraction}
	\label{alg:intiialexpandertree}
\end{algorithm}

To complete the proof we need to make sure that we can construct 
the starting \expandertree~in polynomial time; we will describe our algorithm later.
Apart from that, the main difficulty is that to obtain the shortcut matrix $D$ promised in \autoref{thm:hierarchydecomgeneral} we need to solve $O(\log\log(n))$ many convex programs ($\maincp(\cT_i)$) and each one depends on the solution of the previous ones. In other words, we should be recursively calling $O(\log\log n)$ many ellipsoid algorithms. Therefore, if we find a separating hyperplane for one of the ellipsoids, we should restart the ellipsoid algorithms for all the proceeding convex programs. The resulting algorithm runs in time $n^{O(\log\log n)}$ and has an approximation factor of $\polyloglog(n)$. 
We can also tradeoff the approximation factor with the running time of the algorithm by  modifying \autoref{alg:expandertreeboosting} to have $O(\l)$ number of iterations. For constant values of $\l$ this gives a polynomial time approximation algorithm.

We will give an algorithm to construct an $(\Omega(1/k^2),.)$-expanding $(k/20, 1/4, \cT)$-LCH, $\cT_0$ for some $\alpha\asymp 1/\log^2(n)$. Then, we run a modified version of \autoref{alg:expandertreeboosting} to obtain \expandertrees~$\cT_1,\dots,\cT_{2\l}$; in particular,  we only run the loop for $2\l$ iterations;  to make sure that $\cT_{2\l}$ is $(\Omega(1),.)$-expanding, we need to boost the expansion by $\left(\frac1{\alpha}\right)^{1/2\l}$ in every iteration of the loop. To be more precise, for any $1\leq i\leq 2\l$, instead of \eqref{def:FireffDi}, we let
$$ F_i:=\left\{e\in E: \reff_{D_i}(e)\leq \frac{O((1/\alpha)^{1/2\l}) f_1(k',\lambda')}{k'}\right\}.$$
The proof simply follows by a modification to the expansion boosting lemma. 
The resulting algorithm runs in time $n^{O(\l)}$ and has an approximation factor of $\polyloglog(n)\cdot \log^{1/\l}(n)$.

It remains to find the starting \expandertree~$\cT_0$.
Given a $k\geq 7\log n$-edge-connected graph $G=(V,E)$, all we need is to find a $1/4$-dense $k/20$-edge-connected induced subgraph $G[S]$ whose expansion is $\Omega(1/k^2)$. 
We essentially make the proof of \autoref{lem:expsizedenseexpander} constructive using the spectral partitioning algorithm \cite{AM85,Alon86} at the cost of obtaining an $\Omega(1/k^2)$-expander instead of a $1/k$-expander. This is because, by Cheeger's inequality,  the spectral partitioning algorithm gives a square-root approximation to the problem of approximating $\phi(G)$. The details of the algorithm are described in \autoref{alg:intiialexpandertree}.

%To find $D_i$'s we need to solve $\maincp(\cT_i)$ which seems prohibitive; because for a given matrix $D$ we need to 
%Then, we use the ellipsoid algorithm to find an ATSP tour. To find $D_i$'s we need to solve 

% !TEX root = main.tex
\section{The Dual of \maincp}
\label{sec:dual}
In this section we write down the dual of $\maincp$.
Before explicitly writing down the dual, let us give a few lines of intuition.
We do this by writing down the dual of a few convex programs computing the maximum or average effective resistance of a number of pairs of vertices. %Although these programs are different from \maincp, the form of the duals are similar.

For a pair of vertices, $a,b\in V$, the optimum value of the following expression,
\begin{equation}\label{eq:potentialvector} \max_{x:V\to\R}\frac{(x(a)-x(b))^2}{\sum_{u\sim v} (x(u)-x(v))^2}.	
\end{equation}
 is exactly equal to $\reff_G(a,b)$; in particular, if we fix $x(b)=0,x(a)=\reff(a,b)$, then the optimum $x$ is the  \emph{potential} vector of the electrical flow that sends one unit of flow from $a$ to $b$. It is an easy exercise to cast the above as a convex program.

Now, suppose we want to write a program which computes the maximum effective resistance of pairs of vertices $(a_1,b_1),\dots,(a_h,b_h)$.
In this case we need to choose a separate potential vector for each pair, 
We use a matrix $X$ where the $i$-th row of $X$ is the potential vector associated to the $i$-th pair.
The following program gives the maximum effective resistance of all pairs. 
$$ \max_{X\in \R^{h\times V}} \frac{\sum_{i=1}^h (X_{i,a_i} - X_{i,b_i})^2}{\sum_{i=1}^h \sum_{u\sim v} (X_{i,u}-X_{i,v})                                                                                                                                                                                                                                                                                                                                                                                                                                                                                                                                                                                                                                                                                                                                                                                                                                                                                                                                                                                                                                                                                                        ^2} = \max_{X\in \R^{h\times V}}\frac{\sum_{i=1}^h (X_{i,a_i} - X_{i,b_i})^2}{\sum_{u\sim v} (X_u-X_v)^2}$$
It follows by \eqref{eq:averagefracs} that the optimum of the above is the maximum effective resistance of all pairs $(a_1,b_1),\dots,(a_h,b_h)$.
Recall that $X_u$ is the $u$-th column of $X$.

%In general, we should expect the dual program to be ``coordinate independent''. 
Note that the denominator of the RHS is coordinate independent, i.e., it is rotationally invariant. We can rewrite the numerator in the following way and make it rotationally invariant. Instead of mapping the $i$-th pair to the $i$-th coordinate, we map the $i$-th pair to $z_i$ where $\{z_1,\dots,z_h\}$ are $h$-orthonormal vectors. In other words, to calculate the numerator we need to find a coordinate system of the space such that the sum of the square of the projection of the edges on the corresponding coordinates is as large as possible
$$ \max_{\substack{X\in\R^{h\times V},\\ \{z_1,\dots,z_h\} \text{ are orthonormal}}} \frac{\sum_{i=1}^h \langle z_i, X_{a_i}-X_{b_i}\rangle^2}{\sum_{u\sim v} (X_u-X_v)^2}.$$
Instead of choosing $z_1,\dots,z_h$ we can simply maximize over an orthogonal matrix $U\in\R^{h\times h}$ and let $z_1,\dots,z_h$ be the first $h$ rows of $U$,
\begin{equation}
	\label{eq:maxeffresdemonstration}
 \max_{X\in\R^{h\times V},\text{Orthogonal } U} \frac{\sum_{i=1}^h \langle U^i,X_{a_i}-X_{b_i}\rangle^2}{\sum_{u\sim v} (X_u-X_v)^2},
 \end{equation}
where $U^i$ is the $i$-th row of the matrix $U$.
The above program is equivalent to the dual of the following convex program 
\begin{equation*}
\begin{aligned}
	\min \hspace{3ex} &\cE,&\\
	\st \hspace{4ex} &\reff_D(a_i,b_i) \leq \cE &\forall 1\leq i\leq h,\\
	&D\preceq L_G.&
\end{aligned}	
\end{equation*}
We will give a formal argument later.
When we replace the constraint $D\preceq L_G$ with $D\preceq_{\square}L_G$, we get the additional assumption that $X$ is a cut metric. This can significantly reduce the value of \eqref{eq:maxeffresdemonstration}.

Next, we write a program which computes the expected effective resistance of pairs of vertices $(a_1,b_1),\dots,(a_h,b_h)$ with respect to a distribution $\lambda_1,\dots,\lambda_h$,
\begin{equation} \sum_{i=1}^h \lambda_i\cdot  \reff(a_i,b_i)=\max_{X\in \R^{h\times V}}\sum_{i=1}^h \lambda_i\cdot \frac{(X_{i,a_i}-X_{i,b_i})^2}{\sum_{u\sim v} (X_{i,u}-X_{i,v})^2}.
\label{eq:expectedeffresexdual}
\end{equation}
 where we simply used \eqref{eq:potentialvector}. Equivalently, we can write the above ratio as follows:
\begin{equation}  \max_{X\in \R^{h\times V}}\frac{\left(\sum_{i=1}^h \sqrt{\lambda_i}\cdot  (X_{i,a_i}-X_{i,b_i})\right)^2}{\sum_{u\sim v} (X_{u}-X_{v})^2},
\label{eq:expectedeffreseydual}	
\end{equation}

To see that the above two are the same, first, assume $X$ is normalized such that $\sum_{u\sim v} (X_{i,a_i}-X_{i,b_i})^2=1$ for all $i$. 
This simplifies \eqref{eq:expectedeffresexdual} to $\sum_i \lambda_i (X_{i,a_i}-(X_i,b_i))^2$.
Then let 
$$Y^i = X^i \sqrt{\lambda_i}\cdot (X_{i,a_i}-X_{i,b_i}),
%{\sum_{j=1}^h \sqrt{\lambda_j} (X_{j,a_j}-X_{j,b_j})},
$$
where as usual $Y^i$ is the $i$-th row of $Y$.
Plugging in $Y$ in \eqref{eq:expectedeffreseydual} gives the same value $\sum_i \lambda_i (X_{i,a_i}-X_{i,b_i})^2$.

Lastly, we can write a rotationally invariant formulation of \eqref{eq:expectedeffreseydual} using an orthogonal matrix $U$.
$$ \max_{\substack{X\in\R^{h\times V},\\\text{Orthogonal } U }} \frac{\left(\sum_{i=1}^h \sqrt{\lambda_i} \cdot \langle U^i, X_{a_i}-X_{b_i}\rangle\right)^2}{\sum_{u\sim v} (X_u-X_v)^2}$$
Let $\X_h\in\R^{n\times h}$ be the matrix where the $i$-th column is $\X_{a_i,b_i}$. It follows by
\autoref{lem:matrixtrace} that
$$ \max_{\text{Orthogonal }U} \sum_{i=1}^h \langle U^i, X_{a_i}-X_{b_i}\rangle = \max_{\text{Orthogonal }U}\trace(UX\X_h) = \norm{X\X_h}_*.$$
This is is a key observation in the proof of the technical theorem.

%In the rest of this document we upper-bound the value of \maincp.
In the rest of this section we will prove that a similar expression is equivalent to the dual of \maincp. Then, in \autoref{subsec:dualvariants} we write the dual of $\maxcp,\averagecp$ and we will prove \autoref{thm:avgeffres}.
The following lemma is the main statement that we prove in this section. Recall that for a mapping $X$ of vertices of $G$, $\Xb=X\X$ is the matrix where for every edge $e=\{u,v\}$, $\Xb_e = X_u-X_v$. 
\begin{lemma}
\label{lem:sdpdual}
For any graph $G=(V,E)$ and any $(.,.,T)$-LCH of $G$, the optimum of \maincp~(up to a multiplicative factor of 2) is equal to
\begin{equation}
\label{eq:sdpdual}
 \sup_{U, X} \frac{\sum_{t\in T} \frac{1}{|\cut(t)|}\Big( \sum_{e\in \cut(t)} \langle\Ue,  \Xb_e\rangle\Big)^2}{\sum_{e\in E} \norm{\Xb_e}^2} 
 \end{equation}
where the supremum is over all semiorthogonal matrices $U\in\R^{E \times h}$, and all cut metrics $X\in \{0,1\}^{h\times V}$, for  arbitrary $h>0$.
% and all permutations $\pi$ of the edges of $G$.
\end{lemma}

Note that the dimension $h$ in the above can be arbitrarily large because $X$ is a cut metric. However, only the first $|E|$ rows of $U$ matter. In addition, since $X$ is a cut metric, for any edge $e=\{u,v\}\in E$, $\norm{\Xb_e}^2 = \norm{\Xb_e}_1$; so, throughout the paper, we may use either of the two norms.

\begin{proof}
First, we show \maincp~satisfies  Slater's condition, i.e., that \maincp~has a nonempty interior.
It is easy to see that   $D=\frac12 L_G+ \frac1{3n^2} J$ is a PD matrix that satisfies all constraints strictly.  
In particular, since $G$ is connected, for any set $S$, $\bone_S^\intercal L_G \bone_S \geq 1$, so
$$ \frac{1}{3n^2} \bone_S J\bone_S \leq \frac13 < \frac12\bone_S^\intercal L_G \bone_S.$$
Therefore, $\bone_S^\intercal D \bone_S < \bone_S^\intercal L_G \bone_S$ for all $S$.
Hence,  Slater's condition is satisfied, and
the strong duality is satisfied and the primal optimum
is equal to the Lagrangian dual's optimum (see \cite[Section 5.2.3]{BV06} for more information).

%For of all, the duality gap of SDP \eqref{cp:main} is 0. 
%We write the Lagrange Dual to the SDP \eqref{cp:main}. 
For every $t\in T$ we associate a Lagrange multiplier $\lambda_t$ corresponding to the first set of constraints, and for every set $S$ we associate a nonnegative Lagrange multiplier $y_S$ corresponding to the second set of constraints of the \maincp. % and we associate a PSD Lagrange  matrix $C$ to the constraint $D\succeq 0$. 
The
Lagrange function is defined as follows:
$$ g(\lambda, y) = \inf_{D\succ 0} \cE + \sum_{t\in T} \lambda_t \Big(\frac{1}{{|\cut(t)|}}\sum_{e\in \cut(t)} \X_e^\intercal D^{-1} \X_e - \cE\Big) + \sum_{ S\subset V} y_S (\bone_S^\intercal D\bone_S - \bone_S^\intercal L_G \bone_S) $$

%In the rest of the proof we just upper-bound the LHS of the above equation.
First, we differentiate the RHS with respect to $\cE,D$ to eliminate the $\inf$. This gives us the Lagrangian dual. Then, we homogenize the dual expression by normalizing the entries of $y$; finally we eliminate the dependency on $\lambda$ by an application of the Cauchy-Schwarz inequality.

First of all, differentiating $g(\lambda,y)$ w.r.t. $\cE$ we obtain that
\begin{equation}
\label{eq:sumlambdat}
\sum_{t\in T} \lambda_t =1.
\end{equation}
Let 
$$A:=\sum_{t\in T} \frac{\lambda_t}{|\cut(t)|} \Big(\sum_{e\in \cut(t)} \X_e \X_e^\intercal\Big) \text{ and } Z:=\sum_{\emptyset \subset S\subset V} y_S \bone_S \bone_S^\intercal.$$ 
Note that by definition $A$ and $Z$ are symmetric PSD matrices. 
The Lagrange dual function simplifies to
$$g(A,Z) = \inf_{D\succ 0} %\sup_{\substack{\lambda, y\geq 0, C\succeq 0,\\ }} \inf_{D} 
A\bullet D^{-1} +Z\bullet D - Z\bullet L_G,$$
subject to $\sum_{t} \lambda_t=1$.
 Now, we find the optimum $D$ for fixed $A,Z$.
 First, %without loss of generality we can 
 we assume that $A$ and $Z$ are nonsingular. This is without loss of generality by the continuity of $g(.)$ and because the assumption $\sum_t \lambda_t=1$ can be satisfied by adding arbitrarily small perturbations. % just by assuming that $\lambda_e >0$ and $y_S>0$ for all edges $e$ and sets $S$. % (note that the all one vector, $\b1$ is not in the range of any of these two matrices).  
% This modification increases the RHS of above a little that we can ignore. 
% Second, by \autoref{lem:matrixinvineq},
% $$ (D+L_G/2)^\dagger I_{\image(D)} \preceq D^\dagger.$$
Differentiating with respect to $D$ we obtain
$$ D^{-1} A D^{-1} = Z.$$
Since, $A,D$ are nonsingular there is a unique solution to the above equation,
%Since we only need an upper-bound on $\sdp$, we may  assign any value to $D$ and use that as an upper-bound, e
%We let 
$$D=Z^{-1/2} (Z^{1/2} A Z^{1/2})^{1/2} Z^{-1/2}$$ 
We refer interested readers to \cite{SLB74} to solve the above matrix equation. 
%Since $A J = \bzero$, by \autoref{lem:kernelpsd} $DJ=\bzero$. 
Using
$$D^{-1} = Z^{1/2} (Z^{1/2} A Z^{1/2})^{-1/2} Z^{1/2},$$
we have
\begin{eqnarray*} A\bullet D^{-1} + Z\bullet D &=& \trace(AZ^{1/2}(Z^{1/2}AZ^{1/2})^{-1/2}Z^{1/2}) + \trace(Z^{1/2}(Z^{1/2}AZ^{1/2})^{1/2}Z^{-1/2}) \\
&=& 2\trace((Z^{1/2}AZ^{1/2})^{1/2}).
\end{eqnarray*}
%On the other hand, since $C,D$ are symmetric PSD matrices, $C\bullet D \geq 0$. 
Therefore,
$$ g(A,Z) = %\sup_{\lambda,g\geq 0\sum_t \lambda_t=1} 
2\trace( (Z^{1/2} A Z^{1/2})^{1/2}) - Z\bullet L_G %= \norm{(Z^{1/2} A Z^{1/2})^{1/2}}_* - Z\bullet L_G. 
$$

Let $\cE^*$ be the optimum value of \maincp.
 By the strong duality,
$$ \cE^* = \sup_{\lambda, y\geq 0} g(A,Z) = \sup_{\lambda,y\geq 0} 2\trace( (Z^{1/2} A Z^{1/2})^{1/2}) - Z\bullet L_G.$$

It remains to characterize values of $\lambda, y$ that maximize the above function. 
Let  $W\in\R^{E\times E}$ be a diagonal  matrix where
for each edge $e\in E$,
\begin{equation}
\label{eq:defW}
 W_{e,e} = \sqrt{\sum_{t\in T: e\in \cut(t)} \frac{\lambda_t}{|\cut(t)|}}.
\end{equation}
Note that the above sum is over zero, one, or two terms because each edge is in at most two sets $\cut(t)$.
%$B\in \R^{n\times |E|}$ where
%for every $t\in T$ and every  edge $e\in \cut(t)$ there is a column $\sqrt{\frac{\lambda_t}{|\cut(t)|}} \cdot \X_e $ in $B$. 
Observe that 
$$A=\X W^2 \X^\intercal.$$ 
Furthermore the nonzero eigenvalues of $Z^{1/2} A Z^{1/2} = Z^{1/2} \X W^2 \X^\intercal  Z^{1/2}$ are the same as the nonzero eigenvalues of $W \X^\intercal Z \X W$. 
Therefore,
\begin{eqnarray}
%\sup_{\lambda,y\geq 0} 2\norm{Z^{1/2} A Z^{1/2}}* - Z\bullet L_G & = &  
\cE^* = \sup_{\lambda,y\geq 0} 2\trace( (W \X^\intercal  Z \X W)^{1/2})  - Z\bullet L_G
\label{eq:lastconvexEstar}
\end{eqnarray}
Observe that the above quantity is not homogeneous in $y$ as $Z\bullet L_G$ scales linearly with $y$ and $\trace( (W \X^\intercal  Z \X W)^{1/2})$ scales with $\sqrt{y}$. 
It is an easy exercise to see that by choosing the right scaling for $y$ we can rewrite the above as follows:
\begin{eqnarray*} 
\cE^*=\sup_{\lambda,y\geq 0} \frac{ \trace( (W \X^\intercal  Z \X W)^{1/2})^2} { Z\bullet L_G}.
\end{eqnarray*}
Note that although \eqref{eq:lastconvexEstar} is convex, the above quantity is not necessarily convex but we prefer to work with the above quantity because it is homogeneous. 

Write $Z=X^\intercal X$ where $X\in\R^{2^n\times V}$ and each row of $X$ corresponds to a vector $y_S\bone_S$ for a set $S\subseteq V$. 
Observe that $X$ defines a weighted cut metric on the vertices of $G$ which can be embedded into an unweighted cut metric (see \autoref{sec:prelim:balls} for properties of weighted/unweighted cut metrics). So, we assume $X\in\{0,1\}^{h\times V}$ for an $h$ possibly larger than $2^n$.
If $h<|E|$ then we extend $X$ by adding all zeros rows to make $h\geq |E|$. 
Let $X_v$
be the mapping of $v$ in that metric, i.e., $X_v$ is the column $v$ of $X$.
By the definition of the nuclear norm, 
$$\trace( (W \X^\intercal  Z \X W)^{1/2})^2 = \norm{X\X W}_*^2 = \norm{\Xb W}_*^2.$$ 
Therefore,
\begin{eqnarray*} \cE^* = \sup_{X,\lambda} \frac{\norm{\Xb W}_*^2}{\sum_{\{u,v\}\in E} \norm{\Xb_e}^2_2} 
\end{eqnarray*}
In the denominator we used the fact that $Z\bullet L_G = \sum_{\{u,v\}} \norm{X_u-X_v}_2^2=\sum_e \norm{\Xb_e}^2$.

Note that $\Xb\in \R^{h\times E}$. 
Since the number of rows of $\Xb$ is at least the number of its columns, 
by \autoref{lem:matrixtrace}, we can rewrite the nuclear norm as $\sup_U \trace(U\Xb W)$ over all semiorthogonal matrices $U\in\R^{E\times h}$, so
\begin{eqnarray}
\cE^* &=& \sup_{\substack{X\in\{0,1\}^h,\lambda\geq 0,\\ \text{Semiorthgonal }U}} \frac{\Big( \sum_{t\in T} \sum_{e \in \cut(t)} W_{e,e} \cdot \langle\Ue, \Xb_e\rangle \Big)^2}{\sum_{e\in E} \norm{\Xb_e}^2}\nonumber \\
&\asymp& \sup_{\substack{X\in\{0,1\}^h,\lambda\geq 0,\\ \text{Semiorthgonal }U}}  \frac{\Big( \sum_{t\in T} \sum_{e \in \cut(t)} \sqrt{\lambda_t/|\cut(t)|} \cdot \langle\Ue, \Xb_e\rangle \Big)^2}{\sum_{e\in E} \norm{\Xb_e}^2}\label{eq:lambdabefcauchy} 
\end{eqnarray}
Note that the second equation  is an equality up to a factor of 2 because each edge is contained in at most two sets $\cut(t)$.
%The  supremum on the RHS is over all cut matrices $X$, unitary matrices $U$. %In the last equality we use 
%\autoref{lem:matrixtrace} and that $X$ is a cut metric. 
In particular, by \eqref{eq:defW}, for any edge $e$,
$$ \frac1{\sqrt{2}} \sum_{t\in T: e\in\cut(t)} \sqrt{\lambda_t/|\cut(t)|} \leq W_{e,e} \leq \sum_{t\in T: e\in \cut(t)} \sqrt{\lambda_t/|\cut(t)|}.$$
%Recall that $\Xb_e$ is defined as  $X\X_e$.
%Note that,
% since $XB$ has rank $n-1$, only $n-1$ entries of   the diagonal of $UXB$ are nonzero. % and without loss of generality we assume these are located in the first $|E|$ rows.
 %the first equality follows by the definitions of %the nuclear norm and the Laplacian matrix, 
%The second inequality uses %\autoref{lem:matrixtrace} and 
%that  $X$ is a cut metric; 

Finally, using  the  Cauchy-Schwarz inequality we can write
$$ \cE^*\lesssim  
\sup_{X,U} \frac{ \Big(\sum_{t\in T} \lambda_t\Big)\cdot \Big(\sum_{t\in T} \frac{1}{|\cut(t)|}\Big(\sum_{e \in \cut(t)} \langle U^e, \Xb_e\rangle \Big)^2 \Big) }{\sum_{e\in E} \norm{\Xb_e}^2}\nonumber
$$
The above inequality is tight because in the worst case we can  let
$$\lambda_t \propto \frac{1}{|\cut(t)|}\bigg(\sum_{e\in\cut(t)} \langle \Ue, \Xb_e\rangle \bigg)^2 
%\Big(\sum_t\in T \sum_{e\in\cut(t)} |\cut(t)|^{-1} Ue \Xb_e\Big)^2}
,$$
%the Cauchy-Schwarz inequality.
such that  $\sum_t \lambda_t=1$.
\end{proof}

\subsection{The Dual for Variants of the Problem}
\label{subsec:dualvariants}
In the rest of this section we prove simple positive and negative results on the value of the dual. We will not use these results in the proof of the technical theorem; we present them to provide some intuition on how one can approach the dual. 

First of all, using similar ideas as the proof of the above lemma, we can also write the dual of $\maxcp$ and $\averagecp$.
We write these quantities, without proof, as we do not need them in the proof of our main theorem.
First, we write the dual of $\maxcp$.
\begin{equation}
\left\{\begin{aligned}
\min \hspace{3ex} & \max_e \reff_D(e),\\	
\st \hspace{2ex} &D\preceq_\square L_G\\
D\succ 0
\end{aligned}\right\}
=
\sup_{\substack{X\in\{0,1\}^{h\times V}\\ \text{Semiorthogonal } U}} \frac{\sum_{e\in E} \langle U^e, \Xb_e\rangle^2}{\sum_{e\in E} \norm{\Xb_e}^2}.
\label{eq:maxcpdual}
\end{equation}
Now, we write the dual of $\averagecp$.
\begin{equation}
\left\{\begin{aligned}
\min \hspace{3ex} & \max_{S\subset V} \mE_{e\sim E(S,\overline{S})}  \reff_D(e),\\
\st \hspace{2ex} & D\preceq_\square L_G,\\
& D\succ 0
\end{aligned}\right\}
=
\sup_{\substack{X\in\{0,1\}^{h\times V},\lambda \\ \text{Semiorthogonal }U}} \frac{\left(\sum_{e\in E} \sqrt{\gamma_e}\cdot \langle U^e,\Xb_e\rangle\right)^2}{\sum_{e\in E} \norm{\Xb_e}^2},
\label{eq:averagecpdual}
\end{equation}
where for any edge $e$, $\gamma_e = \sum_{S: e\in E(S,\overline{S})} \frac{\lambda_{(S,\overline{S})}}{|E(S,\overline{S})|}$ and $\lambda_{(S,\overline{S})}$ is a probability distribution on all cuts of $G$.

In the following lemma, we show that for any pair of vertices of a $k$-edge-connected graph there is a shortcut matrix that reduces the effective resistance of that pair to $1/k$.
\begin{lemma}\label{lem:oneedgeres}
For any $k$-edge-connected graph $G$ and any pair of vertices $a,b$, there is a shortcut matrix $D$ such that $\reff_{D}(a,b)\leq 1/k$. 
\end{lemma}
\begin{proof}
The statement can be proven relatively easy in the primal. Since $G$ is $k$-edge-connected we can simply shortcut the $k$ edge-disjoint paths connecting $a,b$ and $D=k\cdot L_{a,b}$. Then it is easy to see that $\reff_D(a,b) = 1/k$ and $D\preceq_\square L_G$ as desired.

By \eqref{eq:maxcpdual} it is enough to show that
$$  \sup_{\substack{X\in \{0,1\}^{h \times V},\\ \text{Semiorthogonal } U\in\R^{1\times h}}} \frac{\langle U^{\{a,b\}}, X_{a}-X_{b}\rangle ^2 }{\sum_{u\sim v} \norm{X_{u}-X_{v}}^2} \leq O(1/k),$$
First note that in the worst case the vector $U^e$ is parallel to $X_{a}-X_{b}$. Therefore, the numerator is exactly $\norm{X_{a}-X_{b}}^2$. The proof simply follows from the triangle inequality of the cut metrics. 
	
	Since $G$ is $k$-edge-connected there are $k$ edge-disjoint paths from $a$ to $b$. For any such path $P$ we have
	$$ \sum_{e \in P} \norm{\Xb_e}_1 \geq \norm{X_{a}-X_{b}}_1.$$
\end{proof}

In the following theorem we show that there is no  PD shortcut matrix $D$ that reduces the average effective resistance of all cuts of the graph of \autoref{fig:maxeffresbadexample} to $o(1)$.
\begin{theorem}
\label{thm:avgeffres}
For any $h>k>2$, the optimum of $\averagecp$ for the graph of \autoref{fig:maxeffresbadexample}, is at most 
$$ \frac{h^2}{8(h+k)^2}.$$
\end{theorem}
\begin{proof}
%We show that for any integer $h>2$, the optimum of $\averagecp$ for the graph illustrated in \autoref{fig:maxeffresbadexample} is at least
%$\Omega(\frac{h^2}{(h+k)^2})$.
%In other words, we show that for any matrix $0\prec D\preceq_\square L_G$,
%there is a cut $(S,\overline{S})$ such that
%$$\uE{e\sim E(S,\overline{S})} \X_e D^{-1} \X_e \geq \frac{h^2}{8(h+k)^2}. $$ 
Fix $k,h$ and let $G$ be the graph of \autoref{fig:maxeffresbadexample}. By \eqref{eq:averagecpdual} it is enough to construct a cut metric $X$, a semiorthogonal matrix $U$, and a distribution $\lambda$ on the cuts of $G$ such that 
\begin{equation}\label{eq:bigdualbadgraph} \frac{\left(\sum_{e\in E} \sqrt{\gamma_e} \cdot \langle U^e,\Xb_e\rangle\right)^2}{\sum_{e\in E} \norm{\Xb_e}^2} \leq \frac{h^2}{8(h+k)^2}.	
\end{equation}

First, we construct $X$ and we calculate the denominator, then we define $U$ and $\gamma,\lambda$ and we upper bound the numerator.
Let $n=2^h$ (so $G$ has $n+1$ vertices). Let $X\in \{0,1\}^{n\times (n+1)}$  where for any vertex $0\leq i\leq 2^h$,  $X_i:=\bone_{[i]}$, i.e., $X_i$ is $1$ in the first $i$ coordinates and  $0$ otherwise. So, $X_0=\bzero$.
It follows that 
$$ \sum_{\{i,j\}\in E} \norm{X_i-X_j}_1 = n\cdot k +n\cdot h.$$
So, it remains to upper bound the numerator. Next, we define the semiorthogonal matrix $U$.

We define a semiorhogonal matrix $U$ by describing the vectors that we assign to a carefully chosen set $E'$ of ``long'' edges of $G$. 
%We only assign vectors to the long edges. 
For each $1\leq i\leq h$, we assign a vector to each of the edges $\{0,2^i\}, \{2\cdot 2^i, 3\cdot 2^i\}, \{4\cdot 2^i, 5\cdot 2^i\}, \dots$;  we assign the following vector to the edges $\{2j\cdot 2^i, (2j+1)\cdot 2^i\}$:
$$ U^{\{2j\cdot 2^i, (2j+1)\cdot 2^i\}}=
\kbordermatrix{
& 1  & & & &  2j\cdot 2^i & & & & (2j+1)\cdot 2^i  & & & & (2j+2)\cdot 2^i\\
& 0  & \dots & 0 & \vrule & \frac{+1}{\sqrt{2^i}} &  \dots & \frac{+1}{\sqrt{2^i}} & \vrule & \frac{-1}{\sqrt{2^i}}  & \dots & \frac{-1}{\sqrt{2^i}} & \vrule & 0  & \dots & 0
}.
$$
Note that the above vector is only nonzero in the coordinates $2i\cdot 2^j$ to $(2i+2)\cdot 2^j-1$; it is equal to $1/\sqrt{2^j}$ in the first half of these coordinates and $-1/\sqrt{2^j}$ in the second half. For example the rows of $U$ corresponding to the 3 top layers of long edges look as follows:
$$
\kbordermatrix{
      &  1 &     & &   &  n/4 & & & & n/2 & & & & 3n/4 & &\\
%     U^{\{0,n\}} &  \frac{+1}{\sqrt{n}} & \vrule & \ldots & \vrule & \frac{+1}{\sqrt{n}} & \vrule &    \ldots & \vrule & \frac{+1}{\sqrt{n}} & \vrule & \ldots & \vrule & \frac{+1}{\sqrt{n}} & \vrule &\ldots\\
     U^{\{0,n/2\}}& \frac{+1}{\sqrt{n}} & &  & \dots &  &   & \frac{+1}{\sqrt{n}}& \vrule & \frac{-1}{\sqrt{n}} &  &  & \dots &  &  &\frac{-1}{\sqrt{n}} \\
    U^{\{0,n/4\}}& \frac{+\sqrt{2}}{\sqrt{n}} &  & \ldots &  \vrule & \frac{-\sqrt{2}}{\sqrt{n}} &  &    \dots & \vrule & 0 &  &  & \dots &  &  & 0 \\
    U^{\{n/2,3n/4\}} & 0 &  &  & \dots &  &  &  0 & \vrule & \frac{+\sqrt{2}}{\sqrt{n}} &  & \ldots & \vrule & \frac{-\sqrt{2}}{\sqrt{n}} &  &\ldots \\
 }
$$
Note that $X$ can be extended to a matrix in $\{0,1\}^{E\times (n+1)}$ by adding zero rows, and  $U$ can be extended to an orthogonal matrix in $\R^{E\times E}$. 

%\caption{An example of a unitary matrix $U$ for $n=4$ in proof of \autoref{thm:avgeffres}.}
%\label{tab:unitaryexample}
%\end{table}

 By the above construction  for each edge $ e=\{2j\cdot 2^i, (2j+1)\cdot 2^i\}\in E'$,
\begin{equation}
\label{eq:UXxebadgraph}
\langle U^e,\Xb_e\rangle  = \frac{2^i}{2^{-(i+1)/2}} = 2^{(i-1)/2}.
\end{equation}
Therefore, we can write the LHS of \eqref{eq:bigdualbadgraph} as follows:
\begin{eqnarray*}
\frac{\left(\sum_{e\in E} \sqrt{\gamma_e} \cdot  \langle U^e, \Xb_e\rangle\right)^2}{\sum_{e\in E} \norm{\Xb_e}^2 }
&\geq& \frac{\left(\sum_{e=\{j\cdot 2^i,(j+1)\cdot 2^i\}\in E'} \sqrt{\gamma_e} \cdot 2^{(i-1)/2}\right)^2}{n\cdot k+n\cdot h}\\
%&\geq &
%\frac{\Big(\sum_{i=0}^{h-1} \sqrt{} \Big)^2}{n(k+h)}
\end{eqnarray*}
Note that we have an inequality because edges not in $E'$ may have nonzero projection on the corresponding rows of $U$.

%In the above, for an edge $\{u,v\}$, $\gamma_{\{u,v\}}:=\sum_{S: u\in S, v\in \overline{S}} \frac{\lambda_{(S,\overline{S})}}{|E(S,\overline{S})|}$ where $\lambda_{(S,\overline{S})}$ is
%a probability distribution on all cuts of $G$ and the maximum is taken over all choices of $\lambda$.

Now, let us define the distribution $\lambda$. Let $\lambda_{(S,\overline{S})} = 1/n$ for every cut $(\{0,1,\ldots,\l\}, \{\l+1,\ldots,n+1\})$ for all $0\leq \l\leq n-1$. Then, for any edge $\{2j\cdot 2^i, (2j+1)\cdot 2^i\}$, 
$$ \gamma_{\{2j\cdot 2^i, (2j+1)\cdot 2^i\}} = \sum_{2j\cdot 2^i \leq \l < (2j+1)\cdot 2^i} \frac{1}{n \cdot |E(\{0,\ldots,\l\},\{\l+1,\ldots,n+1\})|} \geq \frac{2^i}{n\cdot (h+k)}$$
In the above inequality we use the fact that the sum is over $2^i$ many cuts, and each ``threshold cut'' $(\{0,1,\dots,\l\},\{\l+1,\dots,n\})$ cuts at most $k+h$ edges of $G$.

Therefore, the optimum of \averagecp~is at least,
\begin{eqnarray*} \frac{\Big(\sum_{i=0}^{h-1} \sum_{0\leq 2j < 2^{h-i}} \sqrt{\frac{2^i}{n\cdot (h+k)}} \cdot 2^{(i-1)/2}\Big)^2}{n\cdot (h+k)} &\geq&
 \frac{\Big(\sum_{i=0}^{h-1}  n \cdot 2^{-i-1} \cdot \sqrt{\frac{2^i}{n\cdot (h+k)}}\cdot 2^{(i-1)/2} \Big)^2}{n\cdot (h+k)}\\
 &= & \frac{\Big(2^{-3/2}\cdot h\cdot \sqrt{\frac{n}{h+k}} \Big)^2}{n\cdot (h+k)} = \frac{h^2}{8(h+k)^2}
\end{eqnarray*}
%Now, we show that the optimum of \sdp~\eqref{cp:main} is $\tilde{O}(\sqrt{k})$ for the graph in \eqref{fig:maxeffresbadexample}.
\end{proof}

Let us conclude this section by demonstrating that $\maincp$ performs better than $\averagecp$ for the graph of 
 \autoref{fig:maxeffresbadexample} with respect to the \expandertree~of \autoref{fig:hierarchicaltree}.
  % with that mapping the we described  in \autoref{thm:metricratio} 
%by $O(1/k)$.
Let $\cT$ be the  tree shown in \autoref{fig:hierarchicaltree} and $T=\{1,2,\ldots,2^h\}$.
Let $X$ and $U$ be the cut metric and the orthogonal matrix constructed in the proof of \autoref{thm:avgeffres}, respectively.
Let us estimate $\sum_{e\in \cut(t)} \langle\Ue, \Xb_e\rangle$
for nodes $2^i\in T$; the rest of the terms can be estimated similarly.
%First, for any node $t_i \in \cT$, $V(t_i)=\{0,1,\ldots,i\}$.
For node $2^i$, $\cut(2^i)$ has $k$ copies of the edge $\{2^i-1,2^i\}$ and for each $1\leq j\leq i$, it has an edge $\{2^i-2^j,2^i\}$. 
By \eqref{eq:UXxebadgraph}, for each edge $e=\{2^i-2^j,2^i\}$, $\langle\Ue, \Xb_e\rangle = 2^{(j-1)/2}.$
Therefore, for any node $2^i$,
 $\left(\sum_{e\in \cut(2^i)} \langle\Ue, \Xb_e\rangle\right)^2$ is a geometric sum and we can approximate it with the largest term, i.e., $\max_{e\in\cut(2^i)} \langle\Ue, \Xb_e\rangle^2$. Therefore,
\begin{eqnarray*} \sum_{t\in T} \frac{1}{|\cut(t)|} \left(\sum_{e\in\cut(t)} \langle\Ue, \Xb_e\rangle\right)^2 &\lesssim & \sum_{t\in T} \frac{1}{|\cut(t)|}\max_{e\in\cut(t)} \langle\Ue, \Xb_e\rangle^2 \\
&\lesssim& \sum_{t\in T} \frac{1}{k} \max_{e \in \cut(t)} \langle\Ue, \Xb_e\rangle^2 \leq \frac1k\sum_{e\in E} \norm{\Xb_e}^2.
\end{eqnarray*}
In the second inequality we use the crucial fact that each edge $e$ is contained in $\cut(t)$ for at most two nodes of $\cT$ and that $|\cut(t)|\geq k$ for all $t$. So, $\maincp(\cT)\leq O(1/k)$.

\section{Upper-bounding the Numerator of the Dual}
\label{sec:balls}
In the rest of the paper we prove the following theorem.
\begin{restatable}{theorem}{metricratiothm}
\label{thm:metricratio}
For any $k$-edge-connected graph $G=(V,E)$ and any  $(k,\lambda,T)$-LCH, of $G$, and for $h>0$, any cut metric $X\in\{0,1\}^{h\times V}$, and any semiorthogonal matrix $U\in \R^{E\times h}$,
\begin{equation}
\label{eq:metricratiogoal}
 \frac{\sum_{t\in T} \frac{1}{|\cut(t)|}\left(\sum_{e\in \cut(t)} \langle\Ue, \Xb_e\rangle\right)^2 }{\sum_{e\in E} \norm{\Xb_e}^2} \leq \frac{f_1(k,\lambda)}{k}.
 \end{equation}
\end{restatable}
Recall that $f_1(k,\lambda)$ is a polylogarithmic function of $k,1/\lambda$. 
Observe that the above theorem together with \autoref{lem:sdpdual} implies
\autoref{thm:main}. 

In the rest of the paper we fix $U,X$ and we upper-bound the above ratio by  $\polylog(k,1/\lambda)/k$.
We also identify every vertex $v$ with its map $X_v$.

Before getting into the details of the proof let us describe how  $k$-edge-connectivity blends into our proof.
In the following simple fact we show that to lower bound the denominator it is enough to find many disjoint $L_1$ balls centered at the vertices of $G$ with large radii. 
\begin{fact}\label{fact:ballradii}
For any $X:V\to\{0,1\}^h$ and any set of $\l\geq 2$ disjoint $L_1$ balls $B_1,\dots,B_\l$ centered at vertices of $G$ with radii $r_1,\dots,r_\l$ we have
$$ \sum_{i=1}^\l r_i\cdot k \leq \sum_{e\in E} \norm{\Xb_e}^2.$$	
\end{fact}
Since there are $k$ edge-disjoint paths connecting the center of each ball to the outside, by the triangle inequality, the sum of the $L_1$ length of the edges of the graph is at least $k$ times the sum of the radii of the balls. Note that if $\l=1$, i.e., if we have only one ball, the conclusion does not necessarily hold. This is because $B_1$ may contain all vertices of $G$.

Now, let us give a high-level overview of the proof of \autoref{thm:metricratio}
The main proof consists of two steps; in the rest of this section we upper-bound the numerator by a quantity defined on a geometric object which we call a sequence of bags of balls. Then, in the next section we lower-bound the denominator by  $\Omega(k)$ times the same quantity. 
The main result of this section is \autoref{lem:geometricsequence}, in which we construct a geometric sequence  of bags of $L_1$ balls, $\cB_1,\cB_2,\dots$, 
centered at the vertices of $G$ such that balls in each $\cB_i$ are
disjoint and their radii are exactly equal to $\br_i$, where $\br_1$, $\br_2$, $\dots$ form a $poly(k)$-decreasing geometric sequence.
We guarantee that the numerator is within a $\polylog(k,\lambda)$
factor of the sum of the radii of balls in the geometric sequence.

In \autoref{sec:charging}
we lower-bound the denominator, i.e., the sum of the $L_1$ lengths of the edges by $\Omega(k)$ times the sum of radii of the balls in our geometric sequence. 
At the heart of our dual proof in  \autoref{sec:charging}, we  use an inductive argument with no loss in $n$. 
We prove that under some technical conditions on $\cB_1$,$\cB_2$,$\dots,$ we can construct a set of label-disjoint ({\em hollowed}) balls  such that the sum of the radii of these (hollowed) balls is a constant factor of the sum of the radii of  balls in the given geometric sequence; by label-disjoint balls we mean that we can assign a set of nodes  $\cC(B)\subset \cT$ to each (hollowed) ball $B$, called the \conflictset~of $B$, such that for any two intersecting (hollowed) balls $B$ and $B'$, $\cC(B)\cap \cC(B')=\emptyset$. Furthermore, we use properties of the \expandertree~to ensure that for each (hollowed) ball $B$, there are $\Omega(k)$ edge-disjoint paths, supported on the vertices of $G$ in $\cC(B)$, crossing $B$.
%we can label each ball with a node $t$ of $\cT$ such that  any two intersecting balls are labeled nodes of $\cT$ which are not ancestor-descendant.

%\textcolor{red}{edited to here!}

In the rest of this section we construct a geometric sequence of bags of balls such that the sum of the radii of balls in the sequence  is at least  the numerator of \eqref{eq:metricratiogoal} up to $\polylog(k,1/\lambda)$ factors (see \autoref{lem:geometricsequence} for the final result of this section). First, in \autoref{subsec:disjointballs} we prove a technical lemma; we show that if the average projection of a set $F$ of edges  on $U$ is ``comparable'' to the average squared norm of these edges, then we can construct a large number of disjoint  balls centered at the endpoints of edges of $F$. We use this technical lemma to show that we can reduce the average effective resistance of a set $F$ of edges of any $k$-edge-connected graph to $\tilde{O}(1/k)$.
Then, in \autoref{subsec:bagsofballs} we  group these balls into several bags of balls.
%these balls into bags such that all balls in the same bag have equal radius. 
Finally, in \autoref{subsec:geometricballs}
we partition the edges of $G$ into parts that have similar projections onto $U$ and for each part we  
use the result of \autoref{subsec:bagsofballs} to find a family of bags of balls. Putting these families together we obtain a geometric sequence of families of bags of balls. 
%in the embedding, i.e., edges $e$ where $\norm{\Xb_e}$ and $U^e \Xb_e$ is equal up to constants.

%Roughly speaking, our goal is to prove the following geometric question: Suppose we have an $L_1$ embedding of $G$,  and we are given a geometric sequence of families of $L_1$ balls $\cB_1$, $\cB_2$, $\etc$ such that centers of all balls in all families are vertices of $G$, balls in each $\cB_i$ are
%disjoint and their radii are exactly equal to $\br_i$, where $\br_1$, $\br_2$, $\etc$ form a geometric sequence where the scale factor is $\poly(k)$. 

\begin{figure}\centering
	\begin{tikzpicture}
		\tikzstyle{every node} = [draw,fill,circle,minimum size=1.3mm,inner sep=0];
		\foreach \i/\x/\y  in {1/0/2, 2/3/0, 3/3/4, 4/7/2}{
			\node at (\x,\y) (v_\i) {};
			\foreach \l in {a,b,c}{
				\node at (\x,\y) (\l_\i_0) {};
			}
		}
		\draw (v_1) circle (1.5) (v_4) circle (2.3) (v_2) circle (1) (v_3) circle (1.8);
		\foreach \i/\r/\t in {1/1.5/45,2/1/60,3/1.8/100,4/2.3/280}{
			\foreach \j/\k/\a in {0/1/0, 1/2/-90, 2/3/0, 3/4/-90}{
				\node at ($(a_\i_\j)+.35*\r*(\t+\a:1)$) (a_\i_\k) {} edge [color=red,line width=1.2pt](a_\i_\j);
			}
		}
		\foreach \i/\r/\t in {1/1.5/150,2/1/170,3/1.8/200,4/2.3/30}{
			\foreach \j/\k/\a in {0/1/0, 1/2/-80, 2/3/10, 3/4/-70}{
				\node at ($(b_\i_\j)+.35*\r*(\t+\a:1)$) (b_\i_\k) {} edge [color=blue,line width=1.2pt](b_\i_\j);
			}
		}
		\foreach \i/\r/\t in {1/1.5/270,2/1/260,3/1.8/300,4/2.3/150}{
			\foreach \j/\k/\a in {0/1/0, 1/2/-60, 2/3/30, 3/4/-80}{
				\node at ($(c_\i_\j)+.35*\r*(\t+\a:1)$) (c_\i_\k) {} edge [color=brown,line width=1.2pt](c_\i_\j);
			}
		}
	\end{tikzpicture}
	\caption{Sets of $k$ edge-disjoint paths in disjoint $L_1$ balls.}
	\label{fig:disjointballs}
\end{figure}
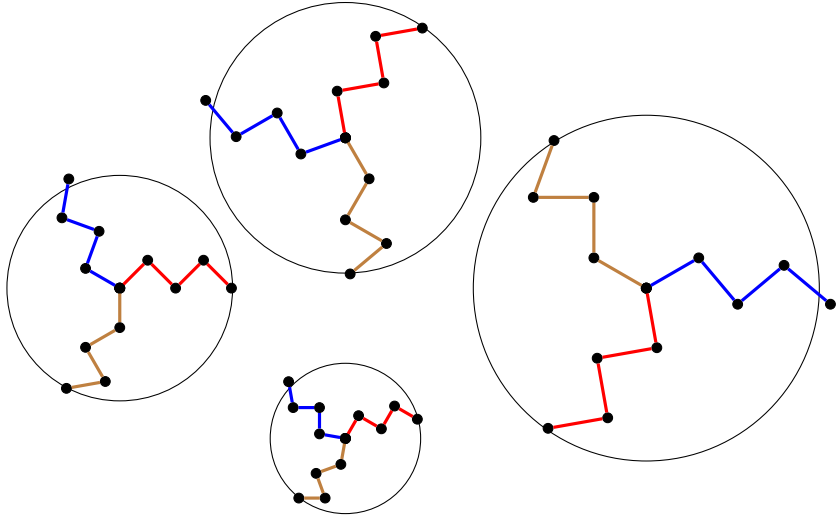

\subsection{Construction of Disjoint $L_1$ Balls}
\label{subsec:disjointballs}
%Let $Y=UX\in\R^{E\times E}$. Since $U$ 
In this section we prove the following proposition; although we do not directly use this proposition in the proof of our main technical theorem, we do use the main tool of the proof, \autoref{lem:orthogonal}, as one of the key components of the proof for the main technical theorem.
\begin{proposition}\label{prop:Favgeffres}
For any $k$-edge-connected graph $G=(V,E)$ and any set $F\subseteq E$, there is a PD shortcut matrix $D$ that reduces the average effective resistance of the edges of $F$ to $\tilde{O}(1/k)$.
\end{proposition}

By \autoref{lem:sdpdual} it is enough to show that for any $X\in\{0,1\}^{h\times V}$ and any semiorthogonal matrix $U\in \R^{E\times h}$,
$$  \frac{\frac{1}{|F|}\left(\sum_{e\in F} \langle U^e,\Xb_e\rangle\right)^2}{\sum_{e\in E} \norm{\Xb_e}^2} = \frac{\left(\mE_{e\sim F} \langle U^e,\Xb_e\rangle\right)^2}{\frac1{|F|}\sum_{e\in E} \norm{\Xb_e}^2}\leq \tilde{O}(1/k).$$
Let $Y=UX$ and $\Yb=Y\X=UX\X$. Note that since $U$ is semiorthogonal, $\norm{\Yb_e}^2 \leq \norm{\Xb_e}^2$ for all $e$. Without loss of generality assume that 
$$ \frac{\left(\mE_{e\sim F} \Yb_{e,e}\right)^2}{\mE_{e\sim F} \norm{\Yb_e}^2} \geq \alpha,$$
for $\alpha = \polylog(k)/k$; otherwise we are done. 
In the following lemma we show that assuming the above inequality we can construct $b$ disjoint $L_2^2$ balls of radius $r$ centered at the vertices of the endpoints of edges of $F$ such that
$$ r\cdot b \geq \frac{\alpha^\eps}{\poly(\eps)} \cdot \left(\mE_{e\sim F} \Yb_{e,e}\right)^2 |F|.$$
On the other hand, since these balls are disjoint, by \autoref{fact:ballradii},
$$ r\cdot b \leq \frac1k \sum_{e\in E} \norm{\Xb_e}^2.$$
Note that we really need to apply \autoref{fact:ballradii} to balls in the space of $X_v$'s, since $Y_v$'s do not necessarily satisfy the triangle inequality. However, given disjoint balls centered around $Y_v$'s, one can take the same balls around the corresponding $X_v$'s and they will remain disjoint, since $U$, the mapping from $X_v$ to $Y_v$, is a contraction.

Now, the above proposition simply follows by the above two inequalities for $\eps=\log k/\log\log k$.

%The following is the main result of this subsection.
\begin{lemma}
\label{lem:orthogonal}
Given  $F\subseteq E$ and a mapping $Y\in\R^{E\times V}$ such that
\begin{eqnarray}
\Upsilon:=\Big( \uE{e\sim F} \Yb_{e,e}\Big)^2 \geq %\aupp\\
%\geq \alpha \cdot 
\alpha\cdot \uE{e\sim F} \norm{\Yb_e}_2^2,
\label{eq:averagebound}
 \end{eqnarray}
 for some $\alpha > 0$, for any $0<\eps<1/2$, there are $b$ disjoint $L_2^2$ balls $B_1,\ldots,B_b$ with radius $r$ such that the center of each ball is an endpoint of an edge in $F$, $b\geq  \alpha |F|/C_1(\eps)$, and 
$$ r\cdot b\geq \frac{\alpha^\eps \cdot\Upsilon\cdot |F|}{C_1(\eps)},$$
where $C_1(\eps)$ is a polynomial function of $1/\eps$.
\end{lemma}
Before getting to the proof of the lemma, let us give an intuitive description of the statement of the lemma.
The extreme case is for $\alpha\approx 1$. Observe that the inequality   \eqref{eq:averagebound} enforces a very strong assumption on the mapping $\Yb$.  Since for any edge $e$, $\Yb_{e,e} \leq \norm{\Yb_e}$, and $\alpha\approx 1$, the  following two conditions must hold for $\Yb$: 
\begin{enumerate}[i)]
\item For most edges $e\in F$, $\Yb_{e,e}\approx \norm{\Yb_e}$,
\item For most pairs of edges $e,f\in F$, $\norm{\Yb_e}\approx \norm{\Yb_f}$.
\end{enumerate}
The above two conditions essentially imply that the vectors $\{\Yb_e\}_{e\in F}$ form an orthonormal basis up to normalizing the size of the vectors. It is an exercise to see that in this case one can select $\Omega(|F|)$ many $L_2^2$ balls of radius $\Omega(\Upsilon)$ around the endpoints of the edges in $F$; one can show that greedily picking balls that do not intersect each other works.

Our proof can be interpreted as a robust version of this argument.

\begin{proofof}{\autoref{lem:orthogonal}}
For a radius $r>0$,   run the following greedy algorithm. Scan the endpoints of the edges in an arbitrary order;  for each point $Y_u$, if the $L_2^2$ ball $B(Y_u,r)$ doesn't touch the balls that we have already selected, select $B(Y_u,r)$. Suppose we manage to select $b$ balls. We say the algorithm succeeds if both of the lemma's conclusions are satisfied. In the rest of the proof we show that this algorithm always succeeds for some value of $r$.

%Perhaps after reording the rows of matrix $Y$, we assume that $1\leq \pi(e)\leq |F|$ for all $e\in F$. Let $A\in\R^{m\times |F|}$ where  for each $1\leq i\leq |F|$, 
%the $i$-th column of $A$ is $Y_{\pi^{-1}(i)}$.
%Note that by definition of $A$, $\trace(A) = \sum_{e\in F} \Ie Y_e.$

Without loss of generality, in the rest of the proof we drop the columns of $\Yb$ corresponding to edges $e\notin F$ and their corresponding rows and we assume $Y\in\R^{F\times F}.$ Note that by removing the rows, we are decreasing $\norm{\Yb}_F^2$, but this only weakens the assumption of the lemma.
Let $\sigma_1\geq \sigma_2\geq \dots\geq \sigma_{|F|}$ be the singular values of $\Yb$. 
We can rewrite the assumption of the lemma as follows:
\begin{equation}
	\label{eq:ratioalpha}
\left(\frac1{|F|}\sum_i \sigma_i\right)^2\geq \left(\frac{\trace(\Yb)}{|F|} \right)^2=\left(\mE_{e\sim F}\Yb_{e,e}\right)^2
\geq\alpha\cdot\mE_{e\sim F}\norm{\Yb_e}^2 = \frac{\alpha}{|F|} \norm{\Yb}_F^2 = \frac{\alpha}{|F|}\sum_{i=1}^{|F|} \sigma_i^2.
\end{equation}
The first inequality follows by \autoref{lem:matrixtrace}.
Note that, for $\alpha = 1$, the LHS is always less than or equal to the RHS by the Cauchy-Schwarz inequality with equality happening only when $\sigma_1=\dots=\sigma_{|F|}$. So, for large $\alpha$ the above inequality can be seen as a reverse Cauchy-Schwarz inequality.

In the next claim, we show that if the above algorithm finds a ``small number'' $b$ of balls for a choice of $r$, this means that $\sigma_{b},\dots,\sigma_{|F|}$ are significantly smaller than $\sigma_1,\dots,\sigma_{b-1}$. In the succeeding claim we use the above reverse Cauchy-Schwarz inequality to show that this is impossible.
\begin{claim}
\label{eq:greedyfailure}
Given $r>0$, suppose that the greedy algorithm finds $b$ disjoint balls of radius $r$. Then
$$ r \geq \frac{1}{16|F|} \sum_{i=b}^{|F|} \sigma_i^2.$$ 
\end{claim}
\begin{proof}
We construct a low-rank matrix $C\in\R^{F\times F}$. Then, we use
 \autoref{thm:hoffmanineq} to prove the claim. 
Let $Y_{w_1},\ldots,Y_{w_b}$ be the centers of the chosen balls. Then, for any endpoint $v$
of an edge in $F$, let $c(v)$ be the closest center to $Y_v$, i.e., 
$$ c(v):=\argmin_{w_i} \norm{Y_{w_i}- Y_v}^2_2$$
We construct a matrix $C\in\R^{F\times F}$ such that the $e$-th column of $C$ is defined as follows: say the $\{u,v\}$-th column of $\Yb$ is $Y_{u}-Y_{v}$ for $\{u,v\}\in F$; we let the $\{u,v\}$-th column of $C$ be $Y_{c(u)}-Y_{c(v)}$.
%Recall that $i$-th column of $A$ is $Y_{u}-Y_{v}$ for $\{u,v\}=\pi^{-1}(i)$. % $\{u_1,v_1\},\ldots,$ is an ordering of edges of $F$.  For each $1\leq i\leq |F|$, we let $\cen(Y_{u_i})-\cen(Y_{v_i})$ be the $i$-th column of $C$.
By definition,  $\rank(C)\leq b-1$, since $C$'s columns are a subset of the differences between $b$ points.

%\begin{equation}
%\label{eq:centerdist}
% \min_{1\leq i\leq b} \norm{Y_v-Y_i}_2^2 \leq 2r.
% \end{equation}
%
%%Let $B\in\R^{n\times|F|}$ where the $i$-th column of $B$ is $\Y_{\sigma^{-1}(i)}$, and let
%has a very low rank, $\rank(C)\leq b$. 
First, notice that
\begin{eqnarray*} \norm{\Yb-C}_F^2 &=&\sum_{\{u,v\}\in F} \norm{(Y_u-Y_v) - (Y_{c(u)}-Y_{c(v)})}_2^2\\
&\leq & \sum_{\{u,v\}\in F} (\norm{Y_u-Y_{c(u)}}_2 + \norm{Y_v - Y_{c(v)}}_2)^2\\
&\leq&  \sum_{\{u,v\}\in F} 2\norm{Y_u- Y_{c(u)}}_2^2 + 2\norm{Y_v-Y_{c(v)}}_2^2 \leq
16r\cdot |F|,
\end{eqnarray*}
where the first inequality follows by the triangle inequality and the last inequality follows by the
definition of greedy algorithm; in particular, for any point $v$, in the worst case there is a point $p$ in the $L_2^2$ ball about $c(v)$ such that $\norm{p-Y_v}^2< r$, so
$$(\norm{Y_{v}-Y_{c(v)}}_2)^2 \leq (\norm{Y_{v}-p} + \norm{Y_{c(v)}-p})^2 \leq (\sqrt{r}+\sqrt{r})^2\leq 4r.$$ 
Now by \autoref{thm:hoffmanineq},
$$ 16r\cdot |F| \geq \norm{\Yb-C}_F^2 \geq \sum_{i=b}^{|F|} \sigma_i^2.$$
where the second inequality uses the fact that $\rank(C)\leq b-1$. 
\end{proof}

All we need to show is that there is a value of $b\geq \alpha|F|/C_1(\eps)$ such that  $\frac{b}{16|F|} \sum_{i=b}^{|F|}\sigma_i^2 \geq \frac{\alpha^\eps \Upsilon\cdot |F|}{C_1(\eps)}$. 
%This is proved in following claim; because the proof of the next claim is quite technical we first prove a simpler bound $r\cot b \gtrsim \sqrt{\alpha}\Upsilon/|F|$. %The reader may skip the proof of the proof of the following 

\begin{claim}
There is a universal  function $C_1(\eps)$ that is polynomial in $1/\eps$  such that for any $0<\eps \leq 1$ %and $b_0:=\alpha|F|/C_1(\eps)$, 
there is an integer $b\geq \alpha|F|/C_1(\eps)$  such that 
$$ \frac{b}{16|F|} \sum_{i=b}^{|F|} \sigma_i^2 \geq \frac{\alpha^
\eps \cdot \Upsilon\cdot |F|}{C_1(\eps)}.$$
\end{claim}
\begin{proof}
Let us first prove the claim for $\eps=1$; this special case reveals the meat of the argument.
We show the claim holds for $b=\alpha|F|/4$.
%First, if $\sum_{i=b_0+1}^{|F|} \sigma_i^2 \gtrsim \Upsilon/|F|^2$ we are done.
%Otherwise, 
By the Cauchy-Schwarz inequality,
\begin{eqnarray*}
\left(\frac{1}{|F|}\sum_{i=1}^{|F|}\sigma_i\right)^2 &\leq& 2\left(\frac{1}{|F|}\sum_{i=1}^{b-1}	\sigma_i\right)^2 + 2\left(\frac{1}{|F|}\sum_{i=b}^{|F|}\sigma_i\right)^2\\
&\leq & \frac{2b}{|F|^2} \sum_{i=1}^{b-1} \sigma_i^2  + \frac{2}{|F|}\sum_{i=b}^{|F|}\sigma_i^2\\
& = & \frac{\alpha}{2|F|}\sum_{i=1}^{b-1} \sigma_i^2 + \frac{8b}{\alpha |F|^2}\sum_{i=b}^{|F|}\sigma_i^2 \leq \frac12\left(\frac1{|F|}\sum_{i=1}^{|F|}\sigma_i\right)^2 + \frac{8b}{\alpha|F|^2}\sum_{i=b}^{|F|}\sigma_i^2.
\end{eqnarray*}
where the equality uses the definition of $b$ and the last inequality uses \eqref{eq:ratioalpha}. Therefore, 
$$\frac{\Upsilon}{2} \leq \frac12\left(\frac1{|F|}\sum_{i=1}^{|F|}\sigma_i\right)^2 \leq \frac{8b}{\alpha|F|^2}\sum_{i=b}^{|F|}\sigma_i^2,$$
where the first inequality uses another application of \eqref{eq:ratioalpha}. This proves the claim for $\eps=1$ and $C_1(\eps) \leq 1/256$. 

Now, we prove the claim for $\eps <1$. Let $b_0 \geq \frac{\alpha|F|}{C_1(\eps)}$ be an integer that we fix later. Let $z:=\max_{b\geq b_0} \frac{b}{16|F|} \sum_{i=b}^{|F|} \sigma_i^2.$ To prove the claim, it is enough to lower bound $z$.
 First, by the definition of $z$, for all $b\geq b_0$,
\begin{equation}
\label{eq:zlowerbounds}
 \frac{z}{|F|^{\eps}\cdot b^{1-\eps}} \geq \frac{b^{\eps}}{16 |F|^{1+\eps}} \sum_{i=b}^{|F|} \sigma_i^2.
 \end{equation}
%Note that the above inequality does not bound the first $b_0$ eigenvalues of $\Yb$. However, 
On the other hand, by \eqref{eq:ratioalpha},
%use the Frobenius norm of $\Yb$ to bound those eigenvalues.
\begin{equation}
\label{eq:frobeniusz}
\frac{1}{|F|}\cdot \sum_{i=1}^{|F|} \sigma_i^2  \leq \Upsilon/\alpha.
%= \frac{\norm{\Yb}_F^2}{|F|} =  \uE{e\sim F} \norm{\Yb_e}^2_2 \leq \frac{1}{\alpha}  \Big(\uE{e\sim F}  \Yb_{e,e}\Big)^2 = \Upsilon/\alpha.
\end{equation}
%where the first inequality uses \eqref{eq:ratioalpha}.

%The inequality follows by  \eqref{eq:averagebound}.

Let $\beta>0$  be a parameter that we fix later. Summing up \eqref{eq:zlowerbounds} for all $b_0 \leq b\leq |F|$ and $\beta$ times \eqref{eq:frobeniusz}, we get
$$
 \sum_{i=1}^{|F|} \Big(\beta+ \frac{\int_{x=b_0-1}^{i} (x-1)^\eps dx}{16 |F|^{\eps}}\Big) \cdot \frac{\sigma_i^2}{|F|} \leq \frac{\beta\cdot \Upsilon}{\alpha} + \frac{z}{|F|^{\eps}} \int_{x=b_0}^{|F|} \frac{dx}{(x-1)^{1-\eps}}.
$$
Note that the integral on the LHS lower-bounds $\sum_{b_0\leq b \leq i-1} b^\eps$ and the integral on the RHS upper-bounds $\sum_{b_0 \leq b < |F|} 1/b^{1-\eps}$.
So,
\begin{equation}
\label{eq:sumsigmas}
 \sum_{i=1}^{|F|} \Big(\beta+ \frac{[(i-1)^{1+\eps} - (b_0-1)^{1+\eps}]_+}{32 |F|^{\eps}}\Big) \cdot \frac{\sigma_i^2}{|F|} \leq \frac{\beta\cdot \Upsilon}{\alpha} + \frac{z}{|F|^{\eps}} \cdot \frac{(|F|-1)^\eps}{\eps} \leq \frac{\beta\cdot \Upsilon}{\alpha} + \frac{z}{\eps}.
\end{equation}
where for $x\in\R$, $[x]_+=\max\{x,0\}$.

%Now, we use equation \eqref{eq:ratioalpha} to lower-bound the singular values of $A$.
%On the other hand, by \eqref{eq:ratioalpha},
%$ \frac{1}{|F|}\cdot \sum_{i=1}^{|F|} \sigma_i 
%\geq \sqrt{\Upsilon}$.
Therefore, by \eqref{eq:ratioalpha} and Cauchy-Schwarz,
\begin{eqnarray}
\Upsilon \leq \Big(\frac{1}{|F|}\cdot \sum_{i=1}^{|F|}\sigma_i\Big)^2&\leq&  \Big(\sum_{i=1}^{|F|} \Big(\beta+ \frac{[(i-1)^{1+\eps}-(b_0-1)^{1+\eps}]_+}{32|F|^{\eps}}\Big) \frac{\sigma_i^2}{|F|}\Big) \cdot \Big(\sum_{i=1}^{|F|} \frac{1/|F|}{\beta+\frac{[(i-1)^{1+\eps}-(b_0-1)^{1+\eps}]_+}{32|F|^{\eps}}}\Big) \nonumber \\
&\leq &
\Big(\frac{\beta\cdot \Upsilon}{\alpha} + \frac{z}{\eps}\Big) \cdot \frac{32(3+1/\eps)}{\beta^{\frac{\eps}{1+\eps}} |F|^{\frac{1}{1+\eps}}}.
\label{eq:betaFepshard}
\end{eqnarray}
To see the last inequality we need to do some algebra. The first term on the RHS follows from \eqref{eq:sumsigmas}.
We obtain the second term in the last inequality by choosing $b_0 = 1+\beta^{\frac{1}{1+\eps}}|F|^{\frac{\eps}{1+\eps}}$; later we will choose $\beta,C_1(\eps)$ making sure that $b_0 \geq \alpha|F|/C_1(\eps)$. In particular,
\begin{eqnarray*}
\sum_{j=1}^{|F|} \frac{1/|F|}{\beta + \frac{[(j-1)^{1+\eps}-(b_0-1)^{1+\eps}]_+}{32|F|^{\eps}}} &\leq& \frac{b_0-1}{\beta |F|} + 
\sum_{i=1}^{\infty} \sum_{j=(b_0-1) i^{1/(1+\eps)}+1}^{(b_0-1)(i+1)^{1/(1+\eps)}} \frac{32}{i\cdot\beta\cdot|F|} \\
&\leq& \frac{b_0-1}{\beta\cdot|F|} \bigg(1+\sum_{i=1}^{\infty}  \frac{32}{i^{\frac{1+2\eps}{1+\eps}} }\bigg) 
\leq \frac{32(3+1/\eps) (b_0-1)}{\beta|F|} \leq \frac{32(3+1/\eps)}{\beta^{\frac{\eps}{1+\eps}} |F|^{\frac{1}{1+\eps}}},
\end{eqnarray*}
where in the first inequality we used $b_0 \geq 1+\beta^{\frac{1}{1+\eps}}|F|^{\frac{\eps}{1+\eps}}$, in second inequality we used $(i+1)^{\frac{1}{1+\eps}} - i^{\frac{1}{1+\eps}} = i^{\frac{1}{1+\eps}} \big((1+1/i)^{\frac{1}{1+\eps}}-1\big) \leq i^{\frac{-\eps}{1+\eps}},$ 
and in the last inequality we used $b_0 \leq 1+\beta^{\frac{1}{1+\eps}}|F|^{\frac{\eps}{1+\eps}}$.
% and the last inequality uses $\beta < 1/600$. 

\medskip
Now, the claim follows directly from \eqref{eq:betaFepshard}.
Letting $\beta=\frac{\alpha^{1+\eps}|F|}{(192+64/\eps)^{1+\eps}}$, we obtain,
\begin{eqnarray*}
z \geq \frac{\eps \cdot\beta^{\frac{\eps}{1+\eps}}\cdot |F|^{\frac{1}{1+\eps}}\cdot \Upsilon}{32(3+1/\eps)} - \frac{\eps \cdot\beta \cdot\Upsilon}{\alpha} \geq \frac{\alpha^\eps \cdot \Upsilon\cdot |F|}{(192/\eps + 64/\eps^2)^{1+\eps} }.
\end{eqnarray*}
The claim follows by letting $C_1(\eps) = (192/\eps + 64/\eps^2)^{1+\eps}$, and 
noting $b_0=1+\beta^{\frac{1}{1+\eps}}|F|^{\frac{\eps}{1+\eps}}\geq \alpha|F|/C_1(\eps)$.
%So, we let 
%$$ \tau = \beta^{1/(1+\eps)}|S| = \frac{\alpha}{4+2/\eps}$$
\end{proof}
Observe that the above claim implies \autoref{lem:orthogonal}. %
It is sufficient to run the greedy algorithm with the infimum value of $r$ such that the greedy algorithm returns at most $b$ balls. 
\end{proofof}

\subsection{Construction of Bags of Balls}
\label{subsec:bagsofballs}
In this subsection we will state the main result of this section, \autoref{lem:geometricsequence}, and we give a high-level overview of the proof of \autoref{thm:metricratio}. Before that we need to define several combinatorial objects called bags of balls. 

To prove \autoref{thm:metricratio}
we would like to follow a path similar to the proof of \autoref{prop:Favgeffres}, i.e., we would like to construct disjoint $L_1$ balls $B_1,B_2,\dots$ centered at the vertices of $G$ of radius $r_1,r_2,\dots$ such that 
\begin{equation}\label{eq:idealdisjointballnumbound} \sum_i r_i \gtrsim \sum_{t\in T} \frac{1}{|\cut(t)|} \left(\sum_{e\in \cut(t)}\langle U^e,\Xb_e\rangle\right)^2,	
\end{equation}
and then use a variant of \autoref{fact:ballradii}.
This approach completely fails for the example of \autoref{fig:maxeffresbadexample} as we will show next.

\begin{example}
Let $G$ be a modification of the graph of \autoref{fig:maxeffresbadexample} 
with $n=2^h+1$ vertices we remove all long edges of length $2^i\leq h$, and we shift all edges of length $2^i>h$ by $i$ to the right, i.e., if we replace an edge $\{j\cdot 2^i,(j+1)\cdot 2^i\}$ with $\{i+j\cdot 2^i, i+(j+1)\cdot 2^i\}$. It is easy to see that in this new graph the degree of every vertex is at most $O(k)$.

Let $X_0, X_1,\dots,X_{2^h}$ be the embedding of $G$ where $X_i=\bone_{[i]}$ and $U$ be the semiorthogonal matrix we constructed in \autoref{thm:avgeffres}.
Suppose $T$ has all vertices of $G$, i.e., we are minimizizing the average effective resistance of all degree cuts. If we follow the approach in the previous section, to prove \eqref{eq:metricratiogoal}, we need to construct disjoint $L_1$ balls $B_1,B_2,\dots,$ with radii $r_1,r_2,\dots$ such that
$$ \sum_i r_i \gtrsim \sum_{v\in V} \frac{1}{d(v)}\left(\sum_{e\in E(v,V\setminus \{v\})} \langle U^e,\Xb_e\rangle\right)^2. $$
It follows that for the particular choice of $X,U$ the RHS is about 
$$\frac{n\log n}{\max_v d(v)} \gtrsim \frac{n\log n}{k} \gg n,$$
for $k\ll \log n$.
Unfortunately, for any set of disjoint $L_1$ balls centered at vertices of $G$ we have $\sum_i r_i \leq n$. So, it is impossible to prove \eqref{eq:idealdisjointballnumbound} for $k\ll \log n$ using disjoint balls. 
\end{example}

We will deviate from the approach of the previous section in two ways. First, the balls that we construct have different radii, in fact the radii of the balls form a geometrically decreasing sequence with a sufficiently large $\poly(k)$ decreasing factor; secondly, only the balls of the same radii are disjoint, but a small ball can completely lie inside a bigger ball. 

To construct these balls we will group the edges of $G$ based on their lengths into $\log(n)$ buckets and we apply \autoref{lem:orthogonal} to each bucket separately; we actually  have a more complicated bucketing  because we want to make sure that any two edges $e,f$ in one bucket  satisfy $\langle U^e,\Xb_e\rangle\approx \langle U^f,\Xb_f\rangle$ and $\norm{\Xb_e}\approx\norm{\Xb_f}$.

%Now, suppose we can find a set $F$ of edges such that
%The major problem is because of a limitation of the statement of \autoref{lem:orthogonal}; in the worst case for any set $F\subseteq E$, we have
%$$ \sum_{t\in T} \frac{1}{|\cut(t)|} \left(\sum_{e\in \cut(t)}\langle U^e,\Xb_e\rangle\right)^2 \gtrsim \log(n)\cdot \left(\mE_{e\sim F} \langle U^e,\Xb_e\rangle\right)^2$$
%the numbers $\sum_{e\in \cut(t)} \langle U^e,\Xb_e\rangle$ are not necessarily comparable; 
%In this part we construct a family of bags of balls for a set of $4$-similar set of edges. The bags of balls are the building blocks of our geometric proof in \autoref{sec:charging} where we lower-bound the denominator of \eqref{eq:metricratiogoal}.
Since the balls that we construct are not disjoint we can no longer use the simple charging argument of \autoref{fact:ballradii}. Instead, we partition the set of balls of each radii into bags. The balls of a bag must satisfy certain properties that we describe next. These properties of bags of balls will be crucially used in \autoref{sec:charging} to lower bound $\sum_e \norm{\Xb_e}_1$.

\begin{definition}[Bag of Balls]
A bag of balls, $\bag$, is a set of  disjoint $L_1$ balls of equal radii such that
the center of each ball is a point $X_v$ for some $v\in V$.
A bag of balls is of type $(\br)$ if each ball in the bag has radius $\br$.
A bag of balls is of type $(\br,\Delta)$
if in addition to above, the  maximum $L_1$ distance between the centers of the balls in the bag is at most $\Delta$,
\begin{equation}\label{eq:bagdiameter} \max_{B(X_v,\delta),B(X_u,\delta)\in\bag} \norm{X_v-X_u}_1 \leq \Delta.	
\end{equation}
%Since each bag of balls is set, 
We write $|\bag|$ to denote the number of balls in $\bag$.
\end{definition}

\begin{definition}[Compact Bag of Balls]
For $\beta > 0$, a  bag of balls, $\bag$,  with type $(\br,\Delta)$ is {\em $\beta$-compact} if $|\bag|\geq 2$ and
\begin{equation}
\label{eq:compactbags}
 \beta \cdot \Delta \leq  |\bag|\cdot \br.
\end{equation}
\end{definition}
\noindent It follows from the definition that for any compact bag of balls of type $(\br,\Delta)$, $\Delta\geq 2\br$.
%for a universal constant $C_3>0$. Think of $C_3$
%as a big constant that we will fix at the very end of the proof. 
%\begin{definition}
%\label{def:bagofballs}

\begin{definition}[Assigned Bag of Balls]
\label{def:assignedbagballs}
For a \expandertree~$\cT$ and $\beta>0$,  a bag of balls, $\bag$, with type $(\br)$ is $\beta$-assigned to a node $t\in \cT$,  if 
\begin{equation}
\label{eq:assignedbag}
  \beta \cdot |\cut(t)| \leq |\bag|,
\end{equation}
and for each ball $B(X_u,\br)\in \bag$, $u \in V(t)$ and there is an edge $\{u,v\}\in \cut(t)$ such that $\norm{X_{u} - X_{v}}_1 < \br$.
\end{definition}
\noindent We use the convention of writing $\bag_t$
for a bag of balls assigned to a node $t$.

In \autoref{sec:charging} we will show that $\beta$-compact bags of balls with $\beta \geq C$  and $\beta$-assigned bags of balls with $\beta \geq C'/k$ for some universal constants $C,C'$ are enough to lower-bound the denominator of \eqref{eq:metricratiogoal}.

In general, compact bags of balls are significantly easier to handle than assigned bags of balls. Roughly speaking, given a number of compact bag of balls we can use the compactness property to ``carve'' them into disjoint \emph{hollowed} balls such that the sum of the widths of the hollowed balls is at least a constant fraction of the sum of the original radii;  then we use an argument analogous to \autoref{fact:ballradii} to lower bound the denominator (see \autoref{subsec:compactchargin} for the details).

On the other hand, it is impossible to construct disjoint (hollowed) balls out of a given number of assigned bags of balls without losing too much on the sum of the widths. Instead of restricting the (hollowed) balls to be disjoint, in the technical proof presented in \autoref{subsec:assignedcharging}, we label the balls of each bag with the node of the \expandertree~to which it is assigned. We construct a \conflictset, $\cC(B)$, by looking at the subtree rooted at the label of $B$, and pruning some of its subtrees.  We carve the balls and modify the labels in such a way that, in the end, for any two intersecting balls, the \conflictset s are disjoint. Then, to charge the denominator, we use the fact that for any node $t$, $G(t)$ is $k$-edge-connected and thus in any of the remaining balls, we can find $\Omega(k)$ edge-disjoint paths contained in $G(t)$; this argument does not overcharge the edges, because throughout the construction we make sure that these edge-disjoint paths are routed through $\cC(B)$.

%\end{definition}

%Now that we defined types of bags of balls, we can define geometric sequence of these bags. 
\begin{definition}[Family of Bags of Balls]
\label{def:familybagofballs}
A family of bags of balls, $\fbag$, is a set  of bags of balls of the same type such that all balls in all bags are disjoint.
We say a family of compact bags of balls has type $(\br,\Delta)$  if all bags in the family have type $(\br,\Delta)$. 
For a \expandertree, $\cT$, and  $T\subseteq \cT$, we say a family of assigned bags of balls has type $(\br,T)$ if the bags in the family are assigned to distinct nodes of $T$.
%each bag in the family is assigned to a unique node of $T$.
\end{definition}

 We abuse notation and write a ball $B\in\fbag$ if there is a $\bag \in\fbag$ such that $B\in\bag$. Note that two distinct bags in $\fbag$ may have unequal numbers of balls.

To upper-bound the value of the dual we need to find a sequence of families of bags of balls with geometrically decreasing radii. 

\begin{definition}[Geometric Sequence of Assigned Bags of Balls]
For a \expandertree, $\cT$, a $\lambda$-geometric sequence of families of assigned bags of balls is a sequence $\fbag_1$, $\fbag_2$, $\dots$ such that  $\fbag_i$ has type $(\br_i,T_i)$ where $T_1,T_2,\dots$ are disjoint subsets of nodes of $\cT$ and
for all $i\geq 1$,
$$\br_i\cdot \lambda  > \br_{i+1}.$$
\end{definition}
%for some universal constant $C_4$. 
\begin{definition}[Geometric Sequence of Compact Bags of Balls]
A $\lambda$-geometric sequence of families of compact bags of balls is defined respectively as a sequence $\fbag_1$, $\fbag_2$, $\dots,$
such that $\fbag_i$ has type $(\br_i,\Delta_i)$, and for all $i\geq 1$,
$$ \br_i \cdot \lambda > \Delta_{i+1}.$$
\end{definition}

Now, we are ready to describe the main result of this section.
Let us justify the assumption of the \autoref{lem:geometricsequence}. 
\begin{definition}[$\alpha$-bad Nodes]
	We say a node $t\in \cT$ is \emph{$\alpha$-bad} if
\begin{equation}\label{eq:nodealphafraction} \left(\mE_{e\sim\cut(t)} \langle U^e,\Xb_e\rangle\right)^2 \geq \alpha\cdot \mE_{e\sim\cut(t)} \norm{\Xb_e}^2. 	
\end{equation}
\end{definition}
First, observe that if there is no bad node in $T$, then we are done with \autoref{thm:metricratio}.
So, to prove \autoref{thm:metricratio} the only thing that we need to upper bound the contribution of the bad nodes in the numerator.
%$$ \frac{}{} \leq \polylog(k)/k.$$
In the following proposition, 
%Given a mapping $X$, a semiorthogonal $U$ and a set $T\subseteq \cT$ 
we construct a $\lambda$-geometric sequence of bags of balls such that the sum of the radii of all balls in the sequence is at least 
$$\frac{1}{\polylog(k,1/\lambda)} \sum_{t\text{ is $\alpha$-bad}} \frac{1}{|\cut(t)|} \left(\sum_{e\in\cut(t)} \langle U^e,\Xb_e\rangle\right)^2.$$
Note that, we construct a geometric sequence of families of either compact or assigned bags of balls.

\begin{proposition}
\label{lem:geometricsequence}
Given a \expandertree~$\cT$ of $G$ and a set $T\subseteq \cT$ of $\alpha$-\bad~nodes, for any $\beta>1$, $\eps<1/3$, and $\lambda < 1$, if  $\alpha$ is sufficiently small such that $(\alpha/C_2(\alpha))^{\eps} \lesssim  \frac{1}{\beta\cdot C_1(\eps)}$, then  one of  the following holds:
\begin{enumerate}
\item There is a $\lambda$-geometric sequence of families of $\beta$-compact bags of balls $\fbag_1$, $\fbag_2$, $\dots,$ where $\fbag_i$ has type $(\br_i,\Delta_i)$ such that
\begin{equation}
\label{eq:constgeomcompact}
 \frac{(\alpha/C_2(\alpha))^\eps}{\beta C_1(\eps)C_2(\alpha)\cdot |\log(\lambda\poly(\alpha))|}\cdot \sum_{t\in T} \frac1{|\cut(t)|} \Big(\sum_{e\in\cut'(t)} \langle\Ue, \Xb_e\rangle \Big)^2 \leq \sum_{i} \sum_{\bag\in\fbag_i} \br_i\cdot |\bag|. 
 \end{equation}
\item There is a $\lambda$-geometric sequence of families of $(\alpha/C_2(\alpha))^{1+2\eps}$-assigned bags of balls $\fbag_1$, $\fbag_2$, $\dots$,
where $\fbag_i$ has type $(\br_i,S_i)$ such that 
\begin{equation}
\label{eq:constgeomassigned}
 \frac{(\alpha/C_2(\alpha))^\eps}{\beta C_1(\eps)C_2(\alpha)|\log(\lambda\poly(\alpha))|}\cdot\sum_{t\in T} \frac1{|\cut(t)|} \Big(\sum_{e\in\cut'(t)} \langle\Ue, \Xb_e\rangle \Big)^2 \leq  \sum_{i}\sum_{\bag\in\fbag_i} \br_i\cdot |\bag|.
 \end{equation}
\end{enumerate}
Here, $C_1$ is the polynomial function that we defined in \autoref{lem:orthogonal} and $C_2$ is a polylogarithmic function that we will define in \autoref{lem:findcutp}.
\end{proposition}
In the proof of \autoref{thm:metricratio},
we invoke the above proposition for $\alpha=\polylog(k)/k$, $\eps=\log k/\log\log k$, and $\lambda=1/\poly(k)$.

%Let us describe a special self-contained case of our geometric proof.
%We call a sequence of families of $L_1$ balls $\cB_1,\cB_2,\dots$, a $1/C$-geometric sequence of $C'$-compact  (bags of) balls if for each $i\geq 1$ the balls in $\cB_i$ are disjoint, all have $L_1$ radius $\delta_i$, and they can be partitioned into groups each having $L_1$ diameter at most $\Delta_i$, such that the number of balls in each group is at least $C'\frac{\Delta_i}{\delta_i}$ (compactness property) and such that $\delta_i>C\cdot \Delta_{i+1}$ (geometrically decreasing property).
%
%\begin{theorem}[\autoref{prop:compactballconstruction} rephrased]	
%	There are universal constants $C,C'>0$ such that the following holds: Given a $k$-edge-connected graph $G=(V,E)$ and a mapping $X:V\to \mathbb{R}^h$, 
%	for any $1/C$-geometric sequence of $C'$-compact (bags of) balls  $\cB_1,\cB_2,\dots$ centered at the vertices of $G$ (i.e. $\{X_v\}_{v\in V}$), the sum of $L_1$ radii of all balls in the sequence is at most $$O\Big(\frac{1}{k}\sum_{\{u,v\}\in E} \|X_u-X_v\|_1\Big).$$
%\end{theorem}

In the rest of this section, we prove the above proposition using \autoref{lem:orthogonal}. We do this in two intermediate steps.
In the first step we extract a $1/\polylog(\alpha)$-\emph{dominating} 2-\emph{homogeneous} set $\cut'(t)$ of edges in each $\cut(t)$ for any bad node $t$ according to the following definitions.
\begin{definition}[Homogeneous Edges]
\label{def:similaredges}
For $c>1$, we say a set $F\subseteq E$ of edges is $c$-homogeneous
if for any two edges $e,f\in F$,
$$ \frac{\langle \Ue, \Xb_e\rangle^2}{\langle\Uf, \Xb_f\rangle^2} < c \text{ and } \frac{\norm{\Xb_e}_2^2}{\norm{\Xb_f}_2^2} < c.$$
\end{definition}
\begin{definition}[Dominating  Subset]
	For a node $t\in\cT$ a set $\cut'(t)\subseteq \cut(t)$ is called $\gamma$-dominating if 
	$$ \Big(\sum_{e\in\cut'(t)} \langle U^e,\Xb_e\rangle\Big)^2 \geq \gamma\cdot\Big(\sum_{e\in\cut(t)} \langle U^e,\Xb_e\rangle\Big)^2.$$
The term dominating refers to the fact that the set $\cut'(t)$ essentially captures the contribution of the edges of $\cut(t)$ to the numerator. 
	%$$  \min_{e\in \cut'(t)}(\langle U^e,\Xb_e\rangle )^2 \geq \alpha\cdot \min_{e\in\cut'(t)}\norm{\Xb_e}^2|\cut(t)|.$$
\end{definition}
Then, we group the bad nodes into sets $T_i$ such that the set $\cup_{t\in T_i}{\cut'(t)}$ is homogeneous for all $i$. 
In the second step, we use \autoref{lem:orthogonal} to construct bags of balls for a give group of homogeneous edges. We postpone the first step to the next subsection.

\begin{lemma}
\label{cor:coradiusbagofballs}
Given a \expandertree~$\cT$ of $G$, a set $T\subseteq \cT$ of $\alpha$-bad nodes, and $\gamma$-dominating sets $\cut'(t)\subseteq \cut(t)$ for each $t\in T$ such that
$\cup_{t\in T}\cut'(t)$ is 4-homogeneous,
% for each $e\in \cut'(t)$, 
%%\begin{equation*}\begin{aligned}
%$c_1\leq (\Ue \Xb_e)^2 < 4c_1$,
%$c_2 \leq \norm{\Xb_e}^2_2 < 4c_2$, and a number $\alpha$ such that
%%\end{aligned}\end{equation*}
%\begin{equation}
%\label{eq:c1c2Palpha}
%\frac{|\cut'(t)|}{|\cut(t)|}\cdot c_1 %\cdot \PP{e\sim\cut(t)}{e\in \cut'(t)} 
%\geq \alpha\cdot c_2,
%\end{equation}
for any $0<\eps<1/2$ and $\beta>1$, if  $\alpha,\gamma$ are sufficiently small such that $(\alpha\cdot\gamma)^{\eps} \lesssim\frac1{\beta \cdot C_1(\eps)}$, %and $k \gtrsim C_1(\eps)$, 
then one of the following holds:
\begin{enumerate}
\item There is a family of $\beta$-compact bags of balls with type $(\delta,\Delta)$, $\fbag$,  such that 
\begin{equation} 
\label{eq:constbagcompact}\sum_{t\in T} \frac1{|\cut(t)|} \Big(\sum_{e\in\cut'(t)} \langle\Ue, \Xb_e\rangle \Big)^2 \lesssim \frac{C_1(\eps)}{(\alpha\cdot\gamma)^\eps} \sum_{\bag\in \fbag} \br\cdot |\bag|.
\end{equation}
\item There is a family of $(\alpha\cdot\gamma)^{1+2\eps}$-assigned bags of balls with type $(\br,S)$,  $\fbag$,  and $S\subseteq T$ such that 
%$\br\geq \alpha^\eps c_1/4C_1(\eps) $ and 
\begin{equation}
\label{eq:constbagassigned}
 \sum_{t\in T} \frac1{|\cut(t)|}\cdot \Big( \sum_{e\in\cut'(t)} \langle\Ue, \Xb_e\rangle\Big)^2 \lesssim \frac{\beta C_1(\eps)}{(\alpha\cdot\gamma)^\eps}\sum_{\bag\in\fbag} \br \cdot |\bag|.
\end{equation}
\end{enumerate}
where in both cases $\br,\Delta=\min_{e\in\cut'(t),t\in T} \langle U^e,X_e\rangle^2$ up to an $O(\alpha\cdot \gamma)$ factor. 
\end{lemma}
\begin{proof}
\def\ca{c_1}
\def\cb{c_2}
Let,
\begin{eqnarray*}
 F&:=&\cup_{t\in T} \cut'(t),\\
\ca &:=& \min_{e\in F}\langle\Ue, \Xb_e\rangle^2,\\
\cb &:=& \max_{e\in F}\norm{\Xb_e}^2_2,\\
N&:=&\big|\cup_{t\in T} \cut(t)\big|,\\
N'&:=&\big|\cup_{t\in T} \cut'(t)\big|=|F|.
\end{eqnarray*}Note that $N\geq N'$ by definition.
First, we show that the edges in $F$ satisfy the assumption of \autoref{lem:orthogonal} with $\alpha$  replaced by $\asymp\alpha\gamma N/N'$. Then, we invoke \autoref{lem:orthogonal} and we obtain many disjoint balls $\cA$ such that the sum of their radii is comparable to LHS of \eqref{eq:constbagcompact} or \eqref{eq:constbagassigned} (see \eqref{eq:gainbagsballsequal}). Then, 
we greedily construct a new set $\cB$ of disjoint large balls of radii $\Delta\geq \cb$. If $|\cB|$ is small, we can partition the balls of $\cA$ into compact bags of balls; otherwise, we use balls of $\cB$ to construct assigned bags of balls.

%Since each edge is in at most two sets $\cut(t),\cut(t')$, $2N\geq \sum_{t\in T}|\cut(t)|$ and $2N'\geq \sum_{t\in T}|\cut'(t)|$. 

First, observe that,
\begin{eqnarray}
\ca \cdot N' \gtrsim \sum_{t\in T} \frac{1}{|\cut(t)|} \big(\sum_{e\in\cut'(t)}\langle U^e,\Xb_e\rangle\big)^2 
&\geq& \sum_{t\in T} \frac{\gamma}{|\cut'(t)|}\big(\sum_{e\in \cut(t)}\langle U^e,\Xb_e\rangle\big)^2  \nonumber \\
&\geq & \sum_{t\in T} \gamma\cdot \alpha\cdot 
\sum_{e\in \cut(t)}\norm{\Xb_e}^2\nonumber \\
&\geq & \alpha\cdot\gamma\cdot N \cdot \cb.
 \label{eq:c1c2alpha} 
\end{eqnarray}
where the first inequality follows by 4-homogeneity of $F$, the second inequality uses the fact that each $\cut'(t)$ is $\gamma$-dominating, the third inequality uses that each node $t$ is $\alpha$-bad, and the last inequality again uses the 4-homogeneity of $F$. This is the only place in the proof that we use $t\in T$ is $\alpha$-bad and $\cut'(t)$ is $\gamma$-dominating.
By the above equation we can choose $\talpha\asymp \alpha\gamma N/N'$ such that
$$\big(\uE{e\sim F} \langle\Ue, \Xb_e\rangle\big)^2 \geq \talpha \cdot \uE{e\sim F} \norm{\Xb_e}_2^2. $$
%Let $\talpha = \frac{\alpha N}{8N'}$. 
Throughout the proof we use that $\talpha\gtrsim \alpha\gamma$. Let $Y_v:=UX_v$ for all $v\in V$. 
Since $U$ is semiorthogonal,
 for each pair $u,v$ 
 $$\norm{Y_u-Y_v}_2^2 \leq \norm{X_u-X_v}^2_2 = \norm{X_u-X_v}_1.$$ 
 %so $Y$ is also a $L^2_2$ metric on $G$.
 Applying \autoref{lem:orthogonal} to $Y$ and $F$, we obtain a family $\cA$ of $b$ disjoint $L^2_2$  balls with radius $\br$
such that
\begin{equation}
\label{eq:numberofballsbetabound}
b \geq \frac{\talpha N'}{C_1(\eps)},
%\gtrsim \frac{\alpha\cdot\gamma\cdot N}{C_1(\eps)},
\end{equation}
and
\begin{equation}
\label{eq:bbrbetabound} 
\br\cdot b \geq \frac{{\talpha}^\eps\cdot N' \cdot \ca}{C_1(\eps)}.
%\gtrsim  \frac{c_1 \alpha^\eps \gamma^\eps N^\eps {N'}^{1-\eps}}{C_1(\eps)}.
\end{equation}
Now, we extract \emph{disjoint} $L_1$ balls in the space of $\{X_v\}_{v\in V}$ with radius $\delta$ out of balls in $\cA$. Balls in $\cA$  correspond to $L^2_2$ balls in the $X$ embedding. 
Since $U$ is a contraction operator, these $L_2^2$ balls are disjoint in the $X$ embedding. 
Now, $L^2_2$ balls with radius $\delta$ are $L_2$ balls with radius $\sqrt{\delta}$, so the $L^2_2$ distance between the centers of any two balls is at least $4\delta$. 
Since $X$ is a cut metric, the $L_2^2$ distance between centers is the same as their $L_1$ distance, so $L_1$ balls with radius $\delta$ around the same centers are disjoint (in fact radius $2\delta$ works as well). So, by abusing notation we let $\cA$ be the $L_1$ balls in the $X$ embedding. 

Next, we construct the large balls. 
Let 
$$V'(t) =  \{u\in V(t): \exists\{u,v\}\in\cut'(t)\}$$
be the endpoints of edges of $\cut'(t)$ that are in $V(t)$. Also, let  $V'=\cup_{t\in T} V'(t)$.
Let $\cB$ be a maximal family of disjoint $L_1$ balls of radius $\Delta$ on the points in $V'$ for $\Delta:=\max\{\delta, \cb\}$. To construct $\cB$,
%We construct another family  of disjoint $L_1$ balls $\cB$ with centers in $V'$ each of radius $\rho$. 
we scan the points in $V'$ in an arbitrary  order; for each point $X_u$ if the ball $B(X_u,\Delta)$ does not touch any of the balls already added to $\cB$ we add $B$ to $\cB$.
We will consider two cases  depending on the size of $\cB$; if $|\cB|$ is small we construct compact bags of balls and we conclude with case (1); otherwise we construct assigned bags of balls and we conclude with (2).

Before getting into the details of the two cases, we prove two facts that are useful for both cases. 
First, without loss of generality, perhaps by decreasing $\br$, we assume $\br\cdot b\asymp \frac{\ca \talpha^\eps N'}{C_1(\eps)}$. We can bound $\delta$ as follows
\begin{equation}
\label{eq:deltabound}
\gamma\cdot \alpha\cdot \ca  \lesssim  \frac{\ca\talpha^\eps N'/C_1(\eps)}{N'}\lesssim \frac{\br\cdot b}{b}=\br = \frac{\br \cdot b}{b}  \lesssim \frac{\ca \talpha^\eps N'}{\talpha N'} \leq  \frac{\ca}{\gamma\cdot \alpha},
\end{equation}
where the first inequality uses the lemma's assumption that $(\gamma\alpha)^{1-\eps}\leq (\gamma\alpha)^\eps \lesssim 1/C_1(\alpha)$, the second inequality uses   $b\leq 2N'$, 
the third inuequality uses $b\gtrsim \frac{\talpha N'}{C_1(\eps)}$ and the last inequality uses $\talpha \geq \gamma\cdot\alpha$.

Secondly, it follows from \eqref{eq:bbrbetabound} that 
\begin{eqnarray} \frac{\sum_{t\in T} \frac1{|\cut(t)|}(\sum_{e\in \cut'(t)} \langle\Ue, \Xb_e\rangle)^2 }{ b \cdot \br} &\lesssim & \frac{\sum_{t\in T} |\cut'(t)| \cdot \ca}{\ca \talpha^\eps {N'}/ C_1(\eps)} \leq \frac{C_1(\eps)}{\talpha^\eps}.
\label{eq:gainbagsballsequal}
\end{eqnarray}
In the above we used $|\cut'(t)|\leq |\cut(t)|$ for all $t$.
To prove the lemma, in the first case we construct a family of compact bags of balls with at least $b/2$ balls of $\cA$, and in the second case we construct a family of assigned bags of balls with at least $|\cB|/2$ balls of $\cB$.

\paragraph{Case 1. $|\cB| < \frac{b\cdot\delta}{12 \beta \cdot\Delta}$.}
We construct a family of compact bags of balls.
For each ball $B=B(X_u,\Delta)\in \cB$ let 
$$f(B):=\big\{B(X_v,\br)\in \cA: \norm{X_u-X_v}_1 = \min_{B(X_{u'},\Delta)\in \cB} \norm{X_{u'} - X_v}_1\big\},$$
be the balls
of $\cA$ that are closer to $B$ than any other ball of $\cB$. We break ties arbitrarily,  making sure that $f(B)\cap f(B')=\emptyset$ for any two distinct balls of $\cB$.
 
First, we show any set $f(B)$ is a bag of balls of type $(\br,6\Delta)$; then we add  those that are $\beta$-compact to $\fbag$.
It is sufficient to show that for any $B(X_u,\Delta)\in\cB$, the $L_1$ distance between the centers of balls of $f(B)$ is at most $6\Delta$.
Fix a ball $B=B(X_u,\Delta)\in\cB$. For any ball $B(X_{v_1},\delta)\in f(B)$ we show that $\norm{X_u-X_{v_1}}_1\leq 3\Delta$. 
Since for all $e\in F$, $\norm{\Xb_e}_1\leq \cb$, there
is a vertex $u_1\in V'$ such that $\norm{X_{v_1}-X_{u_1}}\leq \cb$. Furthermore, by construction of $\cB$, there is a ball $B(X_{u_2},\Delta)\in\cB$ such that $\norm{X_{u_1}-X_{u_2}}_1 \leq 2\Delta$. Putting these together,
$$ \norm{X_{v_1} - X_u}_1 \leq \norm{X_{v_1} - X_{u_2}}_1 \leq \norm{X_{v_1} - X_{u_1}}_1 +
\norm{X_{u_1} - X_{u_2}}_1 \leq \cb + 2\Delta \leq 3\Delta.$$
So, the $L_1$ distance between the centers of balls of $f(B)$ is at most $6\Delta$.

 So, we just need to add those bags that are $\beta$-compact to $\fbag$. For each $B\in\cB$ if $|f(B)| \geq \beta\cdot (6\Delta)/\br$, then  $f(B)$ is  $\beta$-compact, as $|f(B)|\geq 2$ and
$$  \beta\cdot (6\Delta) \leq \br\cdot |f(B)|. $$
So, we add $f(B)$  to $\fbag$. Observe that all balls of $\fbag$ are disjoint because all balls of $\cA$ are disjoint. %so $\fbag(\br,\Delta)$ is a homogeneous family $\beta$-compact bag of balls. 

It remains to verify that $\fbag$ satisfies conclusion (1). First, by \eqref{eq:deltabound} and the fact that $\Delta=\max\{\delta,\cb\}$, 
$6\Delta \gtrsim \alpha\cdot \gamma\cdot \ca$.
On the other hand, by \eqref{eq:c1c2alpha}, $\cb\leq \ca/\alpha$ as shown in , 
$$ 6\Delta \lesssim \max\{\br,\cb\}  \lesssim\{\ca/\alpha\gamma,\ca/\alpha\gamma\}.$$
So we just need to verify \eqref{eq:constbagcompact}. It is easy to see that the number of balls in $\fbag$ is at least $b/2$. This is because,
\begin{eqnarray*}
\sum_{\bag\in\fbag} |\bag| \geq  b - \sum_{B\in\cB} \I{|f(B)| < \frac{\beta\cdot (6\Delta)}{\br}} \cdot |f(B)| \geq b-|\cB|\cdot \frac{\beta\cdot (6\Delta)}{\br}
\geq b/2.
\end{eqnarray*}
The last inequality uses the assumption of case 1, $|\cB|\leq \frac{b\cdot \br}{12\beta\cdot \Delta}$.
So,  \eqref{eq:constbagcompact} follows by  \eqref{eq:gainbagsballsequal}.

\paragraph{Case 2. $|\cB|\geq \frac{b\cdot\br}{12\beta\cdot\Delta}$.}
We construct an assigned family of bags of balls.
For any node $t\in T$, let $\bag_t$ be the set of balls 
in $\cB$ such that their centers are in $V'(t)$.
If the center of a ball $B$ in $\cB$ belongs to multiple $V'(t)$'s we include $B$ in exactly one of those sets arbitrarily.
Note that each $\bag_t$ is a bag of balls with type $(\Delta)$.
For each $t\in T$, if
\begin{equation}
\label{eq:btdef}
\frac{|\bag_t|}{|\cB|} \geq \frac{|\cut(t)|}{4N},
\end{equation}
then we add $\bag_t$ to $\fbag$ and we add $t$ to $S$. Next, we argue that $\fbag$ is a family of $(\alpha\cdot\gamma)^{1+2\eps}$-assigned bag of balls.
%Next, we verify that this is a valid assignment.
First, balls in $\fbag$ are disjoint because they are a subset of balls of $\cB$ and each ball of $\cB$ is in at most one bag of $\fbag$. 

Fix a node $t\in S$. We show $\bag_t$ is $(\alpha\cdot\gamma)^{1+2\eps}$-assigned. Since for any ball $B(X_u,\Delta)\in \bag_t$, 
$u\in V'(t)$, there is an edge $\{u,v\}\in\cut'(t)$ such that
$\norm{X_u-X_v}_1 \leq \cb\leq \Delta$. 
So, we just need to verify \eqref{eq:assignedbag} with $\beta$ replaced by $(\alpha\cdot\gamma)^{1+2\eps}$.
If $\Delta = \br$, by \eqref{eq:btdef},
$$|\bag_t| \geq \frac{|\cB|\cdot |\cut(t)|}{4N} \geq \frac{|\cut(t)|\cdot b\cdot \br}{48\beta \cdot \br \cdot N}\gtrsim  
 \frac{\talpha \cdot |\cut(t)|\cdot N'}{\beta \cdot C_1(\eps)\cdot N} 
%\geq \frac{(\alpha N/4) |\cut(t)|}{2 C_1(\eps) N} \geq \frac{\alpha |\cut(t)|}{8 C_1(\eps)} 
\geq  (\alpha\cdot \gamma)^{1+\eps}\cdot |\cut(t)|,$$
where the second inequality uses the assumption $|\cB|\geq \frac{b\cdot\br}{12\beta\cdot\Delta}$, the third inequality uses \eqref{eq:numberofballsbetabound} and the last inequality uses $(\alpha\cdot \gamma)^\eps \lesssim \frac1{\beta\cdot C_1(\eps)}$.
%and since $|\cut'(t)|\leq k$, $b_t\leq k$.  
Otherwise, $\Delta=\cb$, by \eqref{eq:btdef},
\begin{eqnarray}
|\bag_t| \geq \frac{|\cB|\cdot|\cut(t)|}{4N}\geq \frac{b \cdot \br\cdot |\cut(t)|}{48\beta \cdot\Delta\cdot N} 
&\gtrsim & \frac{\talpha^\eps |\cut(t)|}{C_1(\eps)\beta} \cdot \frac{N'\cdot \ca}{N\cdot \cb} \cdot \nonumber \\
&\gtrsim & \frac{\talpha^{\eps} |\cut(t)|}{C_1(\eps)\beta} \cdot \alpha\cdot\gamma \geq  \alpha^{1+2\eps} \cdot |\cut(t)|.
\label{eq:bagtrholower}
% \frac{\alpha c_1/4800}{2\cb} \geq \frac{2C_1 k c_1}{\alpha \cb} \geq C_1. 
\end{eqnarray}
The third inequality follows by \eqref{eq:bbrbetabound}, the fourth inequality uses \eqref{eq:c1c2alpha}, and the last inequality %of \eqref{eq:bagtrholower} 
uses the assumption that $(\alpha\cdot\gamma)^{\eps} \lesssim \frac1{\beta \cdot C_1(\eps)}$.
%Lastly,  since $b\geq \alpha |T|/32$ and $\br\cdot b = \alpha |T| c_1/2400$, we have 
%$$\br\leq \frac{\alpha^\eps N^\eps {N'}^{1-\eps}}{\alpha N /4C_1(\eps)} \cdot c_1\leq \frac{4C_1(\eps)}{\alpha^{1-\eps}} \cdot c_1
%\leq \Delta/20$$
Therefore, $\fbag$ is a family of $(\alpha\cdot\gamma)^{1+2\eps}$ assigned bags of balls with type $(\Delta,S)$. 

Finally, it remains to verify \eqref{eq:constbagassigned} where $\delta$ is replaced by $\Delta$. %upper-bound $\sum_t \frac1{|\cut(t)|}(\sum_{e\in\cut(t)}(UX\X_e)_{\pi(e)})^2$. 
First, we show that $\sum_{t\in S} |\bag_t| \geq |\cB|/2$. This is because by \eqref{eq:btdef},
$$ \sum_{t\in T\setminus S} |\bag_t| \leq \sum_{t\in T} \frac{|\cut(t)|\cdot|\cB|}{4N}  \leq |\cB|/2.$$
Equation \eqref{eq:constbagassigned} follows by \eqref{eq:gainbagsballsequal}
and the assumption that $|\cB|\geq \frac{b\cdot \br}{12 \beta \cdot \Delta}$. 
%It remains tho verify lemma's conclusions. 
%First, since $b\leq 2N'$ we get that 
%$$\br=\frac{\br \cdot b}{b} \geq \frac{\alpha^\eps N^\eps {N'}^{1-\eps}}{4 C_2(\eps) N'} c_1 \geq \frac{\alpha^\eps c_1}{4 C_2(\eps)}.$$
\end{proof}

\subsection{Construction of a Geometric Sequence of Families of Bags of Balls}
\label{subsec:geometricballs}
In this section we prove \autoref{lem:geometricsequence}.
First, we prove a bucketing lemma. We show that for any $\alpha$-\bad~node $t\in \cT$, we can extract a $1/\polylog(\alpha)$-dominating $2$-homogeneous set $\cut'(t)$ of edges from $\cut(t)$.
\begin{lemma}
\label{lem:findcutp}
For a \expandertree, $\cT$, of $G$, and an $\alpha$-bad node $t\in\cT$, 
if $\alpha$ is sufficiently small, then
there is a 2-homogeneous set $\cut'(t)\subset \cut(t)$ such that $\cut'(t)$ is $1/C_2(\alpha)$-dominating where $C_2(.)$ is a universal polylogarithmic function.
\end{lemma}
\begin{proof}
We fix $t$ throughout the proof and use $\cut$ instead of $\cut(t)$ for brevity.
Throughout the proof all probabilities are measured under the uniform distribution on $\cut$.
Let 
\begin{eqnarray*}
	a_e&:=&\langle\Ue, \Xb_e\rangle,\\
	b_e &:=&\norm{\Xb_e},\\
	\mu &:=& \iEE{e\sim\cut}{a_e}. 
\end{eqnarray*}
Note that since $\norm{U^e}=1$, $a_e\leq b_e$ for any $e$. To prove the claim it is enough to find a 2-homogeneous set $\cut'$
such that
\begin{equation}\label{eq:goalbucketinghom} \P{e\in\cut'}^2 \cdot \min_{e\in\cut'} a_e^2 \geq \frac{\mu^2}{C_2(\alpha)}.	
\end{equation}
Then, the lemma follows by 
$$\Big(\sum_{e\in\cut'} a_e\Big)^2
\geq |\cut|^2\cdot \P{e\in \cut'}^2 \min_{e\in\cut'}a_e^2 \geq \frac{|\cut|^2\cdot \mu^2}{C_2(\alpha)} = \frac{1}{C_2(\alpha)}\cdot \big(\sum_{e\in\cut} a_e\big)^2.$$ 
We prove \eqref{eq:goalbucketinghom} as follows: First, we partition the edges into sets $\cut_1,\cut_2,\dots$ such that for any $e,f\in\cut_i$, $a_e\approx a_f$. Then, we show that there is an index $i$, such that $\P{e\in\cut_i} \cdot \min_{e\in\cut_i} a_e \gtrsim \frac{\mu}{\log(\alpha)}$ (see \eqref{eq:maxOimulogalpha}). Then, we partition $\cut_i$ into sets $\cut_{i,1},\cut_{i,2},\dots$ such that any $\cut_{i,j}$ is 2-similar. Finally, we show that there is an index $j$ such that $\cut_{i,j}$ satisfies \eqref{eq:goalbucketinghom}.

For $i\in\Z$ and $c:=\sqrt{2}$, define,
$$\cut_i:=\{e\in \cut(t): c^{i} \leq a_e/\mu < c^{i+1}\}.$$
  We  write $\cut_{\geq j} = \cup_{i=j}^\infty \cut_i$. Also, for any $i$ let $a_{\wedge i}=\min_{e\in\cut_i} a_e$.
%We can re-write the lemma's assumption as follows:
%\begin{equation}
%\label{eq:muaebeclaim}
%\mu^2 \asymp \Big(\sum_{i=-\infty}^\infty \P{e\in \cut_i} a_{e_i}\Big)^2 \geq \alpha \cdot \sum_{i=-\infty}^\infty \P{e\in \cut_i} b_{e_i}^2.
%\end{equation}

Next, we show that there exists $-4\leq i < 2(2+\log(1/\alpha))$ such that $\P{e\in \cut_i} a_{\wedge i} \gtrsim \mu/\log(1/\alpha)$.
First, observe that,
\begin{equation}\label{eq:Fineginftysmall} \sum_{i=-\infty}^{-6} a_{\wedge i}\cdot \P{e\in \cut_i} \leq \sum_{i=-\infty}^{-6} c^{-5}\mu\cdot \P{e\in \cut_i}  \leq \mu/c^5.
\end{equation}
Let $q =\Theta(\log(1/\alpha))$ %:=2(2+\log\alpha^{-1})$ 
be chosen such that $c^q=c^5/\alpha$. Then,
\begin{eqnarray} 
\frac{c^5\mu}{\alpha}\cdot \sum_{i=q}^\infty a_{\wedge i} \cdot \P{e\in\cut_i} &\leq & \sum_{i=q}^\infty a_{\wedge i}^2 \cdot \P{e\in\cut_i}\nonumber \\ 
&\leq & \iEE{e\sim \cut_{\geq q}}{b_e^2}\cdot 
\P{e\in \cut_{\geq q}} \nonumber \\
%\sum_{i=q}^\infty \iE{e\in F_i} b_e^2 \\
&\leq & \iEE{e\sim\cut}{b_e^2}
\leq \frac{\mu^2}{\alpha}.
\label{eq:Fiaebig}
\end{eqnarray}
The second  inequality uses $a_e\leq b_e$ and the last inequality uses that $t$ is $\alpha$-\bad.
Summing up \eqref{eq:Fineginftysmall} and $\alpha/c^5\mu$ of \eqref{eq:Fiaebig} we get
$$ \sum_{i\geq q \text{ or } i\leq -6} a_{\wedge i} \cdot\P{e\in\cut_i} \leq \mu/c^3 \Rightarrow \sum_{i\geq q\text{ or } i\leq -6} a_e\P{e\in\cut_i} \leq \mu/2, $$
%$$\Rightarrow \sum_{-4\leq i<q} \P{e\in\cut_i}a_{\wedge i}\geq \mu/2.$$
where we used that for any edge $e\in\cut_i$, $a_{\wedge i} \geq a_e/c$.
Therefore,
\begin{equation}
\label{eq:maxOimulogalpha}
 \max_{-5 \leq i < q} \P{e\in \cut_i} \cdot a_{\wedge i} \geq \frac{1}{5+q}\sum_{i=-5}^q \P{e\in \cut_i} a_{\wedge i} \geq \frac{1}{c(5+q)}\sum_{i=-5}^q a_e\P{e\in\cut_i} \geq \frac{1}{c(5+q)}\cdot \frac{\mu}{2}. %\gtrsim \frac{\mu}{\log(1/\alpha)}.
\end{equation}
%\frac{\mu}{16 + 4\log(1/\alpha)}.$$
%where in the last inequality we used that for any $e\in \cut_i$, $a_{\wedge i} \geq a_e/c$. 
Let $i$ be the maximizer of the LHS of the above equation. It remains to choose a subset of $\cut_i$ 
such that $b_e^2/b_f^2 < 2$ for all $e,f$ in that subset.

%First, note that $b_e \geq a_{e^*}$ for all $e\in F_i$.
For any integer $j\geq 0$, we define 
$$\cut_{i,j}:=\{e\in \cut_i: c^j \leq b_e/a_{\wedge i} < c^{j+1}\}.$$
Note that any set $\cut_{i,j}$ is 2-similar. We show that there is an index $j<q$ such that $\cut_{i,j}$ satisfies \eqref{eq:goalbucketinghom}.
Let $\cut_{i,\geq q} = \cup_{j=q}^\infty \cut_{i,j}$. 
Similar to \eqref{eq:Fiaebig},
$$
c^{2q}\cdot \P{e\in \cut_{i,\geq q}} a_{\wedge i}^2 
%\leq \iE{e\in \cut_{i,\geq q}} b_e^2 
%\leq \iEE{e\sim\cut}{a_e^2} 
\leq \iEE{e\sim \cut}{b_e^2} \leq \frac{\mu^2}{\alpha} 
\leq \frac{1}{\alpha}\cdot 8a^2_{\wedge i} \cdot(5+q)^2\cdot \P{e\in\cut_i}^2,
%\leq \frac1\alpha (2(4+q)\iE{e\in \cut_i} a_e)^2
%\leq \alpha(2c(4+q) a_{e_i} \P{e\in \cut_i})^2 .
$$
where the last inequality uses \eqref{eq:maxOimulogalpha}.
Using $c^q = c^5/\alpha$, we obtain
$$ \P{e\in \cut_{i,\geq q}} 
\leq \frac{\alpha}{4} \cdot (5+q)^2\cdot \P{e\in \cut_i}^2 \leq \frac12\cdot \P{e\in\cut_i}^2,$$
for a sufficiently small $\alpha$.
Now, let $j=\argmax_{0\leq j < q} \P{e\in\cut_{i,j}}$.
Then,
$$ \P{e\in\cut_{i,j}}^2 \cdot a_{\wedge i}^2\geq \frac{a^2_{\wedge i}}{q^2}\cdot  (\P{e\in \cut_i} -\P{e\in \cut_{i,\geq q}})^2 \geq \frac{\P{e\in \cut_i}^2\cdot a^2_{\wedge i}}{4q^2} \geq \frac{\mu^2}{32 q^2(5+q)^2}.$$
The last inequality uses \eqref{eq:maxOimulogalpha}.
Now, \eqref{eq:goalbucketinghom} follows by the above inequality and $C_2(\alpha)=32 q^2(5+q)^2$ and $\cut'(t) = \cut_{i,j}$; 
\end{proof}

%,\cut(t')$ for $t,t'\in T\setminus T'$. 

Now, we are ready to prove \autoref{lem:geometricsequence}.
First, by \autoref{lem:findcutp} for each $\alpha$-\bad~node $t\in T$,
there is a 2-homogeneous $\gamma$-dominating set $\cut'(t)\subseteq \cut(t)$ where $\gamma=1/C_2(\alpha)$.
For each $t\in T$, let  
$$a_t=\min_{e\in\cut'(t)} \langle\Ue, \Xb_e\rangle^2 \text{ and }
b_t=\min_{e\in\cut'(t)} \norm{\Xb_e}^2_2.$$
%Without loss of generality, assume that $\min_{t\in T} a_t\geq 1$. 
%Note that by \eqref{eq:findcutpfeas}, $b_t \leq \talpha a_t$. 
Let $\tlambda<1$ be a function of $\lambda$ that we fix later. 
For any integer $i\in\Z$, let
$$ T_i:=\{t\in T:  \tlambda^{i+1/2}\leq a_t < \tlambda^{i-1/2} \}$$
%Without loss of generality (perhaps by slightly perturbing $\lambda$), we assume 
Note that, by definition, for all $i\neq j$, $T_i\cap T_{j} = \emptyset$. % then we arbitrarily assign the nodes in the intersection to either of the two sets.

Next, we partition the bad nodes of each $T_i$ 
into sets $T_{i,j_a,j_b}$ such that each set $\cup_{t\in T_{i,j_a,j_b}}\cut'(t)$ is 4-homogeneous. We will apply \autoref{cor:coradiusbagofballs} to the a $T_{i,j_a,j_b}$ with the largest contribution in the numerator. This will give us a family of either compact or assigned bags of balls. Then, we will drop the bags for odd (or even) $i$ randomly. Since for any $t\in T_i, t'\in T_{i+2}$, $a_{t'} < \tlambda a_{t}$ we will obtain a $\tlambda$-geometric sequence of bags of balls.

First, we partition the nodes of each $T_{i}$ into sets $T_{i,j_a,j_b}$; for all integers $0\leq j_a$  and $0\leq j_b$ let
$$ T_{i,j_a,j_b} := \{t\in T_i: 2^{j_a} \leq \frac{a_t}{\tlambda^{i+1/2}} < 2^{j_a+1}, 2^{j_b}\leq \frac{b_t}{a_t} < 2^{j_b+1} \}.$$
Observe that  for all $i,j_a,j_b$, $\cup_{t\in T_{i,j_a,j_b}} \cut'(t)$ is  4-homogeneous.
Note that by the definition of $T_i$, for $j_a> \log(1/\tlambda)$, $T_{i,j_a,.}=\emptyset$.
On the other hand, since $t$ is $\alpha$-bad and $\cut'(t)$ is $\gamma$-dominating, $a_t\gtrsim \alpha\gamma b_t$ (see \eqref{eq:c1c2alpha}); so for $j_b > \log(1/\alpha\gamma)+O(1)$, $T_{i,..j_b}=\emptyset$. 
Therefore, for any  $i$, the number of nonempty sets $T_{i,j_a,j_b}$  is at most $O(\log(1/\tlambda \alpha\gamma))$.

For a set $S\subseteq T$, let
$$ \Pi(S) := \sum_{t\in S} \frac1{|\cut(t)|}\Big(\sum_{e\in \cut'(t)} \langle\Ue, \Xb_e\rangle\Big)^2.$$
For each $T_{i}$ let 
$$T^*_{i} = \argmax_{T_{i,j_a,j_b}} \Pi(T_{i,j_a,j_b}).$$ 
Since any $t\in T^*_i$ is $\alpha$-\bad~and $\cut'(t)$ is $\gamma$-dominating, and $\cup_{t\in T^*_i}\cut'(t)$ is 4-homogeneous, and 
by the lemma's assumption
$$ (\gamma\alpha)^{\eps} = \frac{\alpha^\eps}{C_2(\alpha)^\eps} \lesssim \frac{1}{\beta \cdot C_1(\eps)},$$
  we may invoke \autoref{cor:coradiusbagofballs} for each set $T^*_i$. This gives us either a family of $\beta$-compact bags of balls $\fbag_i$ with type $(\br_i,\Delta_i)$, 
or  a  family of $(\alpha\gamma)^{1+2\eps}$-assigned bags of balls, $\fbag_i$ of type $(\br_i,S^*_i)$ where $S^*_i\subseteq T^*_i$.
These families satisfy two additional constraints: Firstly, $\br_i,\Delta_i=\min_{t\in T^*_i} a_t$ up to an $O(\alpha\gamma)$ factor, secondly, the sum of the radii of all balls in the family is at least $\frac{(\alpha\gamma)^\eps}{\beta C_1(\eps)} \Pi(T^*_i)$.

We remove half of the families to obtain a geometric sequence.
First, by the definition of $T_i$, 
$$ \tlambda \cdot \min_{t\in T^*_i} a_t \geq  \min_{t\in T^*_{i+2}} a_t.$$
This means that if we remove families for either odd or even $i$'s, then the decaying rate of $\min_{t \in T^*_i} a_t$ is at least $\tlambda$. 
Therefore by the properties guaranteed by \autoref{cor:coradiusbagofballs}, and the above fact, any subsequence of odd or even compact or assigned families of bags of balls is $O(\tlambda/(\alpha\cdot\gamma)^2)$-geometric. Setting $\tlambda\asymp \lambda\cdot(\alpha\cdot\gamma)^2$ produces $\lambda$-geometric sequences.

Without loss of generality we assume that 
$\Pi(\cup_i T_{2i}) \geq \Pi(\cup_i T_{2i+1})$.
Drop the families for odd $i$; consider the sum of radii of balls in the remaining compact families and in the remaining assigned families; one of them is greater.
We let this be our $\lambda$-geometric family. 

It remains to verify \eqref{eq:constgeomcompact} and \eqref{eq:constgeomassigned}.
By \autoref{cor:coradiusbagofballs}, the sum of the radii  in the constructed geometric sequence is at least $\gtrsim \frac{(\alpha\cdot \gamma)^\eps}{\beta C_1(\eps)} \sum_i \Pi(T^*_{2i})$. By the definition of $T^*_i$,
\begin{eqnarray*} 
\sum_i \Pi(T^*_{2i}) \gtrsim
\frac{1}{|\log(\tlambda\alpha\gamma)|} \sum_i \Pi(T_{2i})
\geq \frac{\Pi(T)}{|\log(\lambda\poly(\alpha))|}.
\end{eqnarray*}
%In the last inequality we crudely upper-bound $\log(1/\talpha)+\log(1/\lambda)$
%by $\log(1/\talpha)\cdot \log(1/\lambda)$.
Now, since each $\cut'(t)$ is $\gamma=1/C_2(\alpha)$-dominating,
$$ \Pi(T) \geq 
\frac1{C_2(\alpha)}\cdot \sum_{t\in T} \frac1{|\cut(t)|} \cdot \Big(\sum_{e\in\cut(t)} \langle\Ue, \Xb_e\rangle\Big)^2.
$$
%\end{proofof}
\def\rho{\br}

\section{Lower-bounding the Denominator of the Dual}
\label{sec:charging}
In this part we upper-bound the sum of radii of balls in a geometric sequence. Throughout this section we use $C_3,C_4>0$ as  large universal constants. 
%We start with the easier case of 
%analyzing geometric sequence of compact bags of balls. 
The following two propositions are the main statements that we prove in this section. 

\begin{proposition}
\label{prop:compactballconstruction}
Given a $k$-edge-connected graph $G$, and a $\lambda$-geometric sequence of families of $C_3$-compact bags of balls $\fbag_1$, $\fbag_2$, $\dots$ where $\fbag_i$ has type $(\br_i,\Delta_i)$, if $\lambda\leq 1/12$ and $C_3\geq 36$,
then
$$ \frac{k}{4}\cdot \sum_{i} \br_i \sum_{\bag\in\fbag_i}  |\bag| \leq \sum_{\{u,v\} \in E} \norm{X_u-X_v}_1.$$
\end{proposition}

\begin{proposition}
\label{prop:assignedballconstruction}
Given a $(k,k\cdot \lambda,T)$-LCH, $\cT$, of $G$ and a $\lambda$-geometric sequence of  families of $24C_3/k$-assigned bags of balls, $\fbag_1$, $\fbag_2$, $\dots$ such that each $\fbag_i$ is of type $(\br_i,T_i)$ where  $T_i\subseteq T$,
if $C_4\geq 3$, $\lambda\leq 1/6C_4$ and $C_3 \geq 2((C_4+1)+4(C_4+2)^2)$, then
$$ \frac{k}{8}\cdot \frac{C_4}{12C_3}\cdot \sum_{i} \rho_i \sum_{t\in T_i}  |\bag_t| \leq  \sum_{\{u,v\}\in E} \norm{X_u-X_v}_1.$$
\end{proposition}
\noindent Note that in the above proposition, the assumption $\lambda\leq 1/6C_4$ follows from $k\cdot \lambda < 1$. 

First, we use the above propositions to finish the proof of \autoref{thm:metricratio}. Recall \autoref{thm:metricratio}:
\metricratiothm*
\begin{proof}
%\begin{proofof}{\autoref{thm:metricratio}}
Let $T_{\alpha\text{-\bad}}\subseteq T$ 
%as defined in \autoref{lem:geometricsequence}. 
be the set of $\alpha$-\bad~nodes for a parameter $\alpha$ that we set below. 
It follows that,
\begin{equation}
\label{eq:easynodes}
\alpha \geq \frac{\sum_{t\in T\setminus T_{\alpha\text{-\bad}}} \frac{1}{|\cut(t)|}\cdot \big(\sum_{e\in\cut(t)} \langle U^e, \Xb_e\rangle \big)^2}{\sum_{t\in T\setminus T_{\alpha\text{-\bad}}} \sum_{e\in\cut(t)}\norm{\Xb_e}_1} 
\geq \frac{\sum_{t\in T\setminus T_{\alpha\text{-\bad}}} \frac{1}{|\cut(t)|}\cdot \big(\sum_{e\in \cut(t)} \langle U^e, \Xb_e\rangle \big)^2}{2\sum_{e\in E} \norm{\Xb_e}_1}.
\end{equation}
The second inequality uses the fact that each edge is in at most two sets $\cut(t)$.

We apply \autoref{lem:geometricsequence} to $T_{\alpha\text{-\bad}}$.
Let $C_4=3$, $\beta=36$ and $C_3=104$.
We choose $\alpha=\Theta(\polylog(k)/k),\eps= \Theta(\log\log(k)/\log(k))$ such that the following conditions are satisfied
\begin{eqnarray*}
%\alpha &\leq& 1/800,\\
\bigg(\frac{\alpha}{C_2(\alpha)}\bigg)^\eps &\lesssim& \frac{1}{\beta \cdot C_1(\eps)},\\
 \bigg(\frac{\alpha}{C_2(\alpha)}\bigg)^{1+2\eps} &\geq&  \frac{24C_3}{k}.
\end{eqnarray*}
Recall that $C_1(\eps)$ is an inverse polynomial of $\eps$
and $C_2(\alpha)$ is a polylogarithmic function of $\alpha$ so the above assignment is feasible. 
Also let $\tlambda < \lambda/k$ be such that 
 $\tlambda<1/6C_4.$
 
Now, by \autoref{lem:geometricsequence} either there is a $\tlambda$-geometric sequence of $36$-compact bags of balls $\fbag_1$, $\fbag_2$, $\dots,$ that satisfies \eqref{eq:constgeomcompact}, or there is a $\tlambda$-geometric sequence of $24C_3/k$-assigned
bags of balls $\fbag_1$, $\fbag$, $\dots,$
 that satisfies \eqref{eq:constgeomassigned}.
Now, by \autoref{prop:compactballconstruction} and \autoref{prop:assignedballconstruction} we
get
\begin{eqnarray*}
\frac{\sum_{t\in T_{\alpha\text{-\bad}}} \frac{1}{|\cut(t)|}\cdot\big(\sum_{e \in \cut(t)} \langle U^e, \Xb_e\rangle \big)^2}{\sum_{e\in E} \norm{\Xb_e}_1} \lesssim 
\frac{C_1(\eps)C_2(\alpha)\cdot|\log(\tlambda\poly(\alpha))|}{k\cdot (\alpha/C_2(\alpha))^\eps}
%\frac{\sum_{t\in \cT(G)\setminus T} \Big(\sum_{\{u,v\}\in E(t)} (UX_u - UX_v)_\pi(\{u,v\}) \Big)^2}
%{k\cdot \sum_{i=1}^m \sum_{t\in Z_i} \br_i \cdot b_t }
%\lesssim \frac{k\log^6(k)}{\alpha}.
\end{eqnarray*}
The theorem follows from the above equation together with \eqref{eq:easynodes}.
\end{proof}

\medskip
In the rest of this section we prove above propositions. 
Before getting into the proofs, we give a simple example to show that, in order to bound the denominator, it is necessary to use that the given $\lambda$-geometric sequence of bags of balls is either compact or assigned. 
The following example  is designed based on the dual solution that we constructed in \autoref{thm:avgeffres}.
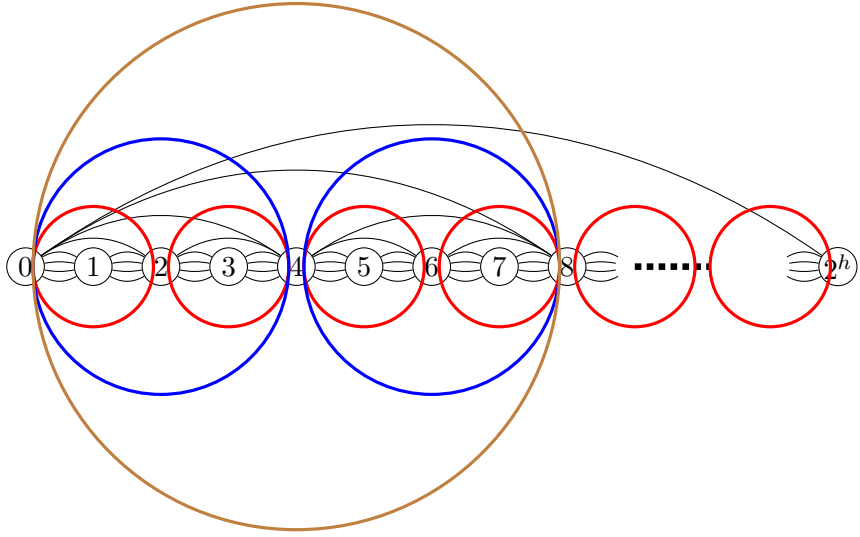
\begin{figure}
\centering
\begin{tikzpicture}[scale=0.5]
\tikzstyle{every node} = [draw,circle,minimum size=5mm,inner sep=0];
\foreach \i/\j in {1/0,2/1,3/2,4/3,5/4,6/5,7/6,8/7,9/8}{
	\node at (1.8*\i,0) (v_\i) {\j};
}
\node at (1.8*10,0) [draw=none] (v_10) {};
\node at (1.8*12,0) [draw=none] (v_11) {};
%\node at (1.8*9,0) (v_n1) {n-1};
\node at (1.8*13,0) (v_n) {$2^h$};
\foreach \i/\j in {1/2, 2/3, 3/4, 4/5, 5/6, 6/7, 7/8, 8/9, 9/10, 11/n}{
	\foreach \a in {30, 10, -10, -30}{
		\path (v_\i) edge [bend left=\a] (v_\j);
	}
}
\path (v_1) edge [bend left=35] (v_3)  (v_3) edge [bend left=35] (v_5) (v_1) edge [bend left=35] (v_5) (v_5) edge [bend left=35] (v_7) (v_1) edge [bend left=35] (v_n)
(v_7) edge [bend left=35] (v_9) (v_5) edge [bend left=35] (v_9) (v_1) edge [bend left=35] (v_9);
\draw [dotted, line width=2.5] (18,0) -- (20,0);
\draw [color=red,line width=1.2pt](v_2) circle (1.6) (v_4) circle (1.6)
(v_6) circle (1.6) (v_8) circle (1.6) (v_10) circle (1.6) (v_11) circle (1.6);
\draw [color=blue,line width=1.2pt] (v_3) circle (3.4) (v_7) ellipse (3.4);
\draw [color=brown,line width=1.2pt] (v_5) circle (7);
\end{tikzpicture}
\caption{Consider the natural $L_1$ mapping of the graph of \autoref{fig:maxeffresbadexample} where vertex $i$ is mapped to the number $i$. Consider  $h$ layers of $L_1$ balls as shown above where the radii of all balls in layer $i$ is $2^i$ and they are disjoint. Although the sum of the radii of all balls in this family is $\Theta(n\cdot h)$, the sum of the $L_1$ lengths of the edges of $G$ is $n\cdot (h+k).$}
\label{fig:badintersectingbagsofballs}
\end{figure}
\begin{example}
\label{ex:intersectingballs}  
Let $G$ be the graph illustrated in \autoref{fig:maxeffresbadexample}, and let $X_0, X_1,\ldots,X_{2^h}$ be an embedding of $G$ where $X_i=\bone_{[i]}$.
Now, for any $1\leq j\leq h-1$, let $\bag_j$ be the union of balls
$$B(X_{2^j}, 2^j), B(X_{3\cdot 2^j}, 2^{j}),B(X_{5\cdot 2^j},2^j),\dots,B(X_{2^h-2^j}, 2^j).$$ 
Note that the center of each of these balls is a vertex of $G$ and that for any $j$, all balls of $\bag_j$ have equal radius and are disjoint (see \autoref{fig:badintersectingbagsofballs}). So we get a $1/2$-geometric sequence of bags of balls (and similarly we can obtain a $\lambda$-geometric sequence by letting $j$ be  multiples of $\log(1/\lambda)$).
As alluded to in the proof of \autoref{thm:avgeffres},
the sum of the radii of balls in the given sequence is $h\cdot 2^h$ while the sum of the $L_1$ lengths of edges of $G$ is only $(h+k)\cdot 2^h$. 
\end{example}

The above example serves as a crucial barrier to both of our proofs. In the proof of \autoref{prop:compactballconstruction} we bypass this barrier
using the compactness of bags of balls. 
Note that in the above example $\bag_j$ is not compact, and indeed the diameter of centers of balls of $\bag_j$ is $2^h$ which is the same as the sum of the radii of balls in $\bag_j$. 
In the proof of \autoref{prop:assignedballconstruction}
we bypass the above barrier using the properties of the \expandertree.

\subsection{Charging Argument for Compact Bags of Balls}
\label{subsec:compactchargin}
In this section we prove \autoref{prop:compactballconstruction}.
We construct a set  of \emph{disjoint} $L_1$ hollowed balls inductively from the given compact bags of balls.
For any integer $\tau\geq 0$, we use $\cZ_\tau$ to denote the set of hollowed balls in the construction at time $\tau$. Initially, we have $\cZ_0=\emptyset$ and $\cZ_\infty$ is the final construction.
We describe the main properties of our construction in
\autoref{fig:compconstprop}.
\begin{figure}
\centering
\fbox{\parbox{6.6in}{ %\vspace*{2mm}
\begin{enumerate}
\item We process bags of balls in phases;
we assume that phase $\l$ starts at time $\tau_{\l-1}+1$ and ends at $\tau_{\l}$. In phase $\l$ we process the bags in $\fbag_\l$; in other words, we process larger balls earlier than smaller ones. In each time step (except the last one) of phase $\l$ we process exactly one bag of $\fbag_\l$.
\item In addition to adding new balls, in each phase we may shrink or delete some of the already inserted (hollowed) balls but when we insert a ball of $\fbag_\l$ we never alter it until after the end of phase $\l$.
\label{propcom:shrink}
\item We keep the invariant that for any $\tau$, all (hollowed) balls in $\cZ_\tau$ are disjoint. 
This crucial property will not hold in our construction of the assigned bags of balls in the next section and it is the main reason that our second construction is more technical. 
\label{propcom:disjointness}
\item For any hollowed ball $B(x,r_1\|r_2)\in\cZ_\tau$,
there are vertices $u,v\in V$ such that $\norm{x-X_u}_1 \leq r_1$ and $\norm{x-X_v}_1\geq r_2$. 
\label{propcom:kconnectivity}
\end{enumerate}
}}
\caption{Properties of the inductive charging argument for compact bags of balls.}
\label{fig:compconstprop}
\end{figure}

\paragraph{Inductive Charging.} Before explaining our construction, we 
describe our inductive charging argument. 
First, by the following lemma, in our construction, we only need to lower-bound the sum of the widths of all hollowed balls of $\cZ_\infty$ by (a constant multiple of) the sum of radii of all balls in the given sequence of compact bags of balls.
\begin{lemma}
\label{lem:countingedgelengthradii}
For any $\tau\geq 0$,
$$ k\cdot \sum_{B(x,r_1\|r_2)\in\cZ_\tau} (r_2-r_1) \leq  \sum_{\{u,v\}\in E} \norm{X_u - X_v}_1.$$
\end{lemma}
\begin{proof}
We simply use the $k$-edge-connectivity of $G$.
First, by property \ref{propcom:kconnectivity} of \autoref{fig:compconstprop} for each hollowed ball $B=B(x,r_1\|r_2)\in \cZ_{\tau}$ there are vertices $u,v\in V$ such that $\norm{x-X_u}_1\leq r_1$ and $\norm{x-X_v}_1\geq r_2$. 
Since $G$ is $k$-edge-connected, there 
are at least $k$ edge-disjoint paths between $u,v$.
 Each of these paths must cross $B$ and, by the triangle inequality, the length of the intersection
with $B$ is at least $r_2-r_1$. 
Finally, since by property \ref{propcom:disjointness} of \autoref{fig:compconstprop}, balls of $\cZ_\tau$ are disjoint, this argument does not overcount the $L_1$-length of any edge of $G$. 
\end{proof}
Suppose at the end of our construction, we allocate $r_2-r_1$ tokens to any hollowed ball 
$B(x,r_1\|r_2)\in \cZ_{\infty}$. Our goal is to distribute these tokens between all bags of balls such that each bag, $\bag$, of type $(\br_i,\Delta_i)$ receives at least $|\bag|\cdot \br_i/4$ tokens. %for some constant $\gain > 0$ that we fix later. 
We prove this by an induction on $\tau$.
Suppose $\tau_{\l-1} < \tau \leq \tau_\l$; for a hollowed ball $B(x,r_1\|r_2)\in\cZ_\tau$, define
\begin{equation}
\label{eq:losscompact}
\token_\tau(B):=\begin{cases}
\br_\l - 6\Delta_{\l+1} & \text{if $B\in \fbag_\l$}\\
[(r_2-r_1)-6\Delta_\l]^+ & \text{otherwise.}
\end{cases}
%\loss(\l) :=  \sum_{i\geq \l} 2(\Delta_i + \br_i)
% \leq \sum_{i\geq \l} 3\Delta_i \leq 3\Delta_\l\cdot \sum_{i\geq 0} \lambda^i \leq 6\Delta_\l
\end{equation}
%where $0<\pshr<1$ is a probability of shrinkage that we fix later.
%be a loss function; 
Instead of allocating $r_2-r_1$ tokens to a ball at time $\tau$, we allocate $\token_\tau(B)$. The term $6\Delta_\l$ takes into account the fact that we shrink balls in $\cZ_\tau$ later in the post processing phase.
We prove the following lemma inductively.

\begin{lemma}
\label{lem:inductioncompact}
 At any time $\tau_{\l-1}+1\leq \tau\leq \tau_\l$, if we allocate  $\token_\tau(B)$ tokens to any hollowed ball $B \in\cZ_\tau$, then
we can distribute these tokens among the bags of balls that we  processed by time $\tau$ such that each $\bag$ of type $(\br_i,\Delta_i)$ receives at least $\br_i\cdot |\bag|/4$ tokens.
\end{lemma}

It is easy to see that \autoref{prop:compactballconstruction} follows by applying the above lemma to the final set of hollowed balls $\cZ_\infty$ and using \autoref{lem:countingedgelengthradii}, since
$$ \frac14\sum_{i} \sum_{\bag\in\fbag_i} \br_i\cdot |\bag|\leq  \sum_{B(x,r_1\|r_2)\in\cZ_\tau} r_2-r_1 \leq \frac1k\cdot \sum_{\{u,v\}\in E} \norm{X_u-X_v}_1.$$
%Note that we used that for any ball $B(x,r_1,r_2)$ and any $\l$, \token_\l(
%Since we are 

\paragraph{Construction.} 
It remains to prove \autoref{lem:inductioncompact}.
%Now, we are ready to describe our construction.
First, we need some definitions.
We say a ball $B=B(X_u,\br_\l)\in \fbag_\l$ is in the {\em interior} of a hollowed ball $B'=B(x,r_1\|r_2)\in\cZ_{\tau}$ if 
$$ r_1 + \br_\l+\Delta_\l \leq \norm{X_u-x}_1 \leq r_2 -\br_\l-\Delta_\l.$$
Note that $B$ is inside $B'$ when $r_1+\br_\l \leq \norm{X_u-x}_1 \leq r_2-\br_\l$; so a ball $B$ may be inside $B'$ but not in the interior of $B'$.
If such a $B'$ exists,  we call $B$ an interior ball. If $B$ is not an interior ball, we call it a border ball.
Since hollowed balls in $\cZ_\tau$ are disjoint, $B$ can be in the interior of at most one hollowed ball of $\cZ_\tau$. 
\begin{fact}
\label{fact:compactinterior}
Any ball $B\in\fbag_\l$ is in the interior of at most one hollowed ball of $\cZ_\tau$.
\end{fact}
%\noindent If $B$ does not intersect with any ball of $\cZ_\tau$ we say $B$ is a {\em border} ball. 
%If $B$ is neither an inside ball nor an outside ball, we say $B$ is a {\em border} ball. Note that in this construction we treat border and outside balls similarly, but this will change in our second construction.
Suppose \autoref{lem:inductioncompact} holds at time $\tau> \tau_{\l-1}$; we show it also holds at time $\tau+1$. At time $\tau$, we process a bag of balls in $\fbag_\l$ that has at least one interior ball (and is not processed yet); if there is no such bag then we run the post processing algorithm that we will describe later. 
%We process all remaining bags simultaneously at the end of phase $\l$, i.e., at time $\tau_\l$.  
Suppose at time $\tau$ we are processing $\bag^*=\{B_1=B(X_{u_1},\br_\l),\dots,B_b=B(X_{u_b},\br_\l)\}$
of $\fbag_\l$ and assume that one of these balls, say $B_1$,  is in the interior of a  hollowed ball $B(x,r_1\|r_2)\in\cZ_\tau$.

First, we show that all balls of $\bag^*$ are  inside of  $B$. 
%Recall that, since $\bag^*$ is $C_3$-compact, 
%\begin{equation}
%\label{eq:useofcompactness}
%b\cdot \br\geq C_3\cdot\Delta_\l.
%\end{equation}
Let 
$$ r'_1 = \min_{1\leq i\leq b} \norm{x-X_{u_i}}_1 \text{ and } r'_2 = \max_{1\leq i\leq b} \norm{x-X_{u_i}}_1$$
It follows that
$$ r'_2 \leq \norm{x-X_{u_1}}_1 + \Delta_\l \leq (r_2-\br_\l-\Delta_\l) + \Delta_\l \leq r_2 - \br_\l,$$
where we used \eqref{eq:bagdiameter}; similarly, $r'_1 \geq r_1 + \br_\l$. Therefore, all balls of $\bag^*$ are inside of $B$ and by property \ref{propcom:disjointness} of \autoref{fig:compconstprop} they do not touch any other (hollowed) ball of $\cZ_\tau$.

Now, we construct $\cZ_{\tau+1}$. We remove $B$ and we add two new hollowed balls $B'_1=B(x,r_1\|r'_1 - \br_\l)$ and $B'_2=B(x,r'_2+\br_\l\|r_2)$. In addition, we add all of the balls of $\bag^*$ (see \autoref{fig:hollowedcompactballs}).
It is easy to see that balls in $\cZ_{\tau+1}$ are disjoint.  We send $\br_\l/4$ tokens of each of $B_1,\ldots,B_b$ to $\bag^*$. 
We send the rest of their tokens and all of the tokens of $B'_1,B'_2$ to $B$ and we re-distribute them by the induction hypothesis.
It follows that $\bag^*$ receives exactly $b\cdot\br_\l/4$ tokens and $B$ receives $\token_\tau(B)$.
\begin{align*} 
\token_{\tau+1}(B'_1) + \token_{\tau+1}(B'_2) &+ \sum_{i=1}^b \token_{\tau+1}(B_i) \\
&\geq r_2 - r_1 - (r'_2-r'_1) - 2\br_\l - 12\Delta_\l + b\cdot(\delta_\l-6\Delta_{\l+1})\\
& \geq \token_{\tau}(B) + b\cdot \br_\l(1-6\lambda) - 7\Delta_\l \\
&\geq \token_\tau(B) + b\cdot\br_\l/2 - C_3\Delta_\l/4\\
&\geq \token_\tau(B) + b\cdot \br_\l/4.
\end{align*}
where the first inequality uses \eqref{eq:losscompact}, the second inequality uses $\Delta_{\l+1}\leq \lambda\cdot \br_\l$ and $\Delta_\l \geq 2\br_\l$, the third inequality uses that $\lambda < 1/12$ and $C_3\geq 28$. 
The last inequality uses that $\bag^*$ is $C_3$-compact, i.e., \eqref{eq:compactbags}; this is the only place that we use the compactness of $\bag^*$. 
Therefore, \autoref{lem:inductioncompact} holds at time $\tau+1$.
\begin{figure}\centering
\def\opacity{.3}
\begin{tikzpicture}[scale=.8]
	\tikzstyle{every node} = [draw=none,line width=1.2pt];
%	\node at (0,0) (a){};
	\draw [fill=black,opacity=\opacity](0,0)  node[below left=1.6cm,draw=none,color=black,opacity=1] {$B$} circle (3);
\foreach \i/\a/\r/\p in {1/50/2/above, 2/25/2.1/right, 3/73/1.5/left, 4/43/1.35/left, 5/13/1.46/below}{
\draw [draw=black] (\a:\r) node [\p=.15cm,color=black,opacity=1] {$B_\i$} circle (.3);
}
{\Huge \draw[double,-implies,double distance between line centers=4.5pt] (3.6,0) -- (5.2,0);}

\begin{scope}[shift={(9,0)}]
	\draw [fill=black,opacity=\opacity] (0,0) node [below left=1.6cm,opacity=1] {$B'_2$} circle (3);
	\draw [fill=white] (0,0) circle (2.43);
	\draw [fill=black,opacity=\opacity] (0,0) node [below left=.5cm,opacity=1]{$B'_1$} circle (1.04);
\foreach \i/\a/\r/\p in {1/50/2/above, 2/25/2.1/right, 3/73/1.5/left, 4/43/1.35/left, 5/13/1.46/below}{
\draw [fill=blue,opacity=\opacity] (\a:\r) 
%node [\p=.15cm,color=black,opacity=1] {$B_\i$} 
circle (.3);
}
\end{scope}
\end{tikzpicture}	
\caption{Balls $B_1,\dots,B_5$ represent the balls of $\bag^*$; $B_1$ is in the interior of  a ball $B\in \cZ_{\tau}$. We decompose $B$ into two hollowed balls, $B'_1,B'_2$ that do not intersect any of the balls in the given compact set as shown on the right.}
\label{fig:hollowedcompactballs}
\end{figure}
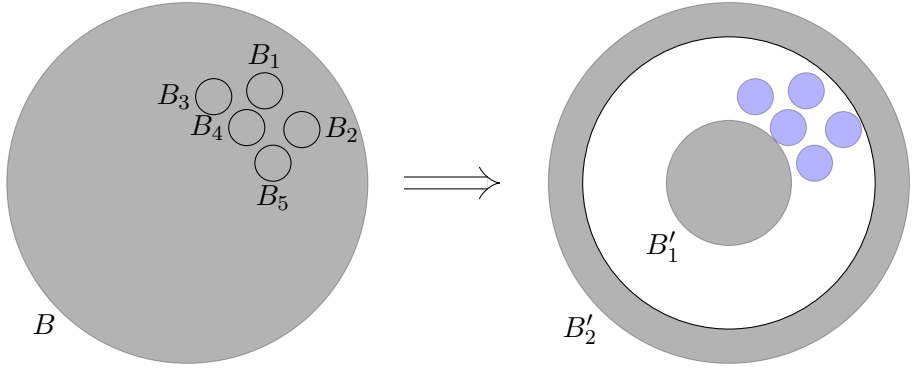
%Assuming $\gain \leq 1/4$ we can send $b\cdot \br_\l\cdot \gain$ tokens to  $\bag^*$.
%It remains to process bags of balls that do not contain any inside balls. 
\paragraph{Post Processing.}
Let $\tau_\l$ be the time by which we have processed all bags of $\fbag_\l$ with at least one interior ball, and
let $\fbag'_\l$ be the set of bags that we have not processed yet, i.e., all balls of $\fbag'_\l$ are border balls with respect to $\cZ_{\tau_\l}$. 
As alluded to, at the end of phase $\l$, i.e., at time $\tau_\l$, we shrink all (hollowed) balls of $\cZ_\tau$ except those that were in $\fbag_\l$. Given a hollowed ball $B=B(x,r_1\|r_2)\in\cZ_{\tau_\l}$, the $\shrink_\l$ operator is defined as follows:
\begin{equation}
\label{eq:shrinkoperator}
\shrink_\l(B):=\begin{cases}
B & \text{if $B\in\fbag_\l$}\\
B(x,r_1+2\br_\l+\Delta_\l\|r_2-2\br_\l-\Delta_\l) & \text{if $B\notin \fbag_\l$ and $r_2-r_1> 2\Delta_\l+4\br_\l$}\\
B(x,0)=\emptyset & \text{otherwise}.
\end{cases}
\end{equation}
At time $\tau_\l$, for any hollowed ball $B\in\cZ_{\tau_\l}$ we add 
$\shrink_\l(B)$ to $\cZ_{\tau_\l+1}$.
In addition, we add all balls of all bags of $\fbag'_\l$ 
to $\cZ_{\tau_\l+1}$. This is the end of phase $\l$ and we consider $\cZ_{\tau+1}$ as our construction in the beginning of phase $\l+1$. 

Let us verify that balls of $\cZ_{\tau+1}$ are disjoint, i.e., $\cZ_{\tau+1}$ satisfies property \ref{propcom:disjointness} of \autoref{fig:compconstprop}.
For any hollowed ball $B=B(x,r_1\|r_2)\in\cZ_{\tau_\l}$ and ball $B'=B(X_{u'},\br_\l)\in\fbag'_\l$, we  show that $\shrink_\l(B)$ and $B'$  do not intersect. First, if $B\in \fbag_\l$, then $\shrink_\l(B)=B$,  by \autoref{def:familybagofballs} any two balls of $\fbag_\l$ do not intersect, so $\shrink_\l(B),B'$ do not intersect. Now, suppose $B\notin\fbag_\l$. Since $B'\in \fbag'_\l$, $B'$ is not in the interior of $B$, i.e., either $\norm{x-X_{u'}}_1 < r_1+\br_\l+\Delta_\l$ or $\norm{x-X_{u'}}_1 > r_2-\br_\l-\Delta_\l$. In both cases, $B'$ does not intersect $\shrink_\l(B)$.
%Note that since balls of $\fbag_\l$ are nonintersecting and  balls of $\fbag'_\l$ are not interior balls, balls of $\cZ_{\tau+1}$ are nonintersecting. 

It remains to distribute the tokens. 
We send all tokens of all balls of all bags of $\fbag'_\l$  to their corresponding bag. Therefore, any $\bag\in\fbag'_\l$,  receives at least 
$$b\cdot (\br_\l - 6\Delta_{\l+1}) \geq b\cdot \br_\l(1-6\lambda) \geq b\cdot\ \br_\l/2$$ 
tokens. 
In addition, for every hollowed ball $B\in \cZ_{\tau_\l}$, 
we send all tokens of $\shrink_\l(B)$ to $B$
and we redistribute by induction. Since % that, say is shrunk to $B'\in\cZ_{\tau_\l+1}$ we have,
$$ \token_{\tau_\l}(B) \leq \token_{\tau_\l+1}(\shrink_\l(B)), $$
$B$ receives at least the same number of tokens. 
%Therefore, we can use the same distribution of the tokens allocated to balls of $\cZ_{\tau_\l}$ for their shrunk version in $\cZ_{\tau_\l+1}$. 
This completes the proof of \autoref{prop:compactballconstruction}.

\subsection{Charging Argument for Assigned Bags of Balls}\label{subsec:assignedcharging}
In this part we prove \autoref{prop:assignedballconstruction}.
Before getting into the details of the proof 
we illustrate the ideas we use to bypass the barrier of \autoref{ex:intersectingballs}.
The first observation is that, unlike  the previous section, we cannot construct a family of \emph{disjoint} hollowed balls in $\cZ_\infty$ in such a way that the sum of widths of hollowed balls of $\cZ_\infty$ is a constant fraction of the sum of radii of all balls in the given geometric sequence. 
Instead, we let hollowed balls of $\cZ_\infty$ intersect and we employ a ball labeling technique that uses the \expandertree, $\cT$.

Let us give a simple example to show the crux of our analysis. Suppose a node $t_1\in\cT$
has exactly two children, $t_2,t_3$. Say at time $\tau_{\l-1} < \tau  \leq \tau_\l$ we are processing $\bag_{t_2}$. Suppose $\cZ_\tau$ has a large 
 ball $B=B(x,r)\in \bag_t$ as shown on the left side of \autoref{fig:L1balls} such that $t$ is an ancestor of $t_1$. Say $\bag_{t_2}$ has four balls $B_1,\ldots,B_4$. Because $\bag_{t_2}$ is not compact, if we remove the part of $B$ that intersects with balls of $\bag_{t_2}$ and add $B_1,\ldots,B_4$, the sum of the widths of hollowed balls in $\cZ_{\tau+1}$ is the same as that sum in $\cZ_{\tau}$, and therefore we gain nothing from adding balls of $\bag_t$. Instead, we add a new ball that intersects $B_1,\ldots,B_4$ as shown on the right side of \autoref{fig:L1balls}.

Say the center of each $B_i$ is $X_{u_i}$ for $u_i\in V(t_2)$; each $X_{u_i}$ corresponds to a blue dot in \autoref{fig:L1balls}. By the definition of assigned bags of balls, \autoref{def:assignedbagballs}, for each $i$ there is a vertex $v_i\in V(t_1)\setminus V(t_2) = V(t_3)$ such that $\norm{X_{u_i}-X_{v_i}}_1\leq \rho_\l$ (each $X_{v_i}$ corresponds to a red dot in \autoref{fig:L1balls}). 
We  add all balls of $\bag_{t_2}$ and a new hollowed ball centered at $x$, the center of $B$, ranging from the closest red vertex to $x$ to the farthest one. We also break $B$ into two hollowed balls and remove the part of it that intersects either of these 5 new (hollowed) balls. 

Observe that, the sum of the widths of hollowed balls of $\cZ_{\tau+1}$ is $\Omega(\rho_\l\cdot |\bag_{t_2}|)$ more than this sum in $\cZ_{\tau}$. The only problem is that, the balls of $\cZ_{\tau+1}$ are intersecting. So, it is not clear if analogous to 
\autoref{lem:countingedgelengthradii}, we can charge  the sum of the widths of hollowed balls of $\cZ_{\tau+1}$ 
to the sum of $L_1$ lengths of edges of $G$. Our idea is to label hollowed balls with different subsets of edges of $G$. 
%We label the red ball in the right of \autoref{fig:L1balls} with $t-t_2$, the blue balls with $t_2$ and the black balls with $t$. 
Although the red hollowed ball and the blue balls intersect, we charge their widths to disjoint subsets of edges of $G$; we charge the width of the red ball with $k$ edge-disjoint paths supported on $G[V(t_1)\setminus V(t_2)]$ going across this hollowed ball and we charge the radius of each blue ball with $k$ edge-disjoint paths supported on $G(t_2)$ going across that ball.

We remark that the above idea is essentially the main new operation we need for the charging argument, compared to the argument for the compact bags of balls. One of the main obstacles in using this idea is that $t_1$ can have more than two children. In that case $G[V(t_1)\setminus V(t_2)]$ is not necessarily $k$-edge-connected. To overcome this, we find a natural decomposition of $G[V(t_1)\setminus V(t_2)]$ into $k/4$-edge-connected components; since each assigned bag of balls, $\bag_t$ has $\gg \cut(t)/k$ balls, the centers of a large  number of balls of $\bag_t$ are neighbors of one of these components; so we can  charge the red ball in the above argument by $k/4$ edge-disjoint paths in that component.

\begin{figure}
\centering
\begin{tikzpicture}
\def\slen{0.49}
\def\hslen{0.3}
\def\opacity{0.25}
\begin{scope}%[rotate=45]
\draw [fill=black] (0,0) circle (2pt);
\draw [fill=black,opacity=\opacity] (0,0) node[below left=2.2cm,draw=none,color=black,opacity=1] {$B$} circle (3.2);
\foreach \i/\bx/\by/\rx/\ry in {1/-2/0/.15/.25, 2/-1/0/-.25/.2, 3/1/0/-.2/-.25, 4/2/0/.2/-.2}{
	\draw  (\bx,\by) node[below=.42cm,draw=none,color=black] {$B_\i$} circle (\slen);
	\draw [fill=blue] (\bx,\by) circle (2pt);
	\draw [fill=red,] (\bx+\rx,\by+\ry) circle (2pt);
}
\end{scope}
\Huge \draw[double,-implies,double distance between line centers=4.5pt] (3.3,0) -- (4.8,0);

\begin{scope}[shift={(8,0)}]
\draw [fill=black] (0,0) circle (2pt);
\draw [fill=black,opacity=\opacity] (0,0) 
%node[below left=2.2cm,draw=none,color=black,opacity=1] {$B$} 
circle (3.2);
\draw [fill=white,draw=none] (0,0) circle (2.5);
\draw [fill=red,opacity=\opacity,draw=none] circle (1.86);
\draw [fill=white,draw=none] circle (.84);
\draw [fill=black,opacity=\opacity] (0,0) circle (.5);
\draw [fill=black] (0,0) circle (2pt);

\foreach \i/\bx/\by/\rx/\ry in {1/-2/0/.15/.25, 2/-1/0/-.25/.2, 3/1/0/-.2/-.25, 4/2/0/.2/-.2}{
	\draw  [draw=none,opacity=\opacity,fill=blue] (\bx,\by) 
	%node[below=.42cm,draw=none,color=black] {$B_\i$} 
	circle (\slen);
	\draw [fill=blue] (\bx,\by) circle (2pt);
	\draw [fill=red,] (\bx+\rx,\by+\ry) circle (2pt);
}
\end{scope}
\end{tikzpicture}
\caption{A simple example of the ball labeling technique. 
%Balls are illustrated as squares to emphasize that these are $L_1$ balls. 
The grey (hollowed) ball $B$ on the left is one of the hollowed balls of $\cZ_\tau$. Small $L_1$ balls with blue vertices as their centers represent balls of $\bag_{t_2}$ that we are processing at time $\tau$. Each red vertex together with the closest blue vertex are the endpoints of an edge of $\cut(t_2)$. The right figure shows new balls added to $\cZ_{\tau+1}$. In particular, each blue vertex is in $V(t_2)$ and each red vertex is in $V(t_3)$ where $t_2,t_3$ are the only children of $t_1$.}
\label{fig:L1balls}
\end{figure}
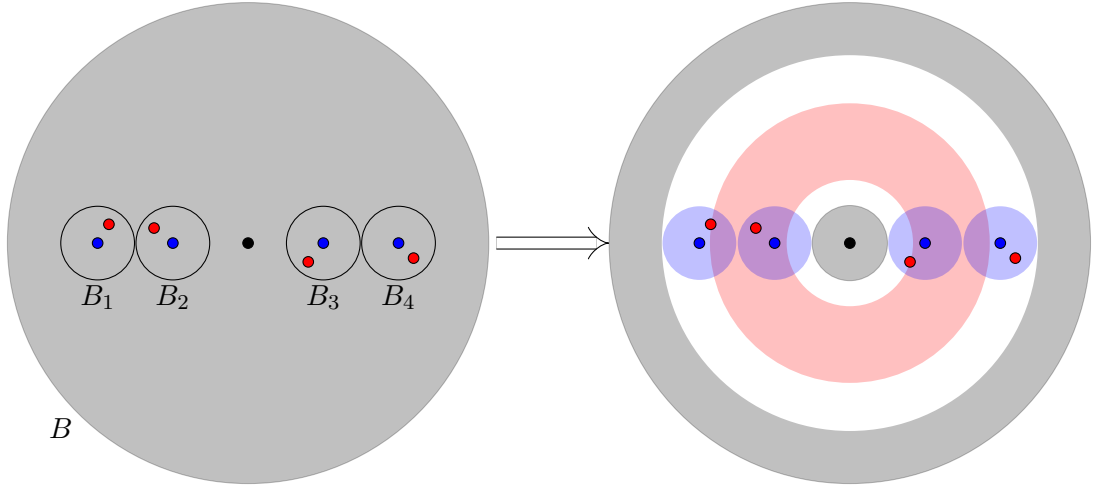

\subsubsection{Ball Labeling} 
In this part we define a valid labeling of hollowed balls
in our construction (see \autoref{fig:labelassignedballs}).
In the proof of \autoref{prop:compactballconstruction}, we used the disjointness property of balls in the construction in two places; namely in the proofs of \autoref{lem:countingedgelengthradii} and \autoref{fact:compactinterior}. 
We address both of these issues by our ball labeling technique.

\paragraph{Basic Label.} In the proof of \autoref{lem:countingedgelengthradii} we used the disjointness property to charge  the sum of the widths of hollowed balls of a set $\cZ_\tau$ to the sum of the $L_1$ lengths of edges of $G$ with no overcounting. Let us give a simple example to show the difficulty in extending this argument to the new setting where balls may intersect.
Suppose $\cZ_\tau$ is a union of $10$ identical copies of $B(x,r)$ with the guarantee that there is a vertex of $G$ at x and one at distance $r$ of $x$. Then,  the sum of the $L_1$ lengths of edges of $G$ can be as small as $k\cdot r$, as $G$ may just be  $k$-edge-disjoint paths from a vertex at $x$ to  a vertex at distance $r$ of $x$.

A hollowed ball $B=B(x,r_1\|r_2)$, 
can be labeled with  $t\in \cT$, denoted by $t(B)=t$, if there are vertices $u,v\in V(t)$ such that $\norm{x-X_u}_1\leq r_1$ and  $\norm{x-X_v}_1 \geq r_2$. 
Recall that, by the definition of $\cT$, for any node $t\in\cT$, $G(t)$ is $k$-edge connected. 
Therefore, if  $B$  is labeled with $t$, then $k$ edge-disjoint paths supported on $E(t)$ cross $B$. 
For any ball $B\in\bag_t$ we let $t(B)=t$. 
Furthermore, when we shrink or divide a ball into smaller ones the label of the shrunk ball or the new subdivisions remain unchanged. 

The simplest definition of the validity of the ball labeling is to make sure that for any two intersecting balls $B$ and $B'$,
$t(B)$ and $t(B')$ are not ancestor-descendant. 
Unfortunately, this simple definition is not enough for our inductive argument,
%We note that $t(B)$ is the basic label of a ball $B$ that is well-defined for all balls in $\cZ_\tau$> 
and as we elaborate next, we will enrich the label of some of the balls $B$ by ``disallowing routing through some of the descendants of $t(B)$''. 
Recall that $t,t'\in\cT$ are {\em ancestor-descendant} if either $t$ is a weak ancestor of $t'$ or $t'$ is a weak ancestor of $t$.
Recall that $t$ is a weak ancestor of $t'$ if either $t$ is an ancestor of $t'$ or $t=t'$.

To this end, we define a \emph{\conflictset}, $\cC(B)$ to be a \emph{connected} subset of the nodes of $\cT$ rooted at $t(B)$ (see \autoref{fig:conflictset}). 
%Now, we address the issue of intersecting balls.
In a valid ball labeling, we  make sure that
for any two intersecting hollowed balls $B$ and $B'$, 
$\cC(B)\cap \cC(B')=\emptyset$. 
For example, if $t(B),t(B')$ are not ancestor-descendant this condition is always satisfied. 
In the charging argument, we may only charge the width of $B$ with edge-disjoint paths supported on the leaves of $\cT$ which are in $\cC(B)$ (see \autoref{fig:conflictset}). Recall that the leaves of $\cT$ are 
identified with the vertices of $G$. 
	\begin{figure}
	\centering
	\begin{tikzpicture}
		\tikzstyle{every node} = [draw,circle,minimum size=5mm,inner sep=0];
		\node [color=red] at (0,0) (a_1) {$t_1$};
		\node [draw=none] at (-1,1) (a_0) {$\ddots$} edge (a_1);
		\node [color=red] at (2,-1) (a_3) {$t_3$} edge (a_1);
		\node at (-2,-1) (a_2) {$t_2$} edge (a_1);
		\node at (-3,-2) () {$v_1$} edge (a_2);
		\node at (-1,-2) () {$v_2$} edge (a_2);
		\node [color=red] at (1,-2) (a_4) {$t_4$} edge (a_3); 
		\node [color=red] at (0,-3) () {$v_3$} edge (a_4);
		\node [color=red] at (1,-3) () {$v_4$} edge (a_4);
		\node at (3,-2) (a_5) {$t_5$} edge (a_3);
		\node  at (2,-3) () {$v_5$} edge (a_5);
		\node  at (3,-3) () {$v_6$} edge (a_5);
	\end{tikzpicture}
		\caption{The red nodes represent the \conflictset~of a ball $B$ with $t(B)=t_1$, i.e., $\cC(B)=\{t_1,t_3,t_4,v_3,v_4\}$. The edge-disjoint paths of $B$ can be routed in the induced subgraph $G[\{v_3,v_4\}]$.}
		\label{fig:conflictset}	
	\end{figure}
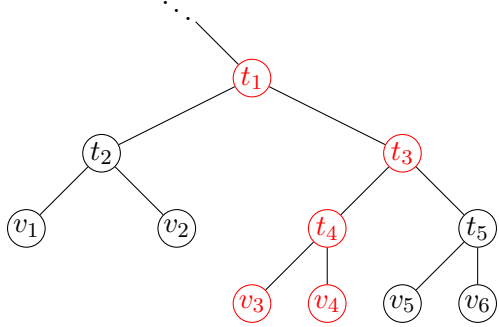
%the two families of  $k$ edge-disjoint paths that cross $B,B'$  are  supported on disjoint subsets of edges. 
%We also assume that $t,t$ are ancestor-descendant.
%The following fact about ancestor-descendancy will be used throughout the proof.
%\begin{fact}
%\label{fact:ancestor-descendancy}
%If $t$ is an ancestor-descendant of $t_1,t_2$, and $t_1$ is not an ancestor-descendant of $t_2$, then both of $t_1,t_2$ must be descendants of $t$.
%\end{fact}

%  In the simplest form let two balls intersect if their labels are not ancestor-descendant.
%
%
%Our idea is to use the hierarchical decomposition, $\cT$, to allow intersecting balls in our construction. 
%Let $t,t'\in\cT$ such that neither $t$ is an ancestor of $t'$ nor $t'$ is an ancestor of $t$.
%If we have two identical copies of $B(x,r)$ one with the promise that $k$ edge-disjoint path supported on $E(t)$ cross it and the other with the promise that $k$ edge-disjoint path supported on $E(t')$ cross it; then we can guarantee that sum of the $L_1$ length of edges of $G$ is at least $2k\cdot r$.
% 
%We now formalize the above idea by labeling balls. 
%It follows that, if $B$ is labeled with $t$, then at least 
%$k$ edge-disjoint path supported on $E(t)$ cross $B$.

\paragraph{\Avoiding~Balls}
As alluded to in \autoref{fig:L1balls}, we may add new (hollowed) balls, called \avoiding~balls, to $\cZ_\tau$ that do not exist in the given geometric sequence.
%We may treat these balls differently and call them {\em \avoiding} hollowed balls, 
An \avoiding~(hollowed) ball $B$, has  an additional label, $t_d(B)$, where $t_d(B)$ is always a descendant of $t(B)$; the name \avoiding~stands for the fact that the edge-disjoint paths of $G(t(B))$ that are crossing $B$ are avoiding the induced subgraph $G(t_d(B))$. Therefore, we exclude the subtree of $t_d(B)$ from $\cC(B)$, i.e., $\cC(B)\cap t_d(B)=\emptyset$.

We insert an \avoiding~hollowed ball only when we shrink or remove part of a nonavoiding (hollowed) ball that already exists in $\cZ_\tau$. 
%Suppose we add an \avoiding~hollowed ball $B'$ to $\cZ_{\tau+1}$  after removing part of a ball $B\in\cZ_\tau$;
%then, we let $t(B')=t(B)$ and $B'$ has an additional label $t_d(B')$  where $t_d(B')$ is  a descendant of $t(B')$.
For example, if $B'$ is the  red ball on the right side of \autoref{fig:L1balls}, then $t(B')=t$, $t_d(B')=t_2$. Note that it is important that avoiding balls are replacing nonavoiding balls; if in the arrangement of \autoref{fig:L1balls} the ball $B$ were an avoiding ball, then the red ball would have to avoid two induced subgraphs; further escalation of this would lead to unmanageable labels. We get around this by never introducing an avoiding ball when the original $B$ is avoiding. Also, for the charging argument to work we need to allocate a fraction of the number of tokens that would be normally allocated to a nonavoiding ball.

For any \avoiding~hollowed  ball $B=B(x,r_1\|r_2)$  there must be  vertices $u,v\in V(t(B)) \setminus V(t_d(B))$
such that $\norm{X_u - x}_1 \leq r_1,$ $\norm{X_v-x}_1 \geq r_2$ and that there are at least $k{\bf/4}$ edge-disjoint paths from $u$ to $v$ in the induced graph $G[V(t(B))\setminus V(t_d(B))]$. Note that if for such a ball, one defines $\cC(B)$ to be the subtree rooted at $t(B)$ minus the subtree rooted at $t_d(B)$, then these $k/4$ edge-disjoint paths must be supported on the leaves of $\cT$ that are in $\cC(B)$.

\paragraph{Non-insertable Balls}  
In \autoref{fact:compactinterior} we used the disjointness property to argue that any ball of $\fbag_\l$ is in the interior of at most one hollowed ball of $\cZ_\tau$. 
Here, this fact may not necessarily hold: 
Suppose at time $\tau$, a ball $B\in \bag_t$ is in the ``interior'' of two balls $B_1,B_2$, i.e., the center of $B$ is far from the boundaries of $B_1,B_2$, and $t$ is an ancestor-descendant of both  $t(B_1),t(B_2)$. Then, $B_1,B_2$ intersect. Assuming that balls of $\cZ_{\tau}$ have a ``valid labeling'', since $B_1,B_2$ are intersecting, $t(B_1),t(B_2)$ are not ancestor-descendant. One would hope that this configuration is impossible. But in fact, it could be the case that $t(B_1), t(B_2)$ are descendants of $t(B)$ that are not ancestor-descendants of each other. In this configuration, one cannot hope to add $B$ with the label $t(B)=t$. 
%Now, if either of $t(B_1)$ or $t(B_2)$ is an ancestor of $t$ we reach to a contradiction; . But a bad case happens when  both of $t(B_1),t(B_2)$ are descendants of $t$. 

In general, the above scenario occurs only if the bags assigned to descendants of a node $t$ appear earlier in the geometric sequence, i.e., if we process $\bag_t$ after processing bags assigned to its  descendants. 
In the first reading of the proof, one can assume that this scenario does not happen and avoid the notation $\Pr(.)$ and (non-)insertable balls that we define below. To address this issue we will use the third property of the \expandertree.
To any (hollowed) ball $B$ in our construction with $t(B)=t$, we will assign $\Pr(B)\subset\cT$ to be a set of descendants of $t$ with the  guarantee that
there are $k$ edge-disjoint paths across $B$ supported on $G[V(t)\setminus \cup_{t'\in\Pr(B)} V(t')]$. 
In other words, we exclude the subtrees rooted at nodes of $\Pr(B)$ from $\cC(B)$. 
We will prune everything from $\bag_t$ except the balls $B$ such that $\Pr(B)$ includes all descendants of $t$ that are processed earlier than $t$. We use the third property of the \expandertree, $\cT$, to show  that the pruning step only removes a small fraction of balls.

\medskip 

Recall that $\fbag_\l$ has type $(\br_\l,T_\l)$. For a node $t\in T_\l$, we say a node $t'$ is a {\em predecessor} of $t$, if  $t'$ is a descendant of $t$ and 
$t'\in T_{i}$ for some $i < \l$.
For any node $t$ and any ball $B=B(X_u,r)\in\bag_t$ we say
$B$ is {\em non-insertable} by $t'$ if $t'$ is a predecessor of $t$ and an endpoint of an edge of $\setdeg(t')$ is in $B$ (see \autoref{subsec:expandertree} for the definition of $\setdeg(t')$).
We say
$B$ is {\em insertable} otherwise. 
For any insertable ball $B\in\bag_t$ we let $\Pr(B)$ be the set of predecessors of $t$. 
In other words, a ball $B=B(X_u,r)\in\bag_t$ is insertable if and only if 
\begin{enumerate}[i)]
\item For any  $t'\in \Pr(B)$, %$u\notin V(t')$ and 
 all endpoints of the edges of $\setdeg(t')$ are outside of $B$, and 
 \item For any $t'\in\Pr(B)$, $u\notin V(t')$, i.e., $u$ does not belong to any of the subtrees rooted at nodes of $\Pr(B)$.
 \end{enumerate}
 Observe that, by the definition of assigned bags of balls, (ii) follows from (i). 
 %if $B=B(X_u,r)$ is insertable, then for any $t'\in\Pr(B)$, $u\notin V(t')$. 
 In particular, since $B\in\bag_t$, there is an edge $\{u,v\}\in \cut(t)$ for $v\notin V(t)$. Therefore,  if $u\in V(t')$, $\{u,v\}\in \setdeg(t')$ which is a contradiction.

\subsubsection{Preprocessing}
In this subsection, we delete all non-insertable balls and we show that they contribute only to a small fraction of the sum of the radii of the given geometric sequence. Then, we formally define a valid labeling and we show that we can lower bound the denominator by the sum of the widths of balls in a valid labeling. 
At the end of this subsection, we reduce \autoref{prop:assignedballconstruction}
to a ``simpler'' statement, that is the existence of an arrangement of a set of hollowed balls with a valid labeling such that the sum of the widths of all hollowed balls in the construction is a constant fraction of the sum of the radii of all balls in the given geometric sequence. 
%(see \autoref{prop:assignedbagreduced}). 

%In \autoref{lem:edgecountingassigned}
%we will prove that for any insertable nonavoiding ball, $B$, there are $k$ edge-disjoint paths across $B$ supported on $G[t(B))\setminus \cup_{t'\in P(B)} V(t')]$.
 
 %Since $B$ is insertable, i.e., $B$ is not  non-insertable by $t'$, and $X_u\in B$, $u\notin V(t')$. 
 
%We say $B\in\bag_t$ is non-insertable by $t'\in\Pr(t)$ if an endpoint of an edge of $\setdeg(t')$ is in $B$.  
%if $t'$ is a predecessor of $t$  and an endpoint of an edge of $\setdeg(t')$ is in $B$, i.e., there is an edge $\{v,w\}\in \setdeg(t')$ such that $\norm{X_u - X_v}_1 < r$. 
In the following lemma we show that for any node $t\in T_\l$, the sum of radii of all balls that are non-insertable by $t$ is 
$\ll \rho_\l \cdot |\bag_t|$.
\begin{lemma}
\label{lem:noninsertable}
For any node $t\in T_\l$,
$$ \sum_{i} \sum_{B\in \fbag_i} \I{B \text{ is non-insertable by } t}\cdot \rho_i \leq \frac{4 \rho_\l \cdot |\bag_t|}{C_3}.$$
\end{lemma}
\begin{proof}
For any $i$ let $b_i$ be the number of balls  in $\fbag_i$ that are non-insertable  by $t$. By definition, $b_i=0$ for $i\leq \l$. We will show that for all $i > \l$, 
\begin{equation}
\label{eq:bidegt}
b_i \leq 2|\setdeg(t)|.
\end{equation}
Then,
\begin{eqnarray*}
\sum_{i}\sum_{B\in \fbag_i} \I{B \text{ is non-insertable by } t}\cdot \rho_i & = &\sum_{i>\l} b_i\cdot \rho_i  \\
&\leq &
2|\setdeg(t)|\sum_{i>\l} \rho_i \\
&\leq& 4\lambda\cdot |\setdeg(t)|\cdot\rho_\l \\
&\leq & \frac{4 |\cut(t)|}{k}\cdot\rho_\l \\&\leq& \frac{4|\bag_t|\rho_\l}{C_3}.
\end{eqnarray*}
where the second to last inequality uses $\cT$ is a $(k,k\lambda,T)$-LCH of $G$, i.e., that $t\in T$ and $\lambda \cdot k\cdot |\setdeg(t)|\leq |\cut(t)|$. The last inequality uses \eqref{eq:assignedbag} and that $\bag_t$ is a $C_3/k$-assigned bag of balls. 

It remains to prove \eqref{eq:bidegt}. Fix $i> \l$. For any ball $B=B(X_u,\rho_i) \in \bag_{t'}$ that is non-insertable by $t$, at least one endpoint of an edge of $\setdeg(t)$ is in $B$. %This is because either both endpoints of an edge of $\setdeg(t)$ is in $B$, or $u\in V(t)$. In the latter case there is an edge $\{u,v\}\in \cut(t')$ for $v\notin V(t')$. 
%Since $t$ is a descendant of $t'$, $V(t)\subseteq V(t'),$ so $\{u,v\}\in \setdeg(t)$.
Since all balls of $\fbag_i$ are disjoint, $b_i \leq 2|\setdeg(t)|$. 
\end{proof}
By the above lemma it is sufficient to prove \autoref{prop:assignedballconstruction} with the assumption that all balls in the given geometric sequence are insertable (see \autoref{prop:assignedbagreduced} at the end of this part).

In \autoref{fig:labelassignedballs} we define a valid labeling of balls.
Later, in our inductive argument we will make sure that at any time $\tau$, $\cZ_\tau$ has a valid labeling.

%Note that, by definition, any insertable ball satisfies property \ref{labelprop:nonavoid}. Therefore, for any insertable ball $B\in \bag_t$,  $B$ can be added to $\cZ_\tau$, if $B,B'$ satisfy condition \ref{labelprop:validlabeling} of \autoref{fig:labelassignedballs} 
%for all already inserted balls $B'$.

\begin{figure}[htb]
\centering
\fbox{\parbox{6.6in}{
\vspace{2mm}
Any set of balls has a valid ball labeling if it satisfies the following properties. 
\begin{enumerate}
%\item\label{labelpropP}
%For any hollowed ball $B=B(x,r_1\|r_2)$ and any $t'\in \Pr(B)$, $t'$ is a predecessor of $t(B)$ and
%all endpoints of the edges of $\setdeg(t')$ are outside of $B$.
\item
\label{labelprop:conflictsetdef}
For any nonavoiding ball $B$, $\cC(B)$ is  the connected subtree rooted at $t(B)$ excluding the subtrees rooted at nodes of $\Pr(B)$.
If $B$ is avoiding, in addition to above, 
$\cC(B)$ excludes the subtree rooted at $t_d(B)$. Note that we always have $t(B)\in\cC(B)$.
%Any non\avoiding~hollowed ball $B=B(x,r_1\|r_2)$ with labels $t(B)=t$, $\Pr(B)$ satisfies the following:
%\begin{enumerate}[i)]
\item \label{labelprop:insertability} For any hollowed ball $B=B(x,r_1\|r_2)$, any $t'\in\Pr(B)$, and $\{u,v\}\in\setdeg(t')$, $\norm{x-X_{u}}_1,\norm{x-X_{v}}_1 \geq r_2$.
%\end{enumerate}
%\item 
%\label{labelprop:avoid}
%Any \avoiding~hollowed ball $B=B(x,r_1\|r_2)$ with labels $t(B)=t$, $t_d(B)=t_d$, $\Pr(B)$ satisfies the following:  
%\begin{enumerate}[i)]
%\item 	$\cC(B)$ is the set nodes $t'$ of the subtree of $t(B)$ excluding the subtrees rooted at nodes of $\Pr(B)$ and $t_d(B)$.
%$t_d$ is a descendant of $t$, and there are vertices 
%$$u\in V(t)\setminus V(t_d)\setminus \cup_{t'\in\Pr(B)} V(t') \text{ and } v\in V(t)\setminus V(t_d)$$ such that
%$\norm{x-X_u}_1 \leq r_1$ and $\norm{x-X_v}_1\geq r_2$ and there are at least $k/4$ edge-disjoint paths in the induced subgraph $G[V(t)\setminus V(t_d)]$ from $u$ to $v$.
%\item For any $t'\in\Pr(B)$ and $\{u',v'\}\in\setdeg(t')$, $\norm{x-X_{u'}}_1,\norm{x-X_{v'}}_1 \geq r_2$.
%\end{enumerate}
%\item For any ball $B=(x,r_1\|r_2)$ and any $t\in\setdeg(B)$
%and any $\{u,v\}\in \cut(t)$, 
%$$ \norm{x-X_u}_1, \norm{x-X_v}_1 \geq r_2.$$
\item \label{labelprop:disjointpaths} For any ball $B=(x,r_1\|r_2)$, there is a vertex $u\in \cC(B)$ such that $\norm{x-X_u}_1\leq r_1$ and there are at least $k/4$ edge-disjoint paths originating from $u$, crossing $B$, supported on $V(t(B))\setminus V(t_d(B))$. In the proof of \autoref{lem:edgecountingassigned} we show that this implies that we have $k/4$ edge-disjoint paths crossing $B$ and supported on leaves of $\cT$ which are in $\cC(B)$.  %originating from one of the  leaves of $\cT$ which are in $\cC(B)$. 
%There are vertices $u\setminus \cup_{t'\in\Pr(B)} V(t')$ and $v\in V(t)$ such that
%$\norm{x-X_u}_1\leq r_1$ and $\norm{x-X_v}_1\geq r_2$.
\item 
For any two intersecting (hollowed) balls $B_1$ and $B_2$, $\cC(B_1)\cap \cC(B_2)=\emptyset$.
Observe that $\cC(B_1)\cap \cC(B_2)\neq\emptyset$ if and only if either $t(B_1)\in\cC(B_2)$ or $t(B_2)\in\cC(B_1)$.
\label{labelprop:validlabeling}
%one of the following four conditions is satisfied: 
%\begin{enumerate}[i)] 
%\item $B_1$ and $B_2$ do not intersect. 
%\item  $t(B_1)$ and $t(B_2)$ are not ancestor-descendant, 
%\item  $\Pr(t(B_1))$ contains a weak ancestor of $t(B_2)$, or $B_1$ is \avoiding~and $t_d(B_1)$ is a weak ancestor of $t(B_2)$, 
%\item  Conversely, $\Pr(t(B_2))$
% contains a weak ancestor of $t(B_1)$, or $B_2$ is  \avoiding~and $t_d(B_2)$
% is a weak ancestor $t(B_1)$.
% \end{enumerate}
\end{enumerate}
}}
\caption{Properties of a valid ball labeling}
\label{fig:labelassignedballs}
\end{figure}

%A ball $B=\in \cZ_{tau}$

The following lemma extends \autoref{lem:countingedgelengthradii} to the new setting where the balls of $\cZ_\tau$ may intersect.

\begin{lemma}
\label{lem:edgecountingassigned}
For any set of hollowed balls $\cZ$ with a valid labeling we have,
$$ \frac{k}{4}\cdot\sum_{B(x,r_1\|r_2)\in\cZ} (r_2-r_1) \leq  \sum_{\{u,v\}\in E} \norm{X_u-X_v}_1.$$
\end{lemma}
\begin{proof}
%First, we show that property \ref{labelprop:insertability}\ref{labelprop:disjointpaths} 
By property \ref{labelprop:disjointpaths}, for any ball $B=(x,r_1\|r_2)$ there are $k/4$ edge-disjoint paths crossing $B$ originating from a vertex $u\in \cC(B)$ such that $\norm{X_u-x}_1\leq r_1$.
We only keep the portion of each of these paths starting from $u$ until  the first vertex that lies outside of $B(x,r_2)$ (and we discard the rest). 
%By property \ref{labelprop:insertability}, 
Next, we show that these paths remain inside $\cC(B)$. 
This is because by property \ref{labelprop:disjointpaths} these paths exclude the subtree rooted at $t_d(B)$. In addition, these paths start at a vertex that does not lie in any of the subtrees rooted at $\Pr(B)$; by property \ref{labelprop:insertability} they can never enter such a vertex.
Therefore, these paths avoid the subtrees rooted at $\Pr(B)$ as well, or in other words they are completely supported on $\cC(B)$. 

%each ball $B$ of $\cZ$; 
We further trim each of these paths from both ends so that the resulting paths lie inside $B$. By the $L_1$ triangle inequality, the $L_1$ length of the trimmed paths is at least the width of $B$.
Now, by property \ref{labelprop:validlabeling}, no edge of $G$ is charged by more than its $L_1$ length. 
\end{proof}
%By the above lemma, to prove \autoref{prop:assignedballconstruction}, it is enough to construct a set of hollowed balls with a valid labeling
%such that the sum of the widths of hollowed balls in our construction is proportional to sum of the radii of balls in the given geometric sequence of assigned bags of balls.
\begin{proposition}
\label{prop:assignedbagreduced}
Given a $(k,k\cdot \lambda,T)$-LCH $\cT$ of $G$ and a $\lambda$-geometric sequence of  families of $12C_3/k$-assigned bags of balls, $\fbag_1$, $\fbag_2$, $\dots$, such that $\fbag_i$ has type $(\rho_i,T_i)$ and $T_i$'s are disjoint subsets of $T$, if all balls of all bags in the sequence are insertable, $C_4\geq 3$, $\lambda\leq 1/6C_4$, and $C_3 \geq 2((C_4+1)+4(C_4+2)^2)$, then
there is a set $\cZ$ of hollowed balls with a valid labeling such that
$$ \frac{C_4}{12C_3} \cdot \sum_{i} \sum_{t\in T_i} \rho_i\cdot |\bag_t| \leq  \sum_{B(x,r_1\|r_2)\in\cZ} (r_2-r_1).$$
\end{proposition}

It is easy to see that the above proposition together with \autoref{lem:edgecountingassigned} implies \autoref{prop:assignedballconstruction}.
\begin{proofof}{\autoref{prop:assignedballconstruction}}
For any $i$ and any $t\in T_i$
we remove all non-insertable balls in $\bag_t$.
If  at least half of the balls of $\bag_t$ are insertable then we will have a $12 C_3/k$-assigned bag of balls. Otherwise, we remove $\bag_t$ from our geometric sequence and we remove $t$ from $T_i$. The resulting geometric sequence satisfies the conditions of  \autoref{prop:assignedbagreduced}.

By \autoref{lem:noninsertable},
 the sum of the radii of balls that we removed, which is at most twice the sum of the radii of all non-insertable balls, is at most half of the radii of all
balls in the given geometric sequence,
\begin{eqnarray*} 
\sum_{j} \sum_{B\in \fbag_j} \I{B \text{ is non-insertable}}\cdot\rho_j 
&\leq & \sum_{i} \sum_{t\in T_i} \sum_{j} \sum_{B\in\fbag_j} \I{B \text{ is non-insertable by $t$}}\cdot \rho_j \\
&\leq& 
\sum_{i}\sum_{t\in T_i} \frac{4 |\bag_t|\cdot\rho_i}{C_3}\leq \sum_{i} \sum_{t\in T_i} \frac{|\bag_t|\cdot\rho_i}{4}
\end{eqnarray*}
where the last inequality uses $C_3\geq 16$. 
Therefore, the proposition follows by \autoref{lem:edgecountingassigned}.
\end{proofof}

We conclude this section with a simple fact. We show that 
any insertable ball $B\in\bag_t$ satisfies properties \ref{labelprop:insertability} and \ref{labelprop:disjointpaths}.
\begin{fact}
	Any insertable ball $B=B(X_u,\br_\l)\in\bag_t$ satisfies properties \ref{labelprop:insertability} and \ref{labelprop:disjointpaths} of \autoref{fig:labelassignedballs}.
\end{fact}
\begin{proof}
Property \ref{labelprop:insertability} follows by the definition of insertable balls. 	To see \ref{labelprop:disjointpaths} note that all balls of $\bag_t$ are nonavoiding; in addition, since $B$ is insertable, $u$ does not belong to any of the subtrees rooted at $\Pr(B)$. Since by definition of $\cT$, $G(t)$ is $k$-edge-connected, there are $k$ edge-disjoint paths from $u$ to a vertex of $V(t)$ outside of $B$ (note that since $|\bag_t|>1$ there is always a vertex of $V(t)$ outside of $B$). 
\end{proof}

%$B=B(x,r_1\|r_2)$, there are $k/4$ %edge-disjoint paths originating from a %vertex $u\in V(t(B))\setminus V(

\subsubsection{Order of Processing}
In the rest of this section we prove \autoref{prop:assignedbagreduced}. 
So from now on, we assume all balls of all bags in the sequence are insertable and that every bag is $12 C_3/k$-assigned.

Similar to \autoref{subsec:compactchargin}, we give an inductive proof. In this part we describe general properties of our construction and we use them to prove two essential lemmas. We process families of bags of balls in phases, and in phase $\l$ we process $\fbag_\l$. %There is one modification. We process bags of balls in $\fbag_\l(\rho_\l,T_\l)$ in an increasing order of the distance to the root of $\cT$. For example, if $\bag_{t_1}(\rho_\l, b_{t_1}), \bag_{t_2}(\rho_\l, b_{t_2})\in\fbag_\l$, such that $t_1$ is an ancestor of $t_2$, then we first process $t_1$ and then $t_2$. 
We need to use slightly larger (compared to the previous section) constants in the definition of interior balls.
\begin{definition}[Interior ball]
%Say at time $\tau$ we are processing $\bag_t(\rho_\l,b_t)$.
We say a ball $B=B(X_u,\br_\l)\in\bag_t$  is {\em in the interior of} a hollowed ball $B'=B(x,r_1\|r_2)$ if  %$\{B,B'\}$'s labeling is invalid 
$\cC(B)\cap \cC(B')\neq\emptyset$
and,
$$ r_1+C_3\cdot \rho_\l < \norm{x-X_u}_1 < r_2-C_3\cdot \rho_\l. $$
%if $B'$ is not a child-avoiding ball and
%$$ r_1+\rho_\l \leq \norm{x-X_u}_1 \leq r_2-\rho_\l, $$
%otherwise.
We say $B$ is an interior ball (with respect to $\cZ$) if $B$ is in the interior of a %non child-avoiding 
hollowed ball (of $\cZ$). If $B$ is not an interior ball, we call it a border ball. Similar to the previous section we insert all border balls of phase $\l$ at time $\tau_\l$. 
\end{definition}

See \autoref{fig:assignedballsorder} for the main properties of our inductive construction.
In the rest of this  part we use these properties to prove lemmas \ref{lem:uniqueinterior} and \ref{lem:invariantborder}.
The following fact follows simply by 
property~\ref{assprop:tau}.
\begin{lemma}
\label{fact:maptBtotau}
%For any time $\tau_{\l-1}+1 \leq \tau\leq \tau_\l$, $s\geq 0$, any ball $B\in\cZ_\tau$,
% $t(B)\in T_i$ for $i\leq \l$. If $i=\l$, then $\bag_t$ must be processed by time $\tau+1$.
Suppose we are processing $\bag_t\in\fbag_\l$ at time $\tau$. For any $s\geq 0$ and any ball $B\in\bag_t$ and $B'\in\cZ_{\tau,s}$, if $\cC(B)\cap \cC(B')\neq \emptyset$, then $t(B')$ is a weak ancestor of $t$. 
\end{lemma}
\begin{proof}
Let $t'=t(B')$.
If $\cC(B)\cap \cC(B')\neq\emptyset$, then by property \ref{labelprop:conflictsetdef} of \autoref{fig:labelassignedballs}, $t,t'$ are ancestor-descendant. So, we just need to show that $t'$ is not a descendant of $t$.

First, by properties \ref{assprop:labelinvariant} 
and \ref{assprop:tau} of \autoref{fig:assignedballsorder},  $t'\in T_i$ for some $i\leq \l$. If $t'\in T_\l$  either $t'=t$ or $\bag_{t'}$ is processed by time $\tau$. 
Therefore, by property \ref{assprop:root} of \autoref{fig:assignedballsorder},
$t'$ is not a descendant of $t$ and we are done.
Otherwise, $t'\in T_i$ and $i <\l$.
If $t'$ is a descendant of $t$, then it is a predecessor of $t$ and since $B$ is an insertable ball, $t'\in\Pr(B)$. So $t'\notin \cC(B)$ and $\cC(B')\cap \cC(B)=\emptyset$, which cannot be the case.
\end{proof}

\begin{figure}[htb]
\fbox{ \parbox{6.6in} {
\begin{enumerate}
\item Phase $\l$ starts at $\tau_{\l-1}+1$ and ends at $\tau_\l$. In phase $\l$, we process assigned bags of  balls in $\fbag_\l$, in the increasing order of the depth\footnotemark of the node to which they are assigned in $\cT$. For example, if $\bag_{t_1}, \bag_{t_2}\in\fbag_\l$ and $t_1$ is an ancestor of $t_2$,  we process $\bag_{t_1}$ before $\bag_{t_2}$. 
\label{assprop:root}
%\item For any $\bag_t\in \fbag_\l$  $B=B(X_u,\rho_j)\in\bag_t$ 
%if there is a node $t'$ that is a descendant of $t$ and $\bag_{t'} \in \fbag_\i$ for $i<j$ such that 
\item  Any ball of $\fbag_\l$ that we insert (in phase $\l$) remains unchanged till the end of phase $\l$. 
All other hollowed balls may be shrunk or be split into several balls but their labels (and their \conflictset s) remain invariant. 
%At any time $\tau_{\l-1} < \tau \leq \tau_\l$ we may add new balls not in $\fbag_\l$ but any of these balls 
\label{assprop:labelinvariant}

\item Say at time $\tau_{\l-1} < \tau< \tau_\l$ we are processing $\bag_t$. We construct $\cZ_{\tau+1}$ inductively by constructing $\cZ_{\tau,0}=\cZ_\tau$, $\cZ_{\tau,1}$, $\ldots$, $\cZ_{\tau,\infty}=\cZ_{\tau+1}$.
We make sure that each set $\cZ_{\tau,s}$ has a valid labeling.
When we are constructing $\cZ_{\tau,s+1}$, we  insert several new (hollowed) balls where only some of them are in $\bag_t$. 
Those not in $\bag_t$ are inserted as a result of a conflict in labeling that would be introduced if we inserted a ball of $\bag_t$.
%, i.e., a ball of $\bag_t$ is in the interior of a non\avoiding~ball $B'=B(x,r'_1\|r'_2)\in \cZ_{\tau,s}$
%where $t(B')$ is an ancestor of $t$, and $\Pr(B')$ does not contain any weak ancestor of $t$.
In these cases, we split or shrink an already inserted nonavoiding ball $B'$ %where $t(B')$ is an ancestor of $t$, 
and we insert
 new hollowed balls $B\subseteq B'$ 
such that $\cC(B)\subseteq \cC(B')$. We also let $t(B)=t(B')$ or $t(B)=t$ depending on whether $B$ is avoiding or nonavoiding.
%in $\cZ_{\tau,s+1}$ %, that is not in $\bag_t$, must
% be {\em inside} a non\avoiding~hollowed ball 
%$B'=B(x,r'_1\|r'_2)\in \cZ_{\tau,s}$, i.e., 
%such that $r'_1\leq r_1$, $r_2\leq r'_2$ and $B'$ has the same center as $B$.
%We define the label of $B$ in such a way 

%More precisely, the label of $B$ is defined as follows:
% The label of $B$ is defined as follows:
 %If $B$ is  \avoiding, %that is a subset of a non child-avoiding ball $B'\in\cZ_\tau$ satisfies 
 %then $t(B)=t(B')$, $t_d(B)=t$ and $\Pr(B)=\Pr(B')$; otherwise,  $t(B)=t$ and  $\Pr(B)$ consists of nodes in $\Pr(B')$ which are descendants of $t$. 
\iffalse
\begin{center}
\begin{tabular}{|l|l|l|l|}
\cline{2-4}
\multicolumn{1}{c|}{}& $t(B)$ & $t_d(B)$ & $\Pr(B)$  \\
\hline
$B$ is nonavoiding & $t$ & NA & $\Pr(B')\cap \{\text{descendants of } t\}$\\
\hline
$B$ is avoiding & $t(B')$ & $t$ & $\Pr(B')$\\
\hline
\end{tabular}
\end{center}
\fi
\label{assprop:tau}

\item At time $\tau_\l$ we process the border balls of all bags of $\fbag_\l$. 
\end{enumerate}
}}
\caption{Properties of our Inductive Construction}
\label{fig:assignedballsorder}
\end{figure}
\footnotetext{Note that the root has depth $0$.}

In the following lemma we show that when we are processing $\bag_t$ (at time $\tau$) any ball in this bag 
is in the interior of at most one hollowed ball of $\cZ_{\tau,s}$.
\begin{lemma}
\label{lem:uniqueinterior}
Say we process $\bag_t\in\fbag_\l$ at time $\tau$. 
For any $s\geq 0$, and any ball $B\in\bag_t$
and $B'\in\cZ_{\tau,s}$, if $\cC(B)\cap \cC(B')\neq\emptyset$, then for any ball $B''\neq B'$ in $\cZ_{\tau,s}$
that intersects $B'$, $\cC(B)\cap \cC(B'')=\emptyset$.

Consequently, if $B$ is in the interior of $B'\in \cZ_{\tau,s}$, then $\cC(B)\cap \cC(B'')=\emptyset$ for any $B''\in \cZ_{\tau,s}$ that intersects $B'$ and $B''\neq B'$. So, $B$ is in the interior of at most one ball of $\cZ_{\tau,s}$.
\end{lemma}
\begin{proof}
Let $t'=t(B')$; fix a ball $B''\in\cZ_{\tau,s}$ and let $t''=t(B'')$.
Assume, for the sake of contradiciton, that $\cC(B)\cap \cC(B'')\neq\emptyset$. 
First, by \autoref{fact:maptBtotau},
$t',t''$ are weak ancestors of  $t$. %Therefore, $t',t''$ are ancestor-descendant. 
Since $t'$ is a weak ancestor of $t$ and $\cC(B)\cap \cC(B')\neq\emptyset$, we have $t\in\cC(B')$. Similarly, $t\in\cC(B'')$. Therefore, $\cC(B')\cap\cC(B'')\neq\emptyset$ which contradicts the validity of the labeling since $B',B''$ intersect.
%Since the center of $B$ is far from the boundary of $B'$ and $B$ intersects $B''$, $B'$ also intersects $B''$.
%Since $B',B'' \in \cZ_{\tau,s}$, $\cC(B')\cap \cC(B'')=\emptyset$. 
%If $t'=t''$, then we reach a contradiction. So, without loss of generality assume that $t'$ is an ancestor of $t''$ (see \autoref{fig:uniqinterior}).
%By property \ref{labelprop:validlabeling} of \autoref{fig:labelassignedballs}, $\Pr(B')$ has a weak ancestor of $t''$, or $B'$ is \avoiding~and $t_d(B')$ is a weak ancestor of $t''$. 
%But, either of these implies that $\{B,B'\}$'s labeling  valid, which is a contradiction. 
\end{proof}
%\begin{figure}
%\centering
%\begin{tikzpicture}
%\tikzstyle{every node} = [draw,circle,minimum size=5mm,inner sep=0];
%	\node at (2,2) (a) {$t'$};
%	\node at (1,1) (b) {$t''$} edge (a);
%	\node at (0,0) (c) {$t$} edge (b);
%	\node at (3,1) [draw=none] (d) {$\vdots$} edge (a);
%	\node at (2,0) (e) [draw=none] {$\vdots$} edge (b);
%\end{tikzpicture}	
%\caption{The setting of \autoref{lem:uniqueinterior}.}
%\label{fig:uniqinterior}
%\end{figure}

Next, we show that once a ball of  $\bag_t$ becomes a border ball, it remains a border ball till the end of phase $\l$.
In the proof we use the following simple fact.
\begin{fact}
\label{fact:shrinkinglabels}
Suppose $B,B'$ have a valid labeling; for any ball $B''\subseteq B$ such that $\cC(B'')\subseteq \cC(B)$, $\{B',B''\}$'s labeling is valid as well.
\end{fact}

\begin{lemma}
\label{lem:invariantborder}
Suppose we are processing $\bag_t\in\fbag_\l$ at time $\tau$.
For any ball  $B\in\bag_t$ if $B$ is an interior ball with respect to (some hollowed ball of) $\cZ_{\tau',s}$ for some $ \tau \leq \tau' < \tau_\l$ and $s>0$, then, it is also in the interior of a ball of $\cZ_{\tau',s-1}$. 
\end{lemma}
This lets us backtrack through $\cZ_{\tau',s}$'s until we reach $\cZ_{\tau,0}$. So, if $B$ is a border ball at the time we start processing  $\bag_t$ it remains a border ball until time $\tau_\l$. 
\begin{proof}
If $B$ is in the interior of a newly inserted ball $B'\in\cZ_{\tau',s}$, by property \ref{assprop:tau}, the \conflictset~of $B'$ is a subset of a \conflictset~of a ball $B''\in\cZ_{\tau'-1,s}$ containing $B'$.  
So, by \autoref{fact:shrinkinglabels}, 
$B$ is also in the interior of $B''$.
\end{proof}
%\begin{figure}
%\centering
%\begin{tikzpicture}
%\tikzstyle{every node} = [draw,circle,minimum size=5mm,inner sep=0];
%	\node at (2,2) (a) {$t''$};
%	\node at (1,1) (b) {$t$} edge (a);
%	\node at (0,0) (c) {$t'$} edge (b);
%	\node at (3,1) [draw=none] (d) {$\vdots$} edge (a);
%	\node at (2,0) (e) [draw=none] {$\vdots$} edge (b);
%\begin{scope}[shift={(5,0)}]
%	\node at (2,2) (a) {$t$};
%	\node at (1,1) (b) {$t''$} edge (a);
%	\node at (0,0) (c) {$t'$} edge (b);
%	\node at (3,1) [draw=none] (d) {$\vdots$} edge (a);
%	\node at (2,0) (e) [draw=none] {$\vdots$} edge (b);	
%\end{scope}
%
%\end{tikzpicture}	
%\caption{The setting of \autoref{lem:invariantborder}. }
%\label{fig:invariantborder}
%\end{figure}

\subsubsection{The Construction} 
At any time $\tau_{\l-1} < \tau \leq \tau_\l$ and $s\geq 0$, we allocate $\token_{\tau,s}(B)$ tokens to any hollowed ball $B=B(x,r_1\|r_2)\in\cZ_{\tau,s}$, where
$$
\token_{\tau,s}(B) = \begin{cases}
\rho_\l - C_4\cdot \rho_{\l+1} & \text{ if $B\in\fbag_\l$}\\
[r_2 - r_1 - C_4 \cdot \rho_\l]^+ & \text{ if $B\notin\fbag_\l$ is non\avoiding}\\
[\frac{r_2-r_1 - C_4\cdot \rho_\l}{2(2+C_4)}]^+ & \text{ otherwise.}
\end{cases}
$$
Note that we allocate significantly smaller number of tokens to the \avoiding~hollowed balls; roughly speaking we allocate $1/2C_4$ fraction of what we allocate for a same-sized nonavoiding ball.

%Our induction hypothesis is as follows:
%{\em At time $\tau_{\l-1}+1\leq \tau\leq \tau_\l$, if we allocate $\token_\l(B)$ tokens to any ball $B(x,r_1\|r_2)\in\cZ_\tau$, then we can distribute these tokens among the bags of balls that we have processed by time $\tau$ such that any $t\in T_{i}$ for $i<\l$ receives at least $\frac{C_4}{6C_3}\cdot b_t \cdot \rho_{i}$ tokens and any $t\in T_\l$, that is processed by time $\tau$, receives at least $\frac{C_4}{6C_3}\cdot (b_t-|\border_t|)\cdot \rho_\l$ tokens.}
%The loss function is also similar to the previous section with one modification. Here, we shrink the balls of $\cZ_{\tau_\l}\setminus \fbag_{\tau_\l}$
%with probability $\pshr$ where $0<\pshr<1$ is a constant that we fix later. 
%We re-write the loss function correspondingly.
%\begin{equation}
%\label{eq:lossassigned}
%\loss(\l):=\pshr \cdot \sum_{i\geq \l} 4\rho_\l \leq 6\pshr\cdot \rho_\l. 
%\end{equation}

Say we are processing $\bag_t$ at time $\tau_{\l-1}+1 \leq \tau < \tau_\l$. 
We process $\bag_t$ in several steps; we start with $\cZ=\cZ_\tau$ and in each iteration of the loop we may add/remove several (hollowed) balls to/from $\cZ$. We use $\cZ_{\tau,s}$ to denote
the set $\cZ$ after the $s$-th iteration of the loop, so, $\cZ=\cZ_{\tau,0}=\cZ_{\tau}$ before entering the loop and $\cZ=\cZ_{\tau,\infty}=\cZ_{\tau+1}$ after the loop. 
Before processing $\bag_t$, we let $\border_t$ be the set of border balls of $\bag_t$ with respect to $\cZ_{\tau,0}$ and $\interior_t$ be the set of interior balls. We update these sets in each iteration of the loop. We use $\border_{t,s},\interior_{t,s}$ to denote the sets $\border_t,\interior_t$ after the $s$-th iteration of the loop, respectively. In addition, we use
$\border_{t,\infty}, \interior_{t,\infty}$ to denote these sets after the execution of the loop. We will process the balls in $\border_{t,\infty}$   at the end of phase $\l$. 
The details of our construction are described in \autoref{alg:ballproc}.

\begin{algorithm}[phtb!]
\begin{algorithmic}[1]
\Input $\cZ_{\tau}$ and $\bag_t\in\fbag_\l$.
\Output $\cZ_{\tau+1}$
\State Let $\cZ=\cZ_\tau$, $t^*$ be parent of $t$ and $\border_t,\interior_t$ be the border balls and interior balls of $\bag_t$ respectively.
Also, let $\cut'(t)=\{ \{u,v\}\in \cut(t): \norm{X_u-X_v}_1 < \rho_\l\}$. 
\While {$|\interior_t|\geq |\bag_t|/2$}
\If {$\exists B'\in\interior_t$ s.t. $B'$ is  in the interior of an {\em\avoiding}~hollowed ball $B\in\cZ$,}
\State Suppose $B'=B(X_u,\rho_\l)$ and $B=B(x,r_1\|r_2)$.
\State \label{step:intchildavoiding} \parbox[t]{\dimexpr\linewidth-\algorithmicindent}{
{\bf Update $\cZ$:} Remove $B$ and add $B_1=B(x,r_1\|\norm{X_u-x}_1-\rho_\l)$ and $B_2=B(x,\norm{X_u-x}_1+\rho_\l\|r_2)$ with the same labels as $B$. Add $B'$ (to $\cZ$) and remove it from $\interior_t$.
{\bf Goto} step \ref{step:updateintt}.
\strut}
\Else
\State \parbox[t]{\dimexpr\linewidth-\algorithmicindent}{
Let $S_1,\ldots,S_j$ be a natural decomposition of 
$G[V(t^*)\setminus V(t)]$ into $k/4$-edge-connected subgraphs
as defined in \autoref{def:decompkconnected}. \Comment{In \autoref{lem:numcomps} we will show that $j\leq 2|\cut(t)|/k$.}
\strut}
\State \label{step:constUiVi}\parbox[t]{\dimexpr\linewidth-\algorithmicindent}{
Let $U\subseteq V(t)$ be the centers of balls of $\interior_t$, 
\begin{eqnarray*}
V_i&:=&\{v\in S_i: \exists u\in U, \{u,v\}\in\cut'(t)\},\\
U_i&:=&\{u\in U: \exists v\in S_i,\{u,v\}\in\cut'(t)\}
\end{eqnarray*}
%be the vertices of $S_i$ that are incident to an edge of $\cut'(t)$ and
\Comment{By \autoref{def:assignedbagballs},  every vertex of $U$ is incident to an edge of $\cut'(t)$, so $\cup_{i=1}^j U_i=U$. Also, since $\bag_t$ is a $12 C_3/k$-assigned bag, $|U|=|\interior_t|\geq |\bag_t|/2 \geq \frac{6C_3 |\cut(t)|}{k}$.}
\strut}
\State Let $i=\argmax_{1\leq i\leq j}|U_i|$. \Comment{So, $|U_i| \geq |U|/j \geq 3 C_3$.} \label{step:Uilowerbound}
\State \parbox[t]{\dimexpr\linewidth-\algorithmicindent}{Let $B=B(x,r_1\|r_2)\in \cZ$ be a non\avoiding ~ball such that a ball of $\interior_t$ with its center in $U_i$ is in the interior of $B$. 
\Comment{We will show that $t(B)$ is an ancestor of $t$.}
\strut}
\State 
We define $r'_1 = \max\{r_1, \min_{v\in V_i} \norm{x-X_v}_1\}$ 
and $r'_2:=\min\{r_2,  \max_{v\in V_i} \norm{x-X_v}_1\}.$
% In \autoref{} we show that $r'_2-r'_1\geq \Omega(C_3 \rho_\l).$ 
\State \parbox[t]{\dimexpr\linewidth-\algorithmicindent}{
Let $\interior_{B'}$ be the balls of $\interior_t$ whose centers are in the hollowed ball $B'=B(x,r'_1-\rho_\l\|r'_2+\rho_\l)$
%\{u\in U: r'_1-\rho_\l\leq \norm{X_U-x}_1\leq r'_2+\rho_\l\}$
and $U_{B'}$ be the centers of balls of $\interior_{B'}$. \Comment{We may have $U_i\not\subseteq U_{B'}$ as some vertices of $U_i$ may not even be in $B$, but all vertices of  $U_{B'}$ are in $B$.}
\strut}
\If {$|\interior_{B'}|\cdot \rho_\l > 3 (r'_2-r'_1) $}
\Comment{We treat $\interior_{B'}$ as if it was a 3-compact bag of balls.}
\State \label{step:compactint} \parbox[t]{\dimexpr\linewidth-\algorithmicindent}{
{\bf Update $\cZ$:} Remove $B$ and add $B_1=B(x,r_1\|r'_1-2\rho_\l)$ and $B_2=B(x,r'_2+2\rho_\l\|r_2)$ with the same labels as $B$.
Add all balls of $\interior_{B'}$ to $\cZ$ and remove them from $\interior_t$.
\strut}
\Else %\Comment{In \autoref{} we show $r'_2-r'_1\geq \Omega(C_3\rho_\l)$.}
\State \label{step:hardint} \parbox[t]{\dimexpr\linewidth-\algorithmicindent}{
{\bf Update $\cZ$:} Remove $B$ and add $B_1=B(x,r_1\|r'_1)$ and
$B_2=B(x,r'_2\|r_2)$ to $\cZ$, with the same labels as $B$.
Add a new (non\avoiding) hollowed ball $B_3=B(x,r'_1+\rho_\l\|r'_2-\rho_\l)$
with $t(B_3)=t$ and $\Pr(B_3)$ consisting of nodes $t'\in \Pr(B)$ such that $t'$ is a descendant of $t$. 
Add an \avoiding~hollowed ball $B_4=B(x,r'_1\|r'_2)$
with $t(B_4)=t(B)$, $t_d(B_4)=t$ and $\Pr(B_4)=\Pr(B)$. 
Remove all balls of $\interior_{B'}$ from $\interior_t$. See \autoref{fig:Case3} for an example.
\Comment{Note that no ball of $\interior_t\setminus\interior_{B'}$ is in the interior of $B_1$ or $B_2$.}
\strut}
\EndIf
\EndIf
\State Move all balls of $\interior_t$ that become  border balls w.r.t. $\cZ$ into $\border_t$. \label{step:updateintt}
\EndWhile
\Return $\cZ$.
\end{algorithmic}
\caption{Construction of $\cZ_{\tau+1}$ by processing $\bag_t$.}
\label{alg:ballproc}
\end{algorithm}

The following is the main result of this part.
\begin{lemma}
\label{lem:inductionassigned}
For any $\tau,s\geq 0$ the following holds. 
The set $\cZ_{\tau,s}$'s labeling is valid. 
%Say we process $\bag_t\in\fbag_\l$ at time $\tau$, and $\border_{t,s},\interior_{t,s}$ are disjoint subsets of $\bag_t$ such that each ball of $\border_{t,s}$ is a border ball and each ball of $\interior_{t,s}$ is an interior ball.
If we allocate $\token_{\tau,s}(B)$ tokens to any hollowed ball $B(x,r_1\|r_2)\in\cZ_{\tau,s}$, then we can distribute these tokens among nodes whose bags we have processed by time $\tau$ such that for any $i<\l$, any $t'\in T_{i}$  receives at least $\frac{C_4}{12C_3}\cdot |\bag_{t'}| \cdot \rho_{i}$ tokens, and any $t'\in T_\l$ that is processed by time $\tau$ receives at least 
$$\frac{C_4}{6C_3}\cdot (|\bag_{t'}|-|\border_{t',\infty}|-|\interior_{t',\infty}|)\cdot \rho_\l$$ 
tokens, and the node $t$ that we are processing at time $\tau$ receives at least 
$$\frac{C_4}{6C_3}\cdot (|\bag_t|-|\border_{t,s}|-|\interior_{t,s}|)\cdot\rho_\l$$ tokens.
%If $|\interior_{\tau,s}\geq |b_t|/2$, then we can construct $\cZ_{\tau,s+1}$, and $\border_{\tau,s+1},\interior_{\tau,s+1}$ such that $\interior_{\tau,s+1}\subseteq\interior_{\tau,s}$ and $\border_{\tau,s+1}\cup \interior_{\tau,s+1}\subseteq \border_{\tau,s}\cup \interior_{\tau,s}$. Furthermore, 
\end{lemma}
Later, in the post processing phase we show that
any node $t$ receives 
at least $\frac{C_4}{6C_3}|\border_{t,\infty}|\cdot \rho_\l$ new tokens.
This implies \autoref{prop:assignedbagreduced} as by the stopping condition of the main loop of \autoref{alg:ballproc}, for any $t\in T_\l$, $|\interior_{t,\infty}| < |\bag_t|/2$.

We prove the above lemma by an induction on $\tau,s$.
From now on, we assume that all conclusions of the lemma hold for $\tau,s$ and we prove the same  holds for $\tau,s+1$. 
%We show that each of the sets $\cZ_{\tau,s}$ has a valid labeling. Furthermore, we show that after $s$ iterations, $\bag_t$ receives at least $\frac{C_4}{6C_3}\cdot (b_t-|\border_{t,s}|-|\interior_{t,s}|)\cdot\rho_\l$ tokens.
We construct  $\cZ_{\tau,s+1}$ (from $\cZ_{\tau,s}$) in one of the three steps of the loop, i.e., steps \ref{step:intchildavoiding}, \ref{step:compactint}, \ref{step:hardint}. 
We analyze these steps in the following three cases. %From now on, we fix $s$. We assume that $\cZ_{\tau,s}$ has a valid labeling and that each 

\paragraph{Case 1:} A ball  $B'\in \interior_{t,s}$ is in the interior of an \avoiding~hollowed ball $B=B(x,r_1\|r_2)\in\cZ_{\tau,s}$.\\
In this case by \autoref{lem:uniqueinterior}, for any  ball $B''\in\cZ_{\tau,s}$ such that $B\neq B''$, $\{B',B''\}$'s labeling is valid. Since, by definition, $B'$  intersects neither of $B_1,B_2$, $\cZ_{\tau,s+1}$'s labeling is valid. 
We send all tokens of $B_1$ and $B_2$ and $\rho_\l/2$ of the tokens of $B'$ to $B$ and we redistribute them by the induction hypothesis. We send the rest of the tokens of $B'$ to $t$. 
Then, $B$ receives,
$$ \token_{\tau,s+1}(B_1)+\token_{\tau,s+1}(B_2)+\frac{\rho_\l}{2} \geq \frac{(r_1-r_2-2\rho_\l)-2C_4\cdot\rho_\l+\rho_\l(2+C_4)}{2(2+C_4)}
= \token_{\tau,s}(B).$$
In the above equation, we crucially use that, roughly speaking, $\token_{\tau,s}(B)$ is a only a constant fraction of the width of $B$ when $B$ is an avoiding ball. This is not the case when we deal with nonavoiding balls in cases 2,3.

On the other hand, $t$ receives
$$ \token_{\tau,s+1}(B')-\rho_\l/2 \geq \rho_\l - C_4\cdot \rho_{\l+1}-\rho_\l/2\geq \rho_\l/4.$$
new tokens, where we used $\rho_{\l+1}\leq \lambda\cdot  \rho_\l$ and $\lambda\leq 1/4C_4$. Since $|\border_{t,s+1}| + |\interior_{t,s+1}| = |\border_{t,s}|+|\interior_{t,s}|-1$
we are done by induction.
\medskip

Now suppose that the above does not happen.
Consider the induced graph $G[V(t^*)-V(t)]$. Note that this graph may be disconnected. Let $S_1,S_2,\ldots,S_j$ be a natural decomposition of this graph as defined in \autoref{def:decompkconnected}. In the following lemma we show that $j\leq 2|\cut(t)|/k$.
\begin{lemma}
\label{lem:numcomps}
%For any node $t\in\cT(k,.,.)$ with parent $t^*$, let $S_1,\ldots,S_j$ be a natural decomposition of $G[V(t^*)\setminus V(t)]$ into $k/2$-edge-connected subgraphs as defined in \autoref{def:decompkconnected}. Then,
$ j\leq \frac{2|\cut(t)|}{k}.$
\end{lemma}
\begin{proof}
By the definition of $\cT$, $G(t^*)$ is $k$-edge-connected. 
Therefore, for any $1\leq i\leq j$,
$$ \deg_{G(t^*)}(S_i) \geq k.$$
Therefore,
$$ j\cdot k \leq \sum_{i=1}^j \deg_{G(t^*)}(S_i) = \deg_{G(t^*)}(V(t)) + \sum_{i=1}^j \deg_{G[V(t^*) \setminus V(t)]} (S_i) = |\cut(t)| +\sum_{i=1}^j \deg_{G[V(t^*) \setminus V(t)]} (S_i).$$
But, by \autoref{lem:decompkconnected}, the second term on the RHS is at most $2(j-1)(k/4-1)$. Therefore,
$j\leq 2|\cut(t)|/k$. 
\end{proof}
As we mentioned in the comments of the algorithm, by the assumption that $\bag_t$ is  $12C_3/k$-assigned, the above lemma implies that
\begin{equation}	
\label{eq:UigeqC3}
|U_i| \geq 3C_3.
\end{equation}
%In the remaining two cases of the algorithm
%we have to be aware that $U_i$ may not be a subset of $U_{B'}$.
Next, we prove a technical lemma which will be used in both of cases 2 and 3. In case 2 we use this lemma together with the above inequality to show that $|\interior_{B'}|\geq 3(C_3-1)$; we will use this in our charging argument to compensate for the tokens lost by splitting $B$.
In case 3, we use the following lemma to show that $r'_2 - r'_1 \geq (C_3-1)\cdot \br_\l$. Similarly, we use this inequality to compensate for the tokens lost by splitting $B$.
\begin{lemma}
\label{lem:rp2-rp1}
Let $U,U_i,V_i$ be defined as in step \ref{step:constUiVi}.
If $U_i\not\subseteq U_{B'}$, then $r'_2-r'_1 \geq (C_3-1)\cdot\rho_\l$.
\end{lemma}
\begin{proof}
First, we show that there is a vertex $v\in V_i$
such that $X_v\notin B$. 
For the sake of contradiction assume $V_i\subset B$.
We show that any vertex $u\in U_i$ is in $U_{B'}$ which is a contradiction. Fix a vertex $u\in U_i$. By \autoref{def:assignedbagballs}, there is a vertex $v\in V_i$ such that  $\{u,v\}\in \cut'(t)$.
Since $X_v\in B$, by the definition of $r'_1,r'_2$, we have
$r'_1 \leq  \norm{X_v-x}_1  \leq r'_2$.
So, $X_u\in B(x,r'_1-\rho_\l\|r'_2+\rho_\l)$, i.e., $u\in U_{B'}$. This is a contradiction.

Now, let $v\in V_i$
be such that  either $\norm{X_v-x}_1 \geq r_2$ or $\norm{X_v-x}_1\leq r_1$. Here, we assume the former; the other case can be analyzed similarly. Then, we have $r'_2=r_2$. 
But by definition of $B$, there is a ball $B(X_u,\rho_\l)\in\interior_{t,s}$ in the interior of $B$ such that $u\in U_i$. Since $u\in U_i$, there is a vertex $w\in V_i$ such that $\norm{X_u-X_w}_1 < \rho_\l$.
Therefore,
$$ r'_1\leq \norm{x-X_w}_1 \leq  \norm{x-X_u}_1+\rho_\l \leq r_2-C_3\rho_\l+\rho_\l. $$
where the last inequality uses that $B(X_u,\br_\l)$ is in the interior of $B$.
So, $r'_2-r'_1 \geq (C_3-1)\rho_\l$.
\end{proof}

\paragraph{Case 2:} $|\interior_{B'}|\cdot\rho_\l > 3(r'_2-r'_1)$.\\
First, we show $\cZ_{t,s+1}$'s labeling is valid.
Then, we distribute the tokens. 
To show that $\cZ_{t,s+1}$'s labeling is valid, first
we argue that all balls of $\interior_{B'}$ are in the interior of $B$.
Fix a ball $A\in \interior_{B'}$, we show $A$ is in the interior of $B$.
First, $\{A,B\}$'s labeling is {\em in}valid. Because i) $A,B$ intersect by the definition of $\interior_{B'}$ and ii) a ball of $\bag_t$ is in the interior of $B$ and all balls of $\bag_t$ have the same labels. 
Secondly, since $\interior_{B'}\subseteq \interior_{t,s}$, $A$ is an interior ball. Therefore, by \autoref{lem:uniqueinterior}, $A$ is in the interior of $B$.
%$\{B,B'\}$ does not have a valid labeling. Therefore, $B$ does not have a valid labeling with each ball of $\bag_t$. But, by definition of $\interior_B$, each ball of $\interior_B$ intersects $B$.  
%by \autoref{lem:uniqueinterior}, all of them must be  in the interior of $B$. 
Now, by \autoref{lem:uniqueinterior},  for any $B''\in\cZ_{\tau,s}$ where $B''\neq B$, $\{A,B''\}$'s labeling is valid. Furthermore, by construction,  $B_1,B_2$ do not intersect any balls of $\interior_{B'}$. Hence, $\cZ_{t,s+1}$'s labeling is valid. 

Next, we describe the distribution of tokens allocated to the balls of $\cZ_{\tau,s+1}$. Before that,  we show that $|\interior_{B'}|\geq 3 (C_3-1)$.
We consider two cases. If $U_i\subseteq U_{B'}$. Then, by \eqref{eq:UigeqC3}, 
$$|\interior_{B'}|=|U_{B'}|\geq |U_i|\geq 3 C_3.$$
%The last inequality holds by the definition of $i$ and \autoref{lem:numcomps} (see the comment of step \ref{step:Uilowerbound} of the algorithm). 
Otherwise, $U_i\not\subseteq U_{B'}$. Then, by  \autoref{lem:rp2-rp1},
$$ |\interior_{B'}| \geq \frac{3(r'_2-r'_1)}{\rho_\l} \geq \frac{3(C_3-1)\cdot \rho_\l}{\rho_\l} = 3(C_3-1).$$
Therefore, $|\interior_{B'}|\geq 3(C_3-1)$.

Now, we send all tokens of $B_1,B_2$ and $3/4$ of the tokens of each ball of $\interior_{B'}$ to $B$ and we redistribute them by the induction hypothesis. $B$ receives,
\begin{align*}
\token_{\tau,s+1}(B_1)+\token_{\tau,s+1}(B_2) &+ \frac34 |\interior_{B'}| (\rho_\l - C_4\cdot \rho_{\l+1}) \\
&\geq r_2-r_1 - 4\rho_\l - (r'_2-r'_1) - 2C_4\cdot \rho_\l + \frac34\cdot|\interior_{B'}|\cdot \frac56\rho_\l\\
&\geq \token_{\tau,s}(B) - (4+C_4)\cdot \rho_\l + \frac7{24}|\interior_{B'}|\cdot\rho_\l \\
&\geq  \token_{\tau,s}(B) -(4+C_4)\cdot \rho_\l + \frac78(C_3-1)\cdot \rho_\l \\
&\geq \token_{\tau,s}(B).
\end{align*}
where the first inequality uses $\rho_{\l+1}<\lambda\cdot \rho_\l$ and $\lambda < 1/6C_4$, the second inequality uses the assumption $3(r'_2-r'_1) < |\interior_{B'}|\cdot\rho_\l$, the third inequality uses $|\interior_{B'}| \geq 3(C_3-1)$ and the last inequality uses $C_3\geq 8(C_4+5)/7$. 
On the other hand, each ball $B'\in\interior_{B'}$ sends
$$ \frac14 \token_\tau(B') \geq \frac14 \cdot \frac56 \rho_\l$$
to $t$. So, $t$ receives $|\interior_{B'}|\cdot \rho_\l/5$ new tokens. Since 
$$|\border_{t,s+1}|+|\interior_{t,s+1}| = |\border_{t,s}|+|\interior_{t,s}|-|\interior_{B'}|,$$ 
and we are done by induction.

%It remains to analyze the last case.
\begin{figure}
\centering
\begin{tikzpicture}
\def\slen{0.35}
\def\hslen{0.3}
\def\opacity{0.25}
\begin{scope}%[rotate=45]
\draw [fill=black,opacity=\opacity] (0,0) node[below left=2.2cm,draw=none,color=black,opacity=1] {$B$} 
	%node[below left=2.2cm,draw=none,color=black,opacity=1] {$r_2$}
circle (3.2);
\draw[fill=white,draw=none] (0,0) node[below left=0cm,draw=none,color=black] {$r_1$} circle (.6);
\draw[dashed,line width=1.3pt] (0,0) -- (-.6,0) (0,0)  -- (0,-3.2);
\node at (.3,-1.5) [draw=none] {$r_2$};
\draw [fill=black] (0,0) circle (2pt);
\foreach \i/\bx/\by/\rx/\ry/\col in {1/-2/0/.15/.25/red, 2/-1.3/-1/-.25/.2/green, 3/1.3/0.4/-.2/-.25/red, 4/2.3/-.3/.2/-.2/red}{
	\draw  (\bx,\by) %node[below=.42cm,draw=none,color=black]  
	circle (\slen);
	\draw [line width=1.2pt] (\bx+\rx,\by+\ry) -- (\bx,\by);
	\draw [fill=blue] (\bx,\by) circle (2pt);
	\draw [fill=\col,] (\bx+\rx,\by+\ry) circle (2pt);
	
	\ifthenelse {\equal{\i}{3}}{\draw [dashed,line width=1.3pt] (\bx+\rx,\by+\ry) node [above left,draw=none] {$r'_1$} -- (0,0);}{}; 
	\ifthenelse {\equal{\i}{4}}{\draw [dashed,line width=1.3pt] (\bx+\rx,\by+\ry) node [below left,draw=none] {$r'_2$} -- (0,0);}{}; 

}
\end{scope}
{\Huge \draw[double,-implies,double distance between line centers=4.5pt] (3.3,0) -- (4.8,0);}

\begin{scope}[shift={(8,0)}]
\draw [fill=black,opacity=\opacity] (0,0) node[below left=1.75cm,draw=none,color=black,opacity=1] {$B_2$} circle (3.2);
\draw[fill=white,draw=none] (0,0) circle (2.55);
\draw[fill=blue,opacity=\opacity,draw=none] (0,0) node [below=1.5cm,color=blue,opacity=1] {$B_3$} circle (2.25);
\draw[fill=white,draw=none] (0,0) circle (1.42);
\draw[fill=red,opacity=\opacity,draw=none] (0,0) node [below=2.1cm,color=red,opacity=1] {$B_4$} circle (2.55);
\draw[fill=white,draw=none] circle (1.12);
\draw[fill=black,opacity=\opacity] (0,0) node [below left=.25cm,draw=none,color=black,opacity=1] {$B_1$} circle (1.12);

\draw[fill=white,draw=none] (0,0) circle (.6);
\draw [fill=black] (0,0) circle (2pt);
\foreach \i/\bx/\by/\rx/\ry/\col/\u/\v in {1/-2/0/.15/.25/red//, 2/-1.3/-1/-.25/.2/green//, 3/1.3/0.4/-.2/-.25/red/u_1/v_1, 4/2.3/-.3/.2/-.2/red/u_2/v_2}{
	\draw [line width=1.2pt] (\bx+\rx,\by+\ry)node [below=1mm,color=black] {$\v$} -- (\bx,\by);
	\draw [fill=blue] (\bx,\by) node [above=1mm,color=black] {$\u$} circle (2pt);
	\draw [fill=\col,] (\bx+\rx,\by+\ry) circle (2pt);
}
\end{scope}
\end{tikzpicture}
\caption{An illustration of Case 3. $U_{B'}$ is the blue vertices.  $V_i$ is the set of red vertices. The green vertex belongs to $V_{i'}$ for $i'\neq i$. The edges between the blue vertices and red/green vertices are in $\cut'(t)$. We update $\cZ_{\tau,s}$ as follows: We split $B$ to balls $B_1,B_2$. We also add an avoiding  $B_4$ from the  closest red point ($r'_1$) to the farthest one ($r'_2$), and  a nonavoiding ball, $B_3=B(.,r'_1+\br_\l\|r'_2-\br_\l)$. 
%Note that there are blue vertices $u_1,u_2$ such that $
}
\label{fig:Case3}
\end{figure}
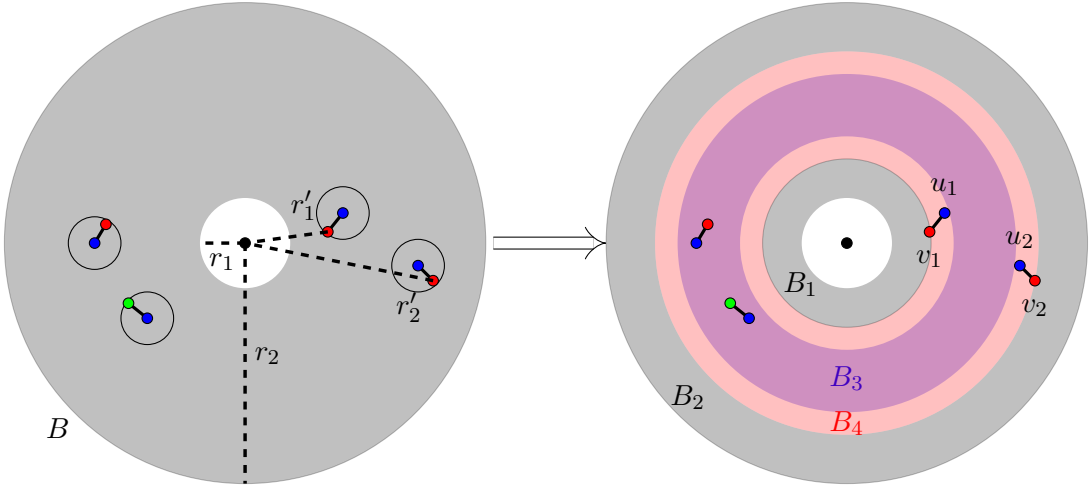

\paragraph{Case 3: } $|\interior_B|\cdot\rho_\l \leq 3(r'_2-r'_1)$.\\
As usual, first we verify the validity of the labeling, then we show that the tokens assigned to $B_3,B_4$ compensate the loss of $B$ and the balls of $\interior_{B'}$ that we delete.
We emphasize that verifying the validity of labeling is more involved in this case compared to cases 1, 2; this is because case  3 is the only one in which we insert new balls, i.e., $B_3,B_4$, that do not exist in the given geometric sequence of bags of balls. 

First, we show that property \ref{assprop:tau} of \autoref{fig:assignedballsorder} is satisfied; then we verify properties \ref{labelprop:validlabeling}, \ref{labelprop:insertability}, \ref{labelprop:disjointpaths} of \autoref{fig:labelassignedballs} in that order. 
Recall that the labels of $B_3$ and $B_4$ are defined as follows:
\begin{center}
\begin{tabular}{|l|l|l|l|l|}
\cline{2-5}
\multicolumn{1}{c|}{}& $t(.)$ & $t_d(.)$ & $\Pr(.)$  & $\cC(.)$ \\
\hline
$B_3$ & $t$ & NA & $\Pr(B)\cap \{\text{descendants of } t\}$ & $\cC(B)\cap \text{ subtree rooted at } t$\\
\hline
$B_4$ & $t(B)$ & $t$ & $\Pr(B)$ & $\cC(B)\setminus \text{ subtree rooted at } t$\\
\hline
\end{tabular}
\end{center}
Note that by  \autoref{fact:maptBtotau} and that a ball of $\bag_t$ is in the interior of $B$, $t(B)$ is a weak ancestor of $t$. 
Therefore, $\cC(B_3),\cC(B_4)\subseteq \cC(B)$ as required by property \ref{assprop:tau} of \autoref{fig:assignedballsorder}.
Let us now verify that $t_d(B_4)=t$ is a proper descendent of $t(B_4)=t(B)$, i.e., $B_4$ is a valid avoiding ball.
Since we showed $t(B)$ is a weak ancestor of $t$, it is enough to show that $t(B)\neq t$. If $t(B)=t$, then 
%either $B\in \bag_t$ or 
$B$ is constructed in an iteration  $s'\leq s$  of the loop. 
%The former cannot happen because balls of $\bag_t$ do not intersect. The latter 
This does not happen because whenever we construct a new ball in step \ref{step:hardint}
we delete all balls of $\interior_t$ that intersect with the new ball; in addition, no new interior balls are added throughout the loop by \autoref{lem:invariantborder}. Therefore $t(B)\neq t$.

%This together with the fact that $B$ is nonavoiding imply that, 

Next, we verify property \ref{labelprop:validlabeling} of \autoref{fig:labelassignedballs}. Since $\cC(B_3),\cC(B_4)\subseteq \cC(B)$, by \autoref{fact:shrinkinglabels}, $B_3,B_4$ do not have a conflict with any ball of $\cZ_{\tau,s}\setminus\{B\}$, i.e., for any ball $B''\in \cZ_{\tau,s}\setminus \{B\}$ that intersects one of them,
$$ \cC(B_3)\cap \cC(B'')=\emptyset \text{ and } \cC(B_4)\cap \cC(B'')=\emptyset.$$
In addition, since $t_d(B_4)=t=t(B_3)$, $\cC(B_3)\cap \cC(B_4)=\emptyset$. Furthermore, $B_3$ and $B_4$ do not intersect $B_1,B_2$.
So the labelings satisfy property \ref{labelprop:validlabeling} of \autoref{fig:labelassignedballs}.

It remains to verify that $B_3,B_4$ satisfy properties \ref{labelprop:insertability} and \ref{labelprop:disjointpaths} of \autoref{fig:labelassignedballs}.
$B_3$ and $B_4$ satisfy property \ref{labelprop:insertability} because $\Pr(B_3),\Pr(B_4)\subseteq \Pr(B)$ and they are inside $B$. Finally, we need to verify property \ref{labelprop:disjointpaths}. 
First, we show $B_3$ satisfies property \ref{labelprop:disjointpaths}. By the definition of $U_i$ there are vertices $u_1,u_2\in U_i$ such that $\norm{x-X_{u_1}} < r'_1+\rho_\l$ and $\norm{x-X_{u_2}} > r'_2-\rho_\l$ (see \autoref{fig:Case3}). 
Since $G(t)$ is $k$-edge-connected there are $k$ edge-disjoint paths between $u_1$ and $u_2$ supported on $V(t)$.
So, we just need to argue that $u_1\in\cC(B_3)$, i.e., for any $t'\in\Pr(B_3)$, $u_1\notin V(t')$. This is because, $u_1\in U_i$ is incident to an edge $e$ of $\cut'(t)$.  Since  $t'$ is a descendant of $t$, if $u_1\in V(t')$ then $e\in \setdeg(t')$ so  an endpoint of an edge of $\setdeg(t')$ has distance less than $r_2$ from the center of $B$ which is contradictory with $t'\in \Pr(B_3)\subseteq \Pr(B)$.

Lastly, we show $B_4$ satisfies property \ref{labelprop:disjointpaths}. By the definition of $V_i$ there are vertices $v_1,v_2\in V_i$ such that
$\norm{x-X_{v_1}} \leq r'_1$ and $\norm{x-X_{v_2}} \geq r'_2$ (see \autoref{fig:Case3}). Since $V_i\subseteq S_i$ and $S_i$ is $k/4$-edge-connected in $G[V(t^*)\setminus V(t)]$, there are $k/4$ edge-disjoint paths from $v_1$ to $v_2$ in $G[V(t(B))\setminus V(t)]$. We  need to argue that $v_1\in\cC(B_4)$, i.e., it is enough to show that for any $t'\in \Pr(B_4)$, we have $v_1\notin V(t')$. This is similar to the argument in the previous paragraph.
First, since $v_1\in V_i$, $v_1$ is incident
to an edge $e\in \cut'(t)$. 
Since $t'\in \Pr(B)$ and $\norm{X_{v_1}-x}_1 \leq r_2$, we must have $e\notin \setdeg(t')$. Therefore, if $v_1\in V(t')$,  $t'$ must be a weak ancestor of $t^*$. But, since $t(B)$ is an ancestor of $t$ and a ball of $\bag_t$ is in the interior of $B$, we must have $t\in\cC(B)$, i.e., $\Pr(B)$ cannot not contain a weak ancestor of $t$.
So, $v_1\notin V(t')$.
\medskip 

It remains to distribute the tokens.
First, we show that $r'_2-r'_1 \geq (C_3-1)\cdot \rho_\l$.
If $U_i\not\subseteq U_{B'}$, then by \autoref{lem:rp2-rp1},
$r'_2-r'_1\geq (C_3-1)\cdot\rho_\l$.
Otherwise, by the assumption of Case 3,
%$\bag_t$ is  $12 C_3/k$-assigned,
$$ r'_2-r'_1\geq \frac13 |\interior_{B'}|\cdot\rho_\l \geq \frac13 |U_i| \cdot\rho_\l  \geq C_3\cdot \rho_\l,$$
where the last inequality follows by \eqref{eq:UigeqC3}.
We send all tokens of $B_1,B_2,B_3$, and $(2C_4+2)\rho_\l$ tokens of $B_4$ to $B$ and we redistribute them by the induction hypothesis. We send the rest of the tokens of $B_4$ to $t$.
Ball $B$ receives
$$\sum_{i=1}^3 \token_{\tau,s+1}(B_i) + (2C_4+2)\cdot \rho_\l \geq r_2-r_1 - 2\rho_\l - 3C_4\rho_\l + (2C_4+2)\rho_\l =\token_{\tau,s}(B).  $$
On the other hand, $t$ receives,
\begin{eqnarray*} 
\token_{\tau,s+1}(B_4) - (2C_4+2)\rho_\l 
&= & \frac{r'_2-r'_1 - C_4\cdot\rho_\l - 4(2+C_4)^2\cdot\rho_\l}{2(2+C_4)}\\
&\geq& \frac{r'_2-r'_1 - (C_3-1)\rho_\l/2}{2(2+C_4)}\\
& \geq &\frac{r'_2-r'_1}{4(2+C_4)} \\
&\geq& \frac{|\interior_B|\cdot\rho_\l}{12(2+C_4)} \geq \frac{C_4 |\interior_B|\cdot\rho_\l}{6C_3},
\end{eqnarray*}
new tokens. In the first inequality we used $(C_3-1) \geq 2(C_4+4(C_4+2)^2)$, the second inequality uses $r'_2-r'_1 \geq (C_3-1)\cdot \rho_\l$, and the third inequality uses 
the assumption $r'_2-r'_1\geq \frac13\cdot|\interior_B|\cdot\rho_\l$.
This concludes the proof of \autoref{lem:inductionassigned}.

\subsubsection{Post-processing}
Say we have processed all $\bag_t\in\fbag_\l$ and we are the end of phase $\l$, i.e., time $\tau_\l$.
%We may add  border balls. 
We need to make sure that each node $t\in T_\l$ receives at least $\frac{C_4}{6C_3}|\border_{t,\infty}|\cdot \rho_\l$ new tokens.
Then, by \autoref{lem:inductionassigned}, each node $t$, altogether,  receives at least
$$ \frac{C_4}{6C_3}(|\bag_t|-|\interior_{t,\infty}|)\cdot\rho_\l \geq \frac{C_4}{12C_3}|\bag_t|\cdot \rho_\l$$
tokens. The above inequality uses that by the  condition of the main loop of \autoref{alg:ballproc},
for any $t\in T_\l$, $|\interior_{t,\infty}|\leq |\bag_t|/2$.

\medskip
We define the shrink operator as follows:
For any hollowed ball $B=B(x,r_1\|r_2)\in\cZ_{\tau_\l}$,
\begin{equation}
\shrink_\l(B)=\begin{cases}
B & \text{if $B\in\fbag_\l$}\\
B(x,r_1+(C_3+1)\rho_\l\|r_2-(C_3+1)\rho_\l) & \text{if $B\notin\fbag_\l$ and $r_2-r_1>2(C_3+1)\rho_\l$}\\
B(x,0)=\emptyset & \text{otherwise}.
\end{cases}
\end{equation}

 Let
\begin{eqnarray*}
 b&:=&\sum_{t\in T_\l} |\border_{t,\infty}|,\\
\excess &:=&\sum_{B\in\cZ_{\tau_\l}} (\token_{\tau_\l+1}(B) -\token_{\tau_\l}(B)).
\end{eqnarray*}
Think of $\excess$ as the additional number of tokens that we gain for all hollowed balls $B\in\cZ_{\tau_\l}$ when we go to the new phase $\l+1$.
Our idea is simple. If $\excess$ is very large then we do not add any of the border balls and we just distribute $\excess$ between all nodes of $T_\l$. Otherwise, we shrink balls of $\cZ_{\tau_\l}$ and we add the border balls.

\paragraph{Case 1: }$ \excess\geq \frac{C_4}{6C_3}\cdot b\cdot\rho_\l$.\\
In this case,  we do not add any of the border balls and 
we simply let $\cZ_{\tau_\l+1}=\cZ_{\tau_\l}$.

Now, observe that for any hollowed ball $B\in \cZ_{\tau_\l}$, we have $\token_{\tau_\l+1}(B)-\token_{\tau_\l}(B)$ additional tokens that $B$ has not used. We distribute these tokens between the nodes of $T_\l$ proportional to their number of border balls.
More precisely, for any ball $B\in \cZ_{\tau_\l}$ and $t\in T_\l$, we send
$$\frac{|\border_{t,\infty}|}{b}\cdot (\token_{\tau_\l+1}(B)-\token_{\tau_\l}(B))$$ 
tokens to $t$.
Therefore, $t$ receives
\begin{eqnarray*} 
\sum_{B\in\cZ_{\tau_\l}} \frac{|\border_{t,\infty}|}{b}\cdot(\token_{\tau_\l+1}(B)-\token_{\tau_\l}(B)) 
&=& \frac{|\border_{t,\infty}|\cdot \excess}{b} \\
&\geq& \frac{C_4}{6C_3}\cdot |\border_{t,\infty}|\cdot \rho_\l,
\end{eqnarray*}
and we are done.

\paragraph{Case 2: } $\excess < \frac{C_4}{6C_3}\cdot b\cdot\rho_\l$.\\
For each hollowed ball $B\in\cZ_{\tau_\l}$
we replace $B$ by $\shrink_\l(B)$ in $\cZ_{\tau_\l+1}$.
We also add all balls of $\border_{t,\infty}$ for all $t\in T_\l$ to $\cZ_{\tau_\l+1}$. 
By \autoref{lem:invariantborder} any border ball $B\in\border_{t,\infty}$ is not in the interior of any ball of $\cZ_{\tau_\l}$. 
%Furthermore, since all balls of $\fbag_t$ are disjoint, any pair of border balls are nonintersecting.
By the definition of the shrink operator, and using the fact that balls of $\fbag_\l$ do not intersect,  
any ball of $\cup_{t\in T_\l} \border_{t,\infty}$ does not intersect any ball of $\cZ_{\tau_\l+1}$.
So, $\cZ_{\tau_\l+1}$'s labeling is valid. 

It remains to distribute the tokens. First, we prove a technical lemma. 
\begin{lemma}
\label{lem:changeoftoken}
If $\excess < \frac{C_4}{6C_3}\cdot b\cdot\rho_\l$, then
$$ b\cdot \rho_\l \geq 2\sum_{B\in\cZ_{\tau_\l}} (\token_{\tau_\l}(B) - \token_{\tau_\l+1}(\shrink_\l(B))).$$
\end{lemma}
\begin{proof}
It is sufficient to show that for any hollowed ball $B=B(x,r_1\|r_2)\in\cZ_{\tau_\l}$ %and $B'=B(x,r_1+C_3\rho_\l, r_2-C_3\rho_\l)$. Then,
\begin{equation}
\label{eq:changeoftoken}
\token_{\tau_\l+1}(B) - \token_{\tau_\l}(B)\geq \frac{C_4}{3 C_3} \cdot (\token_{\tau_\l}(B) - \token_{\tau_\l+1}(\shrink_\l(B))). \end{equation}
Because, then
\begin{eqnarray*}
\sum_{B\in\cZ_{\tau_\l}} \token_{\tau_\l}(B)-\token_{\tau_\l+1}(\shrink_\l(B)) &\leq& \frac{3C_3}{C_4} \sum_{B\in\cZ_{\tau_\l}} \token_{\tau_\l+1}(B)-\token_{\tau_\l}(B)\\
&=& \frac{3C_3}{C_4} \excess\leq \frac{b\cdot \br}{2},
\end{eqnarray*}
as desired. The last inequality follows by the lemma's assumption.

It remains to prove \eqref{eq:changeoftoken}.
First, note that if $\token_{\tau_\l}(B)=0$ then the above holds trivially. So assume $\token_{\tau_\l}(B)>0$. 
We consider three cases. 
i) $B\in\fbag_\l$. In this case both sides of the above inequality is zero. This is because $\shrink_\l(B)=B$ and $\token_{\tau_\l}(B)=\token_{\tau_\l+1}(B)$.
ii) $B$ is a non\avoiding~hollowed ball. Since $\token_{\tau_\l}(B)>0$, $r_2-r_1 > C_4\cdot \rho_\l$. Therefore, 
\begin{eqnarray*}
\token_{\tau_\l+1}(B)-\token_{\tau_\l}(B) &=& C_4\cdot (\rho_\l-\rho_{\l+1}) \geq \frac23\cdot C_4\cdot \rho_l\\
\token_{\tau_\l}(B) - \token_{\tau_\l+1}(\shrink_\l(B)) &\leq & 2(C_3+1)\rho_\l + C_4\cdot (\rho_{\l+1}-\rho_\l) \leq 2C_3\cdot \rho_\l.
\end{eqnarray*}
using $\rho_{\l+1}\leq \rho_\l/3$ and $C_4\geq 3$. So, \eqref{eq:changeoftoken} is correct.
iii) $B$ is an \avoiding~hollowed ball. 
Equation \eqref{eq:changeoftoken} is equivalent
to case (ii) up to a $2(2+C_4)$ factor in both sides of the inequality.
\end{proof}
%by the assumption $c< C_4\cdot b\cdot\rho_\l/6C_3$ and \autoref{lem:changeoftoken}, 
%$$ b\cdot\rho_\l \geq \frac{6C_3}{C_4}\cdot c \geq 2\sum_{B\in\cZ_{\tau_\l}} (\token_\l(B) - \token_{\l+1}(\shrink_\l(B)))=:2c'$$
For any ball $B\in \border_{t,\infty}$ and any ball $B'\in\cZ_{\tau_\l}$, we send 
$$\frac{\rho_\l}{2} \cdot \frac{\token_{\tau_\l}(B') - \token_{\tau_\l+1}(\shrink_\l(B'))}{\sum_{B''\in\cZ_{\tau_\l}} \token_{\tau_\l}(B'') - \token_{\tau_\l+1}(\shrink_\l(B''))}$$ 
tokens to $B'$ and we send the remaining tokens to $t$. For any ball $B\in\cZ_{\tau_\l}$, also send all of the tokens of $\shrink_\l(B)$ to $B$.

Therefore, by \autoref{lem:changeoftoken}, any ball $B\in\cZ_{\tau_\l}$ receives at least
\begin{align*} 
\token_{\tau_\l+1}(\shrink_\l(B)) &+ b\cdot \frac{\rho_\l}{2}\cdot \frac{\token_{\tau_\l}(B) - \token_{\tau_\l+1}(\shrink_\l(B))}{\sum_{B'\in\cZ_{\tau_\l}} \token_{\tau_\l}(B') - \token_{\tau_\l+1}(\shrink_\l(B'))} \\
&
\geq \token_{\tau_\l+1}(\shrink_\l(B)) + (\token_{\tau_\l}(B) - \token_{\tau_\l+1}(\shrink_\l(B))) \\
&= \token_{\tau_\l}(B),
\end{align*}
that we redistribute by the induction hypothesis.
On the other hand, any $t\in T_\l$ receives
$$ |\border_{t,\infty}| \cdot (\rho_\l-\rho_\l/2-C_4\cdot\rho_{\l+1}) \geq |\border_{t,\infty}|\cdot \rho_\l/4 $$
new tokens, and we are done with the induction.
This completes the proof of \autoref{prop:assignedbagreduced}.

\subsubsection*{Acknowledgment.}
We would like to thank James R Lee, Satish Rao, Amin Saberi, Luca Trevisan, and Umesh Vazirani for enlightening and stimulating discussions.
We also would like to specially thank Lap Chi Lau for proofreading and extensive comments on earlier versions of this manuscript. 
%\newpage 
%\printbibliography
\bibliographystyle{alpha}
\bibliography{references}
%\bibliography{ref2.bib}

\appendix
\section{Missing proofs of \autoref{sec:introduction}}
\label{app:missingintroduction}

\atspthinnessthm*
\begin{proof} %of}{\autoref{thm:agmos}}
For a feasible vector $x$ of LP \eqref{lp:tsp},
let $c(x)=\sum_{u,v} c(u,v)\cdot x_{u,v}$.
For two disjoint sets $A,B$ and a set of arcs $T$ let 
$$\Vec{T}(A,B):=\{(u,v): u\in A, v\in B\},$$
be the set of arcs from $A$ to $B$. 
We use the following theorem that is proved in \cite{AGMOS10}.
\begin{theorem}
\label{thm:agmos_version}
For a feasible solution $x$ of LP \eqref{lp:tsp} and a spanning tree $T$ such that
for any $S\subseteq V$,
\begin{equation}\label{eq:directedthinness} |\Vec{T}(S,\overline{S})|- |\Vec{T}(\overline{S},S)| \leq \alpha\cdot \sum_{u\notin S,v\in S} x_{u,v}+x_{v,u}=: \alpha\cdot x(S,\overline{S}),	
\end{equation}
and $\sum_{(u,v)\in T} c(u,v)\leq \beta\cdot c(x)$, there is a polynomial time algorithm that finds a tour of length $O(\alpha+\beta)\cdot c(x)$.
\end{theorem}
%Suppose for some $C\geq 1$, any $C\cdot \log(n)$-edge-connected graph has an $\alpha/\log(n)$-thin tree.
Given a feasible solution $x$ of LP \eqref{lp:tsp},
for a constant $C\geq 4$, we sample $C k\cdot n$ arcs where the probability of choosing each arc $(u,v)$ is proportional to $x_{u,v}$.
We drop the direction of the arcs
and we call the sampled graph $G=(V,E)$.
Since  $x(S,\overline{S}) \geq 2$ for all $S\subseteq V$, and $k\geq \log n$, it follows by the seminal work of Karger~\cite{Kar99}  that for a sufficiently large $C$, with high probability, for any $S\subseteq V$,
$|E(S,\overline{S})|$ is between 1/2 and 2 times $Ck \cdot x(S,\overline{S}).$
Since this happens with high probability, by Markov's 
inequality we can also assume that 
$$c(E) \leq 2C\cdot k \cdot c(x),$$
where for a set $F\subseteq E$, 
$$c(F):=\sum_{\{u,v\}\in F}
\min\{c(u,v),c(v,u)\}.$$

Since $x(S,\overline{S}) \geq 2$ and $C\geq 4$, $G$ is $2k$-edge-connected. Let $\beta=\alpha/k$. By the assumption of the theorem, 
$G$ has a $\beta$-thin tree, say $T_1$. Because of the thinness of $T_1$, $G(V,E\setminus T_1)$ is $2k (1-\beta) \geq k$-edge-connected. 
Therefore, it also has a $\beta$-thin tree.
By repeating this argument, we can find
 $j=\frac{1}{2\beta}$ edge-disjoint $\beta$-thin spanning trees in $G$, $T_1,\dots,T_j$.
 
Without loss of generality, assume that $c(T_1)=\min_{1\leq i\leq j} c(T_i)$.
We show that $T_1$ satisfies the conditions of the above theorem. 
First, since $c(T_1)=\min_{1\leq i\leq j} c(T_i)$,
$$ c(T_1) \leq \frac{c(E)}{j} \leq \frac{2C\cdot  k\cdot c(x)}{j} = 4C\cdot \alpha \cdot c(x). $$
On the other hand, since $T_1$ is $\beta$-thin with respect to $G$, for any set $S\subseteq V$,
$$ |T_1(S,\overline{S})| \leq \beta \cdot |E(S,\overline{S})| \leq 2C\cdot k\cdot \beta\cdot x(S,\overline{S}) = 2C\cdot \alpha\cdot x(S,\overline{S}). $$
Therefore, the theorem follows from an application of \autoref{thm:agmos_version}.
\end{proof} %of}

\end{document}